\documentclass[titlesmallcaps,copyrightpage]{uqthesis}

\input{Qcircuit}

\usepackage[hypertex]{hyperref}
\usepackage{listings}
\usepackage{subfigure}
\usepackage{amsthm}
\usepackage{gensymb}
\usepackage{mathrsfs}
\usepackage[usenames]{color}
\DeclareMathAlphabet{\mathpzc}{OT1}{pzc}{m}{it}
\newtheorem{theorem}{Theorem}[chapter]

\newtheorem{definition}{Definition}[chapter]

\newtheorem{proposition}{Proposition}
\newtheorem{corollary}{Corollary}[theorem]
\newtheorem{lemma}{Lemma}[chapter]
\newtheorem{conjecture}{Conjecture}[section]

\newtheorem{problem}{Problem}
\newenvironment{principle}{\vspace{0.3cm}\begin{it}}{\end{it}\vspace{0.3cm}}

\usepackage[sans]{dsfont}
\def\openone{\mathds{1}}
\def\defeq{\mathrel{\mathop:}=}
\def\defeqinv{\mathrel=\!\!{\mathop:}}
\def\brho{\overline{\rho}}

\usepackage{bm}
\usepackage{rotating} 
\usepackage{color}
\definecolor{gray}{rgb}{0.35,0.35,0.35}
\definecolor{graytable}{rgb}{0.45,0.45,0.45}

\DeclareMathAlphabet\EuScript{U}{eus}{m}{n}
\SetMathAlphabet\EuScript{bold}{U}{eus}{b}{n}

\newcommand{\Y}{\EuScript{Y}}
\newcommand{\X}{\EuScript{X}}
\newcommand{\A}{\EuScript{A}}
\newcommand{\F}{\EuScript{F}}
\newcommand{\Fn}{\EuScript{F}_N}
\newcommand{\D}{\EuScript{D}}
\renewcommand{\O}{\EuScript{O}}
\renewcommand{\H}{\EuScript{H}}
\DeclareMathOperator{\tr}{Tr}

\newcommand {\etal} {\emph{et~al.~}}

\def\braa#1{\mathinner{\langle{#1}|}}
\def\kett#1{\mathinner{|{#1}\rangle}}
\def\aver#1{\mathinner{\langle{#1}\rangle}}

\usepackage{play}
\usepackage[grey,times]{quotchap}
\usepackage{makeidx}

\usepackage[square,comma,numbers,sort&compress]{natbib}

\usepackage[nottoc]{tocbibind}

\usepackage{amsmath,amsfonts,amssymb} 

\usepackage{graphicx} 

\usepackage{epsfig} 

\usepackage{array}
\usepackage{colortbl}
\usepackage{arydshln}
\definecolor{umbra}{rgb}{0.8,0.8,0.5}
\definecolor{gray}{rgb}{0.85,0.85,0.85}
\def\colCell#1#2{\multicolumn{1}{>{\columncolor{#1}}c}{#2}}

\newcommand\TT{\rule{0pt}{3.4ex}}
\newcommand\BB{\rule[-2ex]{0pt}{0pt}}

\usepackage{nccmath}

\newif\ifpdf

\ifpdf
	\DeclareGraphicsExtensions{.pdf}  
\else
	\DeclareGraphicsExtensions{.ps}
\fi

\begin{document}

\frontmatter

\title{Optimal Control of Finite Dimensional Quantum Systems}
\author{Paulo Eduardo Marques Furtado de Mendon\c{c}a}
\department{Department of Physics}  
\principaladvisor{Dr. Andrew C. Doherty}

\titlepage

\chapter{List of Publications}

\textbf{Publications by the Author Incorporated into the Thesis}

\begin{itemize}
    \item Paulo E. M. F. Mendon\c ca, Reginaldo~d.~J.~Napolitano, Marcelo A. Marchiolli, Christopher~J.~Foster and Yeong-Cherng~Liang, {\em Alternative fidelity measure between quantum states}. Physical Review A {\bf 78}, 052330 (2008).

        Incorporated as Section~\ref{sec:altfid}.

    \item Agata M. Bra\'nczyk, Paulo E. M. F. Mendon\c ca, Alexei Gilchrist, Andrew C. Doherty and Stephen D. Bartlett, {\em Quantum control of a single qubit}. Physical Review A {\bf 75}, 012329 (2007).

        Incorporated as Chapter~\ref{chap:aggie}.

    \item Paulo E. M. F. Mendon\c ca, Alexei Gilchrist and Andrew C. Doherty, {\em Optimal tracking for pairs of qubit states}.  Physical Review A {\bf 78}, 012319 (2008).

        Incorporated as Chapter~\ref{chap:tracking}.
\end{itemize}

\chapter{Acknowledgements}

This thesis is the result of four years of hard --- but not lonely --- work.  Here, I would like to express my gratitude to all of those who lent a hand and helped me to come this far.

Following the chronology of the facts, I would first like to thank my associate supervisor Dr. Stephen Bartlett. Steve enabled my PhD by accepting my application, proposing a really interesting research project, going through a lot of paper work and closely supervising me in the initial stage of the program. Unfortunately, Steve moved to Sydney after a while; fortunately, he was substituted by two other great supervisors: Dr. Andrew Doherty and Dr. Alexei Gilchrist. Eventually, Alexei also moved to Sydney, but even from distance he managed to follow my research progress. I am sincerely grateful for his encouragement, useful suggestions and the major role he played in transforming my poorly written manuscripts into readable material.

I owe my principal advisor, Dr. Andrew Doherty, a HUGE thank you! To start with, he has never moved to Sydney. Besides, Andrew marked his presence in Brisbane with endless patience, knowledge and willingness to help. Despite his busy schedule, he closely followed my research and made major contributions to it. It was from him that I learnt one of the most unmistakable lessons of my PhD: when Andrew speculates that something ``has got to be true'' and you prove it is not, you better check your proof.

I would like to thank Dr. Kurt Jacobs and Dr. Andrew Scott for kindly accepting to be on my reading committee and for providing much valuable comments and feedback on my work.

Apart from collaboration with my supervisors, I have also appreciated the opportunity of engaging joint work with Agata M. Bra\'nczyck, Dr. Reginaldo d. J. Napolitano, Dr. Marcelo A. Marchiolli, Dr. Yeong-Cherng Liang and Chris Foster. I am indebted to Dr. Jens Eisert, Dr. Navin Khaneja and Dr. John Gough for the hospitality and opportunity to expand my horizons in visits to their research groups. Thanks also are due to Dr. Luc Bouten, who generously missed most of the coffee breaks of QPIC/2006 and PRACQSYS/2006 to introduce me to the basics of quantum filtering theory. The work presented in this thesis has also benefited from helpful discussions and comments by Marco Barbieri, Joshua Combes, Robin Blume-Kohout, Howard Wiseman, Armin Uhlmann and Karol \.Zyckzkowsky.

Even when around my fellow PhD students, I was privileged to be surrounded by brilliant and kind people. It was a pleasure to share an office and/or time with Joshua Combes, Sukhwinder Singh, Mark de Burgh, Mark Dowling, Eric Cavalcanti, Andy Ferris, Geoff Lee, Andrew Sykes, Terry McRae, Chris Foster and Yeong-Cherng Liang. Special mention is due to Chris and Yeong-Cherng whose contribution towards my PhD do not restrict to the work of Ref.~\cite{08Mendonca1150}. Chris, for example, has not only alleviated my many computational, mathematical and \emph{communicational} shortcomings, but has also introduced and guided me through the terribly addictive practice of juggling (with juggling balls provided!). It would be fair to say that Yeong-Cherng was my fourth supervisor (which explains why he has also moved to Sydney...). I learnt a lot from our many discussions, both the scientific and the non-scientific ones. I am also indebted to him for his generous revision of this thesis.

Terry McRae took Yeong-Cheng's desk when he left, and soon became another dear friend. His random sense of humor and constant presence in the office from early to late hours made the preparation of this thesis a rather amusing time. Thanks are also due to Terry for his kindness in helping me organizing my thoughts and my writing in certain critical moments. Next, I would like to express my gratitude for the friendship of Sukhwinder Singh --- my favorite Indian cook and spiritual mentor (in the lack of \emph{any} other Indian around). We had great curries together, during which I could contemplate his personal way of facing science and everything else. Thanks also to Marcelo Marchiolli for our online conversations, scientific collaboration, and advice on the most diverse circumstances.

Good friends were also outside the Physics department. I could not forget to mention the compatriot families of Wander Barbosa and Gl\'aucia, Jo\~ao Marinho and Cida, Zorano de Souza and Marieta, Paulo Schneider and Bete. We all had many memorable get-togethers, usually accompanied by \emph{Brazilian-style} barbecue. I am specially grateful to Marinho, Cida, Wander and Gl{\'a}ucia for ``taking good care of me'' when I was on my own. Their support was essential in the final stage of this project.

I am deeply thankful to my beloved mother Sonia and sisters Beatriz and Raquel, who managed to make me feel close even from $16,000$ Km away. Their encouragement and enthusiasm have always been crucial in making me move forward. Thank you so much for going through great lengths and literally traveling halfway around the world to visit me in $2007$.

Any attempt to express how grateful I am to my wife Suely and children Beatriz and Laura is almost certainly doomed to failure. Nevertheless, it is worth a try: I am lucky to have the most comprehensive, supportive and loving family ever seen. At the most difficult times, it was reassuring to know that somewhere in town there was this little place where three adorable girls lived, and no matter how dark everything else could get, they would still be there with arms wide open waiting for me to come back. The love and support of my family was certainly \emph{the} crucial ingredient to the completion of this project. For this, it is to them that I dedicate this thesis.

Finally, I thank the Brazilian agency Coordena\c c\~ao de Aperfei\c coamento de Pessoal de N\' ivel Superior (CAPES) for the financial support without which this project would not have been initiated. A special thank you to Vanda Lucena, whose administrative support and kindness saved a lot of my time, allowing me to focus more on physics and less on paper work.

\chapter{Abstract}

This thesis addresses the problem of developing a quantum counter-part of the well established classical theory of control. We dwell on the fundamental fact that quantum states are generally not perfectly distinguishable, and quantum measurements typically introduce noise in the system being measured. Because of these, it is generally not clear whether the central concept of the classical control theory --- that of observing the system and then applying feedback --- is always useful in the quantum setting.

We center our investigations around the problem of transforming the state of a quantum system into a given target state, when the system can be prepared in different ways, and the target state depends on the choice of preparation. We call this the \emph{quantum tracking problem} and show how it can be formulated as an optimization problem that can be approached both numerically and analytically. This problem provides a simple route to the characterization of the quantum trade-off between information gain and disturbance, and is seen to have several applications in quantum information.

In order to characterize the optimality of our tracking procedures, some figure-of-merit has to be specified. Naturally, distance measures for quantum states are the ideal candidates for this purpose. We investigated several possibilities, and found that there is usually a compromise between physically motivated and mathematically tractable measures. We also introduce an alternative to the Uhlmann-Jozsa fidelity for mixed quantum states, which besides reproducing a number of properties of the standard fidelity, is especially attractive because it is simpler to compute.

We employ some ideas of convex analysis to construct optimal control schemes analytically. In particular, we obtain analytic forms of optimal controllers for stabilizing and tracking any pair of states of a single-qubit. In the case of stabilization, we find that feedback control is always useful, but because of the trade-off between information gain and disturbance, somewhat different from the type of feedback performed in classical systems. In the case of tracking, we find that feedback is not always useful, meaning that depending on the choice of states one wants to achieve, it may be better not to introduce any noise by the application of quantum measurements. We also demonstrate that our optimal controllers are immediately applicable in several quantum information applications such as state-dependent cloning, purification, stabilization, and discrimination. In all of these cases, we were able to recover and extend previously known optimal strategies and performances.

Finally we show how optimal single-step control schemes can be concatenated to provide multi-step strategies that usually over-perform optimal control protocols based on a single interaction between the controller and the system.

\newpage
\section*{Keywords}

quantum control, state transformation, quantum channels, distance measures, semidefinite programming.

\section*{ANZSRC (FOR) Classification}

\begin{tabular}{ll}
020603 & Quantum Information, Computation and Communication (50\%);\\
010503 & Mathematical Aspects of Classical Mechanics, Quantum Mechanics and \\
& Quantum Information Theory (50\%).
\end{tabular}

\setcounter{secnumdepth}{3}
\setcounter{tocdepth}{2}

\tableofcontents
\listoffigures
\listoftables

\chapter{List of Abbreviations}

\begin{list}{}{
\setlength{\labelwidth}{24mm}
\setlength{\leftmargin}{35mm}}
\item[lhs\dotfill] left-hand-side
\item[rhs\dotfill] right-hand-side
\item[CP\dotfill] Completely Positive
\item[CPTP\dotfill] Completely Positive and Trace Preserving
\item[EBTP\dotfill] Entanglement Breaking and Trace Preserving
\item[EBTD\dotfill] Entanglement Breaking and Trace Decreasing
\item[HS\dotfill] Hilbert-Schmidt
\item[N/A\dotfill] Not Applicable
\item[POVM\dotfill] Positive Operator Valued Measure
\item[PSD\dotfill] Positive Semidefinite
\item[PPT\dotfill] Positive Partial Transpose
\item[QND\dotfill] Quantum Non Demolition
\item[SDP\dotfill] Semidefinite Program
\item[SVD\dotfill] Singular Value Decomposition
\end{list}

\mainmatter

\chapter{Introduction}\label{chap:intro}

Broadly speaking, to control a physical system is to modify its natural evolution towards some preferred dynamics. An obvious example is steering a car along a sinuous highway. An \emph{attentive} driver \emph{rotates} the steering wheel in order to \emph{preserve} the car on the highway --- obviously preferable than the natural straight line trajectory. Even in this simple example, three pillars of Control Theory can already be recognized:
\begin{enumerate}
    \item Characterization of the possible ways of influencing the dynamics of the system (in the example, a rotation of the steering wheel);
    \item A clear description of the control goal (in the example, to preserve the car on the highway);
    \item A method to determine a control action that, if does not precisely achieve it, at least approximates the goal (in the example this is implemented in the driver's brain, which induces a suitable movement of the arms based on the observation of the highway).
\end{enumerate}

Feedback control, a central concept of this theory, appears in the third point of our example: the direction taken by the car is conditioned on the \emph{observation} of the highway. In fact, a driver who is asked to close their eyes and change lanes usually drives the car off the road at an angle.
This very same idea of conditioning an action on the outcome of a measurement has found many technological applications in aircraft flight control, fabrication of fiber optics cables, robotics and many others.

Behind all these triumphs of classical feedback is the fact that --- at least in principle --- there is no cost associated to the extraction of information from classical systems. In contrast, quantum measurements cannot perfectly distinguish all states and necessarily disturb the measured system. In practice, the information gained from a quantum measurement may not compensate for the disturbance caused by it. As a result, the characterization of an unknown quantum system is fundamentally more difficult, rendering the role of feedback in the quantum domain disputable.

In this thesis we study how measurements on finite dimensional quantum systems can be designed in order to optimize this intrinsic trade-off between information gain and disturbance, aiming at the design of optimal feedback schemes for \emph{quantum} control. As demonstrated by Fuchs and Peres \cite{96Fuchs2038}, there exists an entire range of generalized measurements that trade-off information gain and disturbance, and as recognized by Doherty, Jacobs and Fuchs \cite{01Doherty062306,01Fuchs062305}, it is the basic problem of quantum control to choose one such that feedback offers more help than hinderance.

The theory of quantum feedback control started to be developed in the 70's with the work of Belavkin in the mathematical physics literature \cite{83Belavkin178,79Belavkin3}. However, due to the experimental limitations of that time, the field developed almost exclusively along a mathematically abstract direction, until reappearing in the 90's in the quantum optics literature \cite{94Wiseman2133,93Wiseman548}, encountering much richer experimental possibilities and practical problems where the theory could be successfully applied \cite{02Armen133602,06Bushev043003,04Geremia270,04LaHaye74,04Reiner023819,02Smith133601}.

Here, we introduce a toy feedback control problem that brings new perspectives into the study of the trade-off between information gain and disturbance. We consider the \emph{quantum tracking} task of transforming the state of a quantum system into a given target state, when the system can be prepared in different ways, and the target state depends on the choice of preparation. It is not difficult to see that tracking can be formulated as a transformation between two sequences of density matrices, as illustrated in Fig.~\ref{fig:singlestep_intro}. The source sequence models the uncertainty of the initial preparation; each density matrix has some probability of being the actual state of the system. The target sequence, in turn, is formed by those states that we would like to output for each initial preparation. Ideally, we would like to be able to perfectly map one sequence into the other, no matter what the initial preparation was. Since this is usually impossible, we look for optimal approximations of these sequences. A choice of measurement and feedback that minimizes some notion of distance between these sequences is optimizing some trade-off between information gain and disturbance \cite{96Fuchs2038,01Doherty062306,01Fuchs062305}.
\begin{figure}[h]
\centering
\includegraphics[width=6cm]{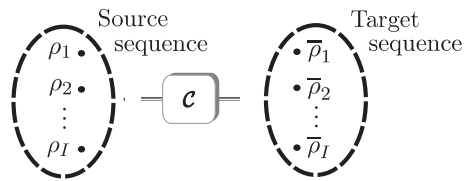}
\caption[Schematic of the basic tracking problem approached in the thesis.]{Schematic of the basic tracking problem approached in the thesis.}\label{fig:singlestep_intro}
\end{figure}

The approach adopted for the discussion of the tracking problem is based on the three pillars of the control theory enumerated before: (1) the theory of \emph{Quantum Operations} (or \emph{Quantum Channels}) provides a mathematically rigorous description of which operations are and are not possible from a physical perspective. (2) The task of enforcing a particular dynamic is formalized via an optimization problem (over the set of quantum channels) that attempts to minimize some \emph{distance measure} between the output states and the target states. (3) Depending on certain technical details, the resulting mathematical problem can be cast as, or approximated by a  \emph{Semidefinite Program} (SDP) \cite{96Vandenberghe49} --- a special type of convex optimization problem that, computationally, can be efficiently solved in polynomial time. The theories of quantum operations, distance measures and convex optimization form the key background material underlying the formulation and solution of the quantum control problems tackled in this thesis.

Regarding applications, the problem of approximating sequences of density matrices turns out to be sufficiently general to encompass a number of interesting problems in quantum information science. Its optimal solution provides optimal schemes for tasks such as state-dependent quantum cloning, quantum error correction, quantum state purification and quantum state discrimination. Besides, it illustrates in a clear-cut way some key departures between the optimal quantum and classical control theories.

In the remainder of this chapter we outline the organization of the thesis and the new contributions made to the field. Loosely, it can be divided in two parts: In Chapters~\ref{chap:background} and~\ref{chap:dist_measure}, we review some background material useful for the formulation of the single-step tracking problem as an optimization problem. From Chapters~\ref{chap:assemble} to \ref{chap:multistep}, we focus on solving this problem and a multi-step variation. A more detailed description of the content of each chapter is as follows.

In Chapter~\ref{chap:background} we review the relevant concepts of quantum operations and semidefinite programming. Although most of this chapter is devoted to well established results of these two fields, some new and more specific material is also introduced here. For example, some of the adopted notation is defined, non-standard representations for quantum operations are derived and, most importantly, the first steps towards the formulation of the single-step tracking problem as a SDP are given in the last section.

Chapter~\ref{chap:dist_measure} merges a review on standard distance measures for quantum states with the introduction of a new measure identified in the course of the present work~\cite{08Mendonca1150}. After a thorough evaluation of its properties, the latter is argued to be an easy-to-compute alternative definition of the Uhlmann-Jozsa fidelity between mixed states. In our quantum control framework, distance measures between quantum states characterize the objective function to be optimized by the controller. In this sense, the results of this chapter are relevant for the identification of figures-of-merit that are both physically motivated and mathematically tractable when used in the formulation of optimization problems for quantum control.

Chapter~\ref{chap:assemble} brings together the results of the two preceding chapters by assembling a number of optimization problems, each of which is associated with one of the distance measures from Chapter~\ref{chap:dist_measure}. The realizability of each problem as an SDP is discussed and their solutions are compared in an attempt to measure the sensitivity of optimal quantum control to different choices of figures-of-merit.

In Chapter~\ref{chap:aggie} we present our first analytical results on the solution of a specific type of the single-step tracking. The source and target sequences are restricted to contain two pure states of a single qubit, in such a way that the tracking task reduces to the stabilization of an uncertain preparation of a single qubit against dephasing noise. Exploiting the symmetry of the states involved and a particular representation of the dephasing map, quantum and ``classical'' channels are analytically constructed and proved to be optimal for the stabilization task. Most results of this chapter were presented in Ref.~\cite{07Branczyk012329}.

In Chapter~\ref{chap:tracking} the single-step tracking problem for pairs of qubits is approached in its full generality and, nevertheless, analytically solved. Several applications of this solution are provided in the grounds of quantum information science. In particular, reproduction and extension of optimal schemes for quantum state discrimination, state purification, state stabilization and state-dependent quantum cloning are obtained. Except for some minor adjustments, this chapter reproduces Ref. \cite{08Mendonca012319}.

Chapter~\ref{chap:multistep} presents some preliminary results in the formulation and solution of a multi-step control variant of the tracking problem. Here, we face the problem of designing multiple control interactions that attempt to track the (uncertain) input density matrices to some desired target. This problem is sufficiently difficult not to be approached from an analytical perspective and we describe a heuristic that, based on the application of the optimal single-step tracking solution, usually leads to multi-step control strategies that over-perform the optimal single-step control scheme.

Chapter~\ref{chap:conclusion} summarizes our main results and presents some possibilities of future research. 

\chapter{Quantum Operations, Semidefinite Programs and Quantum Control}\label{chap:background}
\section{Introduction}
This chapter covers some background material on quantum operations and semidefinite programming. For a long while, these topics were developed as independent research fields of quantum mechanics and optimization theory, respectively. Recently, Audernaert and de Moor~\cite{02Audenaert030302} recognized that the problem of determining optimal quantum operations for a given task can \emph{sometimes} be cast as a semidefinite program. Ever since, connections between these two topics have been further developed and explored in quantum information science. This thesis is one result of this symbiosis.\par

The selection of topics reviewed in this chapter aims to achieve the competing purposes of making the thesis as objective and self-contained as possible: only the essential concepts and a few technical tools required for the remaining chapters are presented. Along these lines, many important theorems are stated without complete proofs, in which cases references are given. In fact, derivations of mathematical results are only presented when they are, on their own, a revision of a useful concept/tool.  In contrast, a rather complete account of the theory of convex optimization and semidefinite programming can be found in the book by Boyd and Varderberghe~\cite{04Boyd} and in the review article by the same authors~\cite{96Vandenberghe49}. Modern reviews on the theory of quantum operations can be found in the Ch.~8 of Ref.~\cite{00Nielsen}, Ch.~10 and~11 of Ref.~\cite{06Bengtsson} and Ch.~5 of Ref.~\cite{06Hayashi}.

This chapter is divided as follows. In Secs.~\ref{sec:QO} and~\ref{sec:SDP} we review, respectively, the key concepts of quantum operations and semidefinite programming. Sec.~\ref{sec:controlSDP} merges the results of the previous sections to show that the optimization problems over the set of quantum operations typically have a ``Semidefinite Programming flavor'', which may or may not be confirmed depending on the specific details of the control task of interest.

\section{Quantum operations}\label{sec:QO}

Is there a well defined mathematical model for every possible dynamics a quantum system can incur? When the
system is \emph{closed}, then the Schr{\"o}dinger equation provides a widely accepted model for quantum evolution, namely, any unitary conjugation of the density operator. The question becomes more subtle for \emph{open} quantum systems. This section gives a short account on the dynamical model we adopt in this case.\par

In Sec.~\ref{sec:CP} we review the concept of \emph{completely positive} maps from a mathematical viewpoint. These maps will be shown, in Sec. \ref{sec:QOasCP}, to provide an adequate model for describing open quantum dynamics. In Sec. \ref{sec:CPuseful} we present a number of technical results related to the set of CP maps and some subsets of interest.\par

\subsection{Completely Positive Maps}\label{sec:CP}
In the most general framework, Completely Positive (CP) maps are certain types of transformations defined between abstract $C^\ast$-algebras \cite{55Stinespring211}. For our purposes, it will suffice to restrict to endomorphisms of the finite dimensional algebra $\mathcal{M}_{\rm d}$ of ${\rm d} \times {\rm d}$ complex matrices\footnote{This is a $C^\ast$-algebra of operators on the Hilbert space ${\sf H}_{\rm d}\cong\mathbb{C}^{\rm d}$ if each matrix $A\in\mathcal{M}_{\rm d}$ is seen as an operator on $\mathbb{C}^{\rm d}$. The algebra is equipped with the standard \emph{operator norm} $\|A\|$ (the largest singular value of the matrix $A$) and the \emph{$\ast$-involution} is taken as the conjugate-transpose, so that the required \emph{$C^\ast$-norm property} $\|A^\dagger A\|=\|A\|^2$ holds.}. In what follows, we define CP maps within this particular framework.


\begin{definition}\label{def:PSD}
A matrix $A \in \mathcal{M}_{\rm d}$ is \emph{positive semidefinite (PSD)}, denoted $A \geq 0$, if $A^\dagger=A$ and every eigenvalue of $A$ is non-negative.
\end{definition}

\begin{definition}\label{def:P}
A linear map $\mathcal{P} : \mathcal{M}_{\rm d} \to \mathcal{M}_{\rm d}$ is positive if $\mathcal{P}(A)\geq 0$ for
every $A \geq 0$.
\end{definition}

\begin{definition}\label{def:CP}
Let $\mathcal{K} : \mathcal{M}_{\rm d} \to \mathcal{M}_{\rm d}$ and define $\mathcal{K}_{\rm n} : \mathcal{M}_{\rm n}\otimes\mathcal{M}_{\rm d} \to \mathcal{M}_{\rm n}\otimes \mathcal{M}_{\rm d}$ by $\mathcal{K}_{\rm n}=\mathcal{I}_{\rm n} \otimes \mathcal{K}$, where $\mathcal{I}_{\rm n}$ is the identity map on the elements of $\mathcal{M}_{\rm n}$. Then $\mathcal{K}$ is called ${\rm n}$-positive if
$\mathcal{K}_{\rm n}$ is positive. A map that is ${\rm n}$-positive for all values of ${\rm n}$ is termed
\emph{completely positive}.
\end{definition}

The definitions above do not make it clear whether positive but not CP maps exist. In addition, definition \ref{def:CP} does not provide a viable way of \emph{confirming} a given linear map as CP: testing {\rm n}-positivity for increasing values of {\rm n} can only be conclusive if a violation is found at some stage.\par

These two points were first addressed by Stinespring \cite{55Stinespring211}, who provided examples of positive maps which failed to be $2$-positive and derived a new criterion for CP-ness, stated below as Theorem \ref{teo:stinespring}. Further developments came by with the work by Kraus \cite{71Kraus311}, where an alternative CP criterion was obtained, Theorem \ref{teo:kraus}. In 1975, Choi \cite{75Choi285} rediscovered Kraus earlier result with an independent proof that led to yet another new criterion, Theorem \ref{teo:choi}. In addition, in Ref.~\cite{72Choi520} Choi proved that for every ${\rm d} \geq 2$ there exists $({\rm d}-1)$-positive maps which are not {\rm d}-positive, but every {\rm d}-positive map is automatically CP. With this, Choi proved it possible to \emph{confirm} a map as CP with ${\rm d}$ ``${\rm n}$-positivity tests''.

\begin{theorem}[Stinespring \cite{55Stinespring211}]\label{teo:stinespring}
A linear map $\mathcal{K}:\mathcal{M}_{\rm d}\to\mathcal{M}_{\rm d}$ is CP if and only if there exists, for some dimension ${\rm d}^\prime$, a rectangular matrix $V$ of size ${\rm d} \times {\rm d d}^\prime $ 
such that $\mathcal{K}$  can be ``factorized'' as
\begin{equation}\label{eq:stinespring_rep}
\mathcal{K}(A)=V(A\otimes \openone_{\rm d^\prime})V^\dagger\qquad\forall A \in \mathcal{M}_{\rm d}\,.
\end{equation}
\end{theorem}
%
%

\begin{proof}
That any map of the form~\eqref{eq:stinespring_rep} is CP is immediate. For a proof of the converse, we refer the reader to Ref.~\cite[pp. 357-361]{06Hayashi}.
\end{proof}
\begin{theorem}[Kraus \cite{71Kraus311}, Choi \cite{75Choi285}]\label{teo:kraus}
A linear map $\mathcal{K}:\mathcal{M}_{\rm d}\to\mathcal{M}_{\rm d}$ is CP if and only if there exists $K_1,\ldots,K_k \in \mathcal{M}_{\rm d}$ such that $\mathcal{K}$ admits a decomposition of the form
\begin{equation}\label{eq:kraus_form}
\mathcal{K}(A)=\sum_{i=1}^k K_i A K_i^\dagger\qquad\forall A \in \mathcal{M}_{\rm d}\,.
\end{equation}
\end{theorem}
\begin{proof} That any map of the form (\ref{eq:kraus_form}) is CP is immediate. To see that every CP map can be written in this form, one can simply define the operators $K_i$ such that
$V^\dagger \bm{r}=\sum_i (K_i^\dagger \bm{r})\otimes \bm{e}_i$, where $V$ is the ${\rm d} \times {\rm d d}^\prime$ matrix from Eq.~\eqref{eq:stinespring_rep}, $\bm{r}$ is an arbitrary vector of size ${\rm d}\times 1$ and $\{\bm{e}_i\}$ is an orthonormal basis of vectors of size ${\rm d}^\prime\times 1$. It is then straightforward to show that Eq.~\eqref{eq:stinespring_rep} reduces to Eq.~\eqref{eq:kraus_form} (see, e.g., Ref. \cite[p. 20]{02Raginsky} for details). Choi's proof of this result follows a different direction and has been nicely revisited by Leung in Ref.~\cite{03Leung528}.
\end{proof}

\begin{theorem}[Choi \cite{75Choi285}]\label{teo:choi}
Let $\mathcal{L}(\mathcal{M}_{\rm d},\mathcal{M}_{\rm d})$ denote the set of linear maps from $\mathcal{M}_{\rm d}$ to $\mathcal{M}_{\rm d}$ and consider the following map from $\mathcal{L}(\mathcal{M}_{\rm d},\mathcal{M}_{\rm d})$ to $\mathcal{M}_{\rm d}\otimes\mathcal{M}_{\rm d}$:
\begin{equation}\label{eq:jamisomorphism}
\mathcal{K} \mapsto \mathfrak{K} = (\mathcal{I}_{\rm d}\otimes\mathcal{K})\left(\sum_{i,j=1}^{\rm d}{\mathbb{E}_{i,j}\otimes \mathbb{E}_{i,j}}\right)\,,
\end{equation}
where each $\mathbb{E}_{i,j}$ is a ${\rm d}\times {\rm d}$ matrix of elements $1$ in the $(i,j)$ position and $0$ elsewhere\footnote{Clearly, $\{\mathbb{E}_{i,j}\}_{i,j=1}^{\rm d}$ forms a basis for $\mathcal{M}_{\rm d}$, the so-called Weyl basis. In physics notation, $\mathbb{E}_{i,j}=\ket{i}\bra{j}$, where $\{\ket{i}\}_{i=1}^{\rm d}$ is an orthonormal basis for the Hilbert space ${\sf H}_{\rm d}\cong\mathbb{C}^{\rm d}$ onto which the matrices of $\mathcal{M}_{\rm d}$ act as the linear (bounded) operators of $\mathcal{B}(\textsf{H}_{\rm d})$. The (unnormalized) density matrix $\ket{\Psi}\!\bra{\Psi}$ defined in this basis via $\ket{\Psi}=\sum_{i=1}^{\rm d}\ket{i}\otimes\ket{i}$ is thus identical to $\sum_{i,j=1}^{\rm d}{\mathbb{E}_{i,j}\otimes \mathbb{E}_{i,j}}$ and Eq.~\eqref{eq:jamisomorphism} is typically written as $\mathfrak{K} = (\mathcal{I}_{\rm d}\otimes\mathcal{K})\ket{\Psi}\!\bra{\Psi}$.\label{foot:choi}}.\par
A linear map $\mathcal{K}:\mathcal{M}_{\rm d}\to\mathcal{M}_{\rm d}$ is CP if and only if $\mathfrak{K} \geq 0$.
\end{theorem}
\begin{proof}
That $\mathfrak{K}\geq 0$ if $\mathcal{K}$ is CP follows directly from the observation that $\sum_{i,j=1}^{\rm d}{\mathbb{E}_{i,j}\otimes \mathbb{E}_{i,j}}$ is a PSD matrix (see footnote \ref{foot:choi}).
 The converse, is typically proved by relying upon theorem \ref{teo:kraus}; see, for example, Ref.~\cite{05Salgado55}.
\end{proof}

In this thesis, the theorems above will not be used for testing the CP-ness of a given map. Instead, Eqs.~\eqref{eq:stinespring_rep},~\eqref{eq:kraus_form} and~\eqref{eq:jamisomorphism} (with $\mathfrak{K}\geq 0$) are taken to be general representations of a CP map.

\subsection{Quantum operations as CP maps}\label{sec:QOasCP}

In this section, we argue that a certain class of CP maps yields a suitable mathematical model for general state changes in quantum theory --- an idea introduced long ago by Kraus \cite{71Kraus311}. In order to make this precise, we start with a brief clarification of what is meant when a quantum operation is said to be a CP map.

In the Schr\"odinger picture, a quantum operation can be thought of as a map on the set of density matrices of dimension {\rm d}, $\mathcal{S}(\textsf{H}_{\rm d})$. Since this is not a linear space (it is not closed under arbitrary linear combinations), a quantum operation is not a linear map, let alone a CP map~\cite{04Salgado054102}. Nevertheless, it is widely accepted (although also debatable, see e.g., Refs. \cite{07Jordan,06Jordan022101} and references therein) that any quantum operation $\mathcal{Q}$ satisfies a weaker form of linearity called \emph{convex-linearity}:
\begin{equation}\label{eq:cvxlnr}
\mathcal{Q}[\lambda\rho + (1-\lambda)\sigma]=\lambda\mathcal{Q}(\rho)+(1-\lambda)\mathcal{Q}(\sigma)\quad\mbox{for}\quad \lambda \in [0,1]\,,
\end{equation}
where $\rho,\sigma \in \mathcal{S}(\textsf{H}_{\rm d})$. In this thesis we adopt the common practice of saying that $\mathcal{Q}$ is a CP map if the map obtained from its linear extension to $\mathcal{B}(\textsf{H}_{\rm d})\cong\mathcal{M}_{\rm d}$ is CP.

Having clarified this subtlety, we now follow justifying why quantum operations are suitably described by certain types of CP maps, namely, those that preserve the trace of their input. These are called completely positive and trace preserving (CPTP) maps.


\subsubsection{Equivalence between quantum operations and CPTP maps}\label{sec:CPTP_qo}

It is constructive to proceed in three steps: first we discuss why the set of quantum operations should lie within the set of positive maps, then we motivate the restriction to the subset of CP maps and finally we justify the equivalence with the set of CPTP maps.

\paragraph{Quantum operations $\Rightarrow$ Positive maps.}
Once Born's rule is accepted as a fundamental axiom of quantum mechanics, quantum states must be modeled by PSD matrices in order to guarantee that they will always yield well-defined probabilities. As a result, any map converting from arbitrary quantum states to arbitrary quantum states must necessarily be positive.

\paragraph{Quantum operations $\Rightarrow$ CP maps.} Not every positive map can model the dynamics of an arbitrary input density matrix. This can be easily understood as follows: Start by assuming that every positive map $\mathcal{P}$ represents a quantum operation. From a physical viewpoint, this means that an experimentalist should (at least in principle) be able to implement such a map in the lab. Obviously, if $\mathcal{P}$ can be implemented, there is no reason why the map $\mathcal{I}_{\rm d}\otimes\mathcal{P}$ could not be built as well, since it corresponds to evolving a bipartite ${\rm d}\otimes {\rm d}$ system by applying $\mathcal{P}$ to the second part and leaving the first part alone. However, as we learnt from Stinespring~\cite{55Stinespring211}, positivity of $\mathcal{P}$ does not imply positivity of $\mathcal{I}_{\rm d}\otimes\mathcal{P}$, or in other words, there exist input states to $\mathcal{I}_{\rm d}\otimes\mathcal{P}$ that are not mapped to PSD operators\footnote{A standard example arises from ${\rm d}=2$, $\rho$ any bipartite ${\rm 2}\otimes{\rm 2}$ entangled density matrix and $\mathcal{P}$ the transposition map with respect to the basis of $\rho$~\cite{96Peres1413,96Horodecki1}.}. As discussed above, this is inconsistent with Born's rule, leading to the conclusion that not every positive map can model the dynamics of an arbitrary input state, as claimed before.

An easy solution is to restrict the set of quantum operations to the set of positive maps $\mathcal{P}$ for which $\mathcal{I}_{\rm d}\otimes\mathcal{P}$ is still a positive map. As demonstrated by Choi in Ref.~\cite{72Choi520}, this restriction characterizes the set of CP maps.

It is worth mentioning an alternative (and more comprehensive) solution. No conflict with Born's rule arises if we allow every positive map to be a quantum operation, but under the understanding that this is only physically meaningful with respect to some particular domains of density matrices. A further step is to let every linear operation mapping density matrices to density matrices to model quantum operations. These ideas are advocated and developed in Refs.~\cite{94Pechukas1060,95Alicki3020,95Pechukas3021,01Stelmachovic062106,03Stelmachovic029902,04Jordan052110,05Shaji48}.

\paragraph{Quantum operations $\Leftrightarrow$ CPTP maps.} It follows from the statistical structure of quantum mechanics that density matrices are PSD matrices of \emph{unit trace}. Because of this, any map transforming density matrices into density matrices must be trace preserving\footnote{Although CP maps that decrease the trace are used to model \emph{stochastic} measurements processes, i.e., those that are selective of certain specific outcomes. In these cases, the trace of the resulting ``density matrix'' gives the probability associated to the specific measurement outcome. This allows for a more comprehensive definition of the term ``quantum operation'' as any completely-positive-trace-non-increasing map, while the term ``quantum channel'' is generally reserved for the \emph{deterministic} subset of trace preserving CP maps. In this thesis, though, both terminologies are used interchangeably to refer to CPTP maps. Whenever needed, we explicitly write \emph{stochastic quantum operations} to refer to CP maps that decrease the trace.}
. This requirement, along with the CP condition explained above, validates the implication quantum operations $\Rightarrow$ CPTP maps.

In order to see that the converse also holds, let us equip the Kraus decomposition of Eq.~\eqref{eq:kraus_form} with the trace preserving condition $\tr \mathcal{K}(\rho)=\tr\rho$. It is a simple exercise to show that every CPTP map must satisfy
\begin{equation}\label{eq:TP_cond}
\sum_i K_i^\dagger K_i=\openone_{\rm d}\,.
\end{equation}
We now demonstrate that any transformation of the form of Eq.~\eqref{eq:kraus_form} satisfying Eq.~\eqref{eq:TP_cond} can be regarded as the evolution of a part of a larger system which evolves unitarily. Since unitary evolution is an obviously valid type of evolution, we conclude that every CPTP map describes a valid quantum operation.

To prove our point, we follow Ref.~\cite[p. 365]{00Nielsen}: Let $U$ be an operator such that $U\ket{\psi}\ket{e_0}=\sum_i K_i\ket{\psi}\ket{e_i}$ where $\rho=\ket{\psi}\!\bra{\psi}$ is an arbitrary state of the system of interest, $\ket{e_0}$ is an arbitrary state of an ancillary system $E$ and $\{\ket{e_i}\}$ is an orthonormal basis of $E$. Using Eq.~\eqref{eq:TP_cond} it is easy to show that $U$ preserves inner-products, i.e., $\bra{\psi}\bra{e_0}U^\dagger U\ket{\varphi}\ket{e_0}=\langle\psi|\varphi\rangle$, therefore $U$ can be realized as a unitary operation on the joint system. Moreover, from the definition of $U$ and the normalization of the basis $\{\ket{e_i}\}$, a straightforward computation gives
\begin{equation}
\tr_E\left[U\left(\rho\otimes\ket{e_0}\!\bra{e_0}\right)U^\dagger\right]=\sum_iK_i\rho K_i^\dagger\,,
\end{equation}
where $\tr_E$ denotes the partial trace over the system $E$. Since the rhs of the equation above is precisely the Kraus decomposition of Eq.~\eqref{eq:kraus_form}, we have proved the implication CPTP maps $\Rightarrow$ quantum operations.

\subsection{Some useful results}\label{sec:CPuseful}

In this section we gather a number of standard results related with quantum operations and their representations. The intention is to offer the reader a self-contained and concise review of tools and concepts employed in the developments to be presented in the following chapters. The experienced reader should feel free to skip the material of this section and refer to it for notational purposes while progressing through the rest of the thesis.

\subsubsection{Entanglement Breaking and Trace Preserving maps}

As argued in the last section, the restriction to set of CPTP maps is necessary and sufficient to provide a mathematical description of every physically admissible transformation of an arbitrary quantum state. Experimentally, though, it is conceivable that not every CPTP can be implemented due to practical (but not fundamental) limitations. From a theoretical viewpoint, it is interesting to consider the impact of additional constraints on the set of realizable quantum state transformations. In this section, we introduce a relevant subset of quantum operations: the set of entanglement breaking and trace preserving (EBTP) maps \cite{98Holevo1295,02Verstraete,03Horodecki629,03Ruskai643}.

Conceptually, EBTP maps convert quantum information into quantum information via a classically mediated process. Quoting Fuchs and Sasaki~\cite{03Fuchs377}, ``the reader should be left with the imagery of a quantum state initially living in a large river of Hilbert space, later to be squeezed through a very small outlet'', and then discharged in a large river of Hilbert space once again. In what follows we give a construction of an arbitrary EBTP map that more precisely formulates the above description.

First, the initial quantum state $\rho$ is measured with a POVM of elements $\{P_\alpha\}_{\alpha=1}^{\rm d}$, and depending on the classical outcome $\alpha$, a pure state $\ket{h_{\rm d}^\alpha}$ from an orthonormal basis set $\{\ket{h_{\rm d}^\alpha}\}_{\alpha=1}^{\rm d}$ of the Hilbert space $\textsf{H}_{\rm d}$ is prepared. The output state of such a process is of the form
\begin{equation}
\mathcal{QC}(\rho)=\sum_{\alpha=1}^{\rm d}\tr\left(P_\alpha\rho\right)\ket{h_{\rm d}^\alpha}\!\bra{h_{\rm d}^\alpha}\,.
\end{equation}
Holevo~\cite{98Holevo1295} coined the term \emph{quantum-classical channel} to designate maps of this form, since they express a decision rule transforming quantum states into probability distributions on an output alphabet $\{\alpha\}$.

After this, the POVM $\{\kett{h_{\rm d}^\beta}\!\braa{h_{\rm d}^\beta}\}_{\beta=1}^{\rm d}$ is applied and a density matrix $Q_\beta$ is prepared conditioned on the classical outcome $\beta$. The resulting state is then
\begin{equation}
\mathcal{CQ}(\varrho)=\sum_{\beta=1}^{\rm d}\tr\left(\varrho \kett{h_{\rm d}^\beta}\!\braa{h_{\rm d}^\beta}\right)Q_\beta\,.
\end{equation}
Maps of this form are termed \emph{classical-quantum} channels \cite{98Holevo1295}. With a straightforward computation, the composition of these two maps can be shown to be of the form
\begin{equation}\label{eq:Holevo_form}
\mathcal{CQ}\circ\mathcal{QC}(\rho)=\sum_{\beta=1}^{\rm d}\tr\left(\rho P_\beta\right)Q_\beta
\end{equation}
which is known as the \emph{Holevo form} --- a defining expression of every EBTP map \cite{03Horodecki629}.\par

Because of the intermediary measurement, any entanglement present between the system evolving under the EBTP dynamics and other degrees of freedom is destroyed. Since the $\mathcal{CQ}$ channel does not allow these systems to interact again, the outcome of an EBTP map is always separable with respect to this partition, motivating the name ``entanglement breaking'' coined in Ref.~\cite{03Horodecki629}.

An important property of EBTP maps is that they form a convex set. This can be promptly verified by noting that any convex sum of Holevo forms is still in the Holevo form, since any convex sum of POVMs is still a POVM. Every $\mathcal{CQ}$ map with a pure density matrix $Q_\beta$ is an extreme point of this convex set, however, for ${\rm d}\geq 3$ there are extreme points which are not $\mathcal{CQ}$ maps. These and other structural characterizations of the set of EBTP maps are provided in Ref.~\cite{03Horodecki629}. In Ref.~\cite{03Ruskai643} the case ${\rm d}=2$ was studied in detail.

In the next section we provide a representation of EBTP maps via certain types of matrices --- analogous to the representation of CP maps via the PSD matrices $\mathfrak{K}$ (cf. Theorem~\ref{teo:choi}). Because this way of characterizing maps is at the heart of most of the developments in this thesis, we also summarize similar results for the sets of positive and CPTP maps.

\subsubsection{Isomorphisms between maps and matrices}\label{sec:isomorphisms}

A key idea to be employed in this thesis is that the problem of searching over a set of maps can be recast as searching over a set of matrices. Mathematically, this arises when these two sets are related through a bijective correspondence, i.e., they are isomorphic. In less technical language, each map in the first set is associated with a different matrix in the second set (injection), and the matrices obtained from this relation fully span the second set (surjection). In this section, we sketch this isomorphism for the sets of positive, CP, CPTP and EBTP maps. An excellent review of the duality between maps and matrices is given in Ref.~\cite{04Zyczkowski3}.

In 1972, Jamio{\l}kowski showed that the set $\mathcal{P}_{\rm d}^{\rm set}$ of positive maps $\mathcal{P}:\mathcal{B}(\textsf{H}_{\rm d})\to\mathcal{B}(\textsf{H}_{\rm d})$ is isomorphic to the set $\mathfrak{P}_{{\rm d}^2}^{\rm set}$ of matrices $\mathfrak{P}\in\mathcal{B}(\textsf{H}_{\rm d}\otimes\textsf{H}_{\rm d})$ such that
\begin{equation}\label{eq:jamconst}
\bra{h_{\rm d}^\alpha}\braa{h_{\rm d}^\beta}\mathfrak{P}\ket{h_{\rm d}^\alpha}\kett{h_{\rm d}^\beta}\geq 0\qquad \mbox{for}\quad\alpha,\beta=1,\ldots,{\rm d}\,, \end{equation}
where $\{\ket{h_{\rm d}^\alpha}\}_{\alpha=1}^{\rm d}$ is some orthonormal bases of $\textsf{H}_{\rm d}$. The isomorphism $J: \mathcal{P}_{\rm d}^{\rm set}\to \mathfrak{P}_{{\rm d}^2}^{\rm set}$ is given by
\begin{equation}
\mathcal{P}\stackrel{J}{\mapsto}\mathfrak{P} = (\mathcal{I}_{\rm d}\otimes\mathcal{P})\left(\ket{\Psi}\!\bra{\Psi}\right)\,,
\end{equation}
where we have defined the unnormalized, bipartite (and maximally entangled) state $\ket{\Psi} \defeq \sum_{\alpha=1}^{\rm d}{\ket{h_{\rm d}^\alpha}\otimes \ket{h_{\rm d}^\alpha}}\in \textsf{H}_{\rm d}\otimes\textsf{H}_{\rm d}$. The reader should recognize the similarity of the map above with the one introduced in Theorem \ref{teo:choi} on page~\pageref{teo:choi} (see also the footnote~\ref{foot:choi} on the same page).\par

A few years later, in his landmark 1975 paper \cite{75Choi285}, Choi obtained a analogous result between the set $\mathcal{K}_{\rm d}^{\rm set}$ of \emph{completely positive} maps $\mathcal{K}:\mathcal{B}(\textsf{H}_{\rm d})\to\mathcal{B}(\textsf{H}_{\rm d})$ and the set $\mathfrak{K}_{{\rm d}^2}^{\rm set}$ of matrices $\mathfrak{K}\in\mathcal{B}(\textsf{H}_{\rm d}\otimes\textsf{H}_{\rm d})$ such that
\begin{equation}\label{eq:choiconst}
\bra{v}\mathfrak{K}\ket{v}\geq 0\qquad \forall \ket{v} \in \textsf{H}_{\rm d}\otimes\textsf{H}_{\rm d}\,,
\end{equation}
that is, the cone of positive semidefinite matrices $\mathfrak{K}$ of dimension ${\rm d}^2$. Once again, the map introduced in Eq.~\eqref{eq:jamisomorphism} --- defined between $\mathcal{K}_{\rm d}^{\rm set}$ and $\mathfrak{K}_{{\rm d}^2}^{\rm set}$ --- gives the expression of this isomorphism. A more recent and very neat demonstration of the isomorphism between $\mathcal{K}_{\rm d}^{\rm set}$ and $\mathfrak{K}_{{\rm d}^2}^{\rm set}$ is given in Ref.~\cite{01DAriano042308}.\par

Remarkably, condition~\eqref{eq:choiconst} is more demanding than condition~\eqref{eq:jamconst}, since the set of product states $\ket{x}\otimes\ket{y}$ is just a subset of the set of states that can be formed in the full Hilbert space $\textsf{H}_{\rm d}\otimes\textsf{H}_{\rm d}$. This is hardly surprising, since complete positivity is a more demanding constraint than mere positivity, as discussed in Sec.~\ref{sec:CP}. It is curious, though, that it is the set of completely positive maps (as opposed to the set of positive maps) that is isomorphic to the cone of positive semidefinite matrices.

For the purposes of this thesis, our main interest is on an isomorphic correspondence for the set $\mathcal{Q}_{\rm d}^{\rm set}$ of CPTP maps on $\mathcal{Q}:\mathcal{B}(\textsf{H}_{\rm d})\to\mathcal{B}(\textsf{H}_{\rm d})$. This is because the elements of this set model quantum operations, and for reasons that will become clearer later, it will be crucial to have a matrix characterization of how quantum mechanical systems can evolve. In Refs. \cite{99Fujiwara3290,99Horodecki1888}, $\mathcal{Q}_{\rm d}^{\rm set}$ was shown to be isomorphic to the set $\mathfrak{Q}^{\rm set}_{{\rm d}^2}$ of matrices $\mathfrak{Q}\in \mathcal{B}(\textsf{H}_{\rm d}\otimes\textsf{H}_{\rm d})$ such that
\begin{equation}\label{eq:CPTPconstraints}
\mathfrak{Q} \geq 0 \quad \mbox{and} \quad \tr_{\rm 2} \mathfrak{Q} = \openone_{\rm d}\,.
\end{equation}
Here, $\mathfrak{Q}$ can be regarded  as an (unnormalized) bipartite density matrix and $\tr_{\rm 2}$ denotes the partial trace over the second subsystem. Clearly, the positivity constraint $\mathfrak{Q}\geq 0$ arises from the CP character of every CPTP map, the partial trace constraint, in turn, follows from the additional trace preserving requirement.\par

Yet another set of maps of interest is $\mathcal{B}_{\rm d}^{\rm set}$, formed by the EBTP maps defined in the last section. Even before the characterization of Eq.~\eqref{eq:Holevo_form} given in Ref.~\cite{03Horodecki629}, these maps had long been employed as CPTP maps whose matrices $\mathfrak{Q}$ are always separable across the two parties of dimension {\rm d} \cite{96Bennett3824,99Horodecki1888,01Rains2921,01Cirac544,02Verstraete}. That is, $\mathcal{B}_{\rm d}^{\rm set}$ is isomorphic to the set of matrices $\mathfrak{B}_{{\rm d}^2}^{\rm set}$ whose elements satisfy
\begin{equation}\label{eq:EBTPisomorphism}
\mathfrak{B}\geq 0\,,\quad\tr_2\mathfrak{B}=\openone_{\rm d}\quad\mbox{and}\quad \mathfrak{B}\mbox{ is separable across }{\rm d}\otimes{\rm d}.
\end{equation}

The matrix elements forming the image of each one of the isomorphisms discussed above are obtained from
\begin{equation}\label{eq:jam_dir}
\mathfrak{C}=(\mathcal{I}_{\rm d}\otimes\mathcal{C})\left(\ket{\Psi}\!\bra{\Psi}\right)=\sum_{\alpha,\beta=1}^{\rm d}\kett{h_{\rm d}^\alpha}\!\braa{h_{\rm d}^\beta}\otimes\mathcal{C}(\kett{h_{\rm d}^\alpha}\!\braa{h_{\rm d}^\beta})
\end{equation}
where $\mathcal{C}$ represents any one of the maps $\mathcal{P}$, $\mathcal{K}$, $\mathcal{Q}$ or $\mathcal{B}$ while $\mathfrak{C}$ corresponds to $\mathfrak{P}$, $\mathfrak{K}$, $\mathfrak{Q}$ or $\mathfrak{B}$, respectively. Henceforth, for any linear map $\mathcal{C}$, we shall refer to the matrix $\mathfrak{C}$ computed via Eq.~\eqref{eq:jam_dir} as the \emph{Choi matrix of $\mathcal{C}$}.

On the other hand, Eq. (\ref{eq:jam_dir}) can be ``inverted'' to characterize the action of any linear map $\mathcal{C}$ from its Choi matrix $\mathfrak{C}$ \cite{01DAriano042308}
\begin{equation}\label{eq:jam_inv}
\mathcal{C}(\varrho)=\tr_1\left[(\varrho^{\sf T}\otimes\openone_{\rm d})\mathfrak{C}\right]\quad\forall \varrho \in \mathcal{B}(\textsf{H}_{\rm d})
\end{equation}
where the transpose operation {\sf T} is taken with respect to $\varrho$ written in the basis $\{\ket{h_{\rm d}^\alpha}\}_{\alpha=1}^{\rm d}$. To see that Eq.~\eqref{eq:jam_inv} actually inverts Eq.~\eqref{eq:jam_dir} for an arbitrary linear map $\mathcal{C}$, note that the substitution of the latter into the former immediately leads to a tautology:
\begin{equation}
\mathcal{C}(\varrho)=\tr_{\rm 1}\left[\left(\varrho^{\sf T}\otimes\openone_{\rm d}\right)\sum_{\alpha,\beta=1}^{\rm d}\kett{h_{\rm d}^\alpha}\!\braa{h_{\rm d}^\beta}\otimes\mathcal{C}(\kett{h_{\rm d}^\alpha}\!\braa{h_{\rm d}^\beta})\right]= \sum_{\alpha,\beta=1}^{\rm d}\braa{h_{\rm d}^\beta}\varrho^{\sf T}\kett{h_{\rm d}^\alpha}\mathcal{C}(\kett{h_{\rm d}^\alpha}\!\braa{h_{\rm d}^\beta})=\mathcal{C}(\varrho)
\end{equation}
where the linearity of $\mathcal{C}$ was used in the last equality.\par

A good example of application of Eq.~\eqref{eq:jam_inv} is the straightforward derivation of the expression of the trace preserving condition of a CPTP map $\mathcal{Q}$ in terms of its Choi matrix $\mathfrak{Q}$. Simple substitution of Eq.~\eqref{eq:jam_inv} into $\tr\mathcal{Q}(\rho)=\tr\rho$ leads to the constraint $\tr_2\mathfrak{Q}=\openone_{\rm d}$, already given in Eq.~\eqref{eq:CPTPconstraints}.

Throughout this thesis, we will often recur to represent quantum operations via their Choi matrices. For this reason, some practical mathematical tools for dealing with these objects are most welcome. In the next section we introduce some of them while approaching a technical problem of interest.

\subsubsection{The Choi matrix of a composed CP map and the {\rm vec}/{\rm mat} operations}\label{sec:choi_composed}

In this section we solve the following problem: Given Choi matrices $\mathfrak{K}_a$ and $\mathfrak{K}_b$ of two CP maps $\mathcal{K}_a$ and $\mathcal{K}_b$, how can we express the Choi matrix of the composed map $\mathcal{K}_b\circ\mathcal{K}_a$ as a function of  $\mathfrak{K}_b$ and $\mathfrak{K}_a$?

The solution for this problem will be useful later in this thesis. We also present it here because it sets the stage for the introduction of some tools for the algebraic manipulation of Choi matrices. These will be employed, for example, in the next section, where we show that Kraus decompositions and CP maps are not isomorphically related.

We start by applying Eqs.~\eqref{eq:jam_dir} and~\eqref{eq:jam_inv} to write
\begin{multline}\label{eq:seq1}
\mathfrak{K}_{ba}= (\mathcal{I}_{\rm d}\otimes\mathcal{K}_b\circ\mathcal{K}_a)\ket{\Psi}\!\bra{\Psi}= \left(\mathcal{I}_{\rm d}\otimes\mathcal{K}_b\right)\circ\left(\mathcal{I}_{\rm d}\otimes\mathcal{K}_a\right)\ket{\Psi}\!\bra{\Psi}=\\ \left(\mathcal{I}_{\rm d}\otimes\mathcal{K}_b\right)\mathfrak{K}_a=
{\rm Tr}_1\left[\left(\mathfrak{K}_a^{\sf T}\otimes\openone_{{\rm d}^2}\right)\widetilde{\mathfrak{K}}\right]\,,
\end{multline}
where $\widetilde{\mathfrak{K}}$ denotes the ${\rm d}^4$-dimensional Choi matrix of the map $\mathcal{I}_{\rm d}\otimes\mathcal{K}_b$ and ${\rm Tr}_1$ the partial trace over the first ${\rm d}^2$-dimensional subsystem of its argument. In order to continue from here, we need to express $\widetilde{\mathfrak{K}}$ in terms of $\mathfrak{K}_b$. The following short review on matrix vectorization will provide adequate tools for the solution of this problem.

\paragraph{The {\rm vec} and {\rm mat} operations.}
The {\rm vec}-operator~\cite{91Horn,04Nielsen} is defined as the transformation of any matrix into a vector by stacking the columns of the original matrix. For example,
\begin{equation}\label{eq:vec_example}
{\rm vec}\left(\begin{array}{cc}a & c \\
b & d \end{array}\right)=\left(\begin{array}{c}a\\b\\c\\d\end{array}\right)\,.
\end{equation}
There are a number of useful properties of {\rm vec} that can be easily verified~\cite{04Nielsen}: For $A,B,C\in\mathcal{M}_{\rm d}$,
\begin{align}
{\rm vec}\,A&=\left(\openone_{\rm d}\otimes A\right)\ket{\Psi}\label{eq:vec_mes}\\
({\rm vec}\, A)^\dagger {\rm vec}\, B &= \tr\left(A^\dagger B\right)\label{eq:vec_trace}\\
{\rm vec}(ABC)&=\left(C^{\sf T}\otimes A\right){\rm vec}\,B\label{eq:roth}\\
{\rm vec}\left( A\otimes B\right)&=P_{{\rm d}^4}\left({\rm vec}\, A\otimes {\rm vec}\, B\right)\label{eq:vecpermbody}
\end{align}
where, in the first line, $A$ is defined in the same basis used to define $\ket{\Psi}$; in the last line, $P_{{\rm d}^4}$ is a permutation matrix whose explicit form is worked out in Appendix~\ref{app:vecperm}. It is also not difficult to see that for every vector $v$ of ${\rm d}^2$ entries, there exists a unique matrix $A\in\mathcal{M}_{\rm d}$ such that ${\rm vec}\,A = v$. This establishes an inverse operation for {\rm vec}, which we shall denote by ${\rm mat}$, i.e., $M={\rm mat}\, v$.

Let us now see how this machinery can be of assistance in solving the composition problem at hand. First, note the following relationship between the Choi matrix of any CP map $\mathcal{K}$ and the Kraus decomposition of Eq.~\eqref{eq:kraus_form}:
\begin{equation}\label{eq:choikrauvec}
\mathfrak{K}=(\mathcal{I}_{\rm d}\otimes\mathcal{K})\ket{\Psi}\!\bra{\Psi}=\sum_{i}\left(\openone_{\rm d}\otimes K_i\right)\ket{\Psi}\!\bra{\Psi}\left(\openone_{\rm d}\otimes K_i^\dagger\right)=\sum_i{\rm vec}\,K_i\left({\rm vec}\,K_i\right)^\dagger\,,
\end{equation}
where Eq.~\eqref{eq:vec_mes} was used to establish the last equality.

Now, for sake of generality, define the CP map $\mathcal{T}\defeq\mathcal{E}\otimes\mathcal{F}$, where $\mathcal{E}$ and $\mathcal{F}$ are CP maps with Kraus operator sets $\{E_i\}$ and $\{F_j\}$, respectively. Using Eqs.~\eqref{eq:choikrauvec},~\eqref{eq:vecpermbody} and some straightforward algebra, we find that the Choi matrix $\mathfrak{T}$ (of $\mathcal{T}$) can be written in terms of the Choi matrices $\mathfrak{E}$ and $\mathfrak{F}$ (of $\mathcal{E}$ and $\mathcal{F}$) as follows:
\begin{align}
\mathfrak{T}&=\sum_{i,j}{\rm vec}\left(E_i\otimes F_j\right)\left[{\rm vec}\left(E_i\otimes F_j\right)\right]^\dagger\\
&=\sum_{i,j}P_{{\rm d}^4}\left[{\rm vec}\,E_i\otimes{\rm vec}\,F_j\right]\left[\left({\rm vec}\,E_i\right)^\dagger\otimes\left({\rm vec}\,F_j\right)^\dagger\right]P_{{\rm d}^4}\\
&=P_{{\rm d}^4}\left[\sum_i{\rm vec}\,E_i\left({\rm vec}\,E_i\right)^\dagger\otimes\sum_j{\rm vec}\,F_j\left({\rm vec}F_j\right)^\dagger\right]P_{{\rm d}^4}\\
&=P_{{\rm d}^4}\left(\mathfrak{E}\otimes\mathfrak{F}\right)P_{{\rm d}^4}\,.
\end{align}
This result and the trivial fact that the Choi matrix of the identity map is $\ket{\Psi}\!\bra{\Psi}$ can be used to write $\widetilde{\mathfrak{K}}$ in terms of $\mathfrak{K}_b$,
\begin{equation}
\widetilde{\mathfrak{K}}=P_{{\rm d}^4}\left(\ket{\Psi}\!\bra{\Psi}\otimes\mathfrak{K}_b\right)P_{{\rm d}^4}\,,
\end{equation}
which substituted into Eq.~\eqref{eq:seq1} leads to
\begin{equation}
\mathfrak{K}_{ba}={\rm Tr}_1\left[\left(\mathfrak{K}_a^{\sf T}\otimes\openone_{{\rm d}^2}\right)P_{{\rm d}^4}\left(\ket{\Psi}\!\bra{\Psi}\otimes\mathfrak{K}_b\right)P_{{\rm d}^4}\right]\,.
\end{equation}
Since both $P_{{\rm d}^4}$ and $\ket{\Psi}\!\bra{\Psi}$ are fixed matrices for a given dimension ${\rm d}$, the equation above gives $\mathfrak{K}_{ba}$ solely in terms of $\mathfrak{K}_a$ and $\mathfrak{K}_b$, as required.

\subsubsection{Non-uniqueness and construction of Kraus decompositions}\label{sec:krausnonunique}
The goal of this section is to show that, contrary to the one-to-one relationship between CP maps and Choi matrices, the relationship between CP maps and Kraus decompositions is one-to-many. Moreover, we describe a practical method to construct any set of Kraus operators from the unique Choi matrix of a given CP map. As a byproduct, we will find that every CPTP map acting on a {\rm d}-dimensional system can be written in the Kraus form with no more than ${\rm d}^2$ Kraus operators.

Let us start by checking that a given CP map can be decomposed in the Kraus form in (infinitely) many different ways. In order to do so, we show that if $K_j$ is regarded as set of Kraus operators of a CP map $\mathcal{K}$, then so is the set $A_i$ defined by
\begin{equation}\label{eq:unitaryfreedom}
A_i\defeq\sum_j u_{i,j} K_j\,,
\end{equation}
where the complex numbers $u_{i,j}$ represent the matrix elements of an isometry $U$. This result has a long history, being firstly stated in a closely related form by Schr{\"o}dinger in 1936 \cite{36Schrodinger446}, and ever since rediscovered numerous times \cite{57Jaynes171,93Hughston14} (see Ref.~\cite{06Kirkpatrick95} for a both technical and historical review).

It suffices to show that the set $A_i$ gives rise to the same Choi matrix $\mathfrak{K}$ as the set $K_j$. Let $\mathfrak{A}$ denote the Choi matrix of the map generated by $A_i$, from Eq.~\eqref{eq:choikrauvec} we have
\begin{equation}
\mathfrak{A}=\sum_i{\rm vec}\,A_i\left({\rm vec}\,A_i\right)^\dagger=\sum_{j,k}\left(\sum_i u_{i,j} u^\ast_{i,k}\right) {\rm vec}\,K_j\left({\rm vec}\,K_k\right)^\dagger=\sum_j{\rm vec}\,K_j\left({\rm vec}\,K_j\right)^\dagger=\mathfrak{K}\,,
\end{equation}
which establishes the claimed result. Notice that in the last line we have only used the fact that $\sum_i u_{i,j} u^\ast_{i,k}=\delta_{k,j}$, which is a common property of any isometry. Since there are infinitely many different isometries to choose from, infinitely many sets of Kraus operators for the same quantum operation can be generated from Eq.\eqref{eq:unitaryfreedom}. Actually, by varying over all isometries, Eq.~\eqref{eq:unitaryfreedom} gives rise to every possible set of Kraus operators of a fixed CP map $\mathcal{K}$ (see \cite[p. 372]{00Nielsen} for a proof).

A closely related question is how a set of Kraus operators can be constructed from a given Choi matrix. In what follows, we present a canonical procedure that, on its own, reinforces the idea that there are (infinitely) many different choices of Kraus operators for a given CP map and clearly demonstrates that every CP map can be specified with a set with no more than ${\rm d}^2$ Kraus operators.

The fact that any Choi matrix $\mathfrak{K}$ is a ${\rm d}^2 \times {\rm d}^2$ PSD matrix can be expressed via the equality \cite{07Bhatia}
\begin{equation}\label{eq:squarerootdecomp}
\mathfrak{K}=S^\dagger S\,,
\end{equation}
where $S$ is a matrix of dimensions ${\rm d}^\prime\times{\rm d}^2$ for arbitrary ${\rm d}^\prime$. A set of Kraus operators for a CP map $\mathcal{K}$ can be obtained from $\mathfrak{K}$ by reshaping the columns of $S^\dagger$ into matrices, as follows:
\begin{equation}\label{eq:fromChoitoKraus}
A_i={\rm mat}\left[{\rm col}_i\left(S^\dagger\right) \right]\qquad\mbox{for}\quad i=1,\ldots,{\rm d}^\prime\,,
\end{equation}
where ${\rm col}_i$ is the operator that extracts the $i$-th column of its matrix argument. Because there are many different choices of matrices $S$ that decompose $\mathfrak{K}$ in the form of Eq.~\eqref{eq:squarerootdecomp}, many different Kraus decompositions can arise from Eq.~\eqref{eq:fromChoitoKraus} --- this is just a different way of observing the already demonstrated one-to-many correspondence between Choi matrices and Kraus operators. The Cholesky factorization of $\mathfrak{K}$~\cite{85Horn}, for example, can always be employed to yield the unique upper triangular matrix $S$ with non-negative diagonal entries that satisfy Eq.~\eqref{eq:squarerootdecomp}. In this case, we have ${\rm d}^\prime={\rm d}^2$, which guarantees that it is always possible to decompose a CP map with no more than ${\rm d}^2$ Kraus operators. In fact, the minimal number of Kraus operators with which a CP map can be decomposed is equal to the rank of the corresponding Choi matrix~\cite{02Verstraete,05Salgado55}.

It is straightforward to verify that formula~\eqref{eq:fromChoitoKraus} gives a valid set of Kraus operators for $\mathcal{K}$. Next, we do that by recovering the Choi matrix $\mathfrak{K}$ when the above defined $A_i$ are used in Eq.~\eqref{eq:choikrauvec}:
\begin{multline}
\sum_i{\rm vec}\,A_i\left({\rm vec}\,A_i\right)^\dagger=\sum_i{\rm vec}\,{\rm mat}\left[{\rm col}_i\left(S^\dagger\right) \right]\left\{{\rm vec}\,{\rm mat}\left[{\rm col}_i\left(S^\dagger\right) \right]\right\}^\dagger=\\\sum_i{\rm col}_i\left(S^\dagger\right) \left[{\rm col}_i\left(S^\dagger\right) \right]^\dagger=S^\dagger S=\mathfrak{K}
\end{multline}

%

\subsubsection{Basis for hermitian matrices}\label{sec:herm_basis}

By construction, every Choi matrix is a hermitian matrix. In many cases, it will be helpful to expand a Choi matrix in some fixed basis of hermitian matrices and characterize the associated map by the real coefficients of the expansion. In this section we set the conventions of the basis we will employ, as well as derive the expansion of an arbitrary CPTP map with respect to it.

We shall denote by $\{H_{\rm d}^\alpha\}_{\alpha=1}^{d^2}$ any set of hermitian matrices forming a basis for the space of hermitian matrices and satisfying the following properties: $H_{\rm d}^1=\openone_{\rm d}$ is the only element of the set with a non-zero trace, that is $\tr H_{\rm d}^\alpha={\rm d}\delta_{\alpha,1}$. Moreover, we would like to think of the remaining matrices as higher dimensional generalizations of the Pauli matrices for {\rm d}=2,
\begin{equation}
\sigma_1=
\left(\begin{array}{cc}
0&1\\
1&0
\end{array}\right)\,,\qquad
\sigma_2=
\left(\begin{array}{cc}
0&-i\\
i&0
\end{array}\right)\quad\mbox{and}\quad
\sigma_3=
\left(\begin{array}{cc}
1&0\\
0&-1
\end{array}\right)\,.
\end{equation}
Since the Pauli matrices satisfy $\tr\left(\sigma_\alpha\sigma_\beta\right)=2\delta_{\alpha,\beta}$, for $\alpha,\beta=1,\ldots,3$, we shall require this same orthonormalization property for $H_{\rm d}^\alpha$, i.e., $\tr\left(H_{\rm d}^\alpha H_{\rm d}^\beta\right)=2\delta_{\alpha,\beta}$ for $\alpha,\beta=2,\ldots,{\rm d}^2$; or, accounting for the convention $H_{\rm d}^1=\openone_{\rm d}$,
\begin{equation}
\tr\left(H_{\rm d}^\alpha H_{\rm d}^\beta\right)=\delta_{\alpha,\beta}\left[{\rm d}\delta_{\beta,1}+2(1-\delta_{\beta,1})\right]\,,
\end{equation}
for $\alpha,\beta=1,\ldots,{\rm d}^2$. For ${\rm d}=3$, the identity matrix and the Gell-mann matrices \cite{62Gell-Mann1067} provide a possible construction of the basis set $\{H_3^{\alpha}\}_{\alpha=1}^9$. For larger dimensions, we have used the generators of SU({\rm d}) (plus the identity matrix) as a basis.
We note that while our generalized basis is still hermitian and orthogonal like the Pauli matrices, the matrices $H_{\rm d}^\alpha$ are generally not unitary for ${\rm d}\geq 3$.

We now move to characterize the Choi matrix of a CPTP map $\mathcal{Q}$ acting on a {\rm d}-dimensional quantum system. As discussed before, $\mathfrak{Q}$ will be a matrix of dimension ${\rm d}^2$. It will be convenient to expand it in a tensor product basis of the form
\begin{equation}\label{eq:choigenbasis}
\mathfrak{Q}=\sum_{\alpha,\beta=1}^{d^2} x_{\alpha,\beta} H_{\rm d}^\alpha\otimes H_{\rm d}^\beta\,.
\end{equation}

Since this represents the Choi matrix of a trace preserving map, we must have that $\tr_2 \mathfrak{Q}=\openone_{\rm d}$. Clearly, this matrix constraint can be rewritten as several scalar constraints, as follows:
\begin{equation}
\tr\left[\left(\tr_2 \mathfrak{Q}\right)H_{\rm d}^\alpha\right]=\tr H_{\rm d}^\alpha\qquad\mbox{for }\alpha=1,\ldots, {\rm d}^2\,.
\end{equation}
From the definition of the partial trace \cite[p. 107]{00Nielsen} and the properties of $H_{\rm d}^\alpha$ established above, we can even rewrite the above as follows
\begin{equation}\label{eq:solvedlineqconst}
\tr\left[\mathfrak{Q}\left(H_{\rm d}^\alpha\otimes\openone_{\rm d}\right)\right]={\rm d}\delta_{\alpha,1}\qquad\mbox{for }\alpha=1,\ldots, {\rm d}^2\,.
\end{equation}
By substituting the expansion of Eq.~\eqref{eq:choigenbasis} in the equation above, and solving the resulting set of equations for the coefficients $x_{\alpha,\beta}$, we find that $x_{\alpha,1}=\delta_{\alpha,1}/{\rm d}$ for all $\alpha=1,\ldots,{\rm d}^2$. This condition, used in Eq.~\eqref{eq:choigenbasis}, gives rise to the following general expansion of any Choi matrix of a CPTP map:
\begin{equation}\label{eq:choiexpandTPin}
\mathfrak{Q}=\frac{\openone_{{\rm d}^2}}{\rm d}+\sum_{\substack{\alpha=1\\\beta=2}}^{{\rm d}^2} x_{\alpha,\beta} H_{\rm d}^{\alpha}\otimes H_{\rm d}^\beta\,.
\end{equation}
Of course, many other constraints apply on the coefficients $x_{\alpha,\beta}$ in order to guarantee that $\mathfrak{Q}\geq 0$. In the next section, we shall explicitly look at them in the case of ${\rm d}=2$.

\subsubsection{CPTP maps on $\mathcal{M}_2$}\label{sec:CPTPonM2}

Up to here, we have surveyed a few properties of CPTP maps that are valid independently of the finite dimension  {\rm d} of the matrices onto which they apply. In this section, we review some important results which are peculiar to the case ${\rm d}=2$. These will be important in our study of qubit state transformations in Ch.~\ref{chap:aggie},~\ref{chap:tracking} and~\ref{chap:multistep}. Most of the results discussed here were obtained by Ruskai~\etal in Refs.~\cite{01King192,02Ruskai159}.

\paragraph{Diagonalizing CPTP maps on qubits}

When we restrict the action of the CPTP maps to $2\times 2$ density matrices, there is a nice alternative representation of these maps due to King and Ruskai \cite{01King192} which, among many other things, enhances our intuition on how such maps transform the state of a single qubit. The derivation of this representation closely resembles a diagonalization procedure and is, in essence, the application of the singular value decomposition to some matrix formed from the expansion coefficients of the Choi matrix, as discussed above. In what follows, this procedure is presented in detail.

Let $\mathfrak{Q}$ be the $4 \times 4$ Choi matrix of a CPTP map $\mathcal{Q}$ on a single qubit state. Then, from Eq.\eqref{eq:choiexpandTPin} we have

\begin{equation}\label{eq:rhoCgeneral}
\mathfrak{Q}=\tfrac{1}{2}\openone_{\rm 4}+\sum_{\substack{j=0\\l=1}}^3 q_{jl} \left(\sigma_j\otimes\sigma_l\right)\,.
\end{equation}
where $q_{jl}$ are the (real) expansion coefficients and we adopted the convention $\sigma_0=\openone_2$.

We now give a general (and convenient) form for the input state of $\mathcal{Q}$. Since any density matrix is necessarily hermitian, we can write
\begin{equation}\label{eq:rhostokes}
2\rho=c\openone_{\rm 2}+\bm{R}\cdot\bm{\sigma}\,,
\end{equation}
where $c\in\mathbb{R}$ and $\bm{R}\in\mathbb{R}^3$ is the vector formed from the expansion coefficients of $\rho$ (multiplied by two) for the basis elements $\sigma_1$, $\sigma_2$ and $\sigma_3$, in such a way that Eq.~\eqref{eq:rhoCgeneral} is an ordinary expansion on the Pauli basis. Besides being written in more compact notation, Eq.~\eqref{eq:rhoCgeneral} highlights the special significance of the vector $\bm{R}$ as the \emph{Bloch vector} of the state $\rho$. Because the Pauli matrices are traceless and density matrices have unit trace, we should set $c=\tr\rho =1$. However, for later use, it will be convenient to proceed leaving $c$ unspecified.

With the aid of Eqs. (\ref{eq:jam_inv}), (\ref{eq:rhoCgeneral}) and (\ref{eq:rhostokes}), the action of the map $\mathcal{Q}$ on $\rho$ can be written in terms of the coefficients $q_{jl}$ as
\begin{equation}\label{eq:Crhofull}
2\mathcal{Q}(\rho)=c\openone_{\rm 2}+\left[c\left(\begin{array}{c}
2 q_{01}\\2 q_{02}\\2 q_{03}\end{array}\right)
+
\left(\begin{array}{ccc}
2 q_{11}& -2 q_{21}& 2 q_{31} \\
2 q_{12}& -2 q_{22}& 2 q_{32} \\
2 q_{13}& -2 q_{23}& 2 q_{33}
\end{array}\right)\left(\begin{array}{c}\bm{R}\cdot\bm{x}\\ \bm{R}\cdot\bm{y}\\ \bm{R}\cdot\bm{z}\end{array}\right)\right]\cdot\left(\begin{array}{c}\sigma_1\\ \sigma_2\\ \sigma_3\end{array}\right)\,.
\end{equation}
We now apply the singular value decomposition to the $3 \times 3$ matrix above. Because $q_{jl}\in\mathbb{R}$ for all $j$ and $l$, the unitary matrices arising from the SVD can be chosen to be real orthogonal matrices,
\begin{equation}
\left(\begin{array}{ccc}
2 q_{11}& -2 q_{21}& 2 q_{31} \\
2 q_{12}& -2 q_{22}& 2 q_{32} \\
2 q_{13}& -2 q_{23}& 2 q_{33}
\end{array}\right)=O_2\left(\begin{array}{ccc}
\widetilde{\mu}_1& 0& 0 \\
0& \widetilde{\mu}_2& 0\\
0& 0& \widetilde{\mu}_3
\end{array}\right) O_1^{\sf T}=
R_2\left(\begin{array}{ccc}
\mu_1& 0& 0 \\
0& \mu_2& 0\\
0& 0& \mu_3
\end{array}\right) R_1^{\sf T}\,,
\end{equation}
where $O_1,O_2 \in O(3,\mathbb{R})$ and $\widetilde{\mu}_{\{1,2,3\}}\geq 0$. The second equality expresses the fact that every element of $O(3,\mathbb{R})$ is either a rotation [i.e., an element of $SO(3,\mathbb{R})$] or the product of a rotation with the inversion $-\openone_3$. As such, the equation is valid for some $R_1, R_2 \in SO(3,\mathbb{R})$ if we absorb the sign of a possible inversion in the coefficients $\mu_{\{1,2,3\}}$. As a result, we have $|\mu_{\{1,2,3\}}|=\widetilde{\mu}_{\{1,2,3\}}$, but the matrix ${\rm diag}\left[\mu_1,\mu_2,\mu_3\right]$ is not necessarily PSD. Applying this decomposition to Eq. (\ref{eq:Crhofull}), we find
\begin{equation}\label{eq:Crhofim}
2\mathcal{Q}(\rho)=c\openone_{\rm 2}+R_2\left[c R_2^{\sf T}\left(\begin{array}{c} 2 q_{01}\\ 2 q_{02} \\ 2 q_{03} \end{array}\right)
+
\left(\begin{array}{ccc}
\mu_1& 0& 0 \\
0& \mu_2& 0\\
0& 0& \mu_3
\end{array}\right)R_1^{\sf T}\left(\begin{array}{c}\bm{R}\cdot\bm{x}\\ \bm{R}\cdot\bm{y}\\ \bm{R}\cdot\bm{z}\end{array}\right)\right]\cdot\left(\begin{array}{c}\sigma_1\\ \sigma_2\\ \sigma_3\end{array}\right)\,,
\end{equation}
which reveals the CPTP map $\mathcal{Q}$ as an affine transformation of the Bloch vector. This can be interpreted as follows: First, the Bloch vector is rotated by $R_1^{\sf T}$, then, the $x-$, $y-$ and $z-$ components of the rotated vector are rescaled by $\mu_1$, $\mu_2$ and $\mu_3$, respectively. Subsequently, the constant vector $c\bm{s}$ with $\bm{s}\defeq R_2^{\sf T}\left(\begin{array}{c}2 q_{01} \\ 2 q_{02} \\2 q_{03} \end{array}\right)$ is added and finally another rotation of the resulting Bloch vector by $R_2$ is performed. This sequence of transformations is illustrated in Fig.~\ref{fig:affine}.

Because every rotation of the Bloch vector corresponds to a unitary transformation of the density matrix, we can write
\begin{equation}\label{eq:ruskai_decomposition_bg}
\mathcal{Q}(\rho)=U\mathcal{D}(V\rho V^\dagger)U^\dagger
\end{equation}
where $V$ is the unitary associated to the rotation $R_1^{\sf T}$, $U$ is the unitary associated to the rotation $R_2$ and $\mathcal{D}$ is said to be a diagonal CPTP map on the Pauli basis, implementing the rescaling and the translation explained above.

\begin{figure}[h]
\centering
\includegraphics[width=14.5cm]{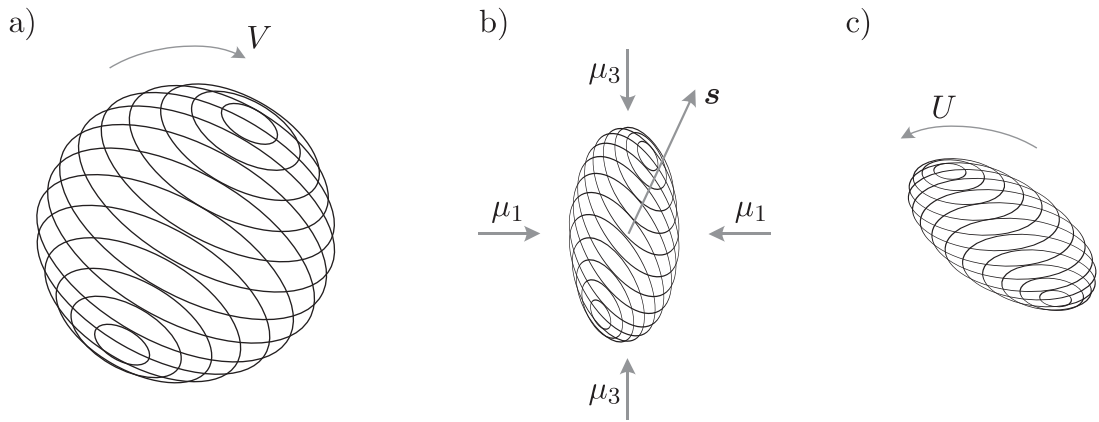}
\caption[Bloch sphere transformation due to a completely positive and trace preserving map]{Affine transformation on the Bloch sphere due to a CPTP map}\label{fig:affine}
\end{figure}

\paragraph{The allowed coefficients $\mu_{1,2,3}$ and $s_{1,2,3}$.}
From the results of Sec. \ref{sec:isomorphisms}, not every real value of coefficients $q_{jl}$ substituted in Eq.~\eqref{eq:rhoCgeneral} will lead to a valid CPTP map $\mathcal{Q}$. In order to produce a sharp characterization, it is necessary to specify every possible choice of coefficients that turn $\mathfrak{Q}$ into as a PSD matrix. This (almost heroic) task was accomplished by Ruskai, Szarek and Werner \cite{02Ruskai159}, who proved that any map $\mathcal{Q}$ in the form of Eq.~\eqref{eq:Crhofim} is CPTP if and only if the following inequalities hold
\begin{equation}
(\mu_1+\mu_2)^2\leq(1+\mu_3)^2-s_3^2-(s_1^2+s_2^2)\left(\frac{1+\mu_3\pm s_3}{1-\mu_3\pm s_3}\right)\,,
\end{equation}
\begin{equation}
(\mu_1-\mu_2)^2\leq(1-\mu_3)^2-s_3^2-(s_1^2+s_2^2)\left(\frac{1-\mu_3\pm s_3}{1+\mu_3\pm s_3}\right)\,,
\end{equation}
\begin{multline}
\left[1-\left(\mu_1^2+\mu_2^2+\mu_3^2\right)-\left(s_1^2+s_2^2+s_3^2\right)\right]^2\geq \\ 4\left[\mu_1^2\left(s_1^2+\mu_2^2\right)+\mu_2^2\left(s_2^2+\mu_3^2\right)+\mu_3^2\left(s_3^2+\mu_1^2\right)-2\mu_1\mu_2\mu_3\right]\,.
\end{multline}

In terms of the representation of Fig.~\ref{fig:affine}, these inequalities establish that not every ellipsoid internal to the Bloch sphere can be obtained from the action of a CPTP map; furthermore, they precisely characterize which ellipsoids can be obtained in such a way.

Remarkably, the inequalities above are saturated for the extreme points of the set of CPTP maps \cite{02Ruskai159}. This turns out to be equivalent to the conditions $s_1=s_2=0$ and
\begin{equation}\label{eq:ruskai_extreme_points}
\mu_3=\mu_1\mu_2\qquad\mbox{and}\qquad s_3^2=(1-\mu_1^2)(1-\mu_2^2)
\end{equation}
which leads to the useful trigonometric parametrization of the extreme points:
\begin{equation}
\mu_1=\cos{u}\,,\quad\mu_2=\cos{v}\,,\quad\mu_3=\cos{u}\cos{v}\,,\quad s_3=\sin{u}\sin{v}
\end{equation}
with $u\in[0,2\pi)$ and $v\in[0,\pi)$.
\paragraph{A representation for the Choi matrix of a general CPTP map on a single qubit}\label{sec:Choimatrep}

Throughout, it will be convenient to express the Choi matrix $\mathfrak{Q}$ of the CPTP map $\mathcal{Q}$ in terms of the unitaries $U$, $V$ and the scalars $\mu_{\{1,2,3\}}$ and $s_{\{1,2,3\}}$ introduced in the last section. Here, this representation is derived in an elementary way.\par

Our first step, is to rewrite Eq.~\eqref{eq:ruskai_decomposition_bg} in terms of the Choi matrices $\mathfrak{Q}$ and $\mathfrak{D}$ of the maps CPTP $\mathcal{Q}$ and $\mathcal{D}$, respectively. For $\{D_i\}$ a set of Kraus operators of $\mathcal{D}$, we have
\begin{align}
\mathfrak{Q}&=\left(\mathcal{I}_2\otimes\mathcal{Q}\right)\ket{\Psi}\!\bra{\Psi}\\
&=\sum_i\left(\openone_2\otimes U D_i V\right)\ket{\Psi}\!\bra{\Psi}\left(\openone_2\otimes U D_i V\right)^\dagger\\
&=\sum_i {\rm vec}\left(U D_i V\right)\left[{\rm vec}\left(U D_i V\right)\right]^\dagger\\
&=\sum_i\left(V^{\sf T}\otimes U\right){\rm vec}\,D_i\left({\rm vec}\,D_i\right)^\dagger\left(V^{\sf T}\otimes U\right)^\dagger\\
&=\left(V^{\sf T}\otimes U\right)\mathfrak{D}\left(V^{\sf T}\otimes U\right)^\dagger\,,\label{eq:KCKD}
\end{align}
where in the third, fourth and fifth equalities were obtained from the application of Eqs.~\eqref{eq:vec_mes}, \eqref{eq:roth} and \eqref{eq:choikrauvec}, respectively.

Eq.~\eqref{eq:KCKD} makes it clear that if we can express $\mathfrak{D}$ in terms of $\mu_{\{1,2,3\}}$ and $s_{\{1,2,3\}}$, then the desired formula for $\mathfrak{Q}$ can be obtained by simple unitary conjugation. Our next step is then to obtain this formula for $\mathfrak{D}$.\par

Since $\ket{\Psi}\!\bra{\Psi}$ can be written in terms of the Pauli matrices as $(\openone_{\rm 4}+\sigma_1\otimes\sigma_1-\sigma_2\otimes\sigma_2+\sigma_3\otimes\sigma_3)/2$, we have
\begin{align}
2\mathfrak{D}&=2(\mathcal{I}_{\rm 2}\otimes\mathcal{D})\ket{\Psi}\!\bra{\Psi}\nonumber\\
&=\openone_{\rm 2}\otimes\mathcal{D}(\openone_{\rm 2})+\sigma_1\otimes\mathcal{D}(\sigma_1)-\sigma_2\otimes\mathcal{D}(\sigma_2)+\sigma_3\otimes\mathcal{D}(\sigma_3)\nonumber\\
&=\openone_{\rm 4} + \openone_{\rm 2}\otimes\left(\bm{s}\cdot\bm{\sigma}\right)+\mu_1\sigma_1\otimes\sigma_1-\mu_2\sigma_2\otimes\sigma_2+\mu_3\sigma_3\otimes\sigma_3\,,\label{eq:ChoiDiag}
\end{align}
where, for the last equation, we used Eq.~\eqref{eq:Crhofim} with $R_1=R_2=\openone_3$ to find that the action of $\mathcal{D}$ on the Pauli basis is given by $\mathcal{D}(\openone_{\rm 2})=\openone_{\rm 2}+\bm{s}\cdot{\bm{\sigma}}$ and $\mathcal{D}(\sigma_k)=\mu_k\sigma_k$, for $k=1,2,3$.
Combining Eqs.~\eqref{eq:ChoiDiag} and \eqref{eq:KCKD}, a formula for $\mathfrak{Q}$ in terms of $U$, $V$ and the scalars $\mu_{\{1,2,3\}}$ and $s_{\{1,2,3\}}$ can be promptly obtained. Next, this formula is presented in terms of two orthonormal sets of real vectors, $\{\bm{v}_k\}_{k=1,2,3}$ and $\{\bm{u}_k\}_{k=1,2,3}$, that specify the unitaries according to
\begin{equation}
\bm{v}_k\cdot \bm{\sigma}\defeq V^\dagger \sigma_k V \quad \mbox{and}\quad
\bm{u}_k\cdot \bm{\sigma}\defeq U\sigma_k U^\dagger \,.\label{eq:unitaryVU}
\end{equation}
or more explicitly
\begin{align}
\bm{v}_k&=\tfrac{1}{2}\left(\tr\left[V^\dagger \sigma_k V\sigma_1\right]\,,\,\tr\left[V^\dagger \sigma_k V\sigma_2\right]\,,\,\tr\left[V^\dagger \sigma_k V\sigma_3\right]\right)\,,\label{eq:vk_intro}\\
\bm{u}_k&=\tfrac{1}{2}\left(\tr\left[U \sigma_k U^\dagger\sigma_1\right]\,,\,\tr\left[U \sigma_k U^\dagger\sigma_2\right]\,,\,\tr\left[U \sigma_k U^\dagger\sigma_3\right]\right)\,.\label{eq:uk_intro}
\end{align}
In terms of these versors, we have
\begin{equation}\label{eq:choigen}
2 \mathfrak{Q} = \openone_{\rm 4} + \sum_{k=1}^{3}{s_k \openone_{\rm 2}\otimes\left(\bm{u}_k\cdot\bm{\sigma}\right)}+\sum_{k=1}^3 \mu_k \left(\bm{v}_k\cdot\bm{\sigma}^{\sf T}\right)\otimes\left(\bm{u}_k\cdot\bm{\sigma}\right)\,.
\end{equation}
From this equation and the above discussion, it should be clear that the specification of the six versors $\bm{v}_k$ and $\bm{u}_k$ and the six scalars $\mu_k$ and $s_k$, fully specify any quantum operation on a single qubit. In fact, it suffices to specify only two versors of each set, since for an arbitrary unitary matrix $W$ and versors $\bm{w}_k$ related to $W$ as in Eq. (\ref{eq:unitaryVU}), we have $\bm{w}_k\times\bm{w}_l=\varepsilon_{klm}\bm{w}_m$, where $\varepsilon_{klm}$ is the Levi-Civita symbol, and $(k,l,m)$ is some permutation of $(1,2,3)$ \footnote{To see that, simply use these two well-known properties of Pauli matrices: (i) $\sigma_k\sigma_l=i\varepsilon_{klm}\sigma_m$ [for every $k\neq l$ and $m$ such that the sequence $(k,l,m)$ is a permutation of $(1,2,3)$] and (ii) $(\bm{a}\cdot\bm{\sigma})(\bm{b}\cdot\bm{\sigma})=(\bm{a}\cdot\bm{b})\openone_{\rm 2}+i(\bm{a}\times\bm{b})\cdot\bm{\sigma}$.}.



\section{Semidefinite Programming}\label{sec:SDP}

In this section we introduce and study a class of optimization problems denominated Semidefinite Programs. Although this is a widely developed chapter of convex optimization theory, our intention here is not to provide a comprehensive revision of the field (see \cite{04Boyd,96Vandenberghe49} for that purpose), but to concisely present the key aspects to be explored in the remainder of this thesis.

\subsection{Definition and common forms}
A Semidefinite program (SDP) is any optimization problem that can be written in the form
\begin{align}
\mbox{maximize}\quad & - \tr\left(E_0 Z\right)\nonumber\\
\mbox{subject to}\quad& Z\geq 0 \label{eq:standardform}\\
\quad&\tr\left(E_i Z\right)  = b_i\,,\quad\mbox{for}\quad i=1,\ldots,n
\end{align}
where the matrix $Z$ is the variable of the problem, $b_i$ are given real numbers and $E_0$, $E_i$ $(i=1,\ldots,n)$, are given hermitian matrices. Any SDP written in the form above is said to be in the \emph{standard form}.

The name \emph{semidefinite} program is motivated by the matrix inequality constraint $Z\geq 0$, which implies that only positive \emph{semidefinite} matrices $Z$ are acceptable as a solution of~\eqref{eq:standardform}. Apart from this restriction, the $n$ equality constraints $\tr\left(E_i Z\right)  = b_i$ completes the characterization of the \emph{feasible set} of  problem~\eqref{eq:standardform}. This is a \emph{convex set}\footnote{A convex set is defined as the set for which every convex combination of its elements is still an element of the same set.}, since it is the intersection between the cone of PSD matrices and a number of (hyper) planes determined by the equality constraints, each of them being a convex set on their own\footnote{Recall that the intersection of convex sets is another convex set.}.

The \emph{objective function} $- \tr\left(E_0 Z\right)$ is a \emph{linear} function of $Z$, and as such is simultaneously concave and convex\footnote{A convex function is any function $f$ that satisfies $f(p x_1 + (1-p) x_2)\leq p f(x_1) + (1-p) f(x_2)$ for every $x_1$ and $x_2$ in the domain of $f$ and $p\in[0,1]$. A function $f$ is concave if the reversed inequality is satisfied, or equivalently, if $-f$ is convex.}. This simple observation --- combined with the convexity of the feasible set --- enable us to recognize SDPs as a special type of convex optimization problem\footnote{A convex optimization problem consists of the minimization of a convex function over a convex set (or either the maximization of a concave function over a convex set).}. This is an important realization since in a convex optimization problem every local optimum is automatically a global optimum, therefore, numerical methods for solving SDPs do not get stuck at suboptimal solutions.

There is another common presentation of a SDP that can be derived from the standard form as follows. Without loss of generality, we can expand $Z$ on a basis of hermitian matrices (this is always possible because $Z$ is constrained to be PSD, hence it is necessarily hermitian). Explicitly, if $Z$ is a matrix of dimension ${\rm d}$, then we can write $Z=\sum_\alpha y_\alpha H_{\rm d}^{\alpha}$, where $H_{\rm d}^\alpha$ ($\alpha=1,\ldots,{\rm d}^2$) are the elements of the hermitian basis (e.g., those described in Sec.~\ref{sec:herm_basis}) and $y_\alpha$ are the ${\rm d}^2$ real coefficients of this expansion. Problem~\eqref{eq:standardform} then assumes the form
\begin{align}
\mbox{maximize}\quad & \sum_{\alpha=1}^{{\rm d}^2}y_\alpha\tr\left(-E_0 H_{\rm d}^\alpha\right)\nonumber\\
\mbox{subject to}\quad& \sum_{\alpha=1}^{{\rm d}^2} y_\alpha H_{\rm d}^{\alpha}\geq 0 \label{eq:from_std_to_ineq}\\
\quad&\sum_{\alpha=1}^{{\rm d}^2}y_\alpha\tr\left(E_i H_{\rm d}^\alpha\right)=b_i\,,\quad\mbox{for}\quad i=1,\ldots,n.\nonumber
\end{align}
The equality constraints in the last line form a linear system of $n$ equations and ${\rm d}^2$ variables, which can be solved to reduce the number of unknowns. Incorporating the solution of this linear system to the objective function and to the matrix inequality constraint, it is straightforward to see that the resulting problem is of the form
\begin{align}
\mbox{minimize}\quad & \bm{c}^{\sf T} \bm{x}\nonumber\\
\mbox{subject to}\quad& F_0+\sum_j x_j F_j \geq 0 \quad\mbox{for}\quad j=1,\ldots,m \label{eq:inequalityform}
\end{align}
which is the so-called \emph{inequality form} of a SDP. In this formulation, the variables of the problem are the entries of the vector $\bm{x}$, while the real vector $\bm{c}$ and the hermitian matrices $F_0$ and $F_j$ are given.

\subsection{Lagrange Duality}\label{sec:LagDuality}

A useful tool for the solution of a SDP is the idea of Lagrange duality, which is basically a recipe to construct another optimization problem  --- the \emph{dual problem} --- whose solution bounds the solution of the original problem --- the \emph{primal problem}. The idea of Lagrange duality is not exclusive of SDPs, but it is particularly useful in this context because the dual of a SDP is another SDP, and the bound it provides on the solution of the primal SDP is \emph{usually} tight.

\subsubsection{The dual problem}
Let us introduce the method of Lagrange multipliers to find the dual problem of problem~\eqref{eq:standardform} --- our primal problem. The key idea is to construct an upper bound to any feasible value of the primal by augmenting its objective function with a weighted sum of the constraint, as follows:
\begin{equation}\label{eq:LDF}
g(\Lambda,\nu)=\max_Z\left\{-\tr\left(E_0 Z\right)+\tr\left(\Lambda Z\right)+\sum_i\nu_i\left[\tr\left(E_i Z\right)-b_i\right]\right\}\,,
\end{equation}
where the matrix $\Lambda$ and the scalars $\nu_i$ are the weights, commonly called \emph{Lagrange multipliers} or \emph{dual variables}.

Clearly, if $Z$ is a primal feasible point, then the last sum vanishes and the term $\tr\left(\Lambda Z\right)$ is non-negative if $\Lambda \geq 0$, in which case $g(\Lambda,\nu)$ is larger than any primal feasible value. The \emph{Lagrange dual problem} (or simply the dual problem), consists of obtaining the tightest possible bound from Eq.~\eqref{eq:LDF}, that is
\begin{align}\label{eq:dualg}
\mbox{minimize}\quad & g(\Lambda,\bm{\nu})\nonumber\\
\mbox{subject to}\quad& \Lambda \geq 0
\end{align}

If we use $\sim$~to designate a feasible point and $\ast$~to designate an optimal point, then the above construction implies that
\begin{equation}\label{eq:wduality}
-\tr\left(E_0\widetilde{Z}\right)\leq -\tr\left(E_0 Z^\ast\right)\leq g(\Lambda^\ast,\nu^*)\leq g(\widetilde{\Lambda},\nu)\,,
\end{equation}
a relation know as \emph{weak duality}. We shall defer until next section a discussion regarding the usefulness of these inequalities.

We now show that the dual problem~\eqref{eq:dualg} is, in fact, a SDP. First, note that the dual objective function~\eqref{eq:LDF} can be rewritten as $g(\Lambda,\bm{\nu})=-\bm{b}^{\sf T}\bm{\nu}+\max_Z \tr\left[(\Lambda-E_0+\sum_i\nu_iE_i) Z\right]$, from which it is clear that
\begin{equation}
g(\Lambda,\nu)=\left\{\begin{array}{ccl}
-\bm{b}^{\sf T} \bm{\nu}&\mbox{if}&\Lambda=E_0-\sum_i\nu_i E_i\\
\infty&&\mbox{otherwise}.
\end{array}
\right.
\end{equation}
As a result, the dual problem~\eqref{eq:dualg} becomes
\begin{align}
\mbox{minimize}\quad & -\bm{b}^{\sf T} \bm{\nu}\nonumber\\
\mbox{subject to}\quad& E_0-\sum_i \nu_i E_i \geq 0 \quad\mbox{for}\quad i=1,\ldots,n \label{eq:dualinequalityform}
\end{align}
which is the inequality form of a SDP in the variable $\bm{\nu}$.

\subsubsection{Weak and strong duality}

We have just seen that, by construction, the dual problem imposes the ordering of Eq.~\eqref{eq:wduality} between feasible and optimal values of the primal and dual problems. In this section, we shall discuss how this observation is useful, and give conditions under which a stronger relation can be obtained.

For sake of notation, let us start restating the weak duality relation from Eq.~\eqref{eq:wduality} as follows:
\begin{equation}\label{eq:weakd}
\mathfrak{p}\leq\mathfrak{p}^\ast\leq\mathfrak{d}^\ast\leq\mathfrak{d}\,,
\end{equation}
where $\mathfrak{p}$ and $\mathfrak{d}$ represent any feasible values of the primal and dual problems, whereas $\mathfrak{p}^\ast$ and $\mathfrak{d}^\ast$ give the optimal values of each problem.

Weak duality is an invaluable tool for \emph{certifying} that ``conjectured optimal solution'' is actually optimal. To see how this works, suppose that some physical problem can be formulated as a SDP. Thanks to physical intuition and/or pattern recognition from numerical analysis, one can tailor an educated guess as to what is the general analytic solution of the problem. As long as the guess is primal feasible, we can promptly compute a value of $\mathfrak{p}$, and the problem boils down to decide whether or not $\mathfrak{p}=\mathfrak{p}^\ast$.

For that purpose, we can start by writing the dual problem. Sometimes, the dual is simpler than the primal and a formula for $\mathfrak{d}^\ast$ can be rigorously derived. In general, though, along the way of ``solving the dual'' we may need to make some hand-waving assumptions to proceed. The whole point is that if after all the assumptions we find a formula for $\mathfrak{d}$ that matches our $\mathfrak{p}$, the problem can be considered optimally solved because the only way to reconcile $\mathfrak{d}=\mathfrak{p}$ with the weak duality relation is to have $\mathfrak{p}=\mathfrak{p}^\ast$.

But, what if the solutions do not match? In this case, there are three (possibly co-existing) possibilities: 1) the guessed $\mathfrak{p}$ is not equal to $\mathfrak{p}^\ast$, 2) the assumptions made in the dual side were not good and yielded some $\mathfrak{d}\neq\mathfrak{d}^\ast$, and/or 3) the optimal duality gap $\mathfrak{d}^\ast-\mathfrak{p}^\ast$ is not zero.

For many primal problems, it is possible to eliminate the third possibility from the above list, that is, many SDPs can have the \emph{strong duality} property $\mathfrak{p}^\ast=\mathfrak{d}^\ast$ guaranteed in advance by the so-called \emph{constraint qualifications}. A very useful and simple-to-check type of constraint qualification is the \emph{Slater condition}, which states that for any convex problem we have strong duality if there is some feasible point for which the inequality constraints are satisfied with \emph{strict} inequalities. It is generally easy to find some trivial point (e.g., the identity matrix) that satisfies Slater's condition in the case of SDPs.

Strong duality is also the underlying basis of most numerical methods for the solution of SDPs \cite{04Nemirovski,96Vandenberghe49,04Boyd}. These methods iteratively generate feasible points of the primal and dual problems that go towards vanishing the \emph{duality gap} $\mathfrak{d}-\mathfrak{p}$. In practice, a small tolerance for the largest acceptable value of $\mathfrak{d}-\mathfrak{p}$ is provided, and the algorithm runs until it is achieved. When that happens, the generated feasible points of the primal and dual problems are (up to numerical precision) their optimal solutions. In practice, \emph{the interior-point algorithm} generally converges after a number of iterations between five and fifty~\cite{96Vandenberghe49}. Furthermore, a worst-case analysis reveals the theoretical complexity as a polynomial in the variables specifying the problem size \cite{04Nemirovski}. In essence, numerical solutions of SDPs can be very efficiently obtained on the basis of strong duality.

\subsubsection{Complementary Slackness}

For SDPs exhibiting the strong duality property, there is a very neat and useful relation between the optimal solution of the primal problem, $Z^\star$, and the matrix $E_0-\sum_i \nu_i^\ast E_i$ from the inequality constraint of the dual problem. This relation is called \emph{complementary slackness}, and is obtained as follows:
\begin{multline}
\mathfrak{p}^\ast=\mathfrak{d}^\ast\Rightarrow -\tr\left(E_0 Z^\ast\right)=-\sum_i b_i \nu_i^\ast\Rightarrow -\tr\left(E_0 Z^\ast\right)=-\sum_i \tr\left(E_i Z^\ast\right)\nu_i^\ast\\\Rightarrow
\tr\left[\left(E_0-\sum_i \nu_i^\ast E_i\right)Z^\ast\right]=0\Rightarrow
\left(E_0-\sum_i \nu_i^\ast E_i\right)Z^\ast=0\,,
\end{multline}
where the last implication follows from the fact that both $E_0-\sum_i \nu_i^\ast E_i$ and $Z^\ast$ are PSD matrices according to the constraints of the dual and primal problems, respectively.\par

Going on with the idea from the previous section of obtaining an optimality certificate for a candidate primal solution, the complementary slackness relation can be of assistance in the following sense: instead of attempting to solve the dual problem, we can solve the set of linear equations $(E_0-\sum_i\nu_i E_i)\widetilde{Z}=0$ for the variables $\nu_i$, with $\widetilde{Z}$ our guess of the primal problem. If the resulting values $\widetilde{\nu}_i$ give rise to a matrix $E_0-\sum_i\widetilde{\nu}_i E_i$ that is positive semidefinite, then we have guaranteed that the candidate is optimal; otherwise, strong duality guarantees that this is not so. Of course, proving positive semidefiniteness of an analytical matrix can still be a challenging task.

\section{Optimal Quantum Control from Semidefinite Programs?}\label{sec:controlSDP}

In this section we begin to formalize the quantum control task of converting a sequence of source density matrices into a sequence of target density matrices via a single quantum channel, as illustrated in Fig.~\ref{fig:singlestep_intro}.

We proceed by first motivating the convertibility problem in the context of quantum information. Then we give some first steps towards a mathematical formulation of the problem. Although a complete formulation will be only achieved in Ch.~\ref{chap:assemble}, the goal here is to introduce the basic technical ideas and to illustrate the adequacy of the SDP formalism for the optimization of quantum operations.

\subsection{Why transforming between sequences of density matrices?}

In quantum information science, the problem of transforming between sequences of density matrices is relevant because it
immediately connects to the problem of transforming a quantum system whose initial preparation is uncertain. Many ``no-go theorems'' for transformations of uncertain inputs are known, as well as quantum machines that attempt to implement these transformations ``as well as possible''.

In general, the construction of optimal quantum machines is conducted in a \emph{ad hoc} basis, envisaging a specific type of impossible transformation. Adopting this approach, optimal schemes for discriminating between non-orthogonal quantum states, cloning, purifying, error correcting, etc. have been designed. The realization that all of these problems can be phrased in terms of transformations between sequences of density matrices leads to a more general problem that encloses the fundamental barriers imposed by quantum mechanics, and whose solution provides a generalized optimal machine that can be applied in many different circumstances.

\subsection{A special case: single-state transformation}\label{sec:singlestate}
As a first exposure to the problem of converting between sequences of density matrices, consider the simplest case where both source and target sequences have a single element ($I=1$). This will lead to a trivial but instructive solution. Our aim is then to determine a quantum channel that transforms a given {\rm d}-dimensional density matrix $\rho$ into another {\rm d}-dimensional density matrix $\overline{\rho}$. In what follows, we prove (by construction) that a channel implementing this conversion always exists, regardless of the specific details of the states $\rho$ and $\overline{\rho}$.\par

Let $\{\ket{j}\}$ be any set of pure states forming a resolution of the identity, i.e., $\sum_{j} \ket{j}\!\bra{j} = \openone_{\rm d}$. For any ${\rm d}$-dimensional density matrix $\overline{\rho}$, we can write
\begin{equation}
\overline{\rho}=\sum_j a_j \ket{j}\!\bra{j}\,,
\end{equation}
 for some set of real coefficients $a_j$ such that $\sum_j a_j = 1$. Now, let the matrices $A_{j,k}=\sqrt{a_j}\ket{j}\!\bra{k}$ represent the Kraus elements of a map $\mathcal{Q}$, as follows:
\begin{equation}\label{eq:Csingleconv}
\mathcal{Q}(\odot)=\sum_{i,j} A_{j,k}\odot A_{j,k}^\dagger\,.
\end{equation}
That $\mathcal{Q}$ is a CP map is immediate from its construction via a Kraus decomposition. Furthermore, $\mathcal{Q}$ is also trace preserving: by exploiting the unit-sum property of the coefficient $a_j$ and the fact that $\{\ket{j}\!\bra{j}\}$ resolves the identity, it is easy to check that $\sum_{j,k} A_{j,k}^\dagger A_{j,k} = \openone_{\rm d}$. A straightforward calculation then shows that $\mathcal{Q}$ perfectly converts $\rho$ into $\overline{\rho}$:
\begin{equation}
\mathcal{Q}(\rho)=\sum_{j,k}\sqrt{a_j} \ket{j}\!\bra{k} \rho \sqrt{a_j} \ket{k}\!\bra{j}=\sum_{j,k}a_j \ket{j}\!\bra{j} \bra{k}\rho\ket{k}=\overline{\rho}\;\tr\left(\rho \sum_k\ket{k}\!\bra{k}\right) = \overline{\rho}\,,
\end{equation}
where, in the last equation, we have used the completeness relation for the sum and the normalization of the density matrix $\rho$.\par

Remarkably, the set of Kraus operators $A_{j,k}$ of $\mathcal{Q}$ is independent of the initial state $\rho$, but exclusively defined in terms of $\brho$. From a physical viewpoint, this means that $\mathcal{Q}$ does not literally \emph{transform} $\rho$ into $\overline{\rho}$, but instead \emph{constructs} $\overline{\rho}$ ``from the scratch'', by completely ignoring the original state $\rho$. The possibility of neglecting available information and yet implementing the conversion with arbitrarily high precision is not a feature inherited by more general convertibility problems. In fact, as shown in Appendix~\ref{app:ptc}, a theorem by Alberti and Uhlmann sets necessary conditions for the existence of a quantum channel accurately transforming between \emph{pairs} of density matrices. Moreover, even when such conditions are met, the resulting channels usually depend on the details of both source and target states. The derivation of these channels is one of the main objectives of this thesis. In the next section, we give the first steps towards their determination via the solution of certain optimization problems.

\subsection{The general case: multi-state transformation}\label{sec:begin_optimize}

In this section, we start to assemble ``mathematical devices'' that take as inputs the sequences of source and target density matrices and outputs a channel of a certain type that implements the conversion between them. In the case where the desired transformation turns out to be physically impossible, such a device is constructed to output (i) a channel that optimally approximates the unphysical transformation and (ii) a number that quantifies the quality of this approximation.

A general realization of the ``device'' we are talking about is an optimization problem of the form
\begin{equation}\label{eq:gen_problem}
\min_{\mathcal{C} \in \mathcal{C}_{\rm d}^{\rm set}}\aver{\mathscr{D}\left[\mathcal{C}(\rho_i),\brho_i\right]}\,,
\end{equation}
%
where $\mathcal{C}^{\rm set}_{\rm d}$ specifies a family of maps formed by all the admissible controllers $\mathcal{C}$, and $\aver{\mathscr{D}\left[\mathcal{C}(\rho_i),\brho_i\right]}$ is some notion of distance between the sequences $\left[\mathcal{C}(\rho_i)\right]_{i=1}^I$ and $\left[\brho_i\right]_{i=1}^I$. Well suited choices of $\aver{\mathscr{D}}$ should yield a solution $\mathcal{C}$ of problem~\eqref{eq:gen_problem} such that $\aver{\mathscr{D}}$ vanishes whenever the target sequence can be approximated with arbitrary accuracy via an element of $\mathcal{C}_{\rm d}^{\rm set}$. When that is not the case, the resulting operation $\mathcal{C}$ yields a non-zero value of $\aver{\mathscr{D}}$ which is interpreted as an estimate of the quality of the best available approximation. Possible measures $\aver{\mathscr{D}}$ will be constructed and discussed in the next chapter.

Strictly speaking, problem~\eqref{eq:gen_problem} represents a large class of optimization problems labeled by $\mathcal{C}_{\rm d}^{\rm set}$ and $\aver{\mathscr{D}}$. In the remainder of this chapter, we let $\aver{\mathscr{D}}$ be arbitrary and focus on the subclasses of problems where the feasible set $\mathcal{C}_{\rm d}^{\rm set}$ is taken to be either the set of CPTP maps, $\mathcal{Q}_{\rm d}^{\rm set}$, or the set of EBTP maps, $\mathcal{B}_{\rm d}^{\rm set}$. In particular, aiming for efficient solution of these problems, we attempt to make the constraints $\mathcal{C}\in\mathcal{Q}_{\rm d}^{\rm set}$ and $\mathcal{C}\in\mathcal{B}_{\rm d}^{\rm set}$ as similar as possible to those constraints appearing in the standard and inequality forms of semidefinite programs.

\subsubsection{The constraint $\mathcal{C}\in\mathcal{Q}_{\rm d}^{\rm set}$ (CPTP maps)}\label{sec:constraintQ}

Thanks to the isomorphism expressed by Eq.~\eqref{eq:CPTPconstraints}, the constraint of problem~\eqref{eq:gen_problem} with $\mathcal{C}_{\rm d}^{\rm set}=\mathcal{Q}_{\rm d}^{\rm set}$ can be written in terms of the Choi matrix $\mathfrak{C}$ as follows
\begin{equation}\label{eq:probSDP_first}
\begin{array}{rl}
\mbox{minimize}&\aver{\mathscr{D}\left[\mathcal{C}(\rho_i),\brho_i\right]}\\
\mbox{subject to}& \mathfrak{C}\geq 0\\
&{\rm Tr}_2\mathfrak{C}=\openone_{\rm d}\,.
\end{array}
\end{equation}
Furthermore, employing the same reasoning used to derive Eq.~\eqref{eq:solvedlineqconst}, the equality constraint $\tr_2{\mathfrak{C}}=\openone_{\rm d}$ can be rewritten to give
\begin{equation}\label{eq:probSDP_step0}
\begin{array}{rl}
\mbox{minimize}&\aver{\mathscr{D}\left[\mathcal{C}(\rho_i),\brho_i\right]}\\
\mbox{subject to}& \mathfrak{C}\geq 0\\
&\tr\left[\mathfrak{C}\left(H_{\rm d}^\alpha\otimes\openone_{\rm d}\right)\right]={\rm d}\delta_{\alpha,1}\quad\mbox{for}\quad \alpha=1,\ldots,{\rm d}^2\,,
\end{array}
\end{equation}
where $\{H_{\rm d}^\alpha\}_{\alpha=1}^{{\rm d}^2}$ is an orthonormal basis of hermitian matrices of dimension {\rm d}, as defined in Sec.~\ref{sec:herm_basis}.

Clearly, the constraints of the optimization problem above closely resemble those of a SDP in the standard form [cf. Eq.~\eqref{eq:standardform}]: while the Choi matrix $\mathfrak{C}$ plays the role of the matrix variable $Z$, the equality constraints are identical to those of Eq.~\eqref{eq:standardform} under the identifications $E_\alpha=H_{\rm d}^\alpha\otimes\openone_{\rm d}$ and $b_\alpha={\rm d}\delta_{\alpha,1}$.

 Following Eq.~\eqref{eq:choiexpandTPin}, we can also explicitly solve the equality constraints to find
\begin{equation}\label{eq:probSDP_step0ineq}
\begin{array}{rl}
\mbox{minimize}&\displaystyle\aver{\mathscr{D}\left[\mathcal{C}(\rho_i),\brho_i\right]}\\
\mbox{subject to}& \displaystyle\frac{\openone_{{\rm d}^2}}{\rm d}+\sum_{\substack{\alpha=1\\\beta=2}}^{{\rm d}^2} x_{\alpha,\beta} H_{\rm d}^{\alpha}\otimes H_{\rm d}^\beta\geq 0\,.
\end{array}
\end{equation}
Now, the optimization constraint resembles those of a SDP in the inequality form [cf. Eq.~\eqref{eq:inequalityform}]: the variables are the real coefficients $x_{\alpha,\beta}$ that expand $\mathfrak{C}$ in the product basis $\{H_{\rm d}^\alpha\otimes H_{\rm d}^\beta\}_{\alpha,\beta=1}^{\rm d}$, while the matrix inequality of Eq.~\eqref{eq:inequalityform} is recognized via the identification $F_0=\openone_{{\rm d}^2}/{\rm d}$ and $F_{\alpha,\beta}=H_{\rm d}^\alpha\otimes H_{\rm d}^\beta$.

\subsubsection{The constraint $\mathcal{C}\in\mathcal{B}_{\rm d}^{\rm set}$ (EBTP maps)}\label{sec:constraintB}

If we make $\mathcal{C}_{\rm d}^{\rm set}=\mathcal{B}_{\rm d}^{\rm set}$ in problem~\eqref{eq:gen_problem}, then the resulting optimization problem can be written as problems~\eqref{eq:probSDP_step0} or~\eqref{eq:probSDP_step0ineq} with the added constraint of  $\mathfrak{C}$ being separable across the partition ${\rm d}\otimes{\rm d}$, as explained in Sec.~\ref{sec:isomorphisms}. Given the types of constraints occurring in a SDP, we are led to ask whether the separability of a bipartite matrix can be expressed with a finite number of linear matrix inequalities. Fortunately, the search for separability criteria is a currently active research field; unfortunately, a complete answer to our question is still an open problem.

In a seminal paper by the Horodecki family \cite{96Horodecki1}, it was proved that a bipartite PSD matrix $A$ of dimension ${\rm d}_1{\rm d}_2$ is separable across the partition ${\rm d_1}\otimes{\rm d_2}$ if and only if
  \begin{equation}\label{eq:HorodeckiLMI}
    (\mathcal{I}_{{\rm d}_1}\otimes\mathcal{P})A \geq 0
  \end{equation}
for every positive map $\mathcal{P}$ acting on $\mathcal{M}_{{\rm d}_2}$. Since the condition above is trivially satisfied by every CP map, it has only to be required for positive maps $\mathcal{P}$ which are not completely positive (PnCP maps). A standard example of such a map is transposition.

In Ref.~\cite{96Peres1413}, Peres showed that for $2 \otimes 2$ systems, the condition $(\mathcal{I}_{{\rm d}_1}\otimes{\sf T})A\geq 0$ --- or in more standard notation, $A^{{\sf T}_2}\geq 0$ --- is not only necessary, but also sufficient to ensure the separability of $A$ across ${\rm d_1}\otimes{\rm d_2}$. The same conclusion was independently drawn in Ref.~\cite{96Horodecki1}, where the sufficiency clause was extended to $2\otimes 3$ and $3\otimes 2$ systems and, sadly, shown not to hold for larger dimensions.\par

For our purposes, this result provides the following equivalent expression of problem~\eqref{eq:gen_problem} when the elements of source and target sequences are qubit density matrices (${\rm d}=2$) and $\mathcal{C}_2^{\rm set}=\mathcal{B}_2^{\rm set}$:
\begin{equation}
\begin{array}{rl}
\mbox{minimize}&\aver{\mathscr{D}\left[\mathcal{C}(\rho_i),\brho_i\right]}\\
\mbox{subject to}& \mathfrak{C}\geq 0\,,\quad \mathfrak{C}^{{\sf T}_2}\geq 0\\
&{\rm Tr}_2\mathfrak{C}=\openone_{\rm 2}\,,
\end{array}
\end{equation}
which is just a restatement of the optimization problem~\eqref{eq:probSDP_first} with the addition of the positive partial transpose (PPT) condition. Once again employing the expansion of $\mathfrak{C}$ from Eq.~\eqref{eq:choigenbasis} and solving the equality constraint, the problem above becomes
\begin{equation}\label{eq:probSDP_step0ineqEBTP}
\begin{array}{rl}
\mbox{minimize}&\displaystyle\aver{\mathscr{D}\left[\mathcal{C}(\rho_i),\brho_i\right]}\\
\mbox{subject to}& \displaystyle\frac{\openone_{8}}{2}+\sum_{\substack{\alpha=1\\\beta=2}}^{4} x_{\alpha,\beta} \left[\left(H_2^{\alpha}\otimes H_2^\beta\right)\oplus\left(H_2^{\alpha}\otimes {H_2^\beta}^{\sf T}\right)\right]\geq 0\,.
\end{array}
\end{equation}
where we have made use of the fact that the direct sum of two matrices is a PSD matrix if and only if each matrix is PSD. Clearly, the resulting constraint is characteristic of a SDP in the inequality form with $F_0=\openone_8/2$ and $F_{\alpha,\beta}=\left(H_2^{\alpha}\otimes H_2^\beta\right)\oplus\left(H_2^{\alpha}\otimes {H_2^\beta}^{\sf T}\right)$.

What about larger dimensional systems? Can we still write the optimization constraint in the form of a linear matrix inequality? Formally, the Horodecki condition~\eqref{eq:HorodeckiLMI} allows the following expression for arbitrary dimension {\rm d}:
\begin{equation}\label{eq:probSDP_step1ineqEBTP}
\begin{array}{rl}
\mbox{minimize}&\displaystyle\aver{\mathscr{D}\left[\mathcal{C}(\rho_i),\brho_i\right]}\\
\mbox{subject to}& \displaystyle\frac{1}{\rm d}\left[\openone_{{\rm d}^2}\oplus\bigoplus_{\mathcal{P}\in\{PnCP\}}\openone_{\rm d}\otimes\mathcal{P}(\openone_{\rm d})\right]+\\
&\displaystyle\sum_{\substack{\alpha=1\\\beta=2}}^{{\rm d}^2} x_{\alpha,\beta} \left[\left(H_{\rm d}^{\alpha}\otimes H_{\rm d}^\beta\right)\oplus\bigoplus_{\mathcal{P}\in\{PnCP\}} H_{\rm d}^{\alpha}\otimes \mathcal{P}\left(H_{\rm d}^\beta\right)\right]\geq 0\,.
\end{array}
\end{equation}
However, unless the direct sum can be restricted to just a few instances of PnCP maps (as the transposition, in the {\rm d=2} case), the problem above is rather useless due to its infinite sized matrix constraint. Alas, the existence of a finite set (possibly dependant on {\rm d}) of PnCP maps that provides a sufficient separability constraint is still an open problem \cite[p. 29]{04Zyczkowski3}.

Nevertheless, for practical applications in larger dimensional systems, it is usually a good first step to consider problem~\eqref{eq:probSDP_step1ineqEBTP} with the direct sum restricted to the transposition map
\begin{equation}\label{eq:probSDP_step2ineqEBTP}
\begin{array}{rl}
\mbox{minimize}&\displaystyle\aver{\mathscr{D}\left[\mathcal{C}(\rho_i),\brho_i\right]}\\
\mbox{subject to}& \displaystyle\frac{\openone_{2{\rm d}^2}}{\rm d}+\sum_{\substack{\alpha=1\\\beta=2}}^{{\rm d}^2} x_{\alpha,\beta} \left[\left(H_{\rm d}^{\alpha}\otimes H_{\rm d}^\beta\right)\oplus\left(H_{\rm d}^{\alpha}\otimes {H_{\rm d}^\beta}^{\sf T}\right)\right]\geq 0\,.
\end{array}
\end{equation}
Of course, the problem above is just a \emph{relaxation} of problem~\eqref{eq:gen_problem} with $\mathcal{C}_{\rm d}^{\rm set}=\mathcal{B}_{\rm d}^{\rm set}$ and, as such, will generally yield a Choi matrix that is a (unnormalized) PPT-entangled state rather than the desired separable $\mathfrak{C}$. Throughout, we will denote by $\widetilde{\mathcal{B}}_{\rm d}^{\rm set}$ the set of CPTP maps whose Choi matrices satisfy the PPT condition, which will allow us to refer to the problem above as an instance of problem~\eqref{eq:gen_problem} with $\mathcal{C}_{\rm d}^{\rm set}=\widetilde{\mathcal{B}}_{\rm d}^{\rm set}$.

Replacing the constraint $\mathcal{C}\in\mathcal{B}_{\rm d}^{\rm set}$ with  $\mathcal{C}\in\widetilde{\mathcal{B}}_{\rm d}^{\rm set}$ is not only useful because it gives a treatable matrix inequality, but also because of the simple observation that we do not always need to tightly bound $\mathcal{B}_{\rm d}^{\rm set}$ in order to obtain a solution which is an element of it. This idea is particularly useful if, \emph{a posteriori}, we can check whether the outcome of the relaxed optimization belongs to the set of interest.

Deciding whether a given state is separable or PPT-entangled state can be remarkably easy or difficult, depending on the particular state under consideration. For example, if the outcome $\mathfrak{C}$ of the optimization~\eqref{eq:probSDP_step2ineqEBTP} turns out to be of rank $2$ or $3$, then we can be sure of its separability, since there are no PPT-entangled states with such ranks \cite{00Horodecki032310,03Horodecki589}. More generally, making this decision has been shown to be a NP-hard problem~\cite{03Gurvits10}. Yet, a number of algorithms performing efficiently in many non-trivial cases do exist \cite{04Doherty022308,04Eisert062317,07Ioannou}.

\chapter{Distance Measures}\label{chap:dist_measure}

In this chapter we present a detailed analysis of several distance measures for the space of density matrices. Apart from being a topic of independent interest, this analysis is relevant in the context of this thesis because it provides a repertoire of distance measures that --- used in the quantum control problem we started to formulate in the previous chapter --- yield optimization problems that are both physically meaningful and mathematically treatable.

We divide our analysis between closeness measures (e.g., the fidelity), and distance measures (e.g., the trace distance). For closeness measures, we focus on reviewing several useful properties of the Uhlmann-Jozsa fidelity, and propose an alternative definition of this quantity that features many attractive properties.

A similar discussion is presented for metrics on density matrices. In particular, we consider the trace distance, and the metrics induced by the Hilbert-Schmidt norm and the spectral (or operator) norm, which have their merits assessed in the same basis as the fidelity-like measures.

All of these quantities are put at work in the next chapter, where they appear to quantify the distance between sequences of density matrices. By minimizing distances or maximizing closeness, we will obtain optimal strategies for quantum control.

\section{Introduction}
From a mathematically rigorous viewpoint, a distance measure $\mathscr{D}$ on a set $\mathbb{S}$ is a function $\mathscr{D}: \mathbb{S} \times \mathbb{S} \to \mathbb{R}$ such that for every $a,b,c \in \mathbb{S}$ the following properties hold:

\begin{enumerate}\label{axioms:metric}
\item[(M1)] $\mathscr{D}(a,b)\geq 0$ (Nonnegativity)\,,
\item[(M2)] $\mathscr{D}(a,b)=0$ iff $a=b$ (Identity of Indiscernibles)\,,
\item[(M3)] $\mathscr{D}(a,b)=\mathscr{D}(b,a)$ (Symmetry)\,,
\item[(M4)] $\mathscr{D}(a,c)\leq \mathscr{D}(a,b)+\mathscr{D}(b,c)$ (Triangle Inequality)\,.
\end{enumerate}
Any such function is called a \emph{metric}

In physics, though, it is common to talk about distance measures that are not metrics. For example, the quantum relative entropy~\cite{62Umegaki59,00Schumacher0004045,02Vedral197} defined on the set $\mathbb{S}=\mathcal{S}(\textsf{H}_{\rm d})$ of {\rm d}-dimensional of density matrices as
\begin{equation}
\EuScript{S}(\rho,\sigma)\defeq\tr\left[\rho\log\rho\right]-\tr\left[\rho\log\sigma\right]\,,\label{eq:relentropy}
\end{equation}
is widely accepted as a distance measure between $\rho$ and $\sigma$, despite not satisfying (M3) nor (M4)\footnote{In fact, $\mathcal{S}(\rho,\sigma)$ is not even well defined for every element of $\mathbb{S} \times \mathbb{S}$ --- consider, for example, $\rho$ any mixed state and $\sigma$ any pure state to see the second term of Eq. (\ref{eq:relentropy}) to go to $-\infty$.}. The Uhlmann-Jozsa fidelity $\F(\rho,\sigma)$ \cite{76Uhlmann273,94Jozsa2315} --- one of the most popular notions of distance between quantum states --- is not technically a metric either: First, it behaves as an ``inverted measure of distance'' being maximal (1) when $\rho=\sigma$ and minimal (0) when $\rho$ and $\sigma$ are orthogonal, in clear disagreement with (M2). Second, even the ``inversion'' $\F^\prime\defeq 1-\F$ which complies with (M2) violates (M4), and is thus not a metric either.

For several quantum information applications, fulfillment of (M1)-(M4) is not as relevant as a compelling operational interpretation for the ``distance measure'' at hand. In this respect, the \emph{quantum hypothesis testing problem} serves as a justification for many standard notions of distance between quantum states.\par

Consider that $n$ copies of a quantum system are identically prepared in one of the states $\rho$ or $\sigma$. One is then asked to perform a measurement on the system in order to reject a certain hypothesis about its identity (the \emph{null hypothesis}) and accept another one (the \emph{alternate hypothesis})
\begin{itemize}
\item[$H_0$] (Null hypothesis): The collective state of the system is $\rho^{\otimes n}$;
\item[$H_1$] (Alternate hypothesis): The collective state of the system is $\sigma^{\otimes n}$.
\end{itemize}
If the measurement procedure suggests rejection of the null hypothesis despite it being correct, we talk about an \emph{error of the first kind}, which occurs with probability $\alpha$. On the other hand, the acceptance of the null hypothesis despite its incorrectness is called an \emph{error of the second kind}, which occurs with probability $\beta$.

Depending on the number $n$ of states available, on the adopted measurement strategy and on whether or not $\alpha$ and $\beta$ are treated symmetrically, different quantities may arise to specify the error probabilities. Intuitively, these quantities form a meaningful ``measure of closeness'' between quantum states, since the ``closer'' $\rho$ and $\sigma$ are, the more likely it is to confuse them.

The relative entropy, for example, arises from the minimal error probability $\beta$ when $\alpha$ is constrained to be smaller than a constant threshold $\epsilon$ and $n\to\infty$. In this case, the quantum analogue of Stein's lemma~\cite{91Hiai99,00Ogawa2428} establishes that
\begin{equation}
\beta\sim\exp{\left[-n \EuScript{S}(\rho,\sigma)\right]}\,,
\end{equation}
hence characterizing $\EuScript{S}$ as a measure of distance.

Likewise, if the error probabilities are treated symmetrically, i.e., $\alpha$ and $\beta$ are simultaneously minimized, then the optimal measurement strategy gives rise to the recently determined~\cite{06Nussbaum0607216,07Audenaert160501} \emph{quantum Chernoff bound} $\xi\defeq -\log \min_{0\leq s \leq 1} \tr(\rho^s\sigma^{1-s})$, via the error probability asymptotic behavior
\begin{equation}
\alpha+\beta\sim \exp{\left[- n \xi(\rho,\sigma)\right]}\,.
\end{equation}

In the case where a single copy of the states is available to generate the measurement statistics ($n=1$), the minimal error probability (symmetrically treated) was shown long ago~\cite{76Helstrom,79Holevo411} to be related to the \emph{trace distance} $\D$ between $\rho$ and $\sigma$ (defined in Sec.~\ref{sec:fidtrdistbnd}) via
\begin{equation}
\alpha+\beta=\tfrac{1}{2}-\D(\rho,\sigma)\,.
\end{equation}

The fidelity between a mixed state $\rho$ and a \emph{pure state} $\sigma=\ket{\psi}\!\bra{\psi}$, $\F(\rho,\sigma)\defeq\bra{\psi}\rho\ket{\psi}$, also arises from the quantum hypothesis testing problem with $n=1$ and measurement effects given by $M_0=\ket{\psi}\!\bra{\psi}$ and $M_1=\openone_{\rm d}-\ket{\psi}\!\bra{\psi}$. A little thought shows that if we adopt the convention that a click of $M_0$ suggests acceptance of $\rho$ and a click of $M_1$ suggests acceptance of $\sigma$, then the probability of an error of the first type vanishes, i.e., $\alpha=0$. An error of the second type can occur, though, if $M_0$ clicks when the actual state was $\sigma$. This happens with probability
\begin{equation}
\beta=\F(\ket{\psi},\rho)\,,
\end{equation}

Besides a satisfying physical interpretation, other physical and mathematical properties of distance measures are desirable for certain specific applications. In this chapter we evaluate a number of distance measures against a list of such properties, namely, nonnegativity, identity of indiscernibles, symmetry, unitary invariance, convexity/concavity properties, multiplicativity under tensor products, monotonicity under CPTP maps, relations with metrics and computability.

Before outlining the structure of this chapter, let us mention two notational points plus a warning that will be pertinent for what follows.

\begin{itemize}
\item The notation $\sqrt{A}$ denotes the unique PSD matrix such that $\sqrt{A}^\dagger\sqrt{A} = A$. Clearly, $\sqrt{A}$ can only exist if $A$ is PSD and is obtained from the following procedure: Let $U$ be the unitary matrix that diagonalizes $A$, i.e., $U^\dagger A U = D$, with $D$ a diagonal matrix of non-negative elements. Then $\sqrt{A}=U\sqrt{D}U^\dagger$ where $\sqrt{D}$ is the entry-wise square root of $D$. Note that there many matrices $B$ (other than $\sqrt{A}$) such that $B^\dagger B = A$. For example, the Cholesky decomposition guarantees that for every PSD matrix $A$, there exists an upper triangular matrix $B$ such that $B^\dagger B=A$. However, $B$ is obviously not hermitian, let alone PSD.

\item For an arbitrary matrix $A$, $|A|$ denotes the matrix $\sqrt{A^\dagger A}$. If $A$ is hermitian, then we have the identity $\sqrt{A^2}=|A|$, which mimics the standard equation $\sqrt{x^2}=|x|$ for real numbers $x$. Likewise, from the definition of the square root given above, we have $|A| \geq 0$ for every matrix $A$, just like $|x|\geq 0$ for all real numbers $x$. However, recall that the statement of a matrix followed by the symbols $\geq 0$ means that the matrix has only non-negative eigenvalues, i.e., it is positive semidefinite.

\item Be aware that the product of two hermitian matrices $A$ and $B$ is not hermitian unless $[A,B]=0$. Similarly, the product of two positive semidefinite matrices  $A$ and $B$ is not positive semidefinite unless $[A,B]=0$. This implies, for example, that the equality $\sqrt{AB}=\sqrt{A}\sqrt{B}$ only makes sense
    for commuting PSD matrices $A$ and $B$.
\end{itemize}

Section~\ref{sec:altfid} --- which is an adapted reproduction of Ref. \cite{08Mendonca1150} --- starts by describing how the Uhlmann-Jozsa fidelity behaves with respect to the several properties mentioned before, and subsequently introduces a new alternative definition of fidelity between mixed states that is thoroughly analyzed along the same lines. In section \ref{sec:metrical} we review some metrics arising from well known norms on the set of density matrices, namely the Hilbert-Schmidt norm, the spectral norm and the trace norm. Sec.~\ref{sec:distset} concludes the chapter with the introduction of two averaging schemes that generalize the notion of a distance measure between density matrices to a distance measure between \emph{sequences} of density matrices.

\section[An alternative fidelity measure for quantum states]{An alternative fidelity measure for quantum states}\label{sec:altfid}

The understanding of the set of density matrices as a
Riemannian manifold \cite{06Bengtsson} implies that a notion of
distance can be assigned to any pair of quantum states. In
quantum information science, for instance, distance measures
between quantum states have proved to be useful resources in
approaching a number of fundamental problems such as
quantifying entanglement \cite{97Vedral2275,98Vedral1619}, the
design of optimized strategies for quantum control
\cite{07Branczyk012329,08Mendonca012319} and quantum error correction
\cite{05Reimpell080501,07Fletcher012338,06Reimpell,
06Kosut,08Kosut020502,07Yamamoto012327,05Yamamoto022322}. In
addition, the concept of distinguishability between quantum
states \cite{95Fuchs} can be made mathematically rigorous and
physically insightful thanks to the close relationship between
certain metrics for the space of density matrices and the error
probability arising from various versions of the quantum
hypothesis testing problem \cite{06Hayashi}. Distance measures
are also regularly used  in the laboratory to verify the
quality of the produced quantum states.

A widely used distance measure in the current literature (or more
precisely, a ``closeness'' measure between two general density
matrices), is the so-called \emph{Uhlmann-Jozsa fidelity}, $\F$.
Historically, this measure had its origins in the 70's through a set of
works by Uhlmann and Alberti \cite{76Uhlmann273,83Alberti5,
83Alberti25,83Alberti107}, who studied the problem of generalizing the
quantum mechanical transition probability to the broader context of
$\ast$-algebras. The usage of the term \emph{fidelity} to designate
Uhlmann's transition probability formula is much more recent and
initiated in the works of Schumacher~\cite{95Schumacher2738} and
Jozsa~\cite{94Jozsa2315}. Indeed, in an attempt to quantify the
``closeness'' between a certain mixed state $\rho$ and a pure state
$\ket{\psi}$, Schumacher dubbed the transition probability
$\bra{\psi}\rho\ket{\psi}$ the fidelity between the two states. In
parallel, Jozsa recognized Uhlmann's transition probability formula as
a sensible extension of Schumacher's fidelity, where now the measure of
``closeness'' is related to a pair of mixed states $\rho$ and $\sigma$.
Ever since, Uhlmann's transition probability formula has been widely
accepted as {\em the} generalization of Schumacher's fidelity.

The prevalence of this measure as one of the most used notions of
distance in quantum information is not accidental, but largely
supported on a number of required and desired properties for the role.
For example, $\F$ satisfies all of \emph{Jozsa's axioms}, that is,
besides recovering Schumacher's fidelity in the case where one of the
states is pure, the following three additional properties also hold:
First, $\F$ equals unity if and only if it is applied to two
identical states; in other cases it lies between zero and one. Second,
it is symmetric, i.e., the fidelity between $\rho$ and $\sigma$ is the
same as that between $\sigma$ and $\rho$. Third, it is invariant under
any unitary transformation on the state space. Nevertheless, $\F$ is
\emph{not} the unique measure satisfying these properties. A prominent
alternative which also complies with Jozsa's axioms and shares many
other properties of $\F$, is given by the nonlogarithmic
variety of the quantum Chernoff bound, $Q$, recently determined in
Ref.~\cite{07Audenaert160501}.\par

Despite fulfilling the properties listed above, both $\F$ and
$Q$ are, in general, unsatisfying measures from a practical
computational viewpoint. Although $\F$ can be expressed in a closed
form in terms of $\rho$ and $\sigma$, it involves successive
computation of the square roots of Hermitian matrices, which often
compromises its usage in analytical computations and numerical
experiments, especially when the fidelity must be computed many times.
Even more serious is the case of $Q$, which to date has only been
defined variationally as the result of an optimization problem. The
question that naturally arises is whether an easy-to-compute
generalization of Schumacher's fidelity can be obtained. In the following,
we provide a positive answer to this question and a thorough analysis
of our proposed alternative fidelity, $\Fn$.

As we were finalizing Ref.~\cite{08Mendonca1150}, we became aware of a very
recent work of Miszczak~\etal \cite{08Miszczak} in which $\Fn$ was
introduced as an upper bound to the Uhlmann-Jozsa fidelity.  In many
ways our analysis of $\Fn$ is complimentary to that provided in
Ref.~\cite{08Miszczak}; results in common are noted in the
corresponding sections of this thesis.

The following sections are structured as follows. In order to provide a concrete
ground for our proposal of $\Fn$ as an alternative fidelity measure,
we firstly revisit, in Sec.~\ref{sec:fid}, a set of basic properties of
the Uhlmann-Jozsa fidelity. In Sec.~\ref{sec:new_meas} we formally
introduce $\Fn$ and analyze it in the spirit of the properties reviewed
in Sec.~\ref{sec:fid}.  The computational efficiency of $\Fn$ is
contrasted with a number of previously known distance measures in
Sec.~\ref{sec:computability}.  We summarize our main results and
discuss some possible avenues for future research in
Sec.~\ref{sec:conclusion}.

\subsection{The Uhlmann-Jozsa Fidelity}\label{sec:fid}

In this section, we will briefly survey some physically appealing
features inherent to the  Uhlmann-Jozsa fidelity $\F$. In
Sec.~\ref{sec:new_meas}, these features will be used as a reference for
characterizing the proposed new fidelity measure.

\subsubsection{Preliminaries}\label{sec:fidprelim}

The fidelity $\F$ was originally introduced as a
transition probability between two generic quantum states $\rho$ and
$\sigma$~\cite{76Uhlmann273},

\begin{equation}\label{eq:fid}
    \F(\rho,\sigma)\defeq\max_{\ket{\psi},\ket{\varphi}}{|\langle\psi |
    \varphi\rangle|^2}= \left(\tr\sqrt{\sqrt{\rho}\sigma \sqrt{\rho}}\right)^2\,.
\end{equation}
Here, $\ket{\psi}$ and $\ket{\varphi}$ are restricted to be
purifications of $\rho$ and $\sigma$, while the second equality indicates that
the maximization procedure can be explicitly evaluated.
At this stage, it is worth noting that it is not uncommon to find
$\sqrt{\F}$ being referred, instead, as the fidelity (e.g.,,
Ref.~\cite{00Nielsen}).\par

In Ref.~\cite{94Jozsa2315}, Jozsa conjectured that Eq.~\eqref{eq:fid}
was the unique expression that satisfies a number of natural properties
expected for any generalized notion of fidelity\footnote{Although, as mentioned before, this conjecture can be
    seen to be  false with the counter-example of the nonlogarithmic
    variety of the quantum  Chernoff bound $Q$, determined in
    Ref.~\cite{07Audenaert160501}.}.
Throughout, we shall refer to these as \emph{Jozsa's axioms}:

\begin{enumerate}
    \item normalization, i.e., $\F(\rho,\sigma) \in [0,1]$ with the
        upper bound attained iff $\rho=\sigma$ (the \emph{identity
        of indiscernible} property);
    \item symmetry under swapping of the two states, i.e.,
        $\F(\rho,\sigma)=\F(\sigma,\rho)$;
    \item invariance under any unitary transformation $U$ of the
        state space, i.e., $\F(U\rho U^\dagger,U\sigma
        U^\dagger)=\F(\rho,\sigma)$; and finally,
    \item consistency with Schumacher's fidelity when one of the states is
    pure, i.e.,
    \begin{equation}\label{eq:fid_schum}
        \F(\rho,\ket{\psi}\!\bra{\psi})=\bra{\psi}\rho\ket{\psi}
    \end{equation}
    for arbitrary $\rho$ and $\ket{\psi}$.
\end{enumerate}

The proof that $\F$ satisfies all of Jozsa's axioms follows
easily from the variational definition of Eq. (\ref{eq:fid}) (see,
e.g., Ref.~\cite{00Nielsen} for technical details). The remainder of
this section discusses a number of less immediate properties of $\F$.

\subsubsection{Concavity Properties}\label{sec:fidconc}

The concavity property of quantities like
entropy, mutual information and fidelity is often of theoretical
interest in the quantum information community~\cite{00Nielsen}. In this
regard, it is worth noting that a useful feature of
$\F$ is its \emph{separate concavity} in each of
its arguments, i.e., for $p_1,p_2\geq 0$, $p_1+p_2=1$ and arbitrary
density matrices $\rho_1$, $\rho_2$, $\sigma_1$ and $\sigma_{2}$, we
have
\begin{equation}
    \F\left(p_1\rho_1+p_2\rho_2,\sigma_1\right)\geq
    p_1 \F(\rho_1,\sigma_1)+p_2 \F(\rho_2,\sigma_1)\,.\label{eq:sepconcF}
\end{equation}
By symmetry, concavity in the second argument follows from Eq.
\eqref{eq:sepconcF}.  Separate concavity can be
proved~\cite{76Uhlmann273,94Jozsa2315} using the variational definition
of $\F$ from Eq. (\ref{eq:fid}).\par

While it is known that $\sqrt{\F}$ is \emph{jointly
concave}~\cite{83Alberti5,00Uhlmann407}, i.e.,
\begin{equation}
    \sqrt{\F}\left(p_1\rho_1+p_2\rho_2,p_1\sigma_1+p_2\sigma_2\right)  \geq  p_1 \sqrt{\F}(\rho_1,\sigma_1)+p_2
    \sqrt{\F}(\rho_2,\sigma_2)\,,\label{eq:jconcsqrtF}
\end{equation}
it is also known that the fidelity $\F$ does \emph{not}, in
general, share the same enhanced concavity property\footnote{Note that joint concavity implies separate
    concavity but not the  other way around. For example, the separate concavity of $\sqrt{\mathcal{F}}(\rho,\sigma)$ can be obtained from Eq.~\eqref{eq:jconcsqrtF} by setting $\sigma_1=\sigma_2$ and using the fact that $p_1+p_2=1$.}.

\subsubsection{Multiplicativity under Tensor Product}\label{sec:fidmultiplic}

Another neat mathematical property of $\F(\rho,\sigma)$ is
that it is multiplicative under tensor products: for any density matrices
$\rho_1$, $\rho_2$, $\sigma_1$ and $\sigma_2$,
\begin{equation}
    \F(\rho_1\otimes\rho_2,\sigma_1\otimes\sigma_2)
    =\F(\rho_1,\sigma_1)\F(\rho_2,\sigma_2)\,.
\end{equation}
This identity follows easily from the following facts: for any
Hermitian matrices $\mathds{A}$ and $\mathds{B}$, (i) $\tr(\mathds{A}\otimes
\mathds{B})=\tr(\mathds{A})\,\tr(\mathds{B})$ and (ii) $\sqrt{\mathds{A}\otimes
\mathds{B}}=\sqrt{\mathds{A}}\otimes\sqrt{\mathds{B}}$.\par

An immediate consequence of this result is that for two physical
systems, described by $\rho$ and $\sigma$, a measure of
their ``closeness'' given by $\F$ remains unchanged even after
appending each of them with an uncorrelated ancillary state $\tau$,
i.e., $\F(\rho\otimes\tau, \sigma\otimes\tau)=\F(\rho,\sigma)$.

\subsubsection{Monotonicity under Quantum Operations}
\label{sec:fidmonotone}

Given that $\F(\rho,\sigma)$ serves as a kind of measure for the
proximity between two quantum states $\rho$ and $\sigma$, one might
expect that any quantum operation $\mathcal{E}$ should bring $\rho$ and $\sigma$
``closer together'' according to $\F$:
\begin{equation}\label{eq:Dmonotone}
    \F(\mathcal{E}(\rho),\mathcal{E}(\sigma))
    \geq \F(\rho,\sigma).
\end{equation}
Indeed, it is now well-known that Eq.~\eqref{eq:Dmonotone} holds
true~\cite{83Alberti107} for an arbitrary quantum operation described
by a completely-positive-trace-preserving (CPTP) map
$\mathcal{E}:\rho\mapsto\mathcal{E}(\rho)$. Inequality (\ref{eq:Dmonotone}) qualifies
$\F$ as a \emph{monotonically increasing measure} under  CPTP
maps and can be considered the quantum analogue of the classical
\emph{information-processing inequality} --- which expresses that the
amount of information should not increase via any information
processing.

On a related note, it is worth noting that any measure $\EuScript{M}$
which is (i) unitarily invariant,  (ii) jointly concave (convex) and
(iii) invariant under the addition of an ancillary system, is also
monotonically increasing (decreasing) under CPTP
maps\footnote{This follows easily from the Stinespring
    representation of  a CPTP map and from the representation of the partial trace operation given in Refs.~\cite{71Uhlmann633,08Carlen107}.}. Clearly, since $\sqrt{\F}$ satisfies all
the above-mentioned conditions, Eq.~\eqref{eq:Dmonotone} also follows
by simply squaring the corresponding monotonicity inequality for
$\sqrt{\F}$.

\subsubsection{Related Metrics}\label{sec:fidmetric}

The fidelity by itself is not a metric. However, one may well expect that a
metric, which is a measure of distance, can be built up from a measure
of ``closeness'' such as $\F$. Indeed, the functionals
\begin{align}
A[\F(\rho,\sigma)]&\defeq\arccos{\sqrt{\F(\rho,\sigma)}}\,,\label{eq:AF}\\
B[\F(\rho,\sigma)]&\defeq\sqrt{2-2\sqrt{\F(\rho,\sigma)}}\,,\label{eq:BF}\\
C[\F(\rho,\sigma)]&\defeq\sqrt{1-\F(\rho,\sigma)}\,,\label{eq:CF}
\end{align}
exhibit such metric properties (see Refs.~\cite{95Uhlmann461,
00Nielsen,69Bures199,92Hubner239,05Gilchrist062310,06Rastegin} and also
Appendix~\ref{app:metric_C} for more details). In particular, these
functionals are now commonly known in the literature, respectively, as
the {\em Bures angle}~\cite{00Nielsen}, the {\em Bures
distance}~\cite{69Bures199,92Hubner239}, and the \emph{sine
distance}~\cite{06Rastegin}.

\subsubsection{Trace Distance Bounds}\label{sec:fidtrdistbnd}

An important distance measure in quantum information is the metric
induced by  the trace norm $\|\cdot\|_{\rm tr}$ (defined in Sec.~\ref{sec:ttd}), which is commonly referred to
as the trace distance~\cite{00Nielsen}:
\begin{equation}\label{eq:trdist}
    \D(\rho,\sigma)=\tfrac{1}{2}\|\rho-\sigma\|_{\rm tr}\,.
\end{equation}
The trace distance is an exceedingly successful distance measure: it is
a metric (as is any distance induced by norms), unitarily invariant
\cite{97Bhatia}, jointly convex \cite{00Nielsen}, decreases under CPTP
maps \cite{94Ruskai1147} and, in the qubit case, is proportional to the
Euclidean distance between the Bloch vectors in the Bloch ball. The
trace distance is also closely related to the minimal probability of
error on attempts to distinguish between two non-orthogonal quantum
states \cite{76Helstrom}. For all of these reasons, one is generally
interested to determine how other distance measures relate with the
trace distance.\par

The following functions of the fidelity were shown in
Ref.~\cite{99Fuchs1216} to provide tight bounds for
$\D$\footnote{Both inequalities in Eq. (\ref{eq:fuchs_ineq}) are
    saturated if $\rho=\sigma$,  and also if $\rho$ and $\sigma$ have
    orthogonal supports. A less trivial example of saturation of the
    upper bound on $\D$ is obtained when both $\rho$ and $\sigma$ are
    pure states, whereas the lower bound on $\D$ can only be
    (non-trivially) saturated in Hilbert spaces of dimension strictly
    greater than $2$ (see Ref.~\cite{01Spekkens012310} for an example
    with $d=3$). Moreover, it is not difficult to show that the
    equality $1-\mathcal{F}=\mathcal{D}$ holds true if
    $[\rho,\sigma]=0$ \emph{and} at least one of the states is pure.}:
\begin{equation}\label{eq:fuchs_ineq}
    1-\sqrt{\F(\rho,\sigma)} \leq
    \D(\rho,\sigma) \leq \sqrt{1-\F(\rho,\sigma)}\,.
\end{equation}
In fact, the stronger lower bound $1-\F \leq \D$
holds if $\rho$ and $\sigma$  have support on a common two-dimensional
Hilbert space \cite{01Spekkens012310} (e.g., any pair of qubit states),
or if at least one of the states is pure \cite{00Nielsen}.\par

From these inequalities, one can conclude a type of qualitative
equivalence between the fidelity $\F$ and the trace distance
$\D$: whenever $\F$ is small, $\D$ is large and whenever $\F$
is large, $\D$ is small.

\subsection{An alternative fidelity measure}\label{sec:new_meas}

\subsubsection{Preliminaries}

We shall now turn attention to our proposed alternative fidelity measure
between two quantum states $\rho$ and $\sigma$, namely,
\begin{equation}\label{eq:fidnew}
    \Fn(\rho,\sigma)=\tr\left[\rho\sigma\right]
    +\sqrt{1-\tr\rho^2}\sqrt{1-\tr\sigma^2}.
\end{equation}
This is simply a sum of the Hilbert-Schmidt inner product between
$\rho$ and $\sigma$ and the geometric mean between their linear
entropies.  It is worth noting that the same quantity --- by the name
\emph{super-fidelity} --- has been independently introduced in
Ref.~\cite{08Miszczak} as an upper bound for $\F$.

Remarkably, when applied to qubit states, $\Fn$ is precisely the same
as $\F$. This observation follows easily from the fact that for density
matrices of dimension $d=2$, it is valid to write
\begin{equation}
    \left.\Fn(\rho,\sigma)\right|_{d=2}=\tr\left[\rho\sigma\right]
    +2\sqrt{\det{\rho}}\sqrt{\det{\sigma}}\,,
\end{equation}
which is just an alternative expression of $\F$ for qubit
states~\cite{93Hubner226,92Hubner239}.

When $d>2$, however, $\Fn$ no longer recovers $\F$, but can be seen as
a simplified version of the fidelity measure $\F_C$ proposed by Chen and
collaborators~\cite{02Chen054304}, which reads as:
\begin{equation}
    \F_C(\rho,\sigma)=\frac{1-r}{2}+\frac{1+r}{2}\Fn(\rho,\sigma),
\end{equation}
where $r=1/(d-1)$, and $d$ is the dimension of the state space of
$\rho$ and $\sigma$. Moreover, it is straightforward to verify that
while $\Fn$ reduces to the Schumacher's fidelity [the rhs of
Eq.~\eqref{eq:fid_schum}] when one of the states is pure; the same
cannot be said for $\F_C$.

It is not difficult to see from Eq.~\eqref{eq:fidnew} that $\Fn$
satisfies Jozsa's axioms 2, 3, and 4 as enumerated in
Sec.~\ref{sec:fidprelim}.  The non-negativity of $\Fn$ required by
axiom 1 is also immediate from the definition.  As a result, $\Fn$ is
an acceptable generalization of Schumacher's fidelity according to
Jozsa's axioms if:
\begin{proposition}\label{prop:jozsa}
    $\F_N(\rho,\sigma)\leq 1$ holds for arbitrary density matrices
    $\rho$ and $\sigma$,  with saturation if and only if $\rho=\sigma$.
\end{proposition}

\begin{proof} To begin with, recall that any $d\times
d$ density matrix can be expanded in terms of an orthonormal basis of
Hermitian matrices $\{\lambda_k\}_{k=0}^{d^2-1}$ such that
$\tr(\lambda_i\lambda_j)=\delta_{ij}$ (see, for example,
Refs.~\cite{03Byrd062322,03Kimura339}). In particular, if we let
$\bm{\Lambda}\defeq(\lambda_0,\ldots,\lambda_{d^2-1})$, then $\rho$ and
$\sigma$ admit the following decomposition:
\begin{equation}\label{eq:param}
    \rho=\bm{r}\cdot \bm{\Lambda}\quad\mbox{and}\quad
    \sigma=\bm{s}\cdot \bm{\Lambda}\,,
\end{equation}
where $\bm{r}$ and $\bm{s}$ are real vectors with $d^2$ entries
(corresponding to the expansion coefficients which can be determined
using the orthonormality condition). Since $\rho$ and $\sigma$ are
density matrices, $\bm{r}$ and $\bm{s}$ satisfy
$0\leq\bm{r}\cdot\bm{s}\leq 1$ and $r,s\leq 1$ where $r = \|\bm{r}\|$
and $s = \|\bm{s}\|$.

Using the expansion of Eq.~\eqref{eq:param} in Eq.~\eqref{eq:fidnew},
we arrive at the following alternative expression of $\F_N$,
\begin{align}\label{eq:fN}
    f_N(\bm{r},\bm{s})&=\bm{r}\cdot\bm{s} + \sqrt{1-r^2}\sqrt{1-s^2}\\
    &= \bm{R}\cdot\bm{S}\,,
\end{align}
where, in the second line, we have defined two \emph{unit} vectors in
$\mathbb{R}^{d^2+1}$, explicitly,
\begin{equation}
    \bm{R}\defeq\left(\bm{r},\sqrt{1-r^2}\right)\quad\mbox{and}
    \quad\bm{S}\defeq\left(\bm{s},\sqrt{1-s^2}\right)\,.
\end{equation}
The normalization of $\bm{R}$ and $\bm{S}$ then implies that
$\Fn(\rho,\sigma) = \bm{R}\cdot\bm{S} \le 1$, with saturation if and only if
$\bm{R}=\bm{S}$, or equivalently $\rho = \sigma$.
\end{proof}

\subsubsection{Concavity Properties}\label{sec:fnconc}

As with $\sqrt{\F}$, the new fidelity measure $\Fn$ is jointly
concave in its two arguments, i.e., for $p_1,p_2\geq 0$, $p_1+p_2=1$
and arbitrary density matrices $\rho_1$, $\rho_2$, $\sigma_1$ and
$\sigma_{2}$, we have,
\begin{equation}
    \Fn\left(p_1\rho_1+p_2\rho_2,p_1\sigma_1+p_2\sigma_2\right)\geq
    p_1 \Fn(\rho_1,\sigma_1)+p_2 \Fn(\rho_2,\sigma_2)\,.\label{eq:fnjconc}
\end{equation}
Since $\F$ fails to be jointly concave in general, $\Fn$ has stronger
concavity property. Remarkably, given the equivalence between $\F$ and
$\Fn$ in the $d=2$ case, the result of this section implies that $\F$
is jointly concave when restricted to qubit states.

The rest of this section concerns a proof of this concavity property of
$\Fn$. We start by proving the following lemma, which provides a useful
alternative expression of inequality~\eqref{eq:fnjconc}.

\begin{lemma}\label{lemmaconc}
Define a function $F:[0,1]\to\mathbb{R}$ by
\begin{equation}\label{eq:Fx}
    F(x) \defeq (\bm{r}+x\bm{u})\cdot(\bm{s}+x\bm{v})
          + \sqrt{1-\|\bm{r}+x\bm{u}\|^2}\sqrt{1-\|\bm{s}+x\bm{v}\|^2}.
\end{equation}
Given density matrices $\rho_1$, $\rho_2$, $\sigma_1$ and $\sigma_{2}$,
there exist vectors $\bm{r}$,$\bm{s}$,$\bm{u}$,$\bm{v} \in \mathbb{R}^{d^2}$
and $x \in [0,1]$ such that the inequality
\begin{equation}
    F(x) \geq (1-x)F(0) + x F(1) \label{eq:linebelow}
\end{equation}
is equivalent to Eq.~\eqref{eq:fnjconc}.
\end{lemma}

\begin{proof}
The proof is by construction. Using the parametrization of Eq.
(\ref{eq:param}) for  the density matrices in inequality
(\ref{eq:fnjconc}), we obtain the following equivalent inequality for
the vectors $\bm{r}_i$ and $\bm{s}_i$:
\begin{equation}
    f_N\left(p_1\bm{r}_1+p_2\bm{r}_2,p_1\bm{s}_1+p_2\bm{s}_2\right)\geq
    p_1 f_N(\bm{r}_1,\bm{s}_1)+p_2 f_N(\bm{r}_2,\bm{s}_2)\,,\label{eq:ineqfn}
\end{equation}
where the function $f_N$ was defined in Eq. (\ref{eq:fN}).

A straightforward computation shows that inequality
(\ref{eq:linebelow}) is identical to inequality (\ref{eq:ineqfn}) when we
identify $x\equiv p_2$, $1-x\equiv p_1$, and set
\begin{equation}\label{eq:rusv}
    \begin{array}{rclcrcl}
        \bm{r}&=&\bm{r}_1\,,&\qquad&\bm{u}&=&\bm{r}_2-\bm{r}_1\,,\\
        \bm{s}&=&\bm{s}_1\,,&\qquad&\bm{v}&=&\bm{s}_2-\bm{s}_1\,.
    \end{array}
\end{equation}
\end{proof}

If $F(x)$ has negative concavity in $x \in [0,1]$, then the inequality
(\ref{eq:linebelow}) is automatically satisfied as it establishes that
the straight line connecting the points $(0,F(0))$ and $(1,F(1))$ lies
below the curve $\{(x,F(x))|x \in [0,1]\}$. As a result, the joint
concavity of $\F_N$ is proved with the following proposition:

\begin{proposition}\label{prop:concFx}
For $x \in [0,1]$, and $\bm{r}$, $\bm{s}$, $\bm{u}$,
$\bm{v}\in\mathbb{R}^{d^2}$ specified in Eq.~(\ref{eq:rusv}), the function
$F(x)$ [cf. Eq.~(\ref{eq:Fx})] satisfies
\begin{equation}
 \frac{d^2F(x)}{dx^2}\leq 0
\end{equation}
and hence $\F_N$ is jointly concave.
 \end{proposition}
\noindent The proof of this Proposition is given in Appendix \ref{app:concFx}.

\subsubsection{Multiplicativity under Tensor Product} \label{sec:fnmultipl}

In contrast with $\F$, the new fidelity measure $\Fn$ is not
multiplicative under tensor products. In fact, it is generally not even
invariant under the addition of an uncorrelated ancilla prepared in the
state $\tau$. In this case, $\Fn$ between the resulting states reads
as:
\begin{equation}
    \Fn(\rho\otimes\tau,\sigma\otimes\tau)
    =\tr\left[\rho\sigma\right]
    \tr\,\tau^2+\sqrt{1-\tr\,\rho^2 \tr\,\tau^2}\sqrt{1-\tr\,\sigma^2 \tr\,\tau^2},
\end{equation}
where the lhs equals $\F_N(\rho,\sigma)$ iff $\tr\,\tau^2=1$, or in other  words, iff $\tau$ is a pure state. More
generally, it can be shown that $\Fn$ is super-multiplicative, i.e.,
\begin{equation}
    \F_N(\rho_1\otimes\rho_2,\sigma_1\otimes\sigma_2) \geq
    \F_N(\rho_1,\sigma_1)\F_N(\rho_2,\sigma_2).
\end{equation}
A proof of this property is given in Appendix~\ref{app:super-multplicativity};
a similar proof was independently obtained in Ref.~\cite{08Miszczak}.

\subsubsection{Monotonicity under Quantum Operations} \label{sec:fnmonotone}

That $\Fn$ is only super-multiplicative may be a first sign that it may
not behave monotonically under CPTP maps. In fact, as we shall see
below, Ozawa's counter-example~\cite{00Ozawa158} to the claimed
monotonicity of the Hilbert-Schmidt distance~\cite{99Witte14} can also
be used to show that $\Fn$  does not behave monotonically under CPTP
maps.\par

Let $\widetilde{\rho}$ and $\widetilde{\sigma}$ be two two-qubit
density matrices, written in the product basis as
\begin{equation}\label{eq:Ozawa}
    \widetilde{\rho}=\tfrac{1}{2}\left(
    \begin{array}{cccc}
    1 & 0 & 0 & 0\\
    0 & 1 & 0 & 0\\
    0 & 0 & 0 & 0\\
    0 & 0 & 0 & 0
    \end{array}
    \right)\quad\mbox{and}\quad
    \widetilde{\sigma}=\tfrac{1}{2}\left(
    \begin{array}{cccc}
    0 & 0 & 0 & 0\\
    0 & 0 & 0 & 0\\
    0 & 0 & 1 & 0\\
    0 & 0 & 0 & 1
    \end{array}
    \right),
\end{equation}
and consider the (trace preserving) quantum operations of tracing over
the first or the second qubit. A straightforward computation shows that
if the first qubit is traced over, then
\begin{equation}
    \Fn(\tr_1\widetilde{\rho},\tr_1\widetilde{\sigma})=1 >
    \tfrac{1}{2}=\F_N( \widetilde{\rho},\widetilde{\sigma}),\label{eq:monot1}
\end{equation}
which satisfies the desired monotonicity property. However, if instead
the second subsystem is discarded, we find
\begin{equation}
    \Fn(\tr_2\widetilde{\rho},\tr_2\widetilde{\sigma})=0 <
    \tfrac{1}{2}=\F_N( \widetilde{\rho},\widetilde{\sigma})\,.\label{eq:monot2}
\end{equation}
Together, Eqs. \eqref{eq:monot1} and \eqref{eq:monot2} show that
$\F_N$ is neither monotonically increasing nor decreasing
under general CPTP maps.

A natural question that follows is whether $\Fn$ features a weaker form
of monotonicity. For example, do arbitrary projective measurements ---
with the measurement outcomes forgotten --- give rise to higher value
of $\Fn$ for the resulting pair of states? An affirmative answer would
follow from a proof of the inequality
\begin{equation}\label{eq:pinching}
\Fn\left(\sum_i P_i \rho P_i , \sum_i P_i \sigma P_i \right) \geq \Fn(\rho,\sigma)
\end{equation}
for any complete set of orthonormal projectors $P_i$, and for arbitrary
density matrices $\rho$ and $\sigma$.

It is a simple  exercise to prove Eq.~\eqref{eq:pinching} for the
particular case where either of the commutation rules $[P_i,\rho]=0$ or
$[P_i,\sigma]=0$ is observed for all values of $i$. Whether the same
conclusion can be drawn from the more general, non-commutative cases
remains to be seen. In this regard, we note that a preliminary
numerical search favors the validity of Eq.~\eqref{eq:pinching}.

\subsubsection{Related Metrics}\label{sec:fnmetric}

In parallel to the metrics $A[\F]$, $B[\F]$ and
$C[\F]$ introduced in Sec. \ref{sec:fidmetric}, we define
\begin{align}
A[\F_N(\rho,\sigma)]&\defeq\arccos{\sqrt{\F_N(\rho,\sigma)}},\label{eq:AFn}\\
B[\F_N(\rho,\sigma)]&\defeq\sqrt{2-2\sqrt{\F_N(\rho,\sigma)}},\label{eq:BFn}\\
C[\F_N(\rho,\sigma)]&\defeq\sqrt{1-\F_N(\rho,\sigma)},\label{eq:CFn}
\end{align}
and prove that while $C[\F_N]$ preserves the metric
properties, both $A[\F_N]$ and $B[\F_N]$ \emph{do
not} always obey the triangle inequality
\begin{equation}
X[\F_N(\rho,\sigma)] \leq X[\F_N(\rho,\tau)] + X[\F_N(\tau,\sigma)]\,,\label{eq:triang}
\end{equation}
where $X$ here refers to either  $A$, $B$
or $C$. For example, consider the qutrit density
matrices, $\rho=\openone_3/3$,
\begin{equation}\label{eq:example_states}
\sigma=\left(\begin{array}{ccc}
1&0&0\\
0&0&0\\
0&0&0\end{array}\right)\mbox{ and }
\tau=\left(\begin{array}{ccc}
0.90&0.04&0.03\\
0.04&0.05&0.02\\
0.03&0.02&0.05\end{array}\right).
\end{equation}

\begin{table}[h!]
\centering \caption{A numerical test of the triangle inequality for
$A[\F_N]$, $B[\F_N]$ and $C[\F_N]$.}
\begin{tabular*}{0.8\textwidth}{@{\extracolsep{\fill}} r|cc}
\hline\hline
\hspace{0.25cm}$X$\hspace{0.25cm}&$X[\F_N(\rho,\sigma)]$&$X[\F_N(\rho,\tau)]+X[\F_N(\tau,\sigma)]$\\
\hline
$A$&$0.9553$ & $0.9241$\\
$B$&$0.9194$ & $0.9137$\\
$C$&$0.8165$ & $0.8828$\\
\hline\hline
\end{tabular*}\label{table:numtest}
\end{table}

Numerical computation of the quantities appearing in the triangle
inequality gives rise to Table \ref{table:numtest}.  Note that for
$X=A, B$, the first column dominates the second, i.e., the triangle
inequality is violated and therefore neither $A[\F_N]$ nor
$B[\F_N]$ are metrics. For $X=C$, no violation is observed for
the above density matrices. Next, we prove that this is the case for
any three density matrices $\rho$, $\sigma$ and $\tau$, thus
$C[\F_N]$ is a metric.
\begin{proposition}\label{Pro:metric:C(Fn)}
The quantity $C[\F_N(\rho,\sigma)]$ is a metric for the space of density matrices.
\end{proposition}

To prove this proposition, we will make use of the following theorem
due to Schoenberg~\cite{38Schoenberg522} (see also \cite[Ch. 3,
Proposition 3.2]{84Berg}). We state here an abbreviated form of the
theorem sufficient for our present purposes.
\begin{theorem}[Schoenberg]\label{thm:schoenberg}
    Let $\EuScript{X}$ be a nonempty set and
    $K:\EuScript{X}\times\EuScript{X}\to\mathbb{R}$ a function such that
    $K(x,y)=K(y,x)$ and $K(x,y)\geq 0$ with saturation iff $x=y$, for all
    $x,y \in \EuScript{X}$. If the implication
    \begin{equation}
    \sum_{i=1}^{n}{c_i}=0 \Rightarrow \sum_{i,j=1}^{n}{K(x_i,x_j)c_i c_j}\leq 0\label{eq:ndk}
    \end{equation}
    holds for all $n\geq 2$, $\{x_1,\ldots,x_n\} \subseteq \EuScript{X}$ and
    $\{c_1,\ldots,c_n\} \subseteq \mathbb{R}$, then $\sqrt{K}$ is a metric.
\end{theorem}

We make a small digression at this point to remark that, in spite of
its successful application in the grounds of classical probability
distance measures~\cite{00Topsoe1602,03Topsoe,04Fuglede}, Schoenberg's
theorem has received almost no attention by the quantum information
community. In this thesis, besides proving the metric properties of $C[\Fn]$, we will also make use Schoenberg's theorem to provide independent proofs of the metric properties of $B[\F(\rho,\sigma)]$, $C[\F(\rho,\sigma)]$ (Appendix~\ref{app:metric_C}), and of the squared Hilbert-Schmidt distance, to be introduced later (Appendix~\ref{app:metricHsq}).\\

\begin{proof}[Proof of Proposition~\ref{Pro:metric:C(Fn)}]
Clearly, from the definition of $C^2[\Fn(\rho,\sigma)]$, it is easy to
see that it inherits from $\Fn(\rho,\sigma)$ the property of being
symmetric in its two arguments, and that
$C^2[\F_N(\rho,\sigma)]\geq 0$ with saturation iff
$\rho=\sigma$. So, to apply Theorem~\ref{thm:schoenberg}, we just have
to show that for any set of density matrices $\{\rho_i\}_{i=1}^n$
($n\geq 2$) and real numbers $\{c_i\}_{i=1}^n$ such that $\sum_{i=1}^n
c_i=0$, it is true that
\begin{equation}
    \sum_{i,j=1}^n C^2[\F_N(\rho_i,\rho_j)] c_i c_j \leq 0\,.
\end{equation}
This follows straightforwardly by exploiting the zero-sum property of
the (real) coefficients $c_i$ and the linearity of the trace,
\begin{align}
    &\sum_{i,j=1}^n \left\{1-\tr\left[\rho_i\rho_j\right]
    -\sqrt{1-\tr\rho_i^2}\sqrt{1-\tr\rho_j^2}\right\} c_i c_j
    \nonumber \\
    =&-\tr\left[\bigg(\sum_{i=1}^n c_i\rho_i\bigg)^2\right]
    -\left(\sum_{i=1}^n c_i\sqrt{1-\tr\rho_i^2}\right)^2\leq 0\,,
\end{align}

which concludes the proof.
\end{proof}

We note that a proof of the metric property of
$\sqrt{2}C[\Fn(\rho,\sigma)]$ --- by the name \emph{modified Bures
distance} --- was independently provided by Ref.~\cite{08Miszczak}. The
proof provided above is significantly shorter thanks to the power of
Schoenberg's theorem.

\subsubsection{Trace Distance Bounds}\label{sec:fntrdistbnd}

In Sec.~\ref{sec:fidtrdistbnd}, we have seen that a kind of qualitative
equivalence between $\D$ and $\F$ can be established through the bounds
on $\D$ given by functions of  $\F$, c.f. Eq.~\eqref{eq:fuchs_ineq}.
Here, we will provide similar bounds on $\D$ in terms of functions of
$\Fn$.

\begin{proposition}\label{prob:lb}
For any two density matrices $\rho$ and $\sigma$ of dimension $d$, the
trace distance $\D(\rho,\sigma)$  satisfies the following
upper bound:
\begin{equation}\label{eq:upbndFn}
    \D(\rho,\sigma)\leq \sqrt{\frac{\mathfrak{r}}{2}}\sqrt{1-\Fn(\rho,\sigma)}\,,
\end{equation}
where $\mathfrak{r}\defeq{\rm rank}(\rho-\sigma)$. Moreover,
this upper bound on $\D$ can be saturated with states of the form
\begin{equation}\label{eq:satupper}
    \rho=\frac{U{\rm diag}\left[\Lambda_d\right]U^\dagger}
    {\tr\left\{{\rm diag}\left[\Lambda_d\right]\right\}}
    \quad\mbox{and}\quad\sigma
    =\frac{U{\rm diag}\left[P(\Lambda_d)\right]U^\dagger}
    {\tr\left\{{\rm diag}\left[\Lambda_d\right]\right\}}\,,
\end{equation}
where $U$ is an arbitrary unitary matrix of dimension  $d$, $\Lambda_d$
is an ordered list of $d$ elements taking values in the set
$\{\lambda_1,\lambda_2\}$ ($\lambda_1, \lambda_2 \geq 0$, but not
simultaneously zero) and $P(\Lambda_d)$ is the list formed by some
permutation of the elements in $\Lambda_d$.
\end{proposition}

\begin{proof}
Note that the product of square roots  in the expression of $\Fn$,
Eq.~\eqref{eq:fidnew}, is the geometric mean between the linear
entropies of $\rho$ and $\sigma$. It then follows from the inequality
of arithmetic and geometric means that
  \begin{equation}\label{eq:AMGM}
  \frac{1-\tr\rho^2}{2}+\frac{1-\tr\sigma^2}{2}\geq \sqrt{1-\tr\rho^2}\sqrt{1-\tr\sigma^2}\,,
  \end{equation}
which can be reexpressed as the following inequality after summation of
$\tr\left[\rho\sigma\right]$ to both sides,
  \begin{equation}\label{eq:HSFn}
\|\rho-\sigma\|_{\rm HS}\leq \sqrt{2\left[1-\F_N(\rho,\sigma)\right]}\,.
  \end{equation}
Here, $\|X\|_{\rm HS}\defeq\sqrt{\tr\left[X^\dagger X\right]}$ is the
Hilbert-Schmidt norm (also known as Frobenius norm), defined for an
arbitrary matrix $X$. The Hilbert-Schmidt norm and the trace norm
$\|X\|_{\rm tr}\defeq\tr\sqrt{X^\dagger X}$ are related according
to\footnote{\label{footnote:prooftrhs}To see that, assume, for simplicity, that $X$ is a
    square matrix  of dimension $d$ and let $\bm{\lambda} \in
    \mathbb{R}^d$ be the vector with entries
    $\lambda_1\geq\lambda_2\geq\ldots\geq\lambda_d$ corresponding to
    the singular values of $X$. In addition, let $\bm{v} \in
    \mathbb{R}^d$ be the vector with the first $\mathfrak{x}={\rm
    rank}\,X$ entries equal to $1$ and the remaining $d-\mathfrak{x}$
    entries equal to $0$. Then, it follows that $\|X\|_{\rm tr}=|\bm{\lambda}\cdot\bm{v}|$, $\|X\|_{\rm HS}=\|\bm{\lambda}\|$
    and $\sqrt{\mathfrak{x}}=\|\bm{v}\|$. In this framework, inequality
    (\ref{eq:normtrhs}) is equivalent to Cauchy-Schwarz inequality
    applied to $\bm{\lambda}$ and $\bm{v}$, i.e.,
    $|\bm{\lambda}\cdot\bm{v}|\leq\|\bm{\lambda}\| \|\bm{v}\|$.}
\begin{equation}\label{eq:normtrhs}
    \|X\|_{\rm tr}\leq \sqrt{\mathfrak{x}}\|X\|_{\rm HS}\,,
\end{equation}
where $\mathfrak{x}\defeq{\rm rank}\,X$. Used in Eq. \eqref{eq:HSFn},
the above inequality leads to the desired result
\begin{equation}\label{eq:DFn_boundrank}
    \D(\rho,\sigma)=\tfrac{1}{2}\|\rho-\sigma\|_{\rm tr}
    \leq\sqrt{\frac{\mathfrak{r}}{2}}\sqrt{1-\F_N(\rho,\sigma)}\,.
\end{equation}

To prove that the states in Eq. (\ref{eq:satupper}) saturate this
bound, we first note that because those states are isospectral, their
linear entropies are identical and hence inequality (\ref{eq:AMGM}) is
saturated. To prove saturation of inequality (\ref{eq:normtrhs}),
simply use Eq. (\ref{eq:satupper}) to compute
\begin{align}
    \|\rho-\sigma\|_{\rm tr}&
    =\tr\sqrt{\left(\rho-\sigma\right)^2}
    =\frac{\mathfrak{r}|\lambda_1-\lambda_2|}
    {\tr\left\{{\rm diag}\left[\Lambda_d\right]\right\}},\\
    \|\rho-\sigma\|_{\rm HS}&
    =\sqrt{\tr\left[\left(\rho-\sigma\right)^2\right]}
    =\frac{\sqrt{\mathfrak{r}}|\lambda_1-\lambda_2|}
    {\tr\left\{{\rm diag}\left[\Lambda_d\right]\right\}},
\end{align}
from which the identity $\|\rho-\sigma\|_{\rm tr}=\sqrt{\mathfrak{r}}\|\rho-\sigma\|_{\rm HS}$ is immediate.
\end{proof}

\begin{figure*}
\centering \subfigure[\hspace{1.5mm}For $d=3$,  a gap can be clearly
noticed between the distribution of states and the \emph{absolute upper
bound}, i.e., the rhs of inequality \eqref{eq:upbndFn} with
$\mathfrak{r}=d$. Such a gap occurs whenever $d$ is odd.] {
    \label{Fig:TrD.Fn_d3:Bounds}
    \includegraphics[width=7.5cm]{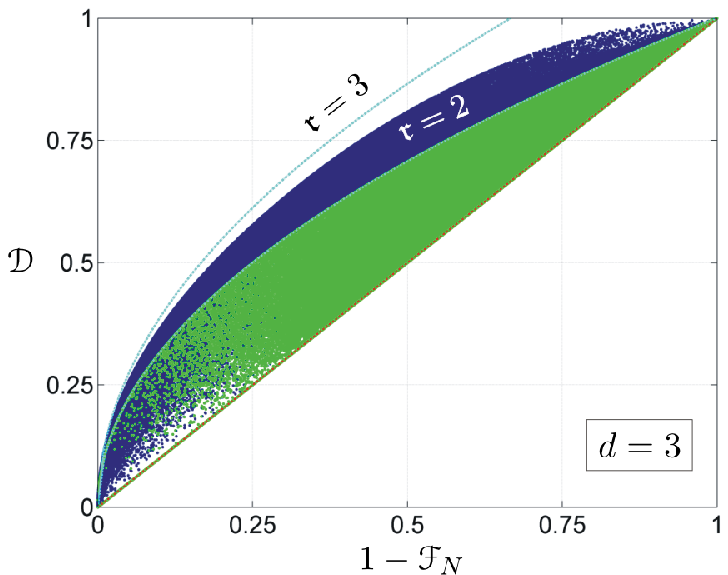}
} \hspace{0.6cm} \subfigure[\hspace{1.5mm} For $d=6$,  no gap is observed
between the bulk of randomly generated states and the \emph{absolute
upper bound}. In fact, this bound can be saturated by density matrices
of the form given by Eq.~\eqref{eq:satupper} whenever $d$ is even.] {
    \label{Fig:TrD.Fn_d6:Bounds}
    \includegraphics[width=7.5cm]{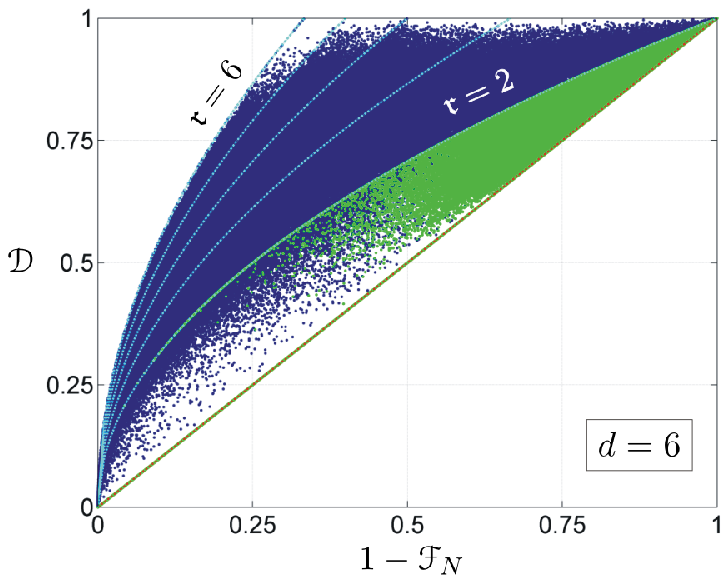}
} \caption[Numerical observation of upper and lower bound relationship between $\D$ and $\Fn$]{(Color online) Plot of the trace distance  $\D(\rho,\sigma)$
vs $1-\Fn(\rho,\sigma)$ for $4\times10^6$ pairs of randomly generated
$\rho$, $\sigma$ with $d=3$ and $d=6$. The darker (blue) points are
generated using pairs of mixed states whereas the lighter (green)
points are generated using at least one pure state. The anti-diagonal
solid line is the conjectured lower bound whereas the upper bounds
given by Eq.~\eqref{eq:upbndFn} are represented by the dashed curves
(cyan) --- one for each integer value of $\mathfrak{r} \in [2,d]$. }
\label{fig:bounds}
\end{figure*}

How good are  these upper bounds? With some thought, it is not
difficult to conclude that the states arising from
Eq.~\eqref{eq:satupper} can only have even $\mathfrak{r}$, and are thus
unable to saturate the upper bound of Eq.~\eqref{eq:upbndFn} for odd
$\mathfrak{r}$. Nonetheless, from our numerical studies, it seems like
the \emph{absolute upper bound} --- corresponding to the choice
$\mathfrak{r}=d$ in the rhs of Eq.~\eqref{eq:upbndFn} --- is actually
unachievable by {\em any} states if $d$ is odd. An illustration of this
peculiarity can be seen in Fig.~\ref{Fig:TrD.Fn_d3:Bounds}, where the
upper bound corresponding to $\mathfrak{r}=3$ is well separated from
the region attainable by physical states. In contrast, for every even
$d$, the states given by Eq.~\eqref{eq:satupper} do trace out a tight
boundary for the region attainable with physical states, as shown in
Fig.~\ref{Fig:TrD.Fn_d6:Bounds} for $d=6$.

On the other hand, it can also be seen from Fig.~\ref{fig:bounds} that
no points occur in the region where $\D\leq1-\Fn$. Indeed, intensive
numerical studies for $d=3,4,\ldots,50$ have not revealed a single
density matrix which contributed to a point in this region. This
suggests that the following lower bound on $\D$, in terms of $\Fn$, may
well be established\footnote{After the publication of Ref.~\cite{08Mendonca1150}, this conjecture was proved by Pucha{\l}a and Miszczak in Ref.~\cite{09Puchala024302}}:

\begin{conjecture}\label{conj:conjbndFnD}
    The trace distance $\D(\rho,\sigma)$ and the fidelity
    $\Fn(\rho,\sigma)$ between two quantum states $\rho$ and $\sigma$
    satisfy
    \begin{equation}\label{eq:conjbnd}
        \D(\rho,\sigma)\ge 1-\Fn(\rho,\sigma).
    \end{equation}
\end{conjecture}

In relation to this, it is also worth noting that the following
(weaker) lower bound can readily be established via a recent result
given in Ref.~\cite{08Miszczak}:
\begin{proposition}
    The trace distance $\D(\rho,\sigma)$ and the fidelity
    $\Fn(\rho,\sigma)$ between two quantum states $\rho$ and $\sigma$
    satisfy the following inequality.
    \begin{equation}\label{eq:polishbnd}
        \D(\rho,\sigma)\ge 1-\sqrt{\Fn}(\rho,\sigma).
    \end{equation}
\end{proposition}
\begin{proof}
    This lower bound on $\D$ follows immediately from the lower bound on
    $\D$ given in inequality~\eqref{eq:fuchs_ineq} and the inequality
    $\F\leq\Fn$ recently established in Ref.~\cite{08Miszczak}.
\end{proof}

As with the fidelity $\F$, we can thus infer that whenever $\Fn$ is
large enough, $\D$ is close to zero and whenever $\Fn$ is close to
zero, $\D$ is close to unity. However --- as should be clear from
Fig.~\ref{Fig:TrD.Fn_d6:Bounds} --- the converse implication is not
necessarily true.

\begin{figure}[h!]
\centering
\subfigure[] {
    \label{Fig:loglogcomplete}
    \includegraphics[height=6.5cm]{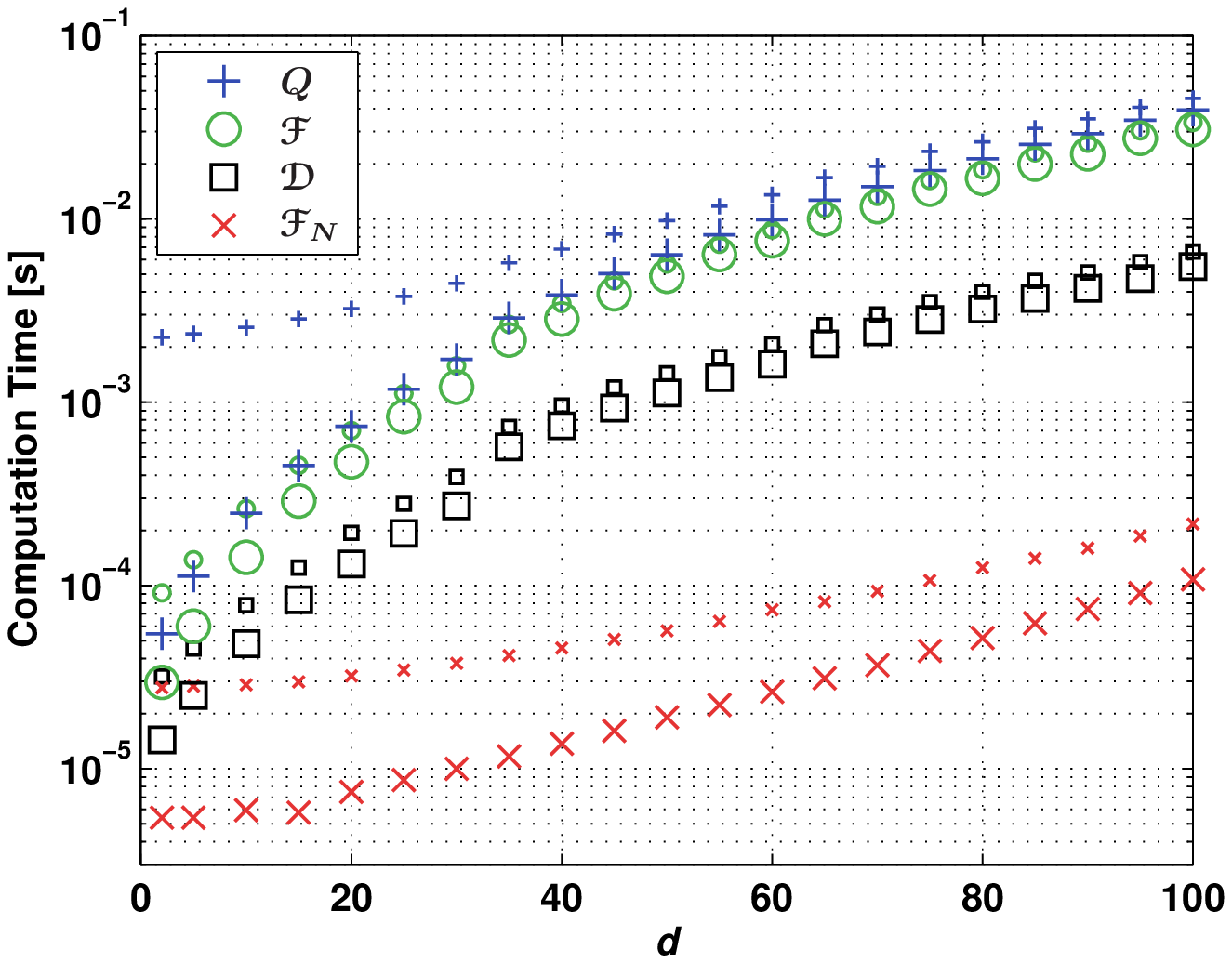}
}
\subfigure[]
{\label{Fig:loglog10pts}
    \includegraphics[height=6.5cm]{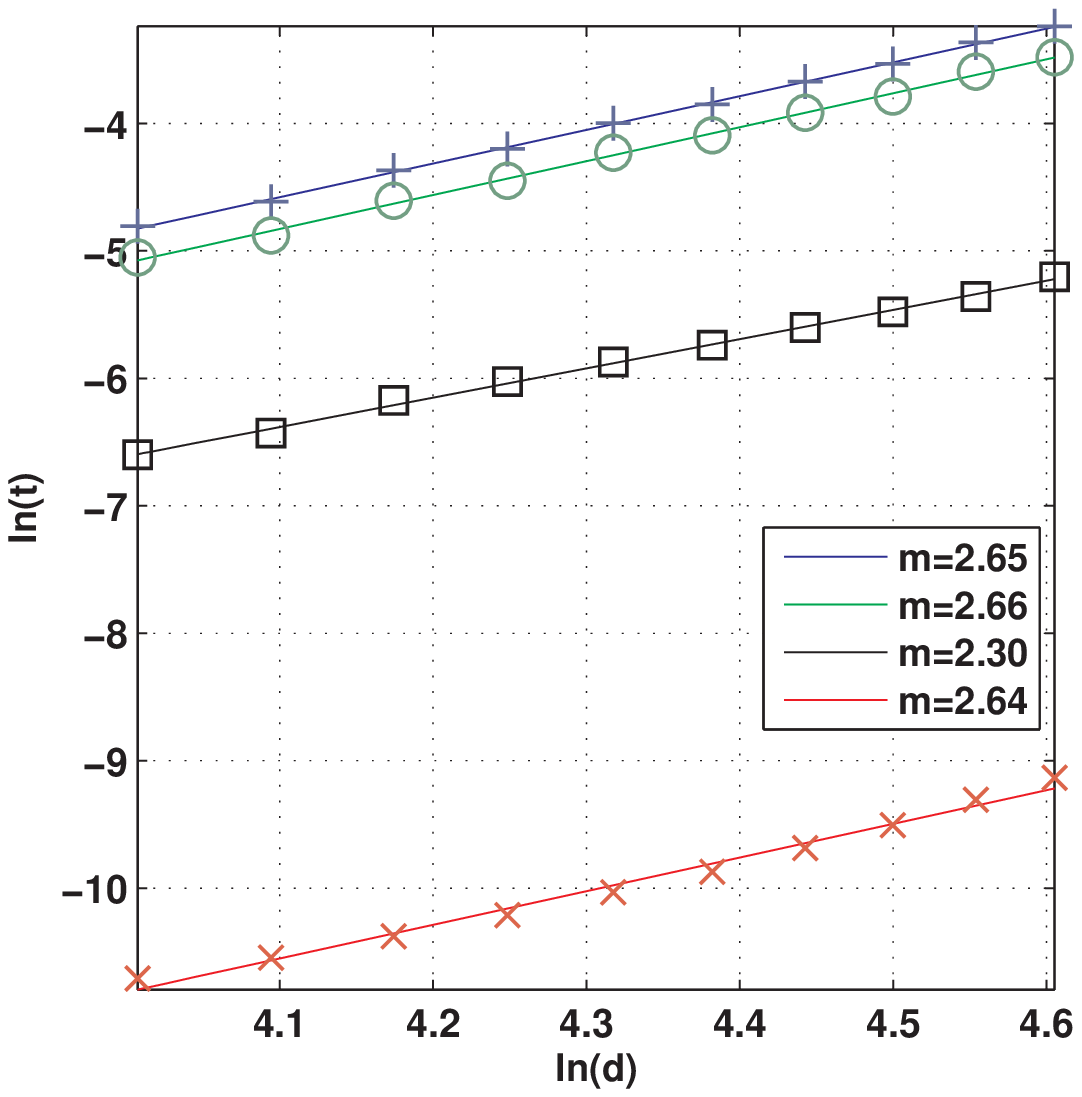}
} \caption[Average computation time of $\F$, $\Fn$, $\D$ and $Q$ estimated with Matlab and C codes.]{\label{Fig:ComputationTime}
    (Color online) Plots of the average computation time for the fidelity functions $\F$ ({$\bigcirc$}), $\Fn$ ($\times$), the nonlogarithmic variety of the quantum Chernoff bound $Q$ ($+$), and the
	trace distance $\D$ ($\square$) as a function of the dimension $d$
	of the state space. Computations were performed on a
	$2.6$~GHz Intel Pentium 4 CPU. (a) The data is presented in a semilog plot where the smaller and larger markers correspond to timings from Matlab and C respectively.  (b) By plotting $\ln(t)\times\ln(d)$ for $55\leq d\leq 100$ and timings from the C code, we obtain straight lines whose angular coefficients ($m$) quantify the ``practical complexity'' for computing each considered measure.}
\end{figure}

\subsection{Computational Efficiency}\label{sec:computability}

For two general density matrices $\rho$ and $\sigma$, analytical
evaluation of the fidelity $\F(\rho,\sigma)$ can be a formidable task.
This is in sharp contrast with $\Fn(\rho,\sigma)$ which involves only
products and traces of density matrices.  Even at the numerical level
--- due to the complication involved in evaluating the square root of a
Hermitian matrix --- the computation of $\F(\rho,\sigma)$ can be rather
resource consuming. For a quantitative understanding of the
computational efficiency, we have performed a numerical comparison of
the time required to calculate the fidelities $\F$ and $\Fn$, the trace
distance $\D$, and the
nonlogarithmic variety of the quantum Chernoff bound $Q\defeq\min_{0\leq s \leq 1} \tr(\rho^s\sigma^{1-s})$.  We have implemented the computations in both Matlab and
C; we present the Matlab codes for reasons of accessibility and
succinctness, while the C codes provide more accurate timings without the
overhead of the Matlab interpreter.

The time required to evaluate each function was estimated by averaging
the times for $100$ pairs of randomly generated $d$-dimensional density
matrices\footnote{Here, we follow the algorithm presented in
    Ref.~\cite{98Zyczkowski883} to generate $d$-dimensional quantum
    states.  In particular, the eigenvalues $\{\lambda_i\}_{i=1}^d$ of
    the quantum states were chosen from a uniform distribution on the
    $d$-simplex defined by $\sum_i \lambda_i=1$.}.  Results are shown in
Fig.~\ref{Fig:loglogcomplete} as a function of $d$.  The Matlab codes
are presented in Appendix~\ref{app:codes}; we attempted to make these
codes as efficient as possible within the constraints of the Matlab
environment.  Corresponding C codes were implemented as Matlab
MEX-files for convenience and can be found online~\cite{C:url}.  Our C
implementation directly calls the LAPACK and BLAS libraries included in
the Matlab distribution for eigenvalue decompositions and matrix
operations.  The minimization required in the computation of $Q$ was
performed using the Brent minimizer from the GNU Scientific Library
\cite{06Galassi}.

The results shown in Fig.~\ref{Fig:loglog10pts} display some consistency with
the expected algorithmic complexity. From the figure, one sees that our C codes for computing $\F$ and $Q$
require approximately $O(d^{2.7})$ operations for values of $d \in [55,100]$. This is
in good agreement with the theoretical asymptotic performance, since $\F$ and $Q$
require two Hermitian diagonalizations, taking an expected $O(d^3)$ operations each
\cite{00Parlett38}.  Computing $Q$ is slowest since it requires both
sets of eigenvectors, while $\F$ requires only eigenvalues from one of
the diagonalizations.

Next fastest is the computation of $\D$, which
requires only eigenvalues from a single diagonalization. In this case, Fig.~\ref{Fig:loglog10pts} suggests that only $O(d^{2.3})$ operations are required to compute the trace distance between density matrices of dimension $d\in [55,100]$. Note, however, that this is significantly less than the expected $O(d^3)$ operations, typical from computations involving matrix diagonalization. Such a discrepancy can be understood as follows: since $\D$ can be computed considerably faster than $Q$ or $\F$, the time taken by other spurious machine processes (not intrinsically related to the execution of our algorithm) becomes relatively important, compromising the accuracy of our timings. This is particularly significant in timing $\Fn$, our fastest-to-compute measure. In principle, one should expect an asymptotic performance $O(d^2)$, since $\Fn$ requires only three Hilbert-Schmidt inner products. However, Fig.~\ref{Fig:loglog10pts} suggests that $O(d^{2.6})$ operations are required.

Nevertheless, Fig.~\ref{Fig:ComputationTime} clearly shows that the practical
numerical evaluation of $\Fn$ is dramatically faster than the
evaluation of $\F$, $\D$ or $Q$.  This raises the prospect of
using $\Fn$ as a numerically efficient estimate of distance measures
such as $\F$~\cite{08Miszczak} and $\D$ --- particularly for small $d$
where the bounds proven in Sec.~\ref{sec:fntrdistbnd} are tighter.  As
the dimension increases, the computational advantage of using $\Fn$
becomes even greater, but the quality of the estimate drops.

\subsection{Concluding Remarks}\label{sec:conclusion}

In the previous sections,  we have proposed an alternative fidelity measure,
$\F_N$, between an arbitrary pair of mixed quantum states. This new
measure, together with the prevailing fidelity $\F$ and the
nonlogarithmic variety of the quantum Chernoff bound $Q$
\cite{07Audenaert160501} are, to the best of our knowledge, the only
known distance measures between density matrices that comply with
Jozsa's axioms~\cite{94Jozsa2315}. That is, $\F$, $Q$, and $\Fn$ are
the only known measures that generalize to pairs of mixed states the
concept of fidelity introduced by Schumacher between a pure and a mixed
state~\cite{95Schumacher2738}.

The simplicity of $\Fn$ is in sharp contrast with $\F$ and $Q$
since it involves only products of density matrices. Numerically, this
leads to significant reduction in computation time for
$\Fn(\rho,\sigma)$ over $\F(\rho,\sigma)$, especially for higher
dimensional systems.

Besides being easier to compute, $\Fn$ has also been shown to preserve
(and even enhance) a number of the useful properties of $\F$
and $Q$. For example, we have shown that $\Fn$ is a jointly concave
measure, that it can be used to place upper and lower bounds on the
value of the trace distance and that it gives rise to a new metric for
the space of density matrices. A remarkable consequence of the joint
concavity of $\Fn$ is that $\F$ is also jointly concave when restricted
to a pair of qubit states --- an interesting problem which remained
unsolved thus far~\cite{Mike:personal,Uhlmann:personal}.

The new measure, nevertheless, is not without its drawbacks. To begin
with --- $\Fn$, unlike measures such as $\F$ or $Q$ --- does not behave
monotonically under completely-positive-trace-preserving (CPTP) maps.
In addition, it does not necessarily vanish when applied to any pair of
{\em mixed} states which are otherwise recognized to be completely
different according to $\F$, $Q$ or their trace distance $\D$. In fact,
the explicit dependence on the linear entropies of $\rho$ and $\sigma$
gives rise to the following undesirable feature: the value of $\Fn$
between two completely mixed states living in disjoint subspaces can
get arbitrarily close to unity as the dimension of the state space
tends to infinity.

The undesirable features of $\Fn$ provide a clue as to when $\Fn$ may
not be the preferred measure of ``closeness'' between two quantum
states: We know that $\Fn$ does not measure the ``closeness'' between
two high-dimensional, highly mixed states (i.e., states having
non-negligible linear entropy) in the same way that measures like $\F$,
$Q$ or $\D$ would. In these cases, the interpretation of $\Fn$ as a
measure of proximity between quantum states must be carried out with
extra caution.\par

With this in mind, we nevertheless see $\Fn$ as an attractive
alternative to $\F$. Even when out of its range of applicability, it
follows from a very recent result of Miszczak {\em
et~al.}~\cite{08Miszczak} that $\Fn$ provides an upper bound on the
Uhlmann-Jozsa fidelity $\F$. Moreover, it seems promising that $\Fn$
between any two quantum states may be measured directly in the
laboratory, without resorting to any state tomography
protocol~\cite{08Miszczak}.

Let us now briefly mention some possibilities for future research that
stem from the present work. To begin with, it would be interesting to
search for a quantitative relationship between $\Fn$ and $Q$ analogous
to that between $\Fn$ and $\D$ established here, or that
between $\Fn$ and $\F$ given in Ref.~\cite{08Miszczak}. An estimate of
$Q$ based on some function of $\Fn$ would be useful given that a closed
form for $Q$ is not currently known, and that $\Fn$ can be computed
relatively easily. In addition, assuming $\Fn$ as an alternative to
$\F$, it seems reasonable to revisit some of the problems where $\F$
has proven useful, but with $\Fn$ playing its role. In particular, it
would be interesting to investigate whether the simplicity associated
with $\Fn$ will offer some advantages over $\F$.

As a first example, we recall from Ref.~\cite{97Vedral2275} that a
standard measure for the amount of entanglement of a state $\rho$ is
given by the shortest \emph{distance} from $\rho$ to the set of
separable density matrices. Given the relative simplicity of $\Fn$ with
respect to $\F$, it is not inconceivable that a distance measure based
on $\Fn$ (such as $C[\Fn]$) may lead to a more efficient determination
of this quantity if compared, for example, to $C[\F]$ or the Bures
distance~\cite{98Vedral1619}. Of course, any serious attempts in this
direction should be preceded by further investigation of the impact of
the nonmonotonicity of $\Fn$ under CPTP maps~\cite{97Vedral2275}.\par

As another example, $\Fn$ can be used as a figure of merit in designing
optimized quantum control and/or quantum error correction strategies:
One is typically interested in determining a quantum operation
$\mathcal{C}$ that minimizes the averaged {\em distance} between the
elements of a sequence of noisy quantum states $\rho_i$ and a pre-defined
sequence of target quantum states $\sigma_i$. In this context, it would be
interesting to investigate if distance measures based on $\Fn$ would
lead to any advantage in terms of computation time. Clearly, this has
potential applications to the implementation of real time quantum
technologies.

Yet another possible direction of research consists of employing $\Fn$
as a distance measure between quantum operations --- as opposed to
quantum states --- via the isomorphism between quantum states and CPTP
maps~\cite{99Horodecki1888,99Fujiwara3290}. In this regard, it is worth
investigating whether distance measures based on $\Fn$ would satisfy
the six criteria proposed in Ref.~\cite{05Gilchrist062310}. Remarkably,
from the results of the present work and Ref.~\cite{08Miszczak}, a few
strengths of $\Fn$-based measures can already be anticipated. Of
special significance are the fulfilment of the criteria ``easy to
calculate'' and ``easy to measure''. Along these lines, some
operational meaning for $\Fn$ would also be highly desirable. Although
we do not presently have a compelling physical interpretation of $\Fn$,
it is not inconceivable that one can be found in an analogous way to
$\F$~\cite{01Spekkens012310}.

\section{Metrical distance measures}\label{sec:metrical}

In this section we study three distance measures on the set of the density matrices that are truly metrics. The fulfilment of the metric axioms is safeguarded by the fact that these distances are ``induced by norms''. In order to make this notion clearer, let us start with the following definition:

\begin{definition} Given a vector space $V$ over $\mathbb{C}$, a function $\|\cdot\| : V \to \mathbb{R}$ is called a norm if and only if, for every $v, w \in V$ and $\lambda \in \mathbb{R}$
\end{definition}
\begin{itemize}
\item[(N1)] $\|v\|\geq 0$ (Nonnegativity)\,,
\item[(N2)] $\|v\| = 0$ \mbox{ iff } $v=0$\,,
\item[(N3)] $\|\lambda v\| = |\lambda|\|v\|$ (Positive Scalability)\,,
\item[(N4)] $\|v+w\| \leq \|v\| + \|w\|$ (Triangle Inequality)\,.
\end{itemize}

The quantity $\mathscr{D}(v,w)\defeq\|v-w\|$ induces a notion of distance between the elements $v$ and $w$ which is, indeed, a metric ``induced by the norm $\|\cdot\|$''. The metric axioms on page~\pageref{axioms:metric} can be easily verified:
\begin{enumerate}\label{proof:inducedmetric}
\item[(M1)] (Nonnegativity): that $\mathscr{D}(v,w)\geq 0$ follows trivially from (N1);
\item[(M2)] (Identity of Indiscernibles): that $\mathscr{D}(v,w)=0$ iff $v=w$ follows trivially that from (N2);
\item[(M3)] (Symmetry): that $\mathscr{D}(v,w)=\mathscr{D}(w,v)$, follows from $\|v-w\|=\|-(w-v)\|$ and (N3);
\item[(M4)] (Triangle Inequality): that $\mathscr{D}(v,w)\leq \mathscr{D}(v,u)+\mathscr{D}(w,u)$ follows by replacing $v \to v-u$ and $w \to u-w$ in (N4), for all $u\in V$.
\end{enumerate}

In Sec.~\ref{sec:threemetrics} we introduce three well established metrics for the space of density matrices: the trace norm, the Hilbert-Schmidt norm and the spectral norm. In Sec.~\ref{sec:evaluatemetrics} these measures are evaluated against the same criteria discussed in the previous section. Table~\ref{table:metrical}, on page~\pageref{table:metrical}, summarizes the main results of this and the previous sections.

\subsection{Three metrics for the space of density matrices}\label{sec:threemetrics}

A standard family of norms for the algebra of matrices was introduced by Schatten in Ref.~\cite{60Schatten} (see also Ref.~\cite{97Bhatia}). For any value of $p\in[1,\infty]$, the Schatten $p$-norms are defined as
\begin{equation}\label{eq:schatten_norms}
\|A\|_p=\left[\tr\left(|A|^p\right)\right]^{\tfrac{1}{p}}\quad\mbox{for}\quad p \in [1,\infty)\quad\mbox{and}\quad \|A\|_\infty=\|A\|\,,
\end{equation}
where $A \in \mathcal{M}_{\rm d_1 \times d_2}$ and $\|A\|$ is the standard operator norm of $A$ [cf. Eq.~\eqref{eq:opnorm}].\par

In this section, we shall focus on the metrics induced by the Schatten $p$-norms with $p=1,2,\infty$. In order to avoid notational confusion with the so-called Ky Fan $k$-norms\footnote{These are the sum of the $k$ largest singular values of $A$.}, we adopt the alternative nomenclature \emph{trace norm} $\|\cdot\|_{\rm tr}$, \emph{Hilbert-Schmidt norm} $\|\cdot\|_{\rm HS}$ and \emph{spectral norm} $\|\cdot\|$ for the Schatten $1$-, $2$- and $\infty$-norms, respectively.

\subsubsection{The Trace Distance}\label{sec:ttd}
From Eq.~\eqref{eq:schatten_norms}, the trace norm of $A$ is simply $\|A\|_{\rm tr}\defeq\tr\sqrt{A^\dagger A}$. Although the \emph{trace distance} --- the metric induced by the trace norm --- should then be given by $\|A-B\|_{\rm tr}$, it is a common practice (adopted here) to define it as half of this number. We have already defined the trace distance between two density matrices in Eq.~\eqref{eq:trdist}. Below, we exploit the above formula for the trace norm and the hermiticity of density matrices to write
\begin{equation}\label{eq:trdist_again}
\D(\rho,\sigma)\defeq\tfrac{1}{2}\|\rho-\sigma\|_{\rm tr}=\tfrac{1}{2}\tr\left[\sqrt{(\rho-\sigma)^2}\right]\,.
\end{equation}

A number of alternative definitions of the trace distance are also known. For example,
\begin{align}\label{eq:trdist_varU}
\D(\rho,\sigma)\defeq \max_{W\in U(\rm{d})} |\tr\left[W (\rho-\sigma)\right]|\,,
\end{align}
where $U({\rm d})$ denotes the group of unitary matrices of dimension {\rm d}. Remarkably, the maximizing $W$ satisfies $W(\rho-\sigma) = |\rho-\sigma|$ (see \cite[Lemma 6]{94Jozsa2315} or \cite[pp. 43--44]{60Schatten} for a proof), so that we recover the definition of Eq.~\eqref{eq:trdist_again}. \footnote{Note that for an invertible matrix $\rho-\sigma$, $W^\dagger$ is the unitary arising from the (unique) left polar decomposition of $\rho-\sigma$ \cite[Theorem 2.3]{00Nielsen}. For non-invertible $\rho-\sigma$, $W$ is not unique but it does exist (see \cite[p. 53]{95Fuchs} and references therein).}

While restricted to the space of density matrices, the trace distance can also be defined as\footnote{See \cite[Lemma 4]{07Rastegin9533} for a closely related definition of the trace distance for arbitrary hermitian matrices.}
\begin{equation}
\D(\rho,\sigma)\defeq \max_{0\leq P\leq \openone_d} \tr\left[P (\rho-\sigma)\right]\,,
\end{equation}
or, alternatively, the optimization can be taken over all projectors $P$~\cite[pp. 404-405]{00Nielsen}.

As a final observation, we note that the trace distance is equal to half of the sum of the singular values of $\rho-\sigma$; or what amounts to be the same in the case of hermitian matrices, half of the sum of the modulus of the eigenvalues. This follows easily from the singular value decomposition $\rho-\sigma=U\Sigma V^\dagger$, where $U$ and $V$ are unitary matrices and $\Sigma$ is a diagonal PSD matrix. Substituted into Eq.~\eqref{eq:trdist_again}, we get
\begin{equation}\label{eq:D_SVD}
\D(\rho,\sigma)\defeq\frac{1}{2}\tr\Sigma\,,
\end{equation}
establishing the claimed result. This definition motivates the Matlab code shown in the Appendix~\ref{app:codes} for the numerical computation of $\D$.

\subsubsection{The Hilbert-Schmidt Distance}\label{sec:HSD}

Whenever an inner product is defined on a set, a norm can be immediately defined for each element of the set via the square-root of the inner product of that element with itself. In this framework, the Hilbert-Schmidt inner product of a matrix $A$ with itself gives rise to the \emph{Hilbert-Schmidt norm} $\|A\|_{\rm HS}\defeq\sqrt{\tr\left[A^\dagger A\right]}$. Clearly, this is precisely the norm arising from Eq.~\eqref{eq:schatten_norms} with $p=2$.\par

The \emph{Hilbert-Schmidt distance} between two density matrices $\rho$ and $\sigma$ is defined as
\begin{equation}\label{eq:HSdef1}
\H(\rho,\sigma)\defeq\|\rho-\sigma\|_{\rm HS}=\sqrt{\tr\left[\left(\rho-\sigma\right)^2\right]}\,.
\end{equation}
From Eq.~\eqref{eq:vec_trace}, it then follows that the Hilbert-Schmidt distance is the Euclidean norm of ${\rm vec}\left(\rho-\sigma\right)$, namely
\begin{equation}\label{eq:Hvectorized}
\H(\rho,\sigma)\defeq\sqrt{\left[{\rm vec}(\rho-\sigma)\right]^\dagger {\rm vec}(\rho-\sigma)}\,.
\end{equation}
This implies, for example, that $\|\rho-\sigma\|_{\rm HS}^2$ is merely the sum of the absolute values squared of every entry of $\rho-\sigma$. Moreover, Eq.~\eqref{eq:Hvectorized} motivates the Matlab code shown in Appendix~\ref{app:codes} for the computation of the Hilbert-Schmidt distance.

Yet another useful definition arises from the singular value decomposition $\rho-\sigma=U\Sigma V^\dagger$ applied to Eq.~\eqref{eq:HSdef1}. A straightforward calculation shows that \begin{equation}\label{eq:H_SVD}
\H^2(\rho,\sigma)\defeq\tr\left(\Sigma^2\right)\,,
 \end{equation}
or in words, the squared Hilbert-Schmidt distance between $\rho$ and $\sigma$ is the sum of the squared \emph{singular values} of $\rho-\sigma$. Due to the hermiticity of density matrices, we can make a further simplification and regard $\H^2(\rho,\sigma)$ as the sum of the squared \emph{eigenvalues} of $\rho-\sigma$.

It is interesting to note that although $\H$ is the induced metric by the Hilbert-Schmidt norm, the function $\H^2$ can also be shown to be a metric. This is proved in the Appendix~\ref{app:metricHsq}.

\subsubsection{The Spectral Distance}
The dual norm \cite{07Recht} of the trace norm is the so-called \emph{spectral norm} (also known as operator norm, Schatten infinity norm, etc). It is defined for an arbitrary matrix $A$ as
\begin{equation}\label{eq:opnorm}
\|A\|\defeq\max_{\|\bm{v}\|=1}{\|A \bm{v}\|}\
\end{equation}
where the norms appearing on the right hand side refer to the Euclidean norm for vectors $\bm{v}\in \mathbb{C}^{\rm d}$.\par

From this variational definition, it is possible to show that the spectral norm of $A$ is equal to the largest singular value of $A$, i.e., for the SVD $A=U \Sigma V^\dagger$, $\|A\|$ is the largest element of $\Sigma$. We can restate this in terms of the eigenvalues of the matrix $A^\dagger A=V \Sigma^2 V^\dagger$. Clearly, each diagonal element of $\Sigma^2$ is an eigenvalue of $A^\dagger A$, so we can write
\begin{equation}\label{eq:O_SVD}
\|A\|\defeq\sqrt{\lambda_{\sf max}\left[A^\dagger A\right]}=\lambda_{\sf max}\sqrt{A^\dagger A}=\lambda_{\sf max}|A|\,,
\end{equation}
where $\lambda_{\sf max}$ is an operator that extracts the largest eigenvalue of its argument. The spectral distance between two density matrices $\rho$ and $\sigma$ is thus defined as
\begin{equation}
\O(\rho,\sigma)=\lambda_{\rm max}|\rho-\sigma|
\end{equation}
Since $\rho-\sigma$ is a normal matrix (it is actually hermitian), simple diagonalization shows that the eigenvalues of $|\rho-\sigma|$ are simply the modulus of the eigenvalues of $\rho-\sigma$. This observation leads to the Matlab code given in Appendix~\ref{app:codes} for the computation of $\O$.

\subsection{Benchmarks of metrical distances}\label{sec:evaluatemetrics}

In this section we present an analysis of the metrics introduced above that parallels the study of the properties of $\F$ and $\Fn$ presented in Sec.~\ref{sec:altfid}. The reader will note, however, one omission: we do not evaluate the metrics against the criterion ``Consistency with Schumacher's fidelity''. Of course, given that metrics are measures of distance, they should not be expected to recover a measure of closeness in some special case.

\subsubsection{Jozsa's Axioms}
\begin{enumerate}
\item \textbf{Normalization.} The metric axioms (M1) and (M2) guarantee that for any metric we have $\|\rho-\sigma\|\geq 0$ with saturation iff $\rho=\sigma$. Noticeably, this establishes a slightly different normalization axiom than that one satisfied by the fidelity-like quantities $\F$, $\Fn$ and $Q$. While these saturate their \emph{upper bounds} when the states are identical, the metrical quantities saturate their \emph{lower bounds} in this case. Of course, this poses no conceptual difficulties and is merely a manifestation of the fact that the fidelity-like functions are \emph{closeness} measures, while metrics are authentic \emph{distance} measures.

    A few comments regarding the saturation of the upper bounds are pertinent: From the inequalities~\eqref{eq:fuchs_ineq}, it is easy to prove that $\D(\rho,\sigma)=1$ if and only if $\F(\rho,\sigma)=0$, which implies that the trace distance upper bound is achieved \emph{with and only with} a pair of orthogonal states. On the other hand, orthogonality does not suffice for the saturation of the upper bounds of $\H$ and $\O$. For example, consider the orthogonal mixed states $\widetilde{\rho}$ and $\widetilde{\sigma}$ of Eq.~\eqref{eq:Ozawa}. It is easy to compute that $\H(\widetilde{\rho},\widetilde{\sigma})=1 < \sqrt{2}$ and $\O(\widetilde{\rho},\widetilde{\sigma})=1/2 < 1$, where the rhs of the inequalities indicate the actual upper bounds of $\H$ and $\O$.

    A little thought shows that the upper bound of the spectral distance is achieved if one of the states is pure and orthogonal to the other state (which is allowed to be mixed). The upper bound of $\H(\rho,\sigma)$, in turn, is saturated if and only if $\Fn(\rho,\sigma)=0$ [this follows from the third inequality in Eq.~\eqref{eq:fuchs_like_bounds}], or equivalently, if and only if $\rho$ \emph{and} $\sigma$ are \emph{pure and orthogonal}.

\item \textbf{Symmetry.} By definition, every metric is symmetric [cf. axiom (M3) on page~\pageref{axioms:metric}].

\item \textbf{Unitary Invariance.} The three metrics studied here are unitarily invariant. This follows from the fact they can be solely expressed in terms of the singular values of $\rho-\sigma$, as shown in Eqs.~\eqref{eq:D_SVD}, \eqref{eq:H_SVD} and \eqref{eq:O_SVD}. Since the singular values of a matrix are invariant under unitary transformations, so are these metrics.

\end{enumerate}

\subsubsection{Convexity Properties}

Every \emph{induced} metric is jointly concave. This is proved in the following via a straightforward application of the norm axioms triangle inequality (N4) and positive scalability (N3).\par

For any non-negative scalars $p_1$ and $p_2$ (no need to require $p_1+p_2=1$) and density matrices $\rho_1$, $\rho_2$, $\sigma_1$ and $\sigma_2$ we can write
\begin{align}
\|(p_1\rho_1+p_2\rho_2)-(p_1\sigma_1+p_2\sigma_2)\|&=\|p_1(\rho_1-\sigma_1)+p_2(\rho_2-\sigma_2)\|\\
&\leq \|p_1(\rho_1-\sigma_1)\|+\|p_2(\rho_2-\sigma_2)\|\\
& = p_1\|\rho_1-\sigma_1\|+p_2\|\rho_2-\sigma_2\|
\end{align}
which establishes the desired property. Contrasted to the effort involved in the proof of joint concavity for $\Fn$, the above proof reveals the value of the underlying structure of induced metrics.

\subsubsection{Multiplicativity under Tensor Product}
With numerical examples, it is straightforward to see that none of the three metrics is multiplicative. In fact, both the trace norm and the operator norm are known to be supermultiplicative under tensor product \cite{05Belavkin062106}, which implies the supermultiplicativity of $\D$ and $\O$. Although this also seems to be the case for $\H$, we have not been able to find or produce a proof of this fact.

\subsubsection{Monotonicity under Quantum Operations}
In Ref.~\cite{94Ruskai1147}, Ruskai proved that $\D$ monotonically decreases (contracts) under arbitrary CPTP maps.

The monotonicity of $\H$ has a somewhat longer history. In the quantum information literature, this property was recognized as desirable feature for entanglement quantification in Refs.~\cite{97Vedral2275,98Vedral1619}. In Ref.~\cite{99Witte14}, a flawed proof of the contractivity of $\H$ was given; the error was detected by Ozawa in Ref.~\cite{00Ozawa158}, who provided an example of a CPTP map and a pair of $4\times 4$ density matrices for which $\H$ was seen to increase\footnote{Recall that with this same example we have shown in Sec.~\ref{sec:fnmonotone} that $\Fn$ is not monotonically decreasing under arbitrary CPTP maps.}.

In Ref.~\cite{06Perez-Garcia083506}, Perez-Garcia \etal showed that neither $\H$ nor $\O$ are generally contractive under arbitrary CPTP maps, but both of them are if restricted to the subset of \emph{unital} CPTP maps\footnote{A unital map $\mathcal{U}$ is characterized by having the identity matrix as a fixed point, i.e., $\mathcal{U}(\openone_{\rm d})=\openone_{\rm d}$.}. Moreover, in the case of qubit states, both $\H$ and $\O$ are contractive under arbitrary CPTP maps\footnote{This fact had already been anticipated by Nielsen in Ref.~\cite{0403Nielsen} for the case of $\H$.}, but this already fails to be true for $\H$ in the case of qutrit states.

\subsubsection{Bounds}
We have already shown in Eq.~\eqref{eq:fuchs_ineq} how $\D$ and $\F$ are related, and in Eqs.~\eqref{eq:upbndFn},~\eqref{eq:conjbnd} and \eqref{eq:polishbnd} how this relation can be modified to place bounds on $\D$ via functions of $\Fn$. Here, we prove the following inequalities between $\D$, $\H$ and $\O$:
\begin{equation}\label{eq:metricbounds}
\O(\rho,\sigma) \leq \H(\rho,\sigma) \leq 2\D(\rho,\sigma) \leq \sqrt{\mathfrak{r}}\,\H(\rho,\sigma) \leq \mathfrak{r}\O(\rho,\sigma)\,,
\end{equation}
where $\mathfrak{r}\defeq {\rm rank}(\rho-\sigma)$.

Apart from mutually relating the metrical distance measures, the sequence of inequalities above can be used in connection with one of the inequalities~\eqref{eq:fuchs_ineq}, \eqref{eq:HSFn}, \eqref{eq:conjbnd} or \eqref{eq:polishbnd} to relate each of metrics $\D$, $\H$ and $\O$ with $\F$ or $\Fn$. For example, it is straightforward to show that
\begin{equation}\label{eq:fuchs_like_bounds}
\begin{array}{rcccl}
\tfrac{2}{\sqrt{\mathfrak{r}}}\left[1-\sqrt{\F(\rho,\sigma)}\right]&\leq&\H(\rho,\sigma)&\leq&2\sqrt{1-\F(\rho,\sigma)}\\
\tfrac{2}{\mathfrak{r}}\left[1-\sqrt{\F(\rho,\sigma)}\right]&\leq&\O(\rho,\sigma)&\leq&2\sqrt{1-\F(\rho,\sigma)}\\
\frac{2}{\sqrt{\mathfrak{r}}}\left[1-\Fn(\rho,\sigma)\right]&\leq&\H(\rho,\sigma)&\leq&\sqrt{2\left[1-\Fn(\rho,\sigma)\right]}\\
\frac{2}{\mathfrak{r}}\left[1-\Fn(\rho,\sigma)\right]&\leq&\O(\rho,\sigma)&\leq&\sqrt{2\left[1-\Fn(\rho,\sigma)\right]}
\end{array}
\end{equation}
where the first inequalities in the third and fourth lines are conditioned on the validity of conjecture~\ref{conj:conjbndFnD}, but guaranteed to hold if $\Fn$ is replaced with $\sqrt{\Fn}$.

We now prove the inequalities in Eq.~\eqref{eq:metricbounds} which are, in fact, a simple restatement of well-known inequalities between the Schatten $p$-norms applied to the matrix $\rho-\sigma$.
The first inequality, for example, is the particular case $A=\rho-\sigma$ of the more general inequality $\|A\|\leq\|A\|_{\rm HS}$. To see that this holds, square and express each norm in terms of the singular values of $A$, $\sigma_i(A)$, to get
\begin{equation}
\max_i{\sigma_i(A)^2}\leq\sum_i\sigma_i(A)^2\,.
\end{equation}
This is obviously true since the term in the left is only one of the many non-negative summands in the right.

Likewise, the second inequality follows from $\|A\|_{\rm HS} \leq \|A\|_{\rm tr}$. This can be proved by squaring and using the singular value expressions to obtain
\begin{equation}
\sum_i\sigma_i(A)^2\leq\left[\sum_i\sigma_i(A)\right]^2\,,
\end{equation}
which clearly holds since the right hand side is a summation of non-negative terms involving every element appearing in the sum on the left hand side.

The third inequality follows from $\|A\|_{\rm tr}\leq \sqrt{{\rm rank}(A)}\|A\|_{\rm HS}$, which has already been stated and proved in this thesis [cf. Eq.~\eqref{eq:normtrhs} and the footnote on page~\pageref{footnote:prooftrhs}].

Finally, the fourth inequality follows from $\|A\|_{\rm HS}\leq\sqrt{{\rm rank}(A)}\|A\|$. Once again, this can be proved by squaring and expressing the norms in terms of the singular values,
\begin{equation}
\sum_i\sigma_i(A)^2\leq {\rm rank}(A)\left[\max_i\sigma_i(A)^2\right]\,,
\end{equation}
which is trivially true since ${\rm rank}(A)$ equals the number of non-zero singular values of $A$.

\subsubsection{Computational Efficiency}
As discussed in Sec.~\ref{sec:computability}, the algorithmic complexity in the computation of the trace distance is $O({\rm d}^3)$ due to the need of diagonalization of a matrix. However, this is considerably more efficient than other measures such as $\F$ or $Q$ since it only requires the computation of the eigenvalues (without eigenvectors). Clearly, the same applies for $\O$.

In contrast, the definition of $\H$ from Eq.~\eqref{eq:Hvectorized} enables its computation with $O({\rm d}^2)$ operations. $\H$ is thus the most efficiently computable metric among the ones considered here.

\begin{sidewaystable}
\centering
\caption[A summary of the properties of the studied distance measures]{A summary of the properties of the studied distance measures. In the column entitled `Monotonic under maps', we present the largest considered family of maps under which each measure is monotonic. The hierarchy of considered families is: CPTP maps $>$ unital CPTP maps $>$ projective measurements. The gray color indicates a conjecture supported by numerical evidence. In the column entitled `Related metrics', the dashes indicate that the corresponding quantities are already metrics.}
\begin{tabular}{c|c|c|c|c|c|c|c|c|c|c}
\hline\hline
& \multicolumn{4}{c|}{Jozsa's Axioms} & Concavity / & Multiplicativity & Monotonic & Related & Bounds & Comput.\\
\cline{2-5}
& Norm. & Sym. & U-inv. & Schum. & Convexity & (tensor product) & under maps& metrics &  & Complexity \\
\hline
$\F$ & yes & yes & yes & yes & sep. concave & multiplicative & CPTP & $A,B,C[\F]$ & P.~\pageref{eq:fuchs_ineq} & $O({\rm d}^3)$\\
$\Fn$ & yes & yes& yes & yes & joint concave & super multipl. & \textcolor{graytable}{(Proj. Meas.)} & $C[\Fn]$ & Pp.~\pageref{eq:DFn_boundrank},\pageref{eq:polishbnd} & $O({\rm d}^2)$\\
$\D$& yes & yes & yes & N/A & joint convex & super multpl.& CPTP & --- & Pp.~\pageref{eq:metricbounds},\pageref{eq:fuchs_like_bounds} & $O({\rm d}^3)$ \\
$\H$& yes & yes & yes & N/A & joint convex & \textcolor{graytable}{(super multipl.)} & unital CPTP & --- & Pp.~\pageref{eq:metricbounds},\pageref{eq:fuchs_like_bounds} & $O({\rm d}^2)$\\
$\O$& yes & yes & yes & N/A & joint convex & super multipl. & unital CPTP & --- & Pp.~\pageref{eq:metricbounds},\pageref{eq:fuchs_like_bounds} & $O({\rm d}^3)$\\
\hline\hline
\end{tabular}\label{table:metrical}
\end{sidewaystable}

\section{Distances between sequences of density matrices}\label{sec:distset}

So far in this chapter, we have studied a number of distance (and closeness) measures between density matrices. In this section we aim to generalize to a pair of \emph{sequences} of density matrices, the notions of distance introduced before to a pair of density matrices. The motivation for this is the construction of objective functions $\aver{\mathscr{D}}$ for problem~\eqref{eq:gen_problem}, introduced in the previous chapter as a general formulation of the problem of transforming between sequences of density matrices.

For any choice of measure $\mathscr{D}\in\{\D,\H,\O\,\F,\Fn\}$, we introduce two averaging schemes $\aver{\mathscr{D}}_1$ and $\aver{\mathscr{D}}_2$, each of which providing a different quantitative estimate of the distance between two sequences of $I\geq 2$ {\rm d}-dimensional density matrices $\left[\rho_i\right]_{i=1}^I$ and $\left[\sigma_i\right]_{i=1}^I$. These are defined as follows:
\begin{align}
\aver{\mathscr{D}(\rho_i,\sigma_i)}_1&\defeq\sum_{i=1}^{I}{\pi_i \mathscr{D}\left[\rho_i,\sigma_i\right]}\,,\label{eq:average1}\\
\aver{\mathscr{D}(\rho_i,\sigma_i)}_2&\defeq\mathscr{D}\left[\bigoplus_{i=1}^I\pi_i \rho_i,\bigoplus_{i=1}^I\pi_i\sigma_i\right]\,.\label{eq:average2}
\end{align}
Here, $\pi_i$ is a chosen probability distribution over the alphabet $i=1,\ldots,I$ and satisfying $\pi_i\in(0,1)$ and $\sum_{i=1}^I\pi_i=1$. The averaging scheme $\aver{\mathscr{D}}_1$ is simply a weighted average distance between the $i$-th elements of each sequence, while $\aver{\mathscr{D}}_2$ measures $\mathscr{D}$ between a single pair of $I{\rm d}$-dimensional density matrices formed by the weighted average --- with respect to the direct sum --- over the elements of each sequence.

By suitably choosing the values of $\pi_i$ in Eqs.~\eqref{eq:average1} and~\eqref{eq:average2} and running the optimization~\eqref{eq:gen_problem}, one is actually setting a hierarchy on the desired accuracy of the implementation of each ``atomic transformation'' $\rho_i\mapsto\overline{\rho}_i$. For example, if there is no precedence of an atomic transformation over the others, then the uniform probability distribution $\pi_i=1/I$ for every $i$ should be chosen. On the other extreme, if $\pi_i$ is chosen to be $1$ for some value of $i$, then the corresponding atomic transformation will be the only one that matters; in these circumstances, the special case of single-state transformation discussed in Sec.~\ref{sec:singlestate} is recovered. Since the single-state case has already been fully solved, we assume without loss of generality that $\pi_i\neq 1$ for any $i$. We can also assume $\pi_i\neq 0$ for every $i$, which is justified as follows: if certain atomic transformations are absolutely irrevelant, then instead of assigning zero weight we can simply remove the corresponding source and target states from the sequences.  Henceforth we shall refer to the point probabilities $\pi_i$ as \emph{priorities}.

The averaging schemes introduced above are interesting because they yield distance measures $\aver{\mathscr{D}}$ for sequences of density matrices which behave much in the same way as  $\mathscr{D}$ behaves for density matrices. This is made more precise in the following:

\begin{theorem}\label{thm:metricsequence}
The functions $\aver{\mathscr{D}(\rho_i,\sigma_i)}_{1,2}$ defined in Eqs.~\eqref{eq:average1} and~\eqref{eq:average2} are metrics for the space of sequences of density matrices if $\mathscr{D}$ is a metric for the space of density matrices.
\end{theorem}
\begin{proof}
The proof is a trivial verification of each one of the metric axioms stated on page~\pageref{axioms:metric}:
\begin{enumerate}
\item[(M1)] (Nonnegativity): That $\aver{\mathscr{D}(\rho_i,\sigma_i)}_1\geq 0$ for all sequences $\left[\rho_i\right]_{i=1}^I$ and $\left[\sigma_i\right]_{i=1}^I$ follows from $\mathscr{D}(\rho_i,\sigma_i)\geq 0$ for $i=1,\ldots,I$ and from the fact that a convex sum of non-negative terms is nonnegative. The nonnegativity of $\aver{\mathscr{D}(\rho_i,\sigma_i)}_2$ is an instance of the nonnegativity of $\mathscr{D}$ with block-diagonal density matrices.
\item[(M2)] (Identity of Indiscernibles): We first prove the `if direction' for both $\aver{\mathscr{D}}_1$ and $\aver{\mathscr{D}}_2$: If $\left[\rho_i\right]_{i=1}^I=\left[\sigma_i\right]_{i=1}^I$, we have $\aver{\mathscr{D}(\rho_i,\sigma_i)}_1=\aver{\mathscr{D}(\rho_i,\rho_i)}_1=\sum_i\pi_i 0 = 0$ and also $\aver{\mathscr{D}(\rho_i,\sigma_i)}_2=\aver{\mathscr{D}(\rho_i,\rho_i)}_2=\mathscr{D}\left(\bigoplus_i\pi_i\rho_i,\bigoplus_i\pi_i\rho_i\right)=0$. Remarkably, notice that the first equality in each case would not hold if we were dealing with sets (as opposed to sequences) of density matrices.

    Conversely, the condition $\aver{\mathscr{D}(\rho_i,\sigma_i)}_1=\sum_i\pi_i\mathscr{D}(\rho_i,\sigma_i)=0$ requires $\mathscr{D}(\rho_i,\sigma_i)=0$ because $\pi_i\neq 0$ for $i=1,\ldots,I$. Due to the metric property of $\mathscr{D}$, this can only happen if $\rho_i=\sigma_i$ for all $i$, or equivalently if $\left[\rho_i\right]_{i=1}^I=\left[\sigma_i\right]_{i=1}^I$.
    Likewise, the condition $\aver{\mathscr{D}(\rho_i,\sigma_i)}_2=\mathscr{D}(\bigoplus_i\pi_i\rho_i,\bigoplus_i\pi_i\sigma_i)=0$
    requires $\bigoplus_i\pi_i\rho_i=\bigoplus_i\pi_i\sigma_i$ which is equivalent to $\rho_i=\sigma_i$ since $\pi_i\neq 0$ for all $i$.
\item[(M3)] (Symmetry): Symmetry of $\aver{\mathscr{D}}_{1,2}$ follows trivially from the symmetry of $\mathscr{D}$.
\item[(M4)] (Triangle Inequality): For three sequences of density matrices $\left[\rho\right]_{i=1}^I$, $\left[\sigma\right]_{i=1}^I$ and $\left[\tau\right]_{i=1}^I$, the multiplication of each one of the valid triangle inequalities $\mathscr{D}(\rho_i,\tau_i)\leq \mathscr{D}(\rho_i,\sigma_i)+\mathscr{D}(\sigma_i,\tau_i)$
    by $\pi_i$ yields another set of $I$ valid inequalities. Summing over all of them, the triangle inequality for $\aver{\mathscr{D}}_1$ is established. The triangle inequality for $\aver{\mathscr{D}}_2$ is just an instance of the triangle inequality for $\mathscr{D}$ with block-diagonal density matrices.
\end{enumerate}
\end{proof}

Apart from the metric axioms, many other properties of $\mathscr{D}$ are inherited by $\aver{\mathscr{D}}_{1,2}$. For example, for any sequences $[\rho_i^{(1)}]_{i=1}^I$, $[\rho_i^{(2)}]_{i=1}^I$, $[\sigma_i^{(1)}]_{i=1}^I$, $[\sigma_i^{(2)}]_{i=1}^I$, and $p_1$ and $p_2$ non-negative numbers such that $p_1+p_2=1$, the inequalities
\begin{equation}
\aver{\mathscr{D}(p_1\rho_i^{(1)}+p_2\rho_i^{(2)},p_1\sigma_i^{(1)}+p_2\sigma_i^{(2)})}_{1,2}\leq p_1\aver{\mathscr{D}(\rho_i^{(1)},\sigma_i^{(1)})}_{1,2}+p_2\aver{\mathscr{D}(\rho_i^{(2)},\sigma_i^{(2)})}_{1,2}\,,
\end{equation}
(or the reversed inequality) follow straightforwardly from the joint convexity (concavity) of $\mathscr{D}$.

Likewise, if $[\varrho_i]_{i=1}^I=[\mathcal{E}(\rho_i)]_{i=1}^I$ and $[\varsigma_i]_{i=1}^I=[\mathcal{E}(\varsigma_i)]_{i=1}^I$, then
\begin{equation}\label{eq:monotsequence}
\aver{\mathscr{D}(\varrho_i,\varsigma_i)}_{1,2}\leq\aver{\mathscr{D}(\rho_i,\sigma_i)}_{1,2}\,,
\end{equation}
(or the reversed inequality) holds if $\mathscr{D}$ is monotonically decreasing (increasing) under the map $\mathcal{E}$. In particular, if $\mathcal{E}$ is a unitary map, saturation of inequality~\eqref{eq:monotsequence} can be shown to hold if $\mathscr{D}$ is unitarily invariant.

Due to the good properties of the distance measures $\D$, $\H$, $\O$, $\F$ and $\Fn$, and the parallelism between $\aver{\mathscr{D}}_{1,2}$ and $\mathscr{D}$, we consider the measures $\aver{\D}_{1,2}$, $\aver{\H}_{1,2}$, $\aver{\O}_{1,2}$, $\aver{\F}_{1,2}$ and $\aver{\Fn}_{1,2}$ sensible choices for quantifying distance between sequences of density matrices. In the next chapter, they are used as the objective function of problem~\eqref{eq:gen_problem} and some of the resulting problems are formulated as semidefinite programs.

\chapter{Assembling Semidefinite Programs for Quantum Control}\label{chap:assemble}

\section{Introduction}
At the end of Ch.~\ref{chap:background}, we proposed the general optimization
\begin{equation}\label{eq:gen_problem_ch4}
\min_{\mathcal{C} \in \mathcal{C}_{\rm d}^{\rm set}}\aver{\mathscr{D}[\mathcal{C}(\rho_i),\brho_i]}\,,
\end{equation}
as a formal expression for the problem of determining a quantum operation converting between sequences of density matrices. In the same chapter, we saw how this problem can be brought very close to the form of a SDP if $\mathcal{C}_{\rm d}^{\rm set}$ is chosen to be either the set of CPTP maps, $\mathcal{Q}_{\rm d}^{\rm set}$, or a relaxed version of the set of EBTP maps, $\widetilde{\mathcal{B}}_{\rm d}^{\rm set}$. In this chapter, the formulation of problem~\eqref{eq:gen_problem_ch4} over these sets is completed with the specification of $\aver{\mathscr{D}}$ using the distance measures introduced in Ch.~\ref{chap:dist_measure}. Our main results are the derivation of several SDPs formalizing our quantum control problem.

Once we have our SDPs assembled, we are in the position to numerically solve them and observe how different choices of distance measures influence the resulting optimal controllers. This analysis is first conducted in a qualitative basis, and restricts to some examples of qubit state transformations. Later on, a more general setting is considered, and preliminary numerical results lead to a quantitative description of the ``compatibility'' between a chosen distance measures and the remaining ones.

This chapter is divided as follows: In Sec.~\ref{sec:min_dist} we derive the SDPs to minimize metrical distances related to the trace distance, $\D$, the Hilbert-Schmidt distance $\H$ and the spectral distance $\O$. In Sec.~\ref{sec:max_close} we discuss the optimization problems arising from the closeness measures $\F$ and $\Fn$, and a SDP is derived in a particular case. Qualitative and quantitative comparisons between the numerical solutions of these problems are presented in Sec.~\ref{sec:ctrlsens}.

\section{Minimizing distances}\label{sec:min_dist}
In this section we consider the formulation of problem~\eqref{eq:gen_problem_ch4} with $\aver{\mathscr{D}}$ taken as one of the metrics\footnote{Note the omission of the metrics $\aver{\H}_1$ and $\aver{\O}_1$ from the roll of metrics considered here. Unfortunately, we have not been able to cast problem~\eqref{eq:gen_problem_ch4} as a SDP for these choices of $\aver{\mathscr{D}}$.} $\aver{\D}_ {1,2}$, $\aver{\H^2}_1$, $\aver{\H}_2$, $\aver{\O}_2$. Exploiting a number of algebraic tricks, we demonstrate how the minimization of these quantities can be cast as SDPs.

As a first trick that will be useful for all the choices listed above, we start by reexpressing problem~\eqref{eq:gen_problem_ch4} in the equivalent \emph{epigraph form} \cite{04Boyd}
\begin{align}
\mbox{minimize}\quad & t \nonumber\\
\mbox{subject to}\quad& \aver{\mathscr{D}} \leq t \label{eq:constassemble}\\
\quad&\mathcal{C} \in \mathcal{C}_{\rm d}^{\rm set}\,, \nonumber
\end{align}
Albeit a new variable $t$ is introduced here, the gain is that the objective function becomes linear in the problem variable. Since we have already seen in Sec.~\ref{sec:begin_optimize} how to handle the constraint $\mathcal{C}\in\mathcal{C}_{\rm d}^{\rm set}$ for the cases of interest [i.e., $\mathcal{C}_{\rm d}^{\rm set}=\mathcal{Q}_{\rm d}^{\rm set}$ and $\mathcal{C}_{\rm d}^{\rm set}=\widetilde{\mathcal{B}}_{\rm d}^{\rm set}$], it only remains to reexpress the new inequality constraint $\aver{\mathscr{D}}\leq t$ in the form of a linear matrix inequality. This is done next, case by case, for each choice of $\aver{\mathscr{D}}$.

\subsection{The Trace Distance}

In principle, we should start by establishing which averaging scheme, $\aver{\D}_1$ or $\aver{\D}_2$, is to be considered first. However, this is dispensable in the case of the trace distance since
\begin{equation}\label{eq:D1eqD2}
\aver{\D[\mathcal{C}(\rho_i),\brho_i]}_1=\aver{\D[\mathcal{C}(\rho_i),\brho_i]}_2\,,
\end{equation}
for all density matrices $\mathcal{C}(\rho_i)$ and $\brho_i$. This ``degeneracy'' follows from the identity $\sum_i\tr A_i = \tr\bigoplus_i A_i$, valid for every set of square matrices $\{A_i\}$. In particular, if $A_i$ is taken to be the diagonal matrix of the eigenvalues of $\pi_i[\mathcal{C}(\rho_i)-\brho_i]$, then we obtain Eq.~\eqref{eq:D1eqD2}. Because of this equivalence between $\aver{\D}_1$ and $\aver{\D}_2$, henceforth we shall adopt the simplified notation $\aver{\D}$, and write
\begin{equation}
\aver{\D[\mathcal{C}(\rho_i),\brho_i]}=\left\|\bigoplus_{i=1}^I\frac{\pi_i}{2}\left[\mathcal{C}(\rho_i)-\brho_i\right]\right\|_{\rm tr}\,.
\end{equation}

We now focus on the inequality constraint of Eq.~\eqref{eq:constassemble} which, thanks to the equation above, is of the form $\|A\|_{\rm tr}\leq t$. The following lemma is then immediately applicable:

\begin{lemma}[Fazel-Hindi-Boyd, \cite{01Fazel4734}]\label{lemma:fazel}
For any square matrix $A$ and $t \in \mathbb{R}$, $\|A\|_{\rm tr}\leq t$ if and only if there exists matrices $Y$ and $Z$ such that
\begin{equation}
\left(\begin{array}{cc}Y&A\\A^\dagger&Z\end{array}\right)\geq 0\quad\mbox{and}\quad\tr Y + \tr Z \leq 2t\,.\label{lemma:conds}
\end{equation}
\end{lemma}
\begin{proof}
Throughout, $U$ and $V$ are defined via the SVD of $A$, i.e., $A=U\Sigma V^\dagger$. We start by assuming that there exists matrices $Y$ and $Z$ such that conditions (\ref{lemma:conds}) hold and then we show that $\|A\|_{\rm tr}\defeq \tr\Sigma \leq t$.

Recall that the trace of the product of two PSD matrices is always non-negative, so
\begin{equation}
\tr\left[XX^\dagger \left(\begin{array}{cc}Y&A\\A^\dagger&Z\end{array}\right) \right]\geq 0\label{eq:posprodtr}
\end{equation}
for every matrix $X$. In particular, take $X^\dagger=\left(\begin{array}{cc}U^\dagger&-V^\dagger\end{array}\right)$ and expand Eq. (\ref{eq:posprodtr}). Exploiting the cyclic property of the trace and the SVD of $A$ one finds
$2 \tr\Sigma \leq \tr Y + \tr Z$, which implies $\tr\Sigma \leq t$.\par
Conversely, suppose $\tr\Sigma\leq t$ and let $Y=U\Sigma U^\dagger$ and $Z=V\Sigma V^\dagger$. Clearly, $\tr Y + \tr Z = 2 \tr \Sigma \leq 2t$. Also note that
\begin{equation}
\left(\begin{array}{cc}Y&A\\A^\dagger&Z\end{array}\right)=\left(\begin{array}{cc}U\Sigma U^\dagger&U\Sigma V^\dagger\\V\Sigma U^\dagger&V\Sigma V^\dagger\end{array}\right)=\left(\begin{array}{c}U\\V\end{array}\right)\Sigma\left(\begin{array}{cc}U^\dagger&V^\dagger\end{array}\right)\geq 0
\end{equation}
which concludes the proof.
\end{proof}
In the light of lemma \ref{lemma:fazel}, we obtain the following optimization problem for the minimization of $\aver{\mathscr{D}}$:
\begin{align}
\mbox{minimize}\quad & \tfrac{1}{2}\left(\tr Y + \tr Z\right) \nonumber\\
\mbox{subject to}\quad&  \left(\begin{array}{cc}Y&\displaystyle\bigoplus_{i=1}^I\tfrac{\pi_i}{2}\left[\mathcal{C}(\rho_i)-\brho_i\right]\\
\displaystyle\bigoplus_{i=1}^I\tfrac{\pi_i}{2}\left[\mathcal{C}(\rho_i)-\brho_i\right]&Z\end{array}\right)\label{eq:nomeeieesse} \geq 0\\
\quad&\mathcal{C} \in \mathcal{C}_{\rm d}^{\rm set} \nonumber
\end{align}
where we have redeemed from the epigraph form by eliminating the variable $t$, in such a way that the objective function becomes a linear function of the matrix variables $Y$ and $Z$.

The optimization problem above can already be recognized as a SDP. To make this more explicit (and to provide a useful form for its numerical implementation), we now reexpress problem~\eqref{eq:nomeeieesse} in the inequality form. This can be done by noting the following:
\begin{enumerate}
\item Due to the positivity requirement, $Y$ and $Z$ can be restricted to the set of hermitian matrices of dimension $I{\rm d}$. As such, they can be expanded in the bases introduced in Sec.~\ref{sec:herm_basis},
    \begin{equation}
    Y=\sum_{\alpha=1}^{(I{\rm d})^2}y_\alpha H_{I{\rm d}}^\alpha\quad\mbox{and}\quad Z=\sum_{\alpha=1}^{(I{\rm d})^2}z_\alpha H_{I{\rm d}}^\alpha\,.
    \end{equation}
    Such a choice, reduces the objective function $\tfrac{1}{2}\left(\tr Y + \tr Z\right)$ to the form $I{\rm d}(y_1+z_1)/2$.
\item The constraint $\mathcal{C}\in\mathcal{C}_{\rm d}^{\rm set}$, reformulated in terms of the expansion coefficients $x_{\mu,\nu}$ of the Choi matrix $\mathfrak{C}$ in the basis $\{H_{\rm d}^\mu\otimes H_{\rm d}^\nu\}_{\mu,\nu=1}^{\rm d^2}$, was shown in Sec.~\ref{sec:constraintQ} to be equivalent to
    \begin{equation}
\frac{\openone_{{\rm d}^2}}{\rm d}+\sum_{\substack{\mu=1\\\nu=2}}^{{\rm d}^2} x_{\mu,\nu} H_{\rm d}^{\mu}\otimes H_{\rm d}^\nu\geq 0\,,
    \end{equation}
in the case of $\mathcal{C}_{\rm d}^{\rm set}=\mathcal{Q}_{\rm d}^{\rm set}$, (i.e., when the optimization runs over the set of CPTP maps); and in Sec.~\ref{sec:constraintB}, to be equivalent to
    \begin{equation}
\frac{\openone_{2{\rm d}^2}}{\rm d}+\sum_{\substack{\mu=1\\\nu=2}}^{{\rm d}^2} x_{\mu,\nu} \left[\left(H_{\rm d}^{\mu}\otimes H_{\rm d}^\nu\right)\oplus\left(H_{\rm d}^{\mu}\otimes {H_{\rm d}^\nu}^{\sf T}\right)\right]\geq 0\,,
    \end{equation}
in the case of $\mathcal{C}_{\rm d}^{\rm set}=\widetilde{\mathcal{B}}_{\rm d}^{\rm set}$, (i.e., when the optimization runs over a certain subset of the set of CPTP maps which is, in general, a superset of the set of EBTP maps).

\item Combining Eqs.~\eqref{eq:jam_inv} and~\eqref{eq:choiexpandTPin}, the density matrix $\mathcal{C}(\rho_i)$ can also be expressed in terms of the coefficents $x_{\mu,\nu}$,
    \begin{equation}
    \mathcal{C}(\rho_i)=\frac{\openone_{\rm d}}{\rm d}+\sum_{\substack{\mu=1\\\nu=2}}^{{\rm d}^2} x_{\mu,\nu}\tr\left(\rho_i^{\sf T}H_{\rm d}^\mu\right)H_{\rm d}^\nu\,.
    \end{equation}

\end{enumerate}

Using these three facts, problem~\eqref{eq:nomeeieesse} can be rewritten as
\begin{align}
\mbox{minimize}&\quad \frac{I{\rm d}}{2}(y_1 + z_1) \quad \nonumber\\ \mbox{subject to}&\quad F_0+\sum_{\substack{\mu=1\\\nu=2}}^{{\rm d}^2} \sum_{\alpha=1}^{(I{\rm d})^{2}}\widetilde{F}(x_{\mu,\nu},y_\alpha,z_\alpha)
 \geq 0 \label{eq:SDP_ineqform_D}
\end{align}
where $\widetilde{F}(x_{\mu,\nu},y_\alpha,z_\alpha)$ is a \emph{linear function} of the variables $x_{\mu,\nu}$, $y_\alpha$, and $z_\alpha$, as explicitly shown below
\begin{multline}
\widetilde{F}(x_{\mu,\nu},y_\alpha,z_\alpha)\defeq\\
\left(
\begin{array}{c@{}c@{}c}
\begin{array}{|c:c|}\hline
y_\alpha H_{I{\rm d}}^{\alpha}  & x_{\mu,\nu}\displaystyle\bigoplus_{i=1\BB}^{I\TT}\frac{\pi_i}{2}\tr\left(\rho_i^{\sf T} H_{\rm d}^\mu\right)H_{\rm d}^\nu\\\hdashline
x_{\mu,\nu}\displaystyle\bigoplus_{i=1 \BB}^{I\TT}\frac{\pi_i}{2}\tr\left(\rho_i^{\sf T} H_{\rm d}^\mu\right)H_{\rm d}^\nu  & z_\alpha H_{I{\rm d}}^\alpha\\\hline
\end{array}&&\\
&
\begin{array}{|c|}\hline
x_{\mu,\nu} H_{\rm d}^{\mu}\otimes H_{\rm d}^\nu\TT \BB\\\hline
\end{array}&\\
&&\begin{array}{c}
\colCell{gray}{x_{\mu,\nu} H_{\rm d}^{\mu}\otimes {H_{\rm d}^\nu}^{\sf T}}\TT\BB\\
\end{array}
\end{array}
\right)\,,
\end{multline}
and $F_0$ is a constant matrix given by
\begin{equation}
F_0\defeq\left(
\begin{array}{c@{}c@{}c}
\begin{array}{|c:c|}\hline
 \TT \BB \text{\large{0}}_{I{\rm d}}&\displaystyle\bigoplus_{i=1 \BB}^{I\TT}\frac{\pi_i}{2}\left(\frac{\openone_{\rm d}}{\rm d}-\brho_i\right)\\\hdashline
\displaystyle\bigoplus_{i=1\BB}^{I\TT}\frac{\pi_i}{2}\left(\frac{\openone_{\rm d}}{\rm d}-\brho_i\right) & \text{\large{0}}_{I{\rm d}} \\\hline
\end{array}& &\\
 &
\begin{array}{|c|}\hline
\frac{\openone_{{\rm d}^2}}{\rm d}\TT \BB\\\hline
\end{array}&\\
&&\begin{array}{c}
\colCell{gray}{\frac{\openone_{{\rm d}^2}}{\rm d}}\TT \BB\\
\end{array}
\end{array}
\right)\,,
\end{equation}
where the highlighted blocks are only considered if $\mathcal{C}_{\rm d}^{\rm set}=\widetilde{\mathcal{B}}_{\rm d}^{\rm set}$.

\subsection{The Hilbert-Schmidt Distance}

As already mentioned, we have been unable to write the minimization of the metric $\aver{\H}_1$ as a SDP. Essentially, the difficulty arises in dealing with the square-root in the definition of this quantity. Indeed, the minimization of the squared version $\aver{\H^2}_1$ --- which is also a metric (cf. Theorem~\ref{thm:metricsequence} and Appendix~\ref{app:metricHsq}) --- can be cast as a SDP. This is demonstrated in Sec.~\ref{sec:minHav}.

In Sec.~\ref{sec:minHalt}, we show how the minimization of $\aver{\H^2}_2$ can also be written as a SDP. Since $\aver{\H^2}_2$ and $\aver{\H}_2$ are monotonically related via $\aver{\H}_2=\sqrt{\aver{\H^2}_2}$, the same SDP will also provide the quantum operation minimizing $\aver{\H}_2$.

Finally, in Sec.~\ref{sec:equivalenceHaltav}, we prove that if every atomic transformation is to be implemented with the same priority, i.e. $\pi_i=1/I$ for $i=1,\dots,I$, then it is irrelevant whether we minimize $\aver{\H^2}_1$ or $\aver{\H}_2$, since both problems lead to the same optimal quantum operation.

\subsubsection{Minimizing $\aver{\H^2}_1$}\label{sec:minHav}
From the definition~\eqref{eq:Hvectorized} of $\H$, we can write $\aver{\H^2}_1$ as
\begin{align}
\aver{\H^2[\mathcal{C}(\rho_i),\brho_i]}_1&=\sum_{i=1}^I \pi_i \left\{{\rm vec}\left[\mathcal{C}(\rho_i)-\brho_i\right]\right\}^\dagger{\rm vec}\left[\mathcal{C}(\rho_i)-\brho_i\right]\nonumber\\
 &= \bm{v}_\mathcal{C}^\dagger\,\Pi\; \bm{v}_\mathcal{C}\,,
\end{align}
where we have defined the diagonal matrix $\Pi\defeq\bigoplus_{i=1}^I\pi_i\openone_{{\rm d}^2}$ and the vector
\begin{equation}
\bm{v}_\mathcal{C}\defeq
\left(\begin{array}{c}
{\rm vec}\left[\mathcal{C}(\rho_1)-\brho_1\right]\\
{\rm vec}\left[\mathcal{C}(\rho_2)-\brho_2\right]\\
\vdots\\
{\rm vec}\left[\mathcal{C}(\rho_I)-\brho_I\right]
\end{array}\right)\,.
\end{equation}
With these provisions, the inequality constraint of Eq.~\eqref{eq:constassemble} assumes the form
 \begin{equation}\label{eq:prepSchurineq}
t-\bm{v}_\mathcal{C}^\dagger\;\Pi\;\bm{v}_\mathcal{C}\geq 0\,,
 \end{equation}
and can be reformulated as a linear matrix inequality with the aid of the following lemma:
\begin{lemma}[Schur complement condition for positive semidefiniteness~\cite{05Zhang,04Boyd}]\label{lemma:schur}
Let $A$ be a Hermitian matrix partitioned as
\begin{equation}
A=\left(\begin{array}{cc}A_{11}&A_{12}\\A_{12}^\dagger& A_{22}\end{array}\right)\,,
\end{equation}
in which $A_{11}$ is square and nonsingular. The \emph{Schur complement of $A$ with respect to $A_{11}$} is defined as ${\rm Sch}(A|A_{11})\defeq A_{22}-A_{12}^\dagger A_{11}^{-1} A_{12}$, and we have $A\geq 0$ if and only if $A_{11}>0$ and ${\rm Sch}(A|A_{11}) \geq 0$.
\end{lemma}

Because the inequality~\eqref{eq:prepSchurineq} is precisely  ${\rm Sch}(M(\mathcal{C},t),\Pi^{-1})\geq 0$ for a matrix $M(\mathcal{C},t)$ defined as
\begin{equation}
M(\mathcal{C},t)\defeq\left(\begin{array}{cc}\Pi^{-1} & \bm{v}_\mathcal{C}\\ \bm{v}_\mathcal{C}^\dagger & t \end{array}\right)\,,
\end{equation}
and because $\Pi^{-1}>0$, Eq.~\eqref{eq:prepSchurineq} can be reexpressed as $M(\mathcal{C},t) \geq 0$ and we arrive at the following optimization problem:
\begin{align}
\mbox{minimize}&\quad t \nonumber\\
\mbox{subject to}&\quad \left(\begin{array}{cc}\Pi^{-1} & \bm{v}_\mathcal{C}\\ \bm{v}_\mathcal{C}^\dagger & t \end{array}\right)\geq 0 \label{eq:nomeeieesseHav}\\
&\quad \mathcal{C}\in\mathcal{C}_d^{\rm set}\,.\nonumber
\end{align}

Just as done in the last section, by adopting the Choi matrix representation of $\mathcal{C}$ and expanding it on a tensor product basis, we obtain a SDP in the inequality form:
\begin{align}
\mbox{minimize}&\quad t  \nonumber\\
\mbox{subject to}&\quad F_0+\sum_{\substack{\mu=1\\\nu=2}}^{{\rm d}^2} x_{\mu,\nu}F_{\mu,\nu} + t T \geq 0\,, \label{eq:SDP_ineqform_H}
\end{align}
where $T=(0_{I{\rm d}^2}\oplus 1\oplus 0_{{\rm d}^2})$ if the optimization is taken over the set $\mathcal{Q}_{\rm d}^{\rm set}$, or $T=(0_{I{\rm d}^2}\oplus 1\oplus 0_{2{\rm d}^2})$ if optimizing over the set  $\widetilde{\mathcal{B}}_{\rm d}^{\rm set}$. The matrices $F_{\mu,\nu}$ and $F_0$ are shown below:
\begin{equation}\label{eq:FmunuHsq1}
F_{\mu,\nu}\defeq
\left(
\begin{array}{c@{}c@{}c}
\begin{array}{|ccc:c|}\hline
&&&\tr\left(\rho_1^{\sf T}H_{\rm d}^\mu\right)\bm{u}_\nu\TT\\
&\text{\large{0}}_{I{\rm d}^2}&&\vdots\\
&&&\tr\left(\rho_I^{\sf T}H_{\rm d}^\mu\right)\bm{u}_\nu\BB\\\hdashline
\tr\left(\rho_1^{\sf T}H_{\rm d}^\mu\right)\bm{u}_\nu^\dagger&\cdots&\tr\left(\rho_I^{\sf T}H_{\rm d}^\mu\right)\bm{u}_\nu^\dagger&0\TT\BB\\\hline
\end{array}&&\\
&
\begin{array}{|c|}\hline
H_{\rm d}^\mu\otimes H_{\rm d}^{\nu}\TT \BB\\\hline
\end{array}&\\
&&\begin{array}{c}
\colCell{gray}{H_{\rm d}^\mu\otimes {H_{\rm d}^{\nu}}^{\sf T}}\TT \BB\\
\end{array}
\end{array}
\right)
\,,
\end{equation}
\begin{equation}
F_0\defeq
\left(
\begin{array}{c@{}c@{}c}
\begin{array}{|ccc:c|}\hline
&&&{\rm vec}\,\left(\frac{\openone_{\rm d}}{\rm d}-\brho_1\right)\TT\\
&\Pi^{-1}&&\vdots\\
&&&{\rm vec}\,\left(\frac{\openone_{\rm d}}{\rm d}-\brho_I\right)\BB\\\hdashline
\left[{\rm vec}\left(\frac{\openone_{\rm d}}{\rm d}-\brho_1\right)\right]^\dagger&\cdots&\left[{\rm vec}\left(\frac{\openone_{\rm d}}{\rm d}-\brho_I\right)\right]^\dagger\TT\BB&0\\\hline
\end{array}&&\\
&
\begin{array}{|c|}\hline
\frac{\openone_{{\rm d}^2}}{\rm d}\TT \BB\\\hline
\end{array}&\\
&&\begin{array}{c}
\colCell{gray}{\frac{\openone_{{\rm d}^2}}{\rm d}\TT \BB}\\
\end{array}
\end{array}
\right)
\,,
\end{equation}
where $\bm{u}_\nu$ is shorthand notation for ${\rm vec}\,H_{\rm d}^\nu$ and the highlighted blocks only occur if the optimization is taken over $\widetilde{\mathcal{B}}_{\rm d}^{\rm set}$.

\subsubsection{Minimizing $\aver{\H}_2$ or $\aver{\H^2}_2$ }\label{sec:minHalt}
As argued before, the monotonicity between $\aver{\H^2}_2$ and $\aver{\H}_2$ guarantees that the quantum operation $\mathcal{C}$ minimizing $\aver{\H^2[\mathcal{C}(\rho_i),\brho_i]}_2$ also minimizes $\aver{\H[\mathcal{C}(\rho_i),\brho_i]}_2$. In what follows, we derive a SDP for the minimization of $\aver{\H^2[\mathcal{C}(\rho_i),\brho_i]}_2$.

Once again, we employ the definition of $\H$ from Eq.~\eqref{eq:Hvectorized} to write $\aver{\H^2[\mathcal{C}(\rho_i),\brho_i]}_2=\bm{w}_\mathcal{C}^\dagger\bm{w}_\mathcal{C}$,
where $\bm{w}_\mathcal{C}$ is given by
\begin{equation}
\bm{w}_\mathcal{C}\defeq{\rm vec}\,\bigoplus_{i=1}^I\pi_i\left[\mathcal{C}(\rho_i)-\brho_i\right]\,.
\end{equation}

Using $\bm{w}_\mathcal{C}^\dagger\bm{w}_\mathcal{C}$ in the inequality constraint of problem~\eqref{eq:constassemble}, we find $t-\bm{w}_\mathcal{C}^\dagger\bm{w}_\mathcal{C}\geq 0$, which, by the application of Lemma~\ref{lemma:schur}, leads to the following equivalent optimization problem:

\begin{align}
\mbox{minimize}&\quad t \nonumber\\
\mbox{subject to}&\quad \left(\begin{array}{cc}\openone_{(I{\rm d})^2} & \bm{w}_\mathcal{C}\\ \bm{w}_\mathcal{C}^\dagger & t \end{array}\right)\geq 0 \label{eq:nomeeieesseHalt}\\
&\quad \mathcal{C}\in\mathcal{C}_d^{\rm set}\,.\nonumber
\end{align}

Reexpressing this problem in terms of the suitable expansions of the Choi matrix of $\mathcal{C}$, we arrive at a SDP of the same form given in Eq.~\eqref{eq:SDP_ineqform_H}, but with different values of $T$, $F_{\mu,\nu}$ and $F_0$. In this case, $T=(0_{(I{\rm d})^2}\oplus 1\oplus 0_{{\rm d}^2})$ if the optimization is taken over the set $\mathcal{Q}_{\rm d}^{\rm set}$, and $T=(0_{(I{\rm d})^2}\oplus 1\oplus 0_{2{\rm d}^2})$ if the feasible set is chosen to be $\widetilde{\mathcal{B}}_{\rm d}^{\rm set}$. The matrices $F_{\mu,\nu}$ and $F_0$, in turn, are defined below according to the convention that the highlighted blocks are to be considered only if the optimization is taken over $\widetilde{\mathcal{B}}_{\rm d}^{\rm set}$:

\begin{equation}\label{eq:FmunuH2}
F_{\mu,\nu}\defeq\left(
\begin{array}{c@{}c@{}c}
\begin{array}{|c:c|}\hline
\text{\large{0}}_{(I{\rm d})^2}  & \displaystyle{\rm vec}\bigoplus_{i=1\BB}^{I\TT}\pi_i\tr\left(\rho_i^{\sf T}H_{\rm d}^\mu\right)H_{\rm d}^\nu\\\hdashline
\displaystyle\left[{\rm vec}\bigoplus_{i=1}^{I}\pi_i\tr\left(\rho_i^{\sf T}H_{\rm d}^\mu\right)H_{\rm d}^\nu\right]^{\dagger \TT}_{\BB} & 0\\\hline
\end{array}&&\\
&
\begin{array}{|c|}\hline
H_{\rm d}^{\mu}\otimes H_{\rm d}^\nu\TT \BB\\\hline
\end{array}&\\
&&\begin{array}{c}
\colCell{gray}{H_{\rm d}^{\mu}\otimes {H_{\rm d}^\nu}^{\sf T}}\TT\BB\\
\end{array}
\end{array}
\right)\,,
\end{equation}

\begin{equation}
F_0\defeq\left(
\begin{array}{c@{}c@{}c}
\begin{array}{|c:c|}\hline
 \TT \BB \openone_{(I{\rm d})^2}&\displaystyle{\rm vec}\bigoplus_{i=1 \BB}^{I\TT}\pi_i\left(\frac{\openone_{\rm d}}{\rm d}-\brho_i\right)\\\hdashline
\displaystyle\left[{\rm vec}\bigoplus_{i=1}^{I}\pi_i\left(\frac{\openone_{\rm d}}{\rm d}-\brho_i\right)\right]^{\dagger \TT}_{\BB} & 0 \\\hline
\end{array}& &\\
 &
\begin{array}{|c|}\hline
\frac{\openone_{{\rm d}^2}}{\rm d}\TT \BB\\\hline
\end{array}&\\
&&\begin{array}{c}
\colCell{gray}{\frac{\openone_{{\rm d}^2}}{\rm d}}\TT \BB\\
\end{array}
\end{array}
\right)\,.
\end{equation}
While the minimal $t$ satisfying the constraint of Eq.~\eqref{eq:SDP_ineqform_H} gives the optimal value of $\aver{\H^2}_2$, its square root gives the optimal value of $\aver{\H}_2$.

\subsubsection{Monotonicity between $\aver{\H^2}_1$ and $\aver{\H}_2$}\label{sec:equivalenceHaltav}

In Sec.~\ref{sec:HSD}, we saw that the squared Hilbert-Schmidt distance between two density matrices $\rho$ and $\sigma$ can be expressed as the sum of the eigenvalues of $\rho-\sigma$ squared, that is, $\tr\left(\Lambda^2\right)$ for $\Lambda$ the diagonal matrix of elements given by the eigenvalues of $\rho-\sigma$. As a result, we can write
\begin{equation}
\aver{\H^2[\mathcal{C}(\rho_i),\brho_i]}_1=\sum_{i=1}^I\pi_i \tr\left(\Lambda_i^2\right)\,,
\end{equation}
for $\Lambda_i$ the diagonal matrix of eigenvalues of $\mathcal{C}(\rho_i)-\brho_i$. Likewise, a simple computation gives
\begin{equation}
\aver{\H[\mathcal{C}(\rho_i),\brho_i]}_2=\sqrt{\sum_{i=1}^I\pi_i^2\tr\left(\Lambda_i^2\right)}\,.
\end{equation}
Due to the different exponents of $\pi_i$ in each equation, the relationship between $\aver{\H^2}_1$ and $\aver{\H}_2$ is not, in general, monotonic. However, it is easy to see that monotonicity takes place if every $\pi_i$ is equally chosen to be $1/I$. In this case, the priorities can be factored out of the sum to give $x/I$ in the first case and $\sqrt{x}/I$ in the second, where $x\defeq\sum_{i=1}^I\tr\left(\Lambda_i^2\right)$. As a result, $\aver{\H^2[\mathcal{C}(\rho_i),\brho_i]}_1$ and $\aver{\H[\mathcal{C}(\rho_i),\brho_i]}_2$ are minimized with the same $\mathcal{C}$ in the case of uniform priorities.

Notice that this is a practically useful observation. As a quick glance at Eqs.~\eqref{eq:FmunuHsq1} and~\eqref{eq:FmunuH2} shows, the dimension of the matrix inequality constraint arising from the optimization of $\aver{\H^2}_1$ scales linearly with $I$, while in the case of $\aver{\H}_2$ the scaling is quadratic. It then follows that the minimizer of $\aver{\H}_2$ can be obtained with a quadratically smaller computational cost if $\pi_i=1/I$ for all $i$.

\subsection{The Spectral Distance}

For $\aver{\mathscr{D}}=\aver{\O}_2$, the inequality constraint of problem~\eqref{eq:constassemble} can be written in terms of the spectral norm, as follows:
\begin{equation}
\left\|\bigoplus_{i=1}^I\pi_i\left[\mathcal{C}(\rho_i)-\brho_i\right]\right\|\leq t\,.
\end{equation}
A well-known reformulation of this type of inequality is given by the following lemma:

\begin{lemma}\label{lemma:sp_dist}
For any {\rm d}-dimensional matrix $A$ and scalar $t \geq 0$, $\|A\|\leq t$ if and only if $A^\dagger A \leq t^2\openone_{\rm d}$
\end{lemma}
\begin{proof}
First assume $\|A\|\defeq\sqrt{\lambda_{\rm max}(A^\dagger A)}\leq t$, then $\lambda_{\rm max}(A^\dagger A)\leq t^2$ which implies that the diagonal form of $A^\dagger A$ satisfies $V^\dagger A^\dagger A V \leq t^2\openone_{\rm d}$. This is trivially equivalent to $A^\dagger A \leq t^2\openone_{\rm d}$. Conversely, assume $A^\dagger A - t^2\openone_{\rm d}\leq 0$ and let $\bm{v}_{\rm max}$ be a normalized eigenvector of $A^\dagger A$ such that $A^\dagger A \bm{v}_{\rm max} = \lambda_{\rm max}(A^\dagger A)\bm{v}_{\rm max}\defeqinv\|A\|^2\bm{v}_{\rm max}$. Then $\bm{v}_{\rm max}^\dagger\left(A^\dagger A-t^2\openone_{\rm d}\right)\bm{v}_{\rm max}=\|A\|^2-t^2=(\|A\|+t)(\|A\|-t)\leq 0$, which implies $\|A\|-t\leq 0$.
\end{proof}
So, the minimization~\eqref{eq:constassemble} becomes
\begin{align}
\mbox{minimize}\quad & t \nonumber\\
\mbox{subject to}\quad&  t\openone_{I{\rm d}} -\left\{\bigoplus_{i=1}^I\pi_i\left[\mathcal{C}(\rho_i)-\brho_i\right]\right\}\; t^{-1}\; \left\{\bigoplus_{i=1}^I\pi_i\left[\mathcal{C}(\rho_i)-\brho_i\right] \right\} \geq 0\\
\quad&\mathcal{C} \in \mathcal{C}_{\rm d}^{\rm set}\,, \nonumber
\end{align}
which, after recognition of the Schur complement (cf. Lemma \ref{lemma:schur}) on the lhs of the inequality constraint, is equivalent to
\begin{align}
\mbox{minimize}\quad & t \nonumber\\
\mbox{subject to}\quad&  \left(\begin{array}{cc}t\openone_{I{\rm d}}&\displaystyle\bigoplus_{i=1}^I\pi_i\left[\mathcal{C}(\rho_i)-\brho_i\right]\\
\displaystyle\bigoplus_{i=1}^I\pi_i\left[\mathcal{C}(\rho_i)-\brho_i\right]&t\openone_{I{\rm d}}\end{array}\right) \geq 0\\
\quad&\mathcal{C} \in \mathcal{C}_{\rm d}^{\rm set}\,. \nonumber
\end{align}

Following the same protocol from previous sections, namely, reexpressing the map $\mathcal{C}$ in terms of a suitable expansion of its Choi matrix, the problem is reduced to the inequality form given in Eq.~\eqref{eq:SDP_ineqform_H} with $T=\openone_{2I{\rm d}}\oplus 0_{{\rm d}^2}$ for $\mathcal{C}_{\rm d}=\mathcal{Q}_{\rm d}^{\rm set}$ and $T=\openone_{2I{\rm d}}\oplus 0_{2{\rm d}^2}$ for $\mathcal{C}_{\rm d}=\widetilde{\mathcal{B}}_{\rm d}^{\rm set}$. Furthermore, the matrices $F_{\mu,\nu}$ and $F_0$ are as follows:

\begin{equation}
F_{\mu,\nu}\defeq\left(
\begin{array}{c@{}c@{}c}
\begin{array}{|c:c|}\hline
\text{\large{0}}_{I{\rm d}}  & \displaystyle\bigoplus_{i=1\BB}^{I\TT}\pi_i\tr\left(\rho_i^{\sf T}H_{\rm d}^\mu\right)H_{\rm d}^\nu\\\hdashline
\displaystyle\bigoplus_{i=1\BB}^{I\TT}\pi_i\tr\left(\rho_i^{\sf T}H_{\rm d}^\mu\right)H_{\rm d}^\nu  & \text{\large{0}}_{I{\rm d}}\\\hline
\end{array}&&\\
&
\begin{array}{|c|}\hline
H_{\rm d}^{\mu}\otimes H_{\rm d}^\nu\TT \BB\\\hline
\end{array}&\\
&&\begin{array}{c}
\colCell{gray}{H_{\rm d}^{\mu}\otimes {H_{\rm d}^\nu}^{\sf T}}\TT\BB\\
\end{array}
\end{array}
\right)\,,
\end{equation}

\begin{equation}
F_0\defeq\left(
\begin{array}{c@{}c@{}c}
\begin{array}{|c:c|}\hline
 \TT \BB \text{\large{0}}_{I{\rm d}}&\displaystyle\bigoplus_{i=1 \BB}^{I\TT}\pi_i\left(\frac{\openone_{\rm d}}{\rm d}-\brho_i\right)\\\hdashline
\displaystyle\bigoplus_{i=1\BB}^{I\TT}\pi_i\left(\frac{\openone_{\rm d}}{\rm d}-\brho_i\right) & \text{\large{0}}_{I{\rm d}} \\\hline
\end{array}& &\\
 &
\begin{array}{|c|}\hline
\frac{\openone_{{\rm d}^2}}{\rm d}\TT \BB\\\hline
\end{array}&\\
&&\begin{array}{c}
\colCell{gray}{\frac{\openone_{{\rm d}^2}}{\rm d}}\TT \BB\\
\end{array}
\end{array}
\right)\,,
\end{equation}
where the highlighted blocks are considered only in the case $\mathcal{C}_{\rm d}^{\rm set}=\widetilde{\mathcal{B}}_{\rm d}^{\rm set}$.

\section{Maximizing Closeness}\label{sec:max_close}

In this section we continue to derive SDP expressions for the control problem of interest, but here $\aver{\mathscr{D}}$ is taken to be a measure of closeness between sequences of density matrices, as opposed to the metrics considered in the previous section. In particular, we shall discuss the maximization of the fidelity-like quantities $\aver{\F}_{1,2}$ and $\aver{\Fn}_{1,2}$, introduced in the last chapter.

To account for this ``inversion'' on how distances are measured, the general optimization problem~\eqref{eq:gen_problem_ch4} must have the minimization replaced with a maximization,
\begin{equation}\label{eq:gen_problem_closeness}
\max_{\mathcal{C} \in \mathcal{C}_{\rm d}^{\rm set}}\aver{\mathscr{D}[\mathcal{C}(\rho_i),\brho_i]}\,.
\end{equation}
The concavity properties of $\aver{\F}_{1,2}$ and $\aver{\Fn}_{1,2}$ (along with the convexity of the constraint $\mathcal{C}\in\mathcal{C}_{\rm d}^{\rm set}$), ensure that for any choice of $\aver{\mathscr{D}}\in\{\aver{\F}_{1,2}, \aver{\Fn}_{1,2}\}$ we obtain a convex optimization. However, we have been unsuccessful in providing a SDP formulation of these optimization problems in the general case of arbitrary priorities and mixed states. For this reason, in what follows we restrict to the maximization of the functions $\aver{\F_{\rm HS}}_{1,2}$, where \begin{equation}
\F_{\rm HS}[\mathcal{C}(\rho_i),\brho_i]\defeq\tr\left[\mathcal{C}(\rho_i),\brho_i\right]\,,
\end{equation}
is the Hilbert-Schmidt inner product between the density matrices $\mathcal{C}(\rho_i)$ and $\brho_i$.

The motivation for this function as a replacement for $\F$ and $\Fn$ is as follows: If the target sequence is exclusively composed by pure states, then we have already seen that both $\F$ and $\Fn$ recover the Schumacher fidelity --- which is, in fact, $\F_{\rm HS}$. This implies that $\aver{\F_{\rm HS}[\mathcal{C}(\rho_i),\ket{\overline{\psi}_i}]}_1=\aver{\F[\mathcal{C}(\rho_i),\ket{\overline{\psi}_i}]}_1=
\aver{\Fn[\mathcal{C}(\rho_i),\ket{\overline{\psi}_i}]}_1$, where $\brho_i=\ket{\overline{\psi}_i}\!\bra{\overline{\psi}_i}$ for all $i$, and thus the SDPs we shall obtain actually maximize any of the three quantities above\footnote{Note, however, that we \emph{cannot} generally write $\aver{\F_{\rm HS}[\mathcal{C}(\rho_i),\ket{\overline{\psi}_i}]}_2=\aver{\F[\mathcal{C}(\rho_i),\ket{\overline{\psi}_i}]}_2=
\aver{\Fn[\mathcal{C}(\rho_i),\ket{\overline{\psi}_i}]}_2$. The direct sum taken over the pure target states multiplied by $\pi_i$, effectively turn them into mixed states of a larger dimensional Hilbert space. In this case, the equivalence between $\F_{\rm HS}$, $\F$ and $\Fn$ is generally invalid.}.

The situation is not as well justified when source and target sequences are mixed states. This can be anticipated by an evaluation of $\F_{\rm HS}$ against the criteria of the last chapter (cf. Table~\ref{table:metrical}). In its favor, $\F_{\rm HS}$ has the properties of symmetry, unitary invariance, compliance with Schumacher's fidelity, separate linearity, multiplicativity under tensor product and computational complexity $O({\rm d}^2)$. However, it fails to be monotonic even under projective measurements and we have not been able to determine a related metric --- in particular, none of $A[\F_{\rm HS}]$, $B[\F_{\rm HS}]$ or $C[\F_{\rm HS}]$ are metrics for the space of density matrices. Even more serious is the fact that, in general, $\F_{\rm HS}(\rho,\sigma)$ does not achieve its maximal value when $\rho=\sigma$.

In spite of this, the maximization of $\aver{\F_{\rm HS}}_{1,2}$ in the case of mixed states is motivated as follows: From the inequalities $\F_{\rm HS} \leq \F \leq \Fn$ (see Refs.~\cite{94Jozsa2315} and~\cite{08Miszczak} for a proof of the first and the second inequalities, respectively), one can easily show that
\begin{equation}
\aver{\F_{\rm HS}}_{1,2} \leq \aver{\F}_{1,2} \leq \aver{\Fn}_{1,2}\,,
\end{equation}
which establishes the maximum value of $\aver{\F_{\rm HS}}_{1,2}$ as a lower bound for both $\aver{\F}_{1,2}$ and $\aver{\Fn}_{1,2}$. In practice, if one is only interested in a control action that guarantees a minimal performance (measured in terms of $\aver{\F}_{1,2}$ or $\aver{\Fn}_{1,2}$), then the maximization of $\aver{\F_{\rm HS}}_{1,2}$ will provide such an operation if its optimal value is larger than the required performance.

\subsection{The Hilbert-Schmidt inner product}\label{sec:theHSIP}
We start noting the following similarity between the expressions of $\aver{\F_{\rm HS}}_1$ and $\aver{\F_{\rm HS}}_2$:
\begin{align}
\aver{\F_{\rm HS}[\mathcal{C}(\rho_i),\brho_i]}_1&=\sum_{i=1}^I\pi_i\tr\left[\mathcal{C}(\rho_i)\brho_i\right]\,,\label{eq:fhs1}\\
\aver{\F_{\rm HS}[\mathcal{C}(\rho_i),\brho_i]}_2&=\tr\left\{\left[\bigoplus_{i=1}^I\pi_i\mathcal{C}(\rho_i)\right]\left[\bigoplus_{j=1}^I\pi_j\brho_j\right]\right\}=\sum_{i=1}^I\pi_i^2\tr\left[\mathcal{C}(\rho_i)\brho_i\right]\,.\label{eq:fhs2}
\end{align}
Since the only difference between the two expressions is the exponent of the priorities $\pi_i$, in the rest of this section we restrict to assemble a SDP for the objective function $\aver{\F_{\rm HS}}_1$. Obviously, a SDP for $\aver{\F_{\rm HS}}_2$ can be readily obtained from our SDP for $\aver{\F_{\rm HS}}_1$  by simply replacing every occurrence of $\pi_i$ with $\pi_i^2$.

In addition, we note that just as $\aver{\H^2}_1$ and $\aver{\H}_2$ are monotonic with respect to each other in the case of uniform priorities (cf. Sec.~\ref{sec:equivalenceHaltav}), we have that $\aver{\F_{\rm HS}}_1$ and $\aver{\F_{\rm HS}}_2$ are \emph{proportional} to each other in the same circumstances. In fact, it is easy to see from Eqs.~\eqref{eq:fhs1} and~\eqref{eq:fhs2} that $\aver{\F_{\rm HS}}_1=I\aver{\F_{\rm HS}}_2$ if $\pi_i=1/I$ for $i=1,\ldots,I$. As a result, the operation maximizing $\aver{\F_{\rm HS}}_1$ also maximizes $\aver{\F_{\rm HS}}_2$ in this case.

For what follows, it will be useful to have two different presentations of Eq.~\eqref{eq:fhs1}. In the first, we write the trace as
\begin{equation}\label{eq:forCPTP}
\tr\left[\mathcal{C}(\rho_i)\brho_i\right]=
\tr\left[\mathfrak{C}(\rho_i^{\sf T}\otimes\brho_i)\right]\,,
\end{equation}
where we used Eq.~\eqref{eq:jam_inv} and the definition of the partial trace. In the second, we substitute $\mathfrak{C}$ with the expansion of Eq.~\eqref{eq:choiexpandTPin} to get
\begin{equation}\label{eq:forEBTP}
\tr\left[\mathcal{C}(\rho_i)\brho_i\right]=\frac{1}{\rm d}+\sum_{\substack{\mu=1\\\nu=2}}^{{\rm d}^2} x_{\mu,\nu}\tr\left(H_{\rm d}^\mu \rho_i^{\sf T}\right)\tr\left(H_{\rm d}^\nu\brho_i\right)\,.
\end{equation}

Next, the formulation of SDPs over $\mathcal{Q}_{\rm d}^{\rm set}$ and $\widetilde{\mathcal{B}}_{\rm d}^{\rm set}$ is presented in an independent fashion.

\subsubsection{Semidefinite program for $\aver{\mathscr{D}}=\aver{\F_{\rm HS}}_1$ and $\mathcal{C}\in\mathcal{Q}_{\rm d}^{\rm set}$}

With the above provisions, expressing problem~\eqref{eq:gen_problem_closeness} as a SDP is immediate for $\aver{\mathscr{D}}=\aver{\F_{\rm HS}}_1$ and $\mathcal{C}\in\mathcal{Q}_{\rm d}^{\rm set}$: The objective function is assembled from Eqs.~\eqref{eq:fhs1} and~\eqref{eq:forCPTP}, whereas the feasible set comes from the constraints of problem~\eqref{eq:probSDP_step0} to give

\begin{align}
\mbox{maximize}\quad & \tr\left[\mathfrak{C} \left(\sum_{i=1}^I{\rho_i^{\sf T}}\otimes\pi_i\brho_i\right)\right] \nonumber\\
\mbox{subject to}\quad& \mathfrak{C}\geq 0\label{eq:opt_hsip}\\
\quad&\tr\left[\mathfrak{C}\left(H_{\rm d}^\alpha\otimes\openone_{\rm d}\right)\right]={\rm d}\delta_{\alpha,1}\quad\mbox{for}\quad \alpha=1,\ldots,{\rm d}^2\,,\nonumber
\end{align}

This is clearly a SDP in the standard form [cf. Eq.~\eqref{eq:standardform}], with
\begin{equation}
E_0=-\sum_{i=1}^I\rho_i^{\sf T}\otimes\pi_i\brho_i\,,\quad E_\alpha=H_{\rm d}^\alpha\otimes\openone_{\rm d}\quad\mbox{and}\quad b_\alpha = {\rm d}\delta_{\alpha,1}\quad\mbox{for}\quad \alpha=1,\ldots,{\rm d}^2\
\end{equation}

\subsubsection{Semidefinite program for $\aver{\mathscr{D}}=\aver{\F_{\rm HS}}_1$ and $\mathcal{C}\in\widetilde{\mathcal{B}}_{\rm d}^{\rm set}$}

In this case, the optimization problem~\eqref{eq:gen_problem_closeness} is just problem~\eqref{eq:opt_hsip} with the extra constraint $\mathfrak{C}^{{\sf T}_2}\geq 0$, as explained in Sec.~\ref{sec:constraintB}. Due to this addition, it turns out to be easier to derive a SDP in the inequality form, as follows.

From the expansion of Eq.~\eqref{eq:choiexpandTPin} for $\mathfrak{C}$ (and using the normalization of density matrices and that $H_{\rm d}^1 = \openone_{\rm d}$), we find that the objective function of problem~\eqref{eq:gen_problem_closeness} can be written as $\tfrac{1}{\rm d}\sum_{i=1}^I\pi_i+\sum_{\substack{\mu=1\\\nu=2}}^{{\rm d}^2}x_{\mu,\nu}a_{\mu,\nu}$ (we have deliberately left the first sum unevaluated in order to obtain a SDP for the maximization of $\aver{\F_{\rm HS}}_2$ via the replacement $\pi_i\rightarrow \pi_i^2$, as explained before), where we have defined
\begin{equation}
a_{\mu,\nu}\defeq\sum_{i=1}^I\pi_i\tr\left(H_{\rm d}^\mu\rho_i^{\sf T}\right)\tr\left(H_{\rm d}^\nu\brho_i\right)\,.
\end{equation}

Using the objective function above and the inequality constraint of Eq.~\eqref{eq:probSDP_step2ineqEBTP}, the optimization problem of interest reduces to the following SDP in the inequality form:
\begin{align}
\frac{1}{\rm d}\sum_{i=1}^I\pi_i-\mbox{minimize}\quad& \sum_{\substack{\mu=1\\\nu=2}}^{{\rm d}^2}x_{\mu,\nu}\left(-a_{\mu,\nu}\right)\\
\mbox{subject to}\quad &F_0 + \sum_{\substack{\mu=1\\\nu=2}}^{{\rm d}^2}x_{\mu,\nu}F_{\mu,\nu}\geq 0
\end{align}
where we have included the minus signs to reexpress the original maximization as a minimization. Furthermore, we have defined
\begin{equation}
F_0=\frac{\openone_{2{\rm d}^2}}{\rm d}\qquad\mbox{and}\qquad F_{\mu,\nu}=\left(H_{\rm d}^{\alpha}\otimes H_{\rm d}^\beta\right)\oplus\left(H_{\rm d}^{\alpha}\otimes {H_{\rm d}^\beta}^{\sf T}\right)\,.
\end{equation}
Finally, we note that by removing the second ${\rm d}^2$-dimensional block from $F_0$ and $F_{\mu,\nu}$, we obtain an inequality form for the SDP of Eq.~\eqref{eq:opt_hsip} (maximization of $\aver{\F_{\rm HS}}_1$ over the set of CPTP maps).

\section{Controller sensitivity to the choice of $\aver{\mathscr{D}}$}\label{sec:ctrlsens}

In the previous sections we derived a number of SDPs for the problem of optimally transforming between sequences of density matrices. This multiplicity of optimization problems arises from the many available choices of distance measures between sequences of density matrices. In this section, these problems are numerically solved and by comparing their solutions we attempt to provide estimates on how different are the optimal operations (controllers) resulting from each problem.

From a practical viewpoint, there is at least one good reason for the proposed analysis: as we will see next, some of the SDPs derived here are harder to solve than others. It is thus interesting to find how well the solution of an easy problem approximates the solution of a difficult one.

The celebrated efficiency in solving a SDP is a consequence of the fact that the \emph{interior-point algorithm} \cite{04Nemirovski,04Boyd} requires only a polynomial number of operations (with respect to the ``problem size'') to find an optimal solution. More specifically, if $n$ is the number of variables of the SDP in the inequality form and $m$ is the dimension of the matrix inequality constraint, then the number of necessary operations to find a solution is not larger than $O(m^2 n^2 \sqrt{n})$ \cite{96Vandenberghe49,04Nemirovski}. However, if $m$ and/or $n$ are large, this can be a formidable task.

The values of $m$ and $n$ for the SDPs derived in the preceding section are shown in Table~\ref{table:dimSDP}, and a practical estimate of the time required for their solution over the set $\mathcal{Q}_{\rm d}^{\rm set}$ is presented in Fig.~\ref{fig:timeSDPs} for a few values of $I$ and ${\rm d}$. A quick glance at Fig.~\ref{fig:timeSDPs} and/or Table~\ref{table:dimSDP} reveals a clear computational advantage of $\aver{\F_{\rm HS}}_{1,2}$ over the metrics $\aver{\D}$, $\aver{\H^2}_1$, $\aver{\H}_2$ and $\aver{\O}_2$. Both the table and the figure are consistent in that no change of the problem size occurs with a variation of $I$. Moreover, for a fixed value of ${\rm d}$, $\aver{\F_{\rm HS}}_{1,2}$ gives rise to the SDPs with the smallest values of $n$ and $m$.\footnote{It should be noted, however, that the SDP formulations of the last section are not guaranteed to be the best ones in each case --- it is possible that the optimization of the metrics can be formulated as smaller SDPs. In particular, if symmetries are introduced in the problem (e.g., by restricting to source and target states symmetrically distributed in the Hilbert space), then symmetry reduction techniques \cite{04Gatermann95} can be of assistance. These ideas will be put at work in the next chapter.}

Table~\ref{table:dimSDP} also shows that the minimization of $\aver{\H}_2$ is the only one that yields a quadratic scaling  of the dimension $m$ of the matrix constraint with $I$. This is born out in Fig.~\ref{fig:timeSDPs}, where the minimization of $\aver{\H}_2$ is seen to be dramatically slower than the optimization of the other measures. Typically, the second slower minimization is that of $\aver{\D}$. This is justified in Table~\ref{table:dimSDP}, where the number $n$ of variables  involved in the minimization of $\aver{\D}$ is seen to scale quadratically with $I$, whereas the other measures do not show any scaling of $n$ with $I$.

Finally, although the minimizations of both $\aver{\H^2}_1$ and $\aver{\O}_2$ yield SDPs with precisely the same number of variables (and matrices whose dimension scale linearly with $I$ and quadratically with ${\rm d}$), the actual dimension of the matrices is larger for $\aver{\H^2}_1$. Once again, this is substantiated in Fig.~\ref{fig:timeSDPs}, where the minimization of $\aver{\H^2}_1$ is seen to be always slower than that of $\aver{\O}_2$.

In the next sections, these measures are compared not from the view point of computational cost, but in terms of how different are the transformations produced by the optimization of each of them.

\begin{table}[h!]
\centering \caption[A comparative analysis of the ``size'' of some semidefinite programs.]{A comparative analysis of the ``size'' of each SDP formulated in Secs.~\ref{sec:min_dist} and~\ref{sec:max_close} for different objective functions $\aver{\mathscr{D}}$ and feasible sets $\mathcal{C}_{\rm d}^{\rm set}$. The scaling on the number $n$ of variables and the dimension $m$ of the matrix constraints is shown as a function of $I$ and ${\rm d}$. For ease of comparison, the values for $I={\rm d}=2$ are shown in brackets.}

\begin{tabular*}{0.96\textwidth}{c|c|c|c|c}
\hline\hline
&\multicolumn{2}{c|}{$\mathcal{Q}_{\rm d}^{\rm set}$} &\multicolumn{2}{c}{$\widetilde{\mathcal{B}}_{\rm d}^{\rm set}$}\\\hline
&$n$&$m$&$n$&$m$\\\cline{2-5}
$\aver{\D}$& $\medmath{{\rm d}^2({\rm d}^2+2 I^2-1)}$, $\scriptstyle (44)$ &$\medmath{2 I {\rm d} + {\rm d}^2}$, $\scriptstyle (12)$ & $\medmath{{\rm d}^2({\rm d}^2+2 I^2-1)}$, $\scriptstyle (44)$ & $\medmath{2 I {\rm d} + 2{\rm d}^2}$, $\scriptstyle (16)$\\
$\aver{\H^2}_1$& $\medmath{{\rm d}^2({\rm d}^2-1)+1}$, $\scriptstyle (13)$ &$\medmath{{\rm d}^2(1+I)+1}$, $\scriptstyle (13)$& $\medmath{{\rm d}^2({\rm d}^2-1)+1}$, $\scriptstyle (13)$ & $\medmath{{\rm d}^2(2+I)+1}$, $\scriptstyle (17)$\\
$\aver{\H}_2$&  $\medmath{{\rm d}^2({\rm d}^2-1)+1}$, $\scriptstyle (13)$ &$\medmath{{\rm d}^2(1+I^2)+1}$, $\scriptstyle (21)$ & $\medmath{{\rm d}^2({\rm d}^2-1)+1}$, $\scriptstyle (13)$ &$\medmath{{\rm d}^2(2+I^2)+1}$, $\scriptstyle (25)$\\
$\aver{\O}_2$& $\medmath{{\rm d}^2({\rm d}^2-1)+1}$, $\scriptstyle (13)$ &$\medmath{2 I {\rm d}+{\rm d}^2}$, $\scriptstyle (12)$ & $\medmath{{\rm d}^2({\rm d}^2-1)+1}$, $\scriptstyle (13)$ &$\medmath{2 I {\rm d}+2{\rm d}^2}$, $\scriptstyle (16)$\\
$\aver{\F_{\rm HS}}_1$& $\medmath{{\rm d}^2}$, $\scriptstyle (4)$ &${\rm d}^2$, $\scriptstyle (4)$ & $\medmath{{\rm d}^2({\rm d}^2 - 1)}$, $\scriptstyle (12)$ & $\medmath{2{\rm d}^2}$, $\scriptstyle (8)$\\
$\aver{\F_{\rm HS}}_2$& $\medmath{{\rm d}^2}$, $\scriptstyle (4)$ &$\medmath{{\rm d}^2}$, $\scriptstyle (4)$ & $\medmath{{\rm d}^2({\rm d}^2 - 1)}$, $\scriptstyle (12)$  &$\medmath{2{\rm d}^2}$, $\scriptstyle (8)$\\
\hline\hline
\end{tabular*}\label{table:dimSDP}
\end{table}

\begin{figure*}[h!]
\centering
\subfigure[\hspace{1.5mm} ${\rm d}=2$] {
    \label{fig:compdistmes_a}
    \includegraphics[width=7.5cm]{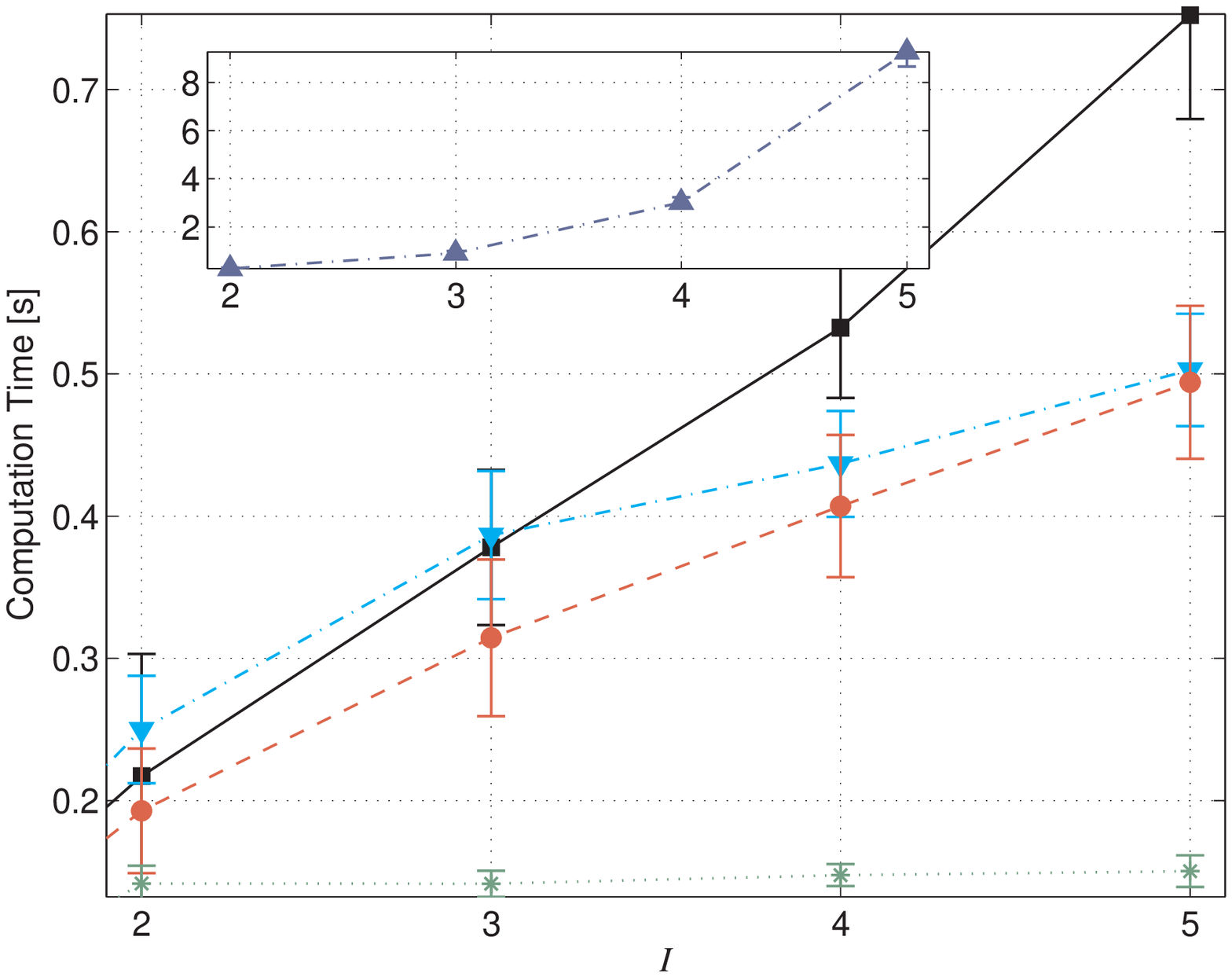}
}
\hspace{0.6cm}
\subfigure[\hspace{1.5mm} ${\rm d}=3$ ] {
    \label{fig:compdistmes_b}
    \includegraphics[width=7.5cm]{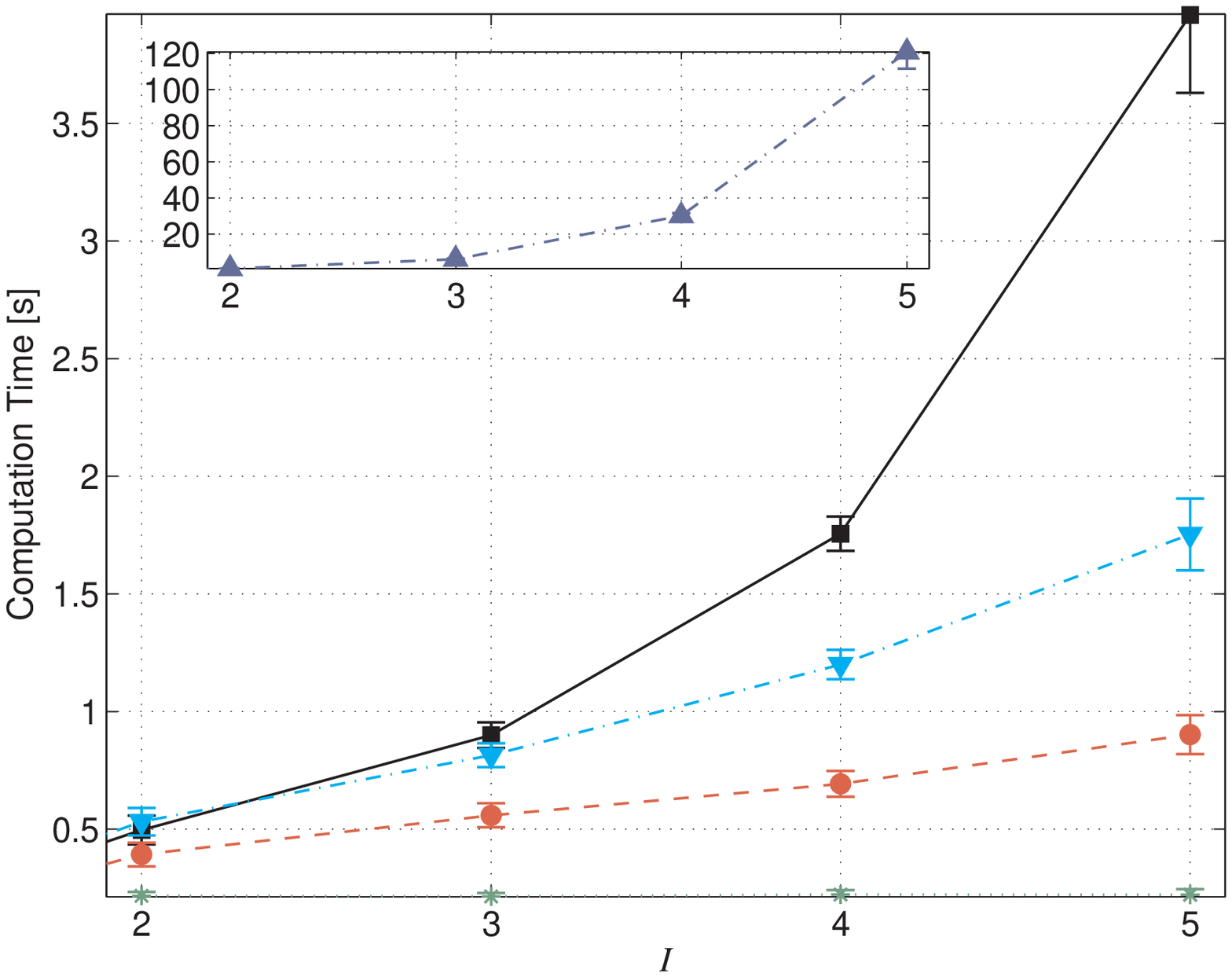}
}
\subfigure[\hspace{1.5mm} ${\rm d}=4$ ] {
    \label{fig:compdistmes_b}
    \includegraphics[width=7.5cm]{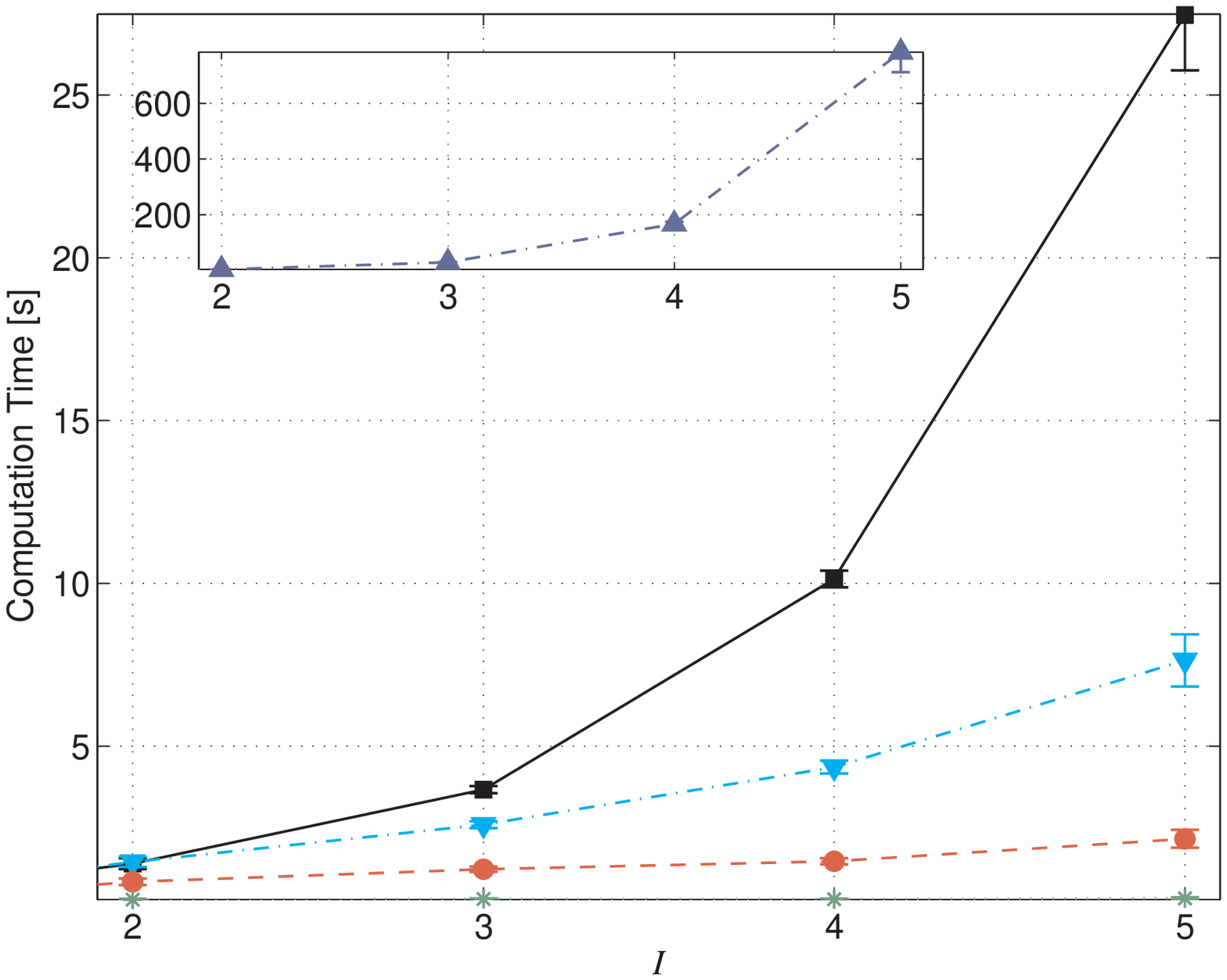}
}
\begin{picture}(0,0)
\put(-245,345){$\textstyle \aver{\H}_2$}
\put(-174,352){$\textstyle \aver{\D}$}
\put(-174,309){$\textstyle \aver{\H^2}_1$}
\put(-174,271){$\textstyle \aver{\O}_2$}
\put(-174,225){$\textstyle \aver{\F_{\rm HS}}$}
\end{picture}

\caption[Average computation time of optimal $\aver{\D}$, $\aver{\H^2}_1$, $\aver{\H}_2$, $\aver{\O}_2$ and $\aver{\F_{\rm HS}}_{1,2}$ over CPTP maps.]{(Color online) Average time required to solve the SDPs from Secs.~\ref{sec:min_dist} and~\ref{sec:max_close} over the set of CPTP maps. First, a sequence of $I$ source and target {\rm d}-dimensional density matrices is randomly generated. These sequences are then input to SeDuMi \cite{99Sturm625} and the time taken to find the optimal solution for each function $\aver{\mathscr{D}}$ is recorded. The procedure is repeated for $100$ different sequences with a fixed value of $({\rm d},I)$. Each point in the plots represents the mean time for a given pair $({\rm d}, I)$, while the error bars represent the standard deviation associated to the averaging process.}
\label{fig:timeSDPs}
\end{figure*}

\subsection{Qualitative analysis: The Bloch disk photo gallery}

In this section we present a sequence of plots representing the numerical solution of each one of the SDPs formulated in the preceding sections. Our aim is to provide a first qualitative analysis of how severely an optimal operation depends on the choice of distance/closeness measure being optimized.

In the present analysis, we restrict to the case of qubit states ${\rm d}=2$. This is done with the intent of visualizing the output states of each optimization problem as vectors on a three dimensional unit ball --- the Bloch ball. Furthermore, we consider source and target sequences of only two states each ($I=2$), in such a way that we can restrict to a plane within the Bloch ball --- the Bloch disk.

The specific type of transformation we look at is the purification of a pair of mixed qubit states. The source states are taken to be separated by a Bloch sphere angle $\Theta=90^o$ and have lengths $R_1, R_2 < 1$. The purification task consists of an attempt to increase these lengths up to unit (pure states), while preserving the angle $\Theta$ between them. We note that due to the choice of pure targets, the maximization of $\aver{\F_{\rm HS}}_{1}$ conducted here is equivalent to the maximization of $\aver{\F}_{1}$ or $\aver{\Fn}_{1}$, as discussed in Sec.~\ref{sec:max_close}.

In the following plots we show the Bloch vectors of source and target states, and also those of the states effectively optimizing each distance measure over the feasible sets $\mathcal{Q}_2^{\rm set}$ and $\mathcal{B}_2^{\rm set}$. For sake of comparison, we present separate sets of plots for situations where (i) $R_1=R_2$ and $\pi_1=\pi_2$, (ii) $R_1\neq R_2$ and $\pi_1 = \pi_2$ and (iii) $R_1 = R_2$ and $\pi_1 \neq \pi_2$.

\subsubsection{Unbiased purification of equally mixed states}\label{sec:unb_eq}
The case $\pi_1=\pi_2=0.5$ and $R_1=R_2=0.7$ is presented in Fig.~\ref{fig:compdistmes}. Fig.~\ref{fig:compdistmes_a} shows that all the metrics lead to a common operation; likewise, Fig.~\ref{fig:compdistmes_b} shows that both $\aver{\F_{\rm HS}}_1$ and $\aver{\F_{\rm HS}}_2$ are maximized with another common operation. The coincidence between the minimizers of $\aver{\H^2}_1$ and $\aver{\H}_2$, and between the maximizers of $\aver{\F_{\rm HS}}_1$ and $\aver{\F_{\rm HS}}_2$, should not come as a surprise --- in Secs.~\ref{sec:equivalenceHaltav} and~\ref{sec:theHSIP} we saw that these coincidences are inherent to any unbiased transformation.\par

All the remaining coincidences in Fig.~\ref{fig:compdistmes_a} are somewhat unexpected, and should be interpreted as a peculiarity of the particular task of this section \emph{in the qubit case}. Indeed, numerical simulations for qutrit states (under the same circumstances of mixedness of states and uniform priorities) do not show the degeneracy observed here. Next, we provide a half-technical-half-intuitive clarification of the origins of these coincidences in the qubit case.\par

For $\rho$ and $\sigma$ any qubit density matrices of Bloch vectors $\bm{r}$ and $\bm{s}$, the eigenvalues of $\rho-\sigma$ can be explicitly calculated to be $\pm 1/2|\bm{r}-\bm{s}|$, and hence
$\D(\rho,\sigma)=\H(\rho,\sigma)/\sqrt{2}=\O(\rho,\sigma)=1/2|\bm{r}-\bm{s}|$, or equivalently,
\begin{equation}\label{eq:DHO}
\aver{\D(\rho_i,\sigma_i)}=\frac{1}{\sqrt{2}}\aver{\H(\rho_i,\sigma_i)}_1=\aver{\O(\rho_i,\sigma_i)}_1= \tfrac{1}{2}\sum_{i=1}^I\pi_i|\bm{r}_i-\bm{s}_i|
\end{equation}
for any sequences $[\rho_i]_{i=1}^I$ and $[\sigma_i]_{i=1}^I$ of qubit states with Bloch vectors given by $[\bm{r}_i]_{i=1}^I$ and $[\bm{s}_i]_{i=1}^I$. Moreover, from the eigenvalue formula of $\rho_i-\sigma_i$, it is also simple to see that
\begin{align}
\aver{\H(\rho_i,\sigma_i)}_2&={\textstyle \sqrt{\frac{1}{2}\sum_{i=1}^I\pi_i^2|\bm{r}_i-\bm{s}_i|^2}}\,,\label{eq:Hrisi}\\
\aver{\O(\rho_i,\sigma_i)}_2&=\tfrac{1}{2}\max_i \pi_i|\bm{r}_i-\bm{s}_i|\,.\label{eq:Orisi}
\end{align}

Now, note that Eqs.~\eqref{eq:DHO},~\eqref{eq:Hrisi} and~\eqref{eq:Orisi} all become proportional to each other if we make $\pi_i=1/I$ and assume $|\bm{r}_i-\bm{s}_i|=k$ for some constant $k$. Of course, with these extra constraints, the optimization of any of the quantities above would lead to a common operation. That is it for the technical part.

Intuitively, the extra constraints found above can be incorporated into the problem of interest without loss of generality: First, $\pi_i=1/I$ is already there by hypothesis. Second, since all the source Bloch vectors have the same length and all the target Bloch vectors have the same length (in particular, equal to one), it would be very odd if the Bloch vectors $\bm{r}_i$ arising from the optimization of any metric
would not dispose perfectly symmetric with respect to their  target vectors. Assuming that this oddness would never occur, we can include the ``redundant constraint'' $|\bm{r}_i-\bm{s}_i|=k$ for every $i$. Hence --- under this intuitive assumption --- the control problem should really be insensitive to the choice of metric, as Fig.~\ref{fig:compdistmes_a} demonstrates it is.

A comparison of Figs.~\ref{fig:compdistmes_a} and~\ref{fig:compdistmes_b}, suggests that the maximization of fidelity-like quantities tends to provide an improved elongation of the lengths of the source vectors than the corresponding elongation arising from the minimization of the metrics. On the other hand, the minimization of the metrics give operations that better approximate the angle between the target vectors.

Finally, it is interesting to compare how the restriction to the set of EBTP maps affect each case. The minimization of the metrics over the set of EBTP maps leads to Bloch vectors that are approximately $5.3\%$ shorter and separated by an angle $9.2\%$ smaller than the lengths and angles arising from the minimization of the same metrics over the set of CPTP maps. On the other hand, the same restriction for the fidelity-like quantities leads to vectors that are actually $1\%$ \emph{longer} than the corresponding vectors from the CPTP case, however, as it should be the case, this is compensated with a substantial angle drop of $31\%$.

\begin{figure*}[h!]
\centering
\subfigure[\hspace{1.5mm} CPTP: $R_1^{\rm eff}=R_2^{\rm eff}=0.75$, $\Theta=74.5^o$ EBTP: $R_1^{\rm eff}=R_2^{\rm eff}=0.71$, $\Theta=67.65^o$] {
    \label{fig:compdistmes_a}
    \includegraphics[width=7.5cm]{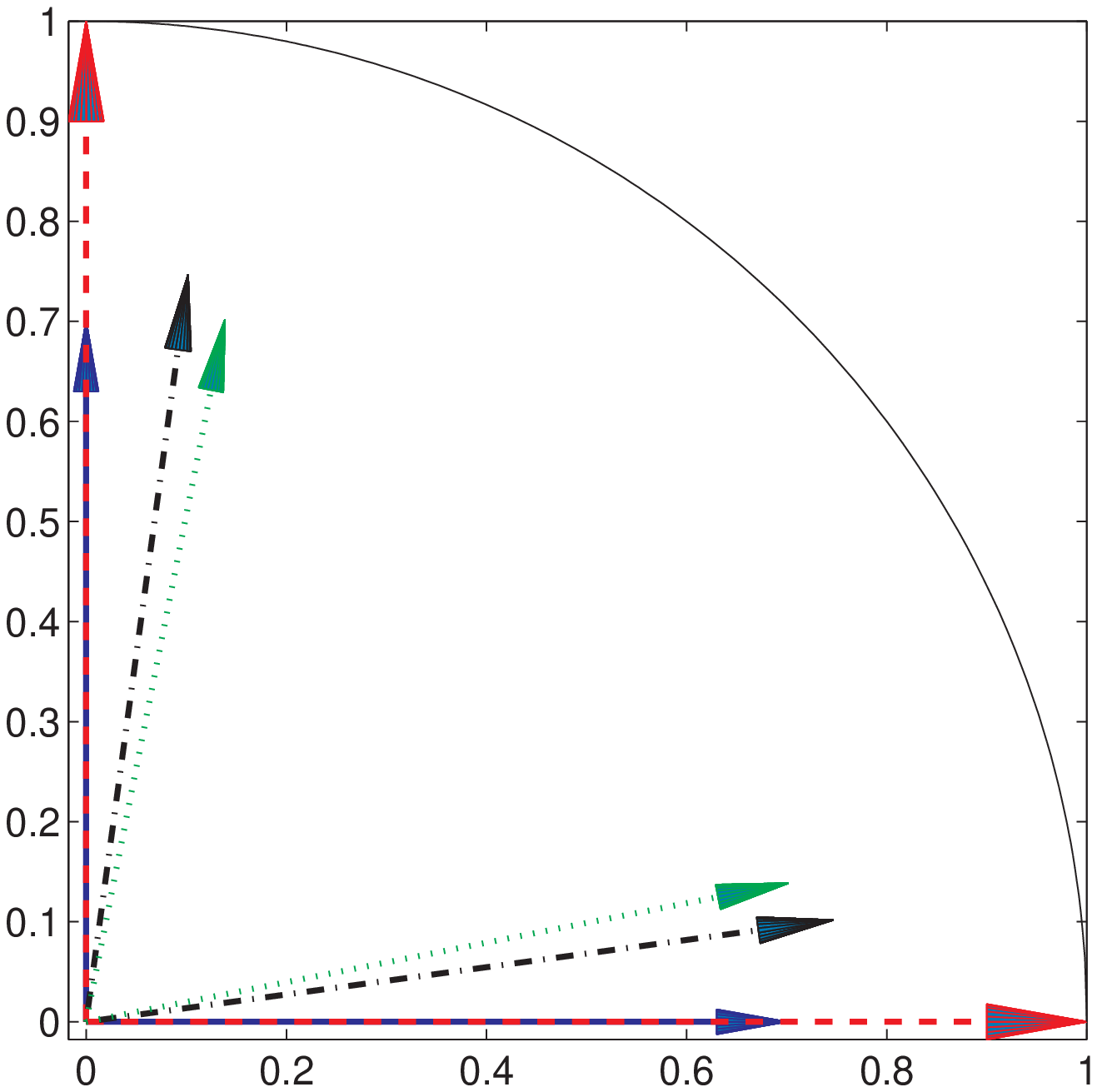}
}
\hspace{0.6cm}
\subfigure[\hspace{1.5mm} CPTP: $R_1^{\rm eff}=R_2^{\rm eff}=0.91$, $\Theta=35.96^o$ EBTP: $R_1^{\rm eff}=R_2^{\rm eff}=0.92$, $\Theta=27.53^o$ ] {
    \label{fig:compdistmes_b}
    \includegraphics[width=7.5cm]{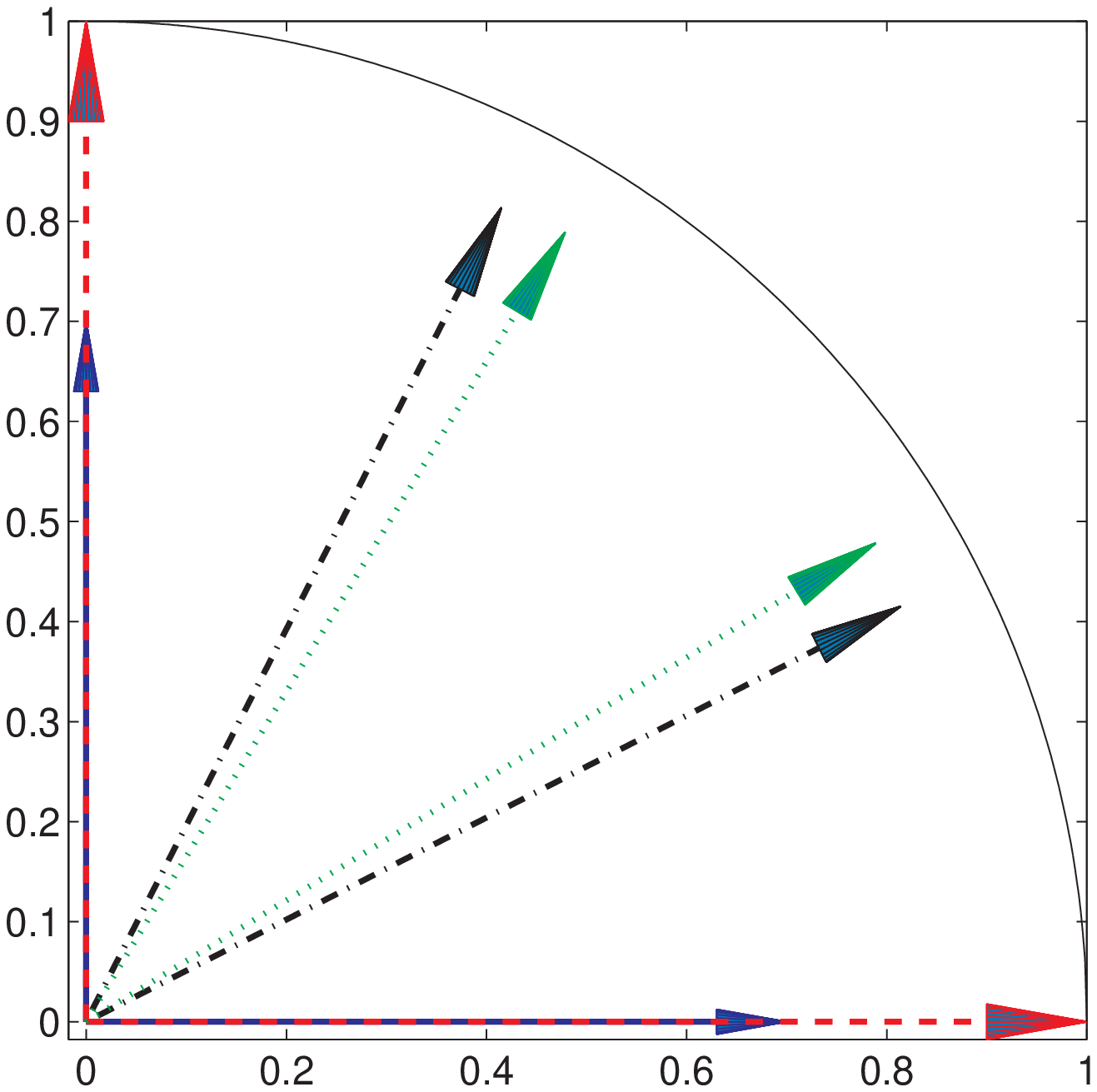}
}

\caption[Bloch vectors visualization of unbiased purification of equally mixed states.]{(Color online) Comparison between the Bloch vectors resulting from the numerical optimization of $\aver{\D}$, $\aver{\H^2}_1$, $\aver{\H}_2$, $\aver{\O}_2$ and $\aver{\F}_{1,2}$ for the unbiased transformation (i.e., $\pi_1=\pi_2=1/2$) of a pair of mixed source states (dashed, blue) with $R_1=R_2=0.7$ into a pair of pure states (dashed, red). The dash-dotted (black) arrows designate the optimal Bloch vectors arising from the optimization over the set of CPTP maps, while the dotted (green) arrows refer to the optimal transformation over the set of EBTP maps.}
\label{fig:compdistmes}
\begin{picture}(0,0)
\put(-190,320){\large$\bm{\aver{\D}}$, $\bm{\aver{\H^2}_1}$, $\bm{\aver{\H}_2}$, $\bm{\aver{\O}_2}$}
\put(160,320){\large$\bm{\aver{\F_{\rm HS}}_{1,2}}$}
\end{picture}
\end{figure*}

\subsubsection{Unbiased purification of states of different mixedness}\label{sec:unb_dif}
The case $\pi_1=\pi_2=0.5$, $R_1=0.7$ and $R_2=0.9$ is presented in Fig.~\ref{fig:compdistmes2}. Notably, by removing the symmetry of equally mixed source states, most of the degeneracies observed in the last section are removed. The only exceptions are $\aver{\H^2}_1$ and $\aver{\H}_2$ [Fig.~\ref{fig:compdistmes2_b}], and $\aver{\F_{\rm HS}}_1$ and $\aver{\F_{\rm HS}}_2$ [Fig.~\ref{fig:compdistmes2_d}], where the degeneracies survive due to the unbiased choice of priorities, as discussed before.

Once again, we find that the minimizers of the metric quantities perform better at approximating the target angle, while the maximizers of the fidelity-like quantities are better at approximating target lengths. In fact, it is now possible to see a smoother transition of this behavior while following the sequence of plots $\aver{\D}\rightarrow\aver{\H^{(2)}}_{(1),2}\rightarrow\aver{\O}_2\rightarrow\aver{\F_{\rm HS}}_{1,2}$ in Fig.~\ref{fig:compdistmes2}. Noticeably, the CPTP map minimizing $\aver{\D}$ yields the shortest vectors, but more widely separated. This is followed by the optimal CPTP for $\aver{\H^2}_1$ or $\aver{\H}_2$, which gives slightly longer vectors, but separated by a smaller angle. Following the same trend comes $\aver{\O}_2$ and finally $\aver{\F_{\rm HS}}_{1,2}$, which gives the longest vectors separated by the smallest angle.

When the restriction to EBTP maps is made, the same pattern applies for the angles, which are seen to decrease along the way. However, the length of one of the vectors breaks the pattern by decreasing while we progress along $\aver{\D}\rightarrow\aver{\H^{(2)}}_{(1),2}\rightarrow\aver{\O}_2$. Nevertheless, even in the EBTP case, the maximal length of both vectors is achieved with the maximization of $\aver{\F_{\rm HS}}_{1,2}$.

Still regarding the restriction to EBTP maps, we note that the same rule observed in the last section still applies. For all metrics, the output vectors are shorter and less separated than the corresponding vectors in the CPTP case. Only for the fidelity like measures, we have an lengthening of the vectors and a more substantial decrease of angle.

\begin{figure*}[h!]
\centering
\subfigure[\hspace{1.5mm} CPTP: $R_1^{\rm eff}=0.73$, $R_2^{\rm eff}=0.91$, $\Theta=81.78^o$ EBTP: $R_1^{\rm eff}=0.68$, $R_2^{\rm eff}=0.83$, $\Theta=73.83^o$] {
    \label{fig:compdistmes2_a}
    \includegraphics[width=7.5cm]{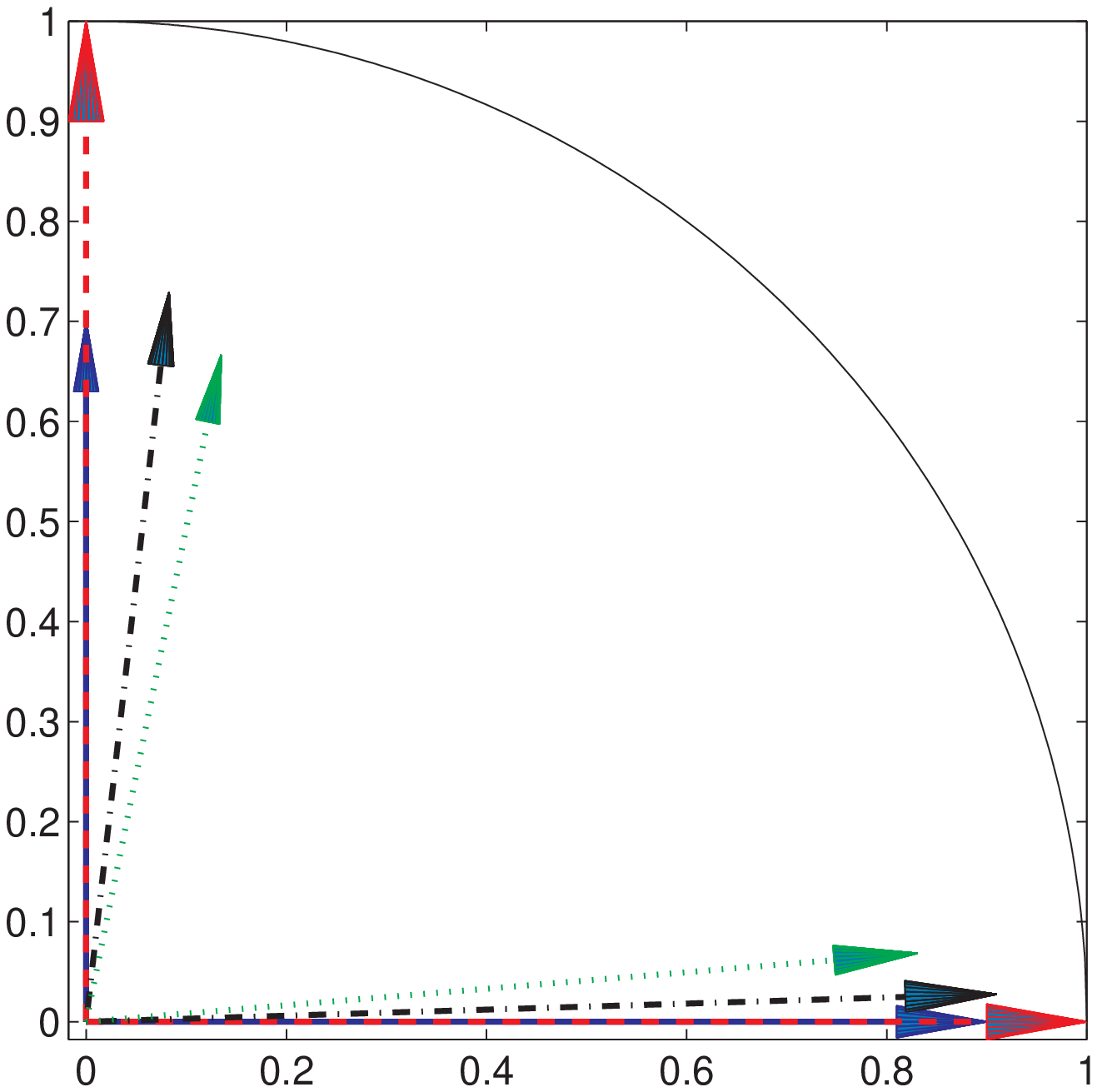}
}
\hspace{0.6cm}
\subfigure[\hspace{1.5mm} $R_1^{\rm eff}=0.74$, $R_2^{\rm eff}=0.91$, $\Theta=80.31^o$ EBTP: $R_1^{\rm eff}=0.71$, $R_2^{\rm eff}=0.79$, $\Theta=73.74^o$ ] {
    \label{fig:compdistmes2_b}
    \includegraphics[width=7.5cm]{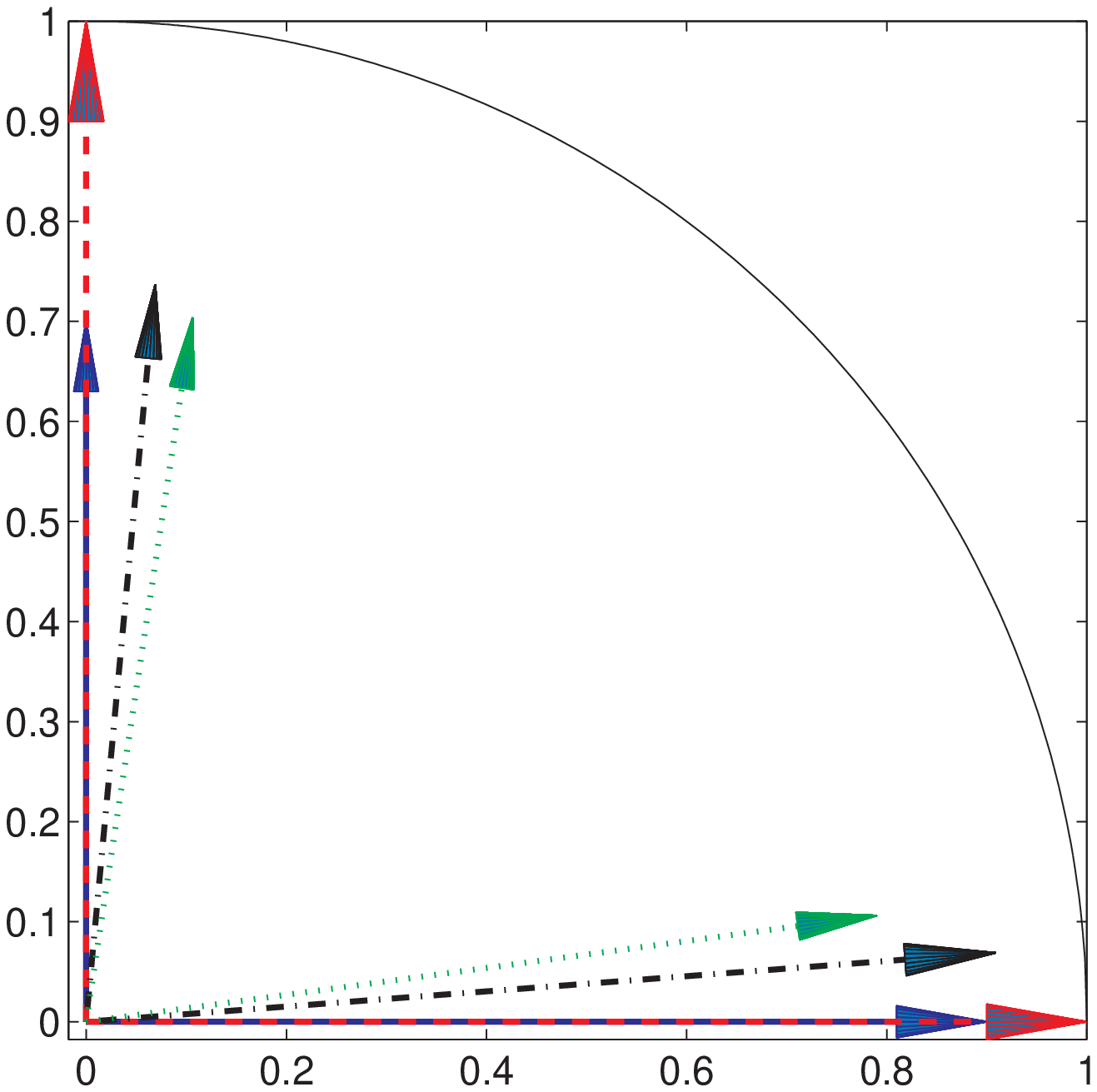}
}
\hspace{0.6cm}
\subfigure[\hspace{1.5mm} $R_1^{\rm eff}=0.77$, $R_2^{\rm eff}=0.92$, $\Theta=72.47^o$ EBTP: $R_1^{\rm eff}=0.74$, $R_2^{\rm eff}=0.77$, $\Theta=72.26^o$ ] {
    \label{fig:compdistmes2_c}
    \includegraphics[width=7.5cm]{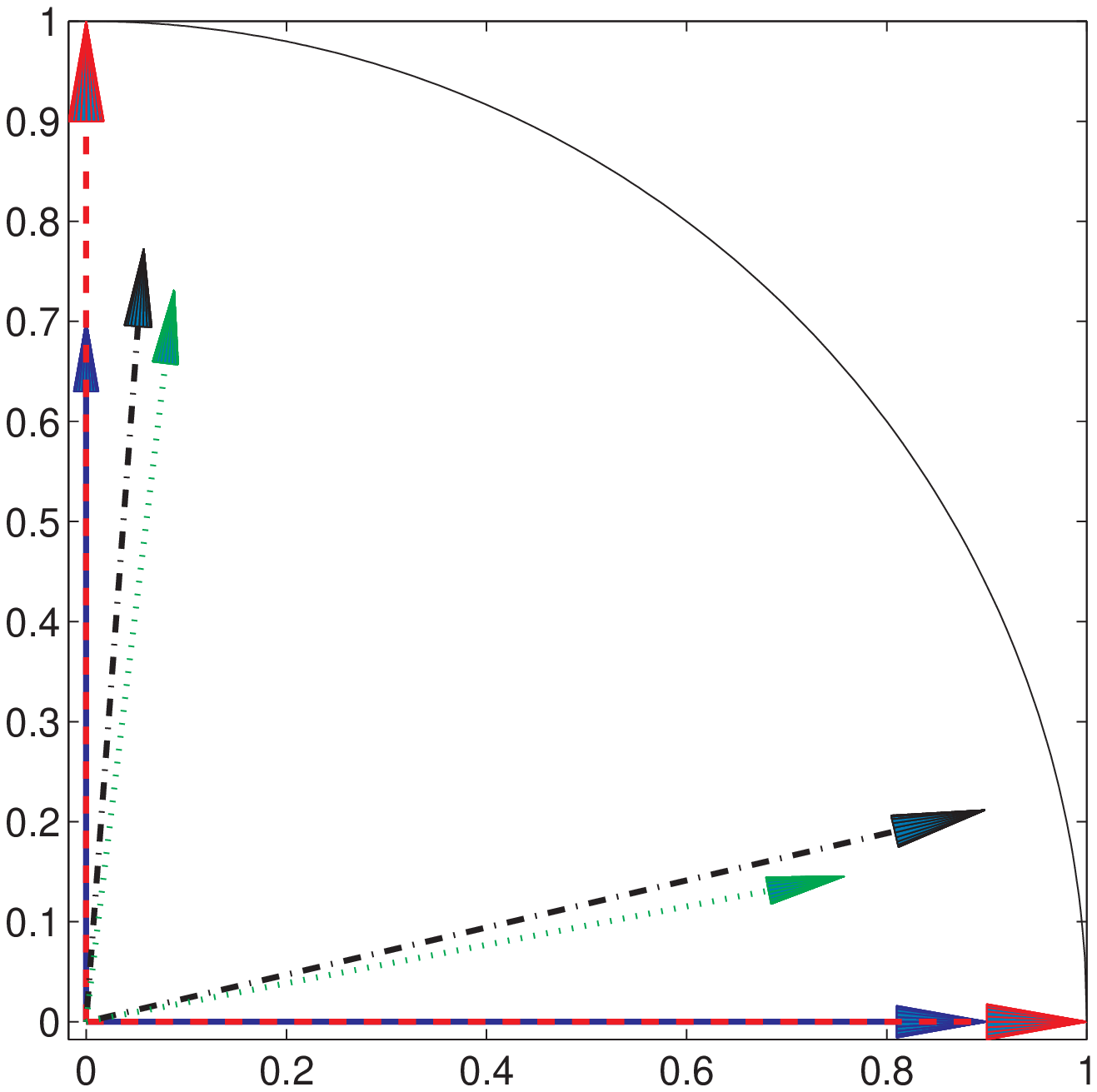}
}
\hspace{0.6cm}
\subfigure[\hspace{1.5mm} $R_1^{\rm eff}=0.87$, $R_2^{\rm eff}=0.95$, $\Theta=50.50^o$ EBTP: $R_1^{\rm eff}=0.88$, $R_2^{\rm eff}=0.94$, $\Theta=36.32^o$] {
    \label{fig:compdistmes2_d}
    \includegraphics[width=7.5cm]{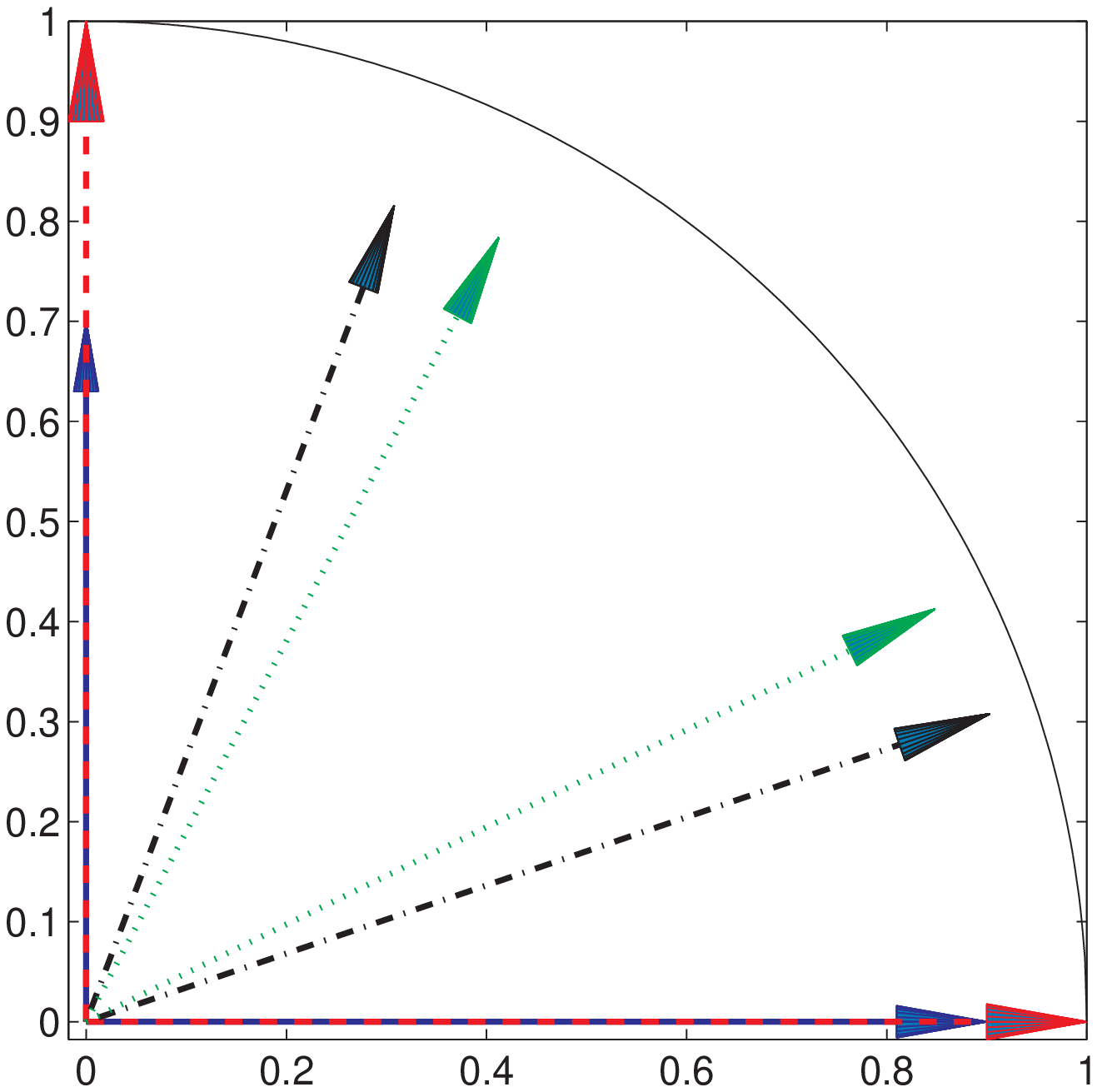}
}
\caption[Bloch vectors visualization of unbiased purification of states of different mixedness.]{(Color online) Comparison between the Bloch vectors resulting from the numerical optimization of $\aver{\D}$, $\aver{\H^2}_1$, $\aver{\H}_2$, $\aver{\O}_2$ and $\aver{\F}_{1,2}$ for the unbiased transformation (i.e., $\pi_1=\pi_2=1/2$) of a pair of mixed source states (dashed, blue) with $R_1=0.7$ and $R_2=0.9$ into a pair of pure states (dashed, red). The dash-dotted (black) arrows designate the optimal Bloch vectors arising from the optimization over the set of CPTP maps, while the dotted (green) arrows refer to the optimal transformation over the set of EBTP maps.}
\label{fig:compdistmes2}
\begin{picture}(0,0)
\put(-65,572){\large$\bm{\aver{\D}}$}
\put(135,572){\large$\bm{\aver{\H^2}_1}$, $\bm{\aver{\H}_2}$}
\put(-65,320){\large$\bm{\aver{\O}_2}$}
\put(160,320){\large$\bm{\aver{\F_{\rm HS}}_{1,2}}$}
\end{picture}
\end{figure*}

\subsubsection{Biased purification of equally mixed states}
The case $\pi_1=0.3$, $\pi_2=0.7$ and $R_1=R_2=0.7$ is presented in Fig.~\ref{fig:compdistmes3}. Here, all the degeneracies are removed: Biasing not only destroys the monotonicity between $\aver{\H^2}_1$ and $\aver{\H}_2$ and the proportionality between $\aver{\F_{\rm HS}}_1$ and $\aver{\F_{\rm HS}}_2$, but in setting a higher hierarchy to the atomic transformation of $i=2$ it also breaks the (intuitively expected) symmetry between the effective Bloch vectors and the targets noted in Fig.~\ref{fig:compdistmes}.

Nevertheless, the optimal CPTP transformations preserve some important features from the unbiased cases. From the case of equally mixed sources (Sec.~\ref{sec:unb_eq}), we note the commonality that the two states resulting from the optimization of each measure have the same length. This shows that it is only via an adjustment of the angle between the effective and target states that biasing is accounted for, as opposed to a possible enhanced lengthening of the vector of higher weight.

From the case with different degrees of mixedness (Sec.~\ref{sec:unb_dif}), we note that the sequence of measures leading to vectors of increasing length and decreasing angle is still approximately the same: $\aver{\H^2}_1\to\aver{\D}\to\aver{\H}_2\to\aver{\O}_2\to\aver{\F_{\rm HS}}_1\to\aver{\F_{\rm HS}}_2$. The difference is that now it is the minimization of $\aver{\H^2}_1$ that yields the shortest and most angularly separated Bloch vectors.

None of the above holds for the optimal EBTP transformations. In this case, biasing is accounted for not only by an adjustment of angle, but also by making longer the Bloch vector arising from the transformation of higher weight. Moreover, although $\aver{\F_{\rm HS}}_2$ and $\aver{\F_{\rm HS}}_1$ are still easily identified as the measures that give the longest vectors separated by the smallest angles, it is not so clear how the metrics should be ordered because the angles do not always decrease as the lengths increase. As a general observation, we have the metrics $\aver{\D}$, $\aver{\H}_2$ and $\aver{\O}_2$ leading to vectors which have approximately the same lengths, but the angle between them decreases in the sequence $\aver{\O}_2\to\aver{\H}_2\to\aver{\D}$, which is actually reversed with respect to the ordering found in Sec.~\ref{sec:unb_dif}.

It is also worth noting how the angle between the two vectors drop as we go from the optimal CPTP transformation to the optimal EBTP transformation with respect to a fixed measure. Consistently with the previous sections, the highest angle drops occur for the fidelity-like measures, and is of approximately $24\%$. Second in the rank is $\aver{\D}$ with a drop of $18.3\%$ ---  much higher than the observed in the unbiased cases. The measures $\aver{\H^2}_1$ and $\aver{\H}_2$ come with an approximately equal angle drop of $10\%$. Finally, we have $\aver{\O}_2$, with only $2.7\%$. As expected, the largest angle drops occur for those measures that, in the EBTP case, attempt to lengthen the Bloch vectors with respect to the  vectors obtained in the CPTP case. This property --- which was noticed just in the fidelity-like measures in the unbiased cases --- is now also detected in $\aver{\D}$ (and very slightly in $\aver{\H}_2$).

\begin{figure*}[h!]
\centering
\subfigure[\hspace{1.5mm} CPTP: $R_1^{\rm eff}=R_2^{\rm eff}=0.77$, $\Theta=69.86^o$ EBTP: $R_1^{\rm eff}=0.69$, $R_2^{\rm eff}=0.80$, $\Theta=57.10^o$] {
    \label{fig:compdistmes3_a}
    \includegraphics[width=5cm]{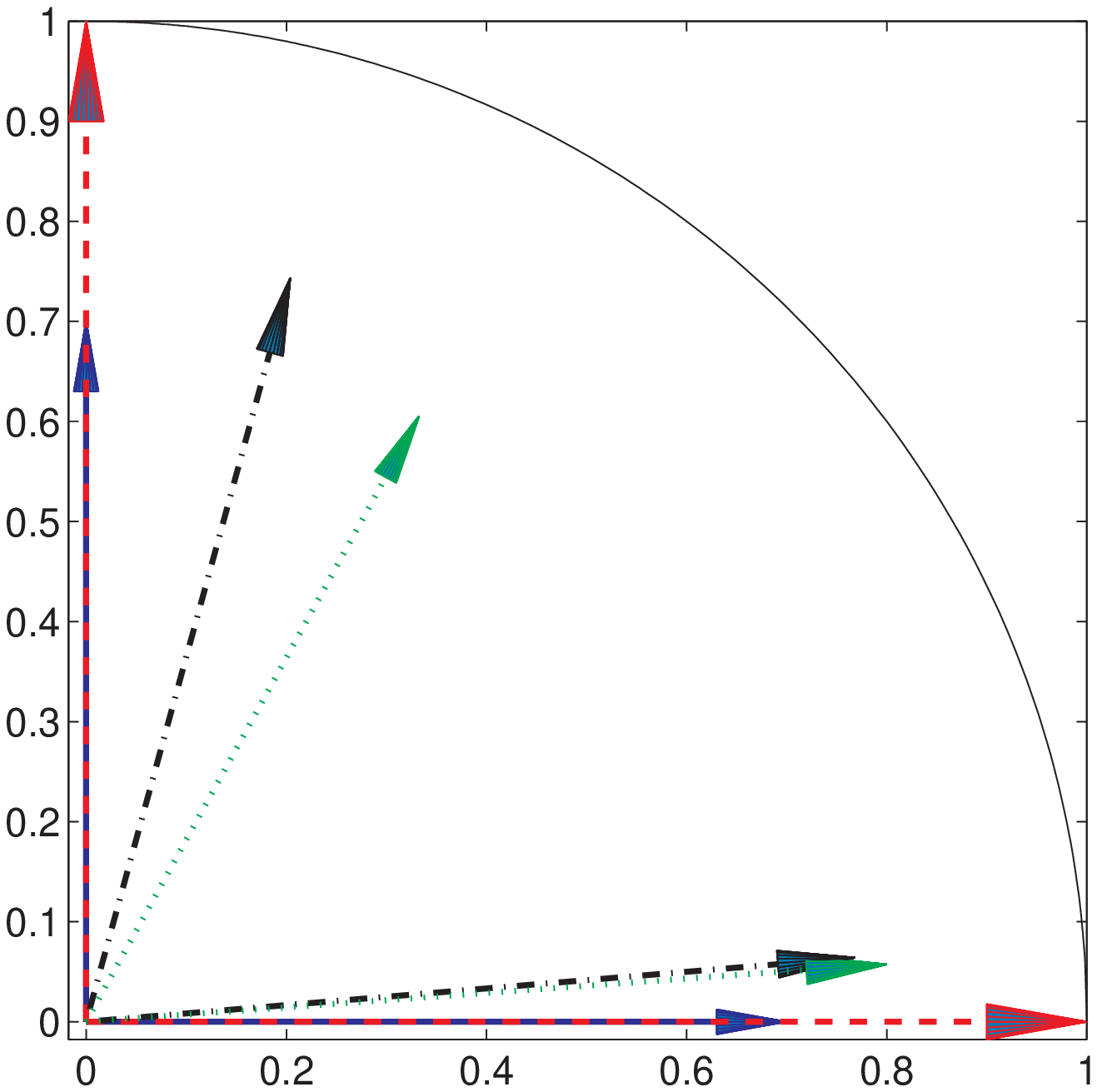}
}
\subfigure[\hspace{1.5mm} $R_1^{\rm eff}=R_2^{\rm eff}=0.76$, $\Theta=72.05^o$ EBTP: $R_1^{\rm eff}=0.68$, $R_2^{\rm eff}=0.75$, $\Theta=65.15^o$ ] {
    \label{fig:compdistmes3_b}
    \includegraphics[width=5cm]{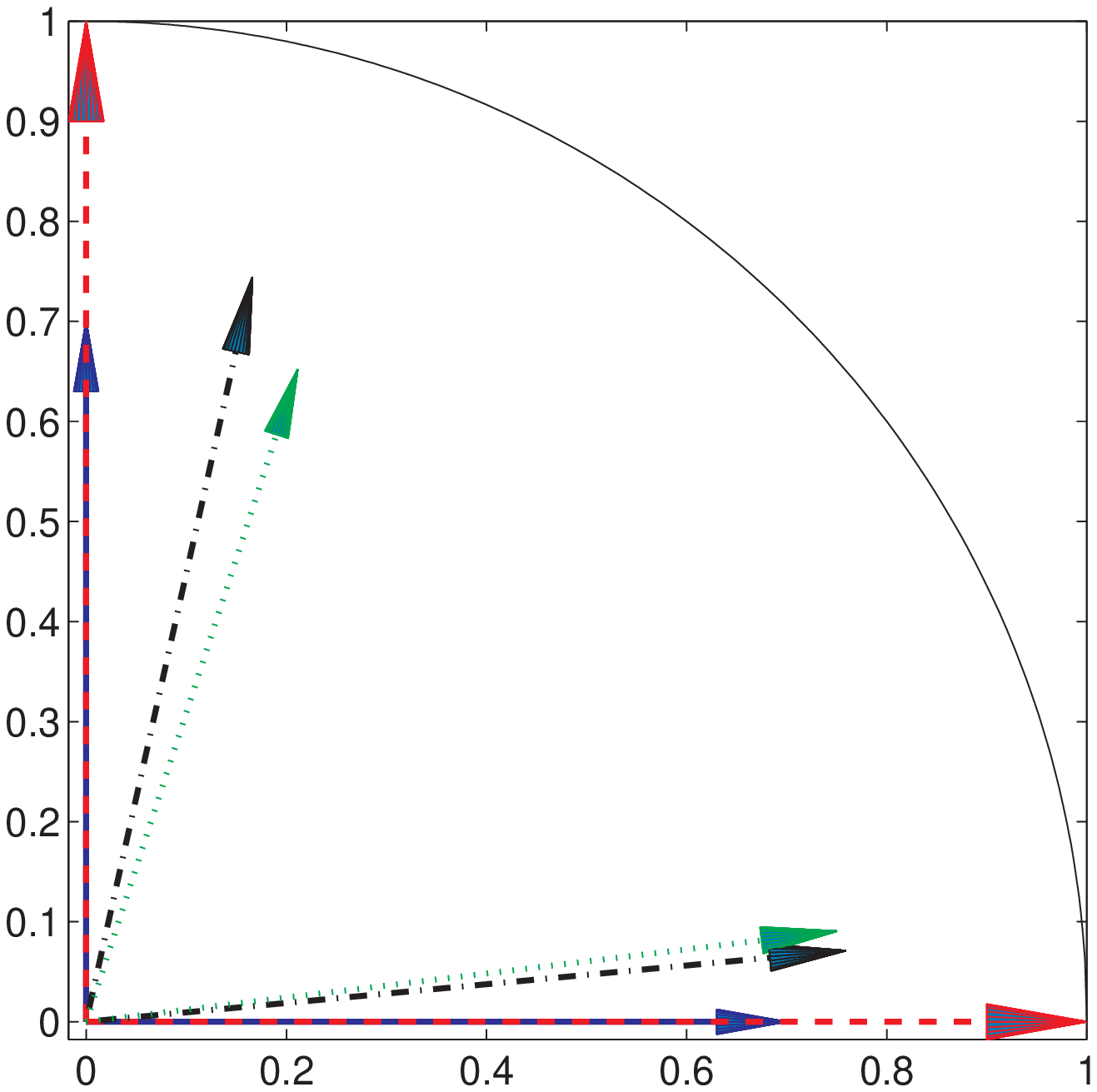}
}
\subfigure[\hspace{1.5mm} $R_1^{\rm eff}=R_2^{\rm eff}=0.79$, $\Theta=64.42^o$ EBTP: $R_1^{\rm eff}=0.69$, $R_2^{\rm eff}=0.80$, $\Theta=57.79^o$ ] {
    \label{fig:compdistmes3_c}
    \includegraphics[width=5cm]{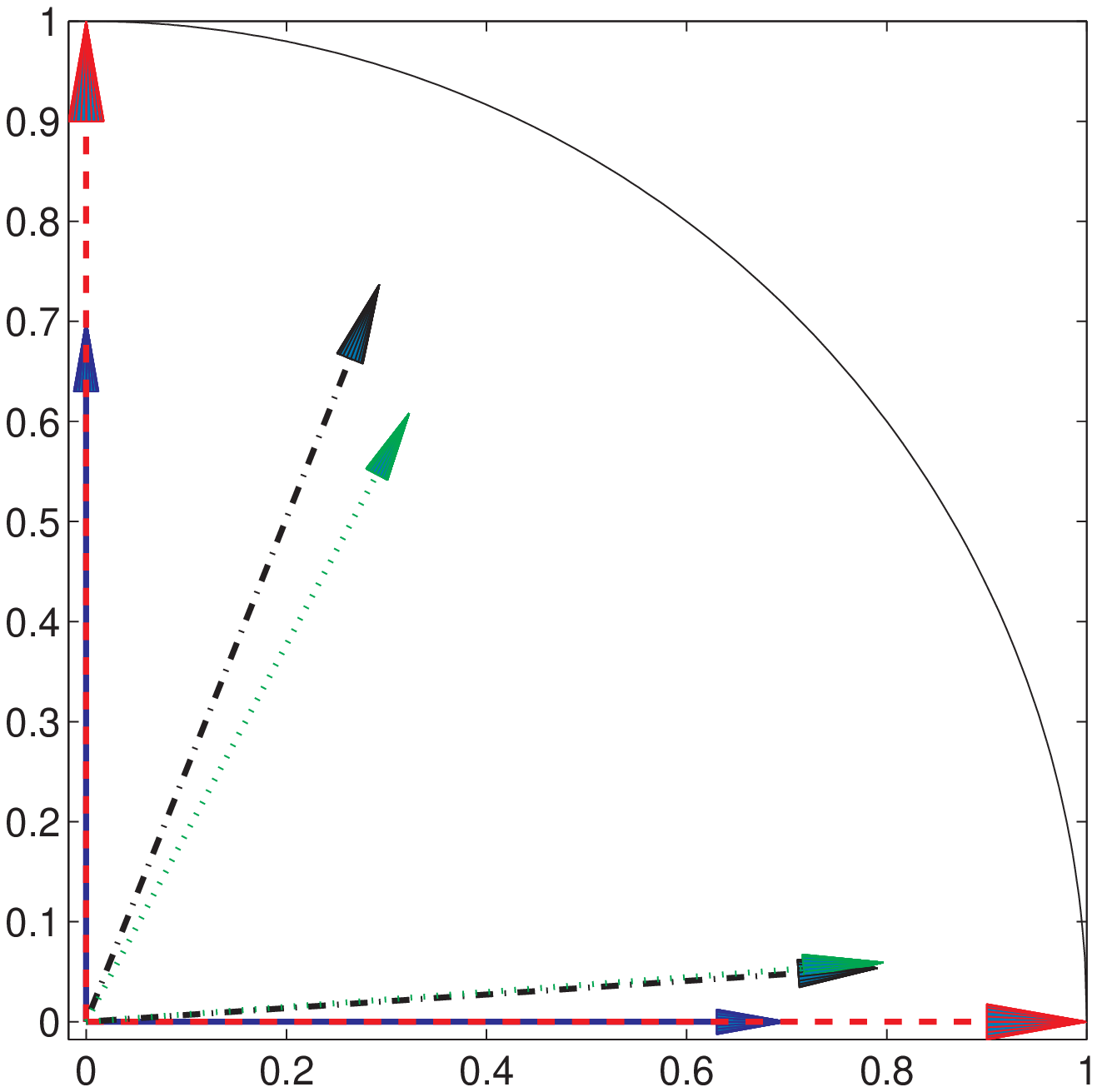}
}
\subfigure[\hspace{1.5mm} $R_1^{\rm eff}=R_2^{\rm eff}=0.81$, $\Theta=59.92^o$ EBTP: $R_1^{\rm eff}=0.69$, $R_2^{\rm eff}=0.80$, $\Theta=58.30^o$] {
    \label{fig:compdistmes3_d}
    \includegraphics[width=5cm]{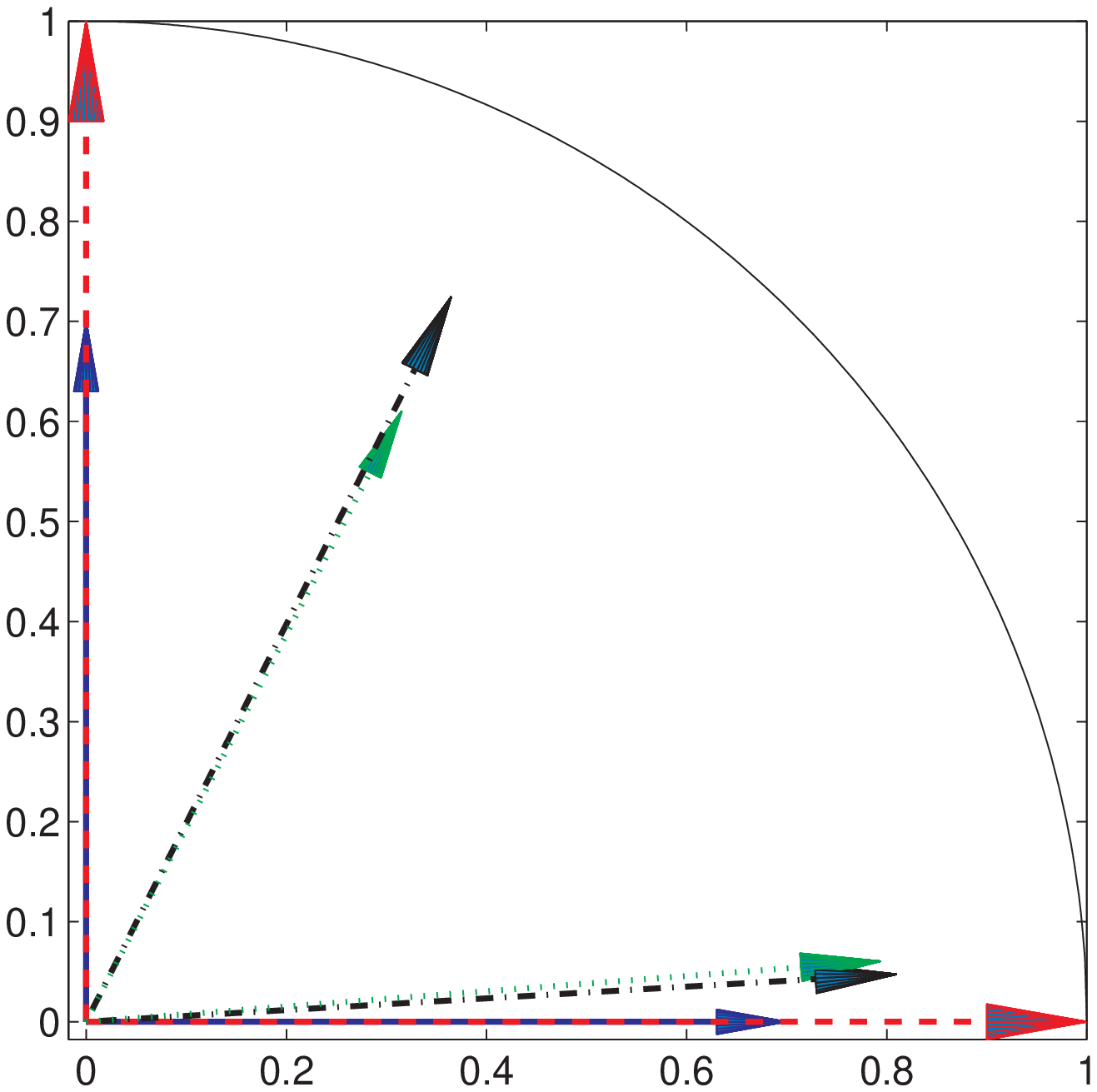}
}
\subfigure[\hspace{1.5mm} $R_1^{\rm eff}=R_2^{\rm eff}=0.94$, $\Theta=30.20^o$ EBTP: $R_1^{\rm eff}=0.93$, $R_2^{\rm eff}=0.95$, $\Theta=23.00^o$] {
    \label{fig:compdistmes3_e}
    \includegraphics[width=5cm]{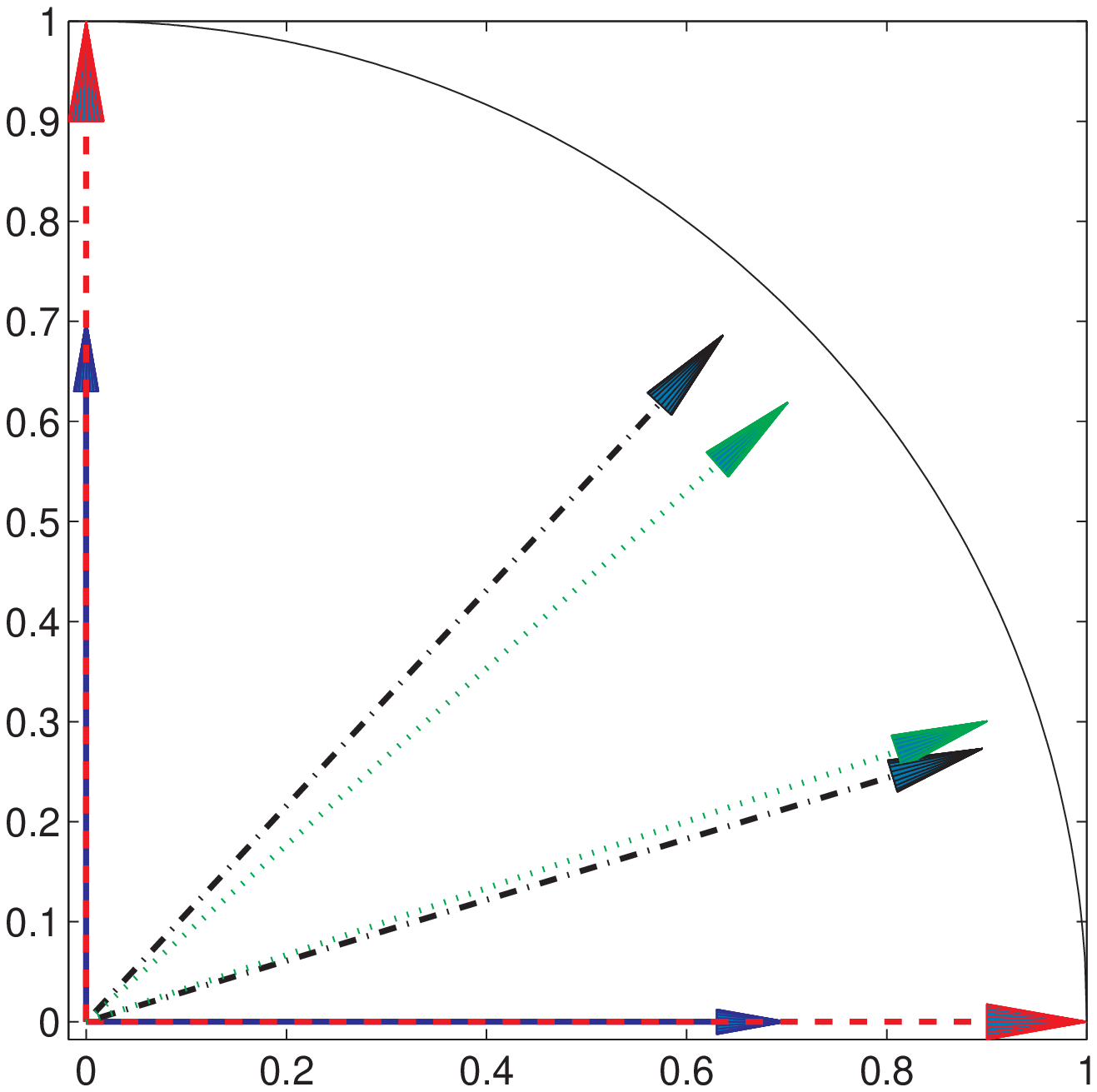}
}
\subfigure[\hspace{1.5mm} $R_1^{\rm eff}=R_2^{\rm eff}=0.97$, $\Theta=17.92^o$ EBTP: $R_1^{\rm eff}=0.97$, $R_2^{\rm eff}=0.98$, $\Theta=13.56^o$] {
    \label{fig:compdistmes3_f}
    \includegraphics[width=5cm]{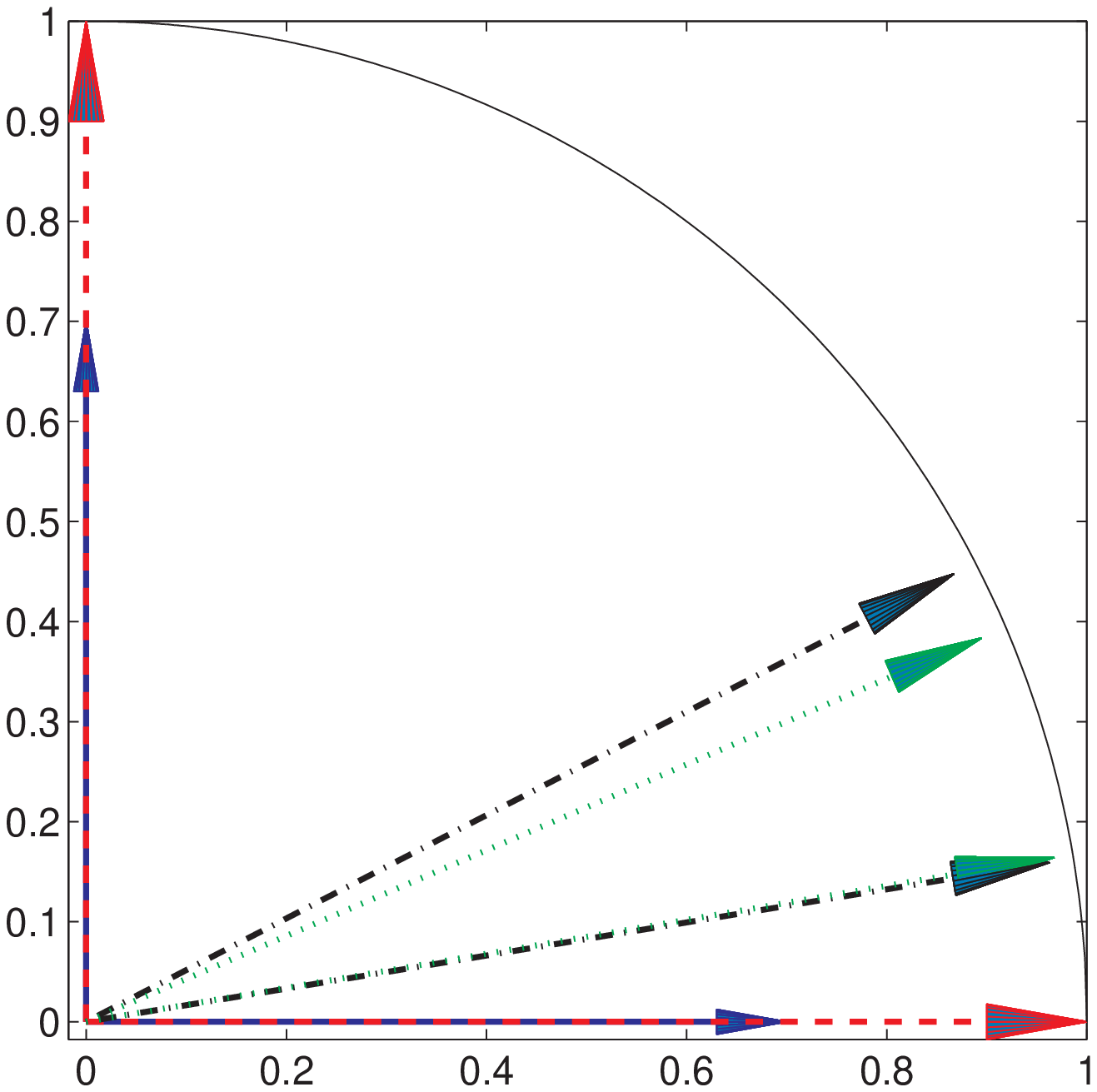}
}
\caption[Bloch vectors visualization of biased purification of equally mixed states.]{(Color online) Comparison between the Bloch vectors resulting from the numerical optimization of $\aver{\D}$, $\aver{\H^2}_1$, $\aver{\H}_2$, $\aver{\O}_2$ and $\aver{\F}_{1,2}$ for a biased transformation with $\pi_1=0.3$, $\pi_2=0.7$ of a pair of mixed source states (dashed, blue) with $R_1=R_2=0.7$ into a pair of pure states (dashed, red). The dash-dotted (black) arrows designate the optimal Bloch vectors arising from the optimization over the set of CPTP maps, while the dotted (green) arrows refer to the optimal transformation over the set of EBTP maps.}
\label{fig:compdistmes3}
\begin{picture}(0,0)
\put(-125,460){\large$\bm{\aver{\D}}$}
\put(25,460){\large$\bm{\aver{\H^2}_1}$}
\put(180,460){\large$\bm{\aver{\H}_2}$}
\put(-125,265){\large$\bm{\aver{\O}_2}$}
\put(20,265){\large$\bm{\aver{\F_{\rm HS}}_1}$}
\put(170,265){\large$\bm{\aver{\F_{\rm HS}}_2}$}
\end{picture}
\end{figure*}

\subsection{Quantitative analysis}

In the previous section we looked at some Bloch disks to visualize discrepancies between the optimization of different distance measures for a common purification problem. From the plots, we have seen that the extent to which these discrepancies occur varies according to which two measures we choose to compare. For example, a comparison between  Figs.~\ref{fig:compdistmes2_a} and~\ref{fig:compdistmes2_b} reveals a much closer resemblance than a comparison between Figs.~\ref{fig:compdistmes2_a} and \ref{fig:compdistmes2_d}. This suggests that --- \emph{for the specific control problem under consideration} --- the controller that solves the minimization of $\aver{\D}$ is not so different from the controller minimizing $\aver{\H^2}_1$, at least not as much as the controller maximizing $\aver{\F_{\rm HS}}_1$.

In this section, we attempt to make this idea of ``resemblance between optimal controllers'' a little more formal and less qualitative, in such a way to enable analogous comparisons between the role of different distance measures for problems involving higher dimensional systems and/or a larger number of atomic transformations. Our ultimate goal is to sort our distance measures in a decreasing \emph{order of compatibility} with respect to a chosen reference. For example, if $\aver{\D}$ is taken to be the reference, we would like to know how to order the remaining measures in such a way that the optimization of the first element in the list yields the controller which is, in some sense, the closest one to that produced by the minimization of $\aver{\D}$. In practice, such a list should provide the guidelines for choosing a computationally cheaper optimization problem to replace a more expensive one.

We quantify the closeness between optimal controllers according to the following construction: Start by choosing a particular convertibility problem, i.e., for some chosen values of $I$ and ${\rm d}$ select sequences of {\rm d}-dimensional density matrices $[\rho_i]_{i=1}^I$ and $[\brho_i]_{i=1}^I$. For definiteness, let $\aver{\D}$ be our reference measure, and $\aver{\D}^\ast$ its minimal value for the problem at hand. The most compatible measure $\aver{\Y}\in\{\aver{\H^2}_1,\aver{\H}_2,\aver{\O}_2,\aver{\F_{\rm HS}}_1,\aver{\F_{\rm HS}}_2\}$ \emph{with respect to the specific convertibility problem at hand}, is defined to be the one whose optimal controller yields a value of $\aver{\D}$ [denoted $\aver{\D(\Y)}$] such that $\aver{\D(\Y)}-\aver{\D}^\ast$ is the smallest over all possible choices of $\aver{\Y}$ (the second, third, etc. positions being decided in the obvious way). Throughout, we shall refer to the difference $\aver{\X(\Y)}-\aver{\X}^\ast$ as the \emph{performance drop} in units of $\aver{\X}$ due to the optimal controller for $\aver{\Y}$, or, for brevity, $\Delta(\X|\Y)$. From the above construction, it should be clear that the value of $\Delta(\X|\Y)$ is not only dependent on the choices of reference measure $\aver{\X}$ and replacement measure $\aver{\Y}$, but also on the specific choice of states involved in the transformation we want to implement.

In the following subsection, we present the details of a numerical analysis (based on unbiased transformations of random sequences  with ${\rm d}=2,3,4$ and $I=2,3,4$) that led to averaged values of $\Delta(\X|\Y)$ over many transformations. These results suggest the ``typical compatibility orderings'' proposed in Table~\ref{table:compatibility}. Because these orderings were identified from the consideration of only \emph{unbiased transformations}, we avoided redundancies and did not include $\aver{\H}_2$ and $\aver{\F_{\rm HS}}_2$ in the Table~\ref{table:compatibility} (cf. Secs.~\ref{sec:equivalenceHaltav} and~\ref{sec:theHSIP}).

Interestingly, we found that these typical orders do not seem to depend on the dimension {\rm d} of the quantum system, nor on the number $I$ of atomic transformations involved. However, our results along this direction are still preliminary and further numerical support would be required before more reliable conclusions could be drawn.

\begin{table}[h!]
\centering \caption[Conjectured compatibility orderings for unbiased state transformations.]{Conjectured compatibility orderings between distance measures based on unbiased state transformations between randomly generated sequences of ${\rm d}=2,3,4$ and $I=2,3,4$.}
\begin{tabular*}{0.4\textwidth}{@{\extracolsep{\fill}} c|c}
\hline\hline
Reference& Decreasing order\\
measure & of compatibility\\
\hline
$\aver{\D}$&$\aver{\H^2}_1$\,,\,$\aver{\O}_2$\,,\,$\aver{\F_{\rm HS}}_1$\\
$\aver{\H^2}_1$&$\aver{\O}_2$\,,\,$\aver{\D}$\,,\,$\aver{\F_{\rm HS}}_1$\\
$\aver{\O}_2$&$\aver{\H^2}_1$\,,\,$\aver{\D}$\,,\,$\aver{\F_{\rm HS}}_1$\\
$\aver{\F_{\rm HS}}_1$&$\aver{\D}$\,,\,$\aver{\H^2}_1$\,,\,$\aver{\O}_2$\\
\hline\hline
\end{tabular*}\label{table:compatibility}
\end{table}

Note that in the first line of Table~\ref{table:compatibility}, the order of compatibility with respect to $\aver{\D}$ reproduces the sequence of measures (found on Sec.~\ref{sec:unb_dif}) that yields Bloch vectors of increasing lengths and decreasing angles for the case $I={\rm d}=2$. This is a nice property, since the similarity between the Bloch vectors noted in that section could be regarded as a measure of compatibility for qubit state transformations. It is then interesting (and reassuring for the establishment of a generalized notion of compatibility between distance measures) to find that even in more general transformations, classified by a  more general compatibility measure, the same order is still observed.

From the first and third lines of Table~\ref{table:compatibility}, we note that $\aver{\H^2}_1$ is the most compatible measure to the metrics $\aver{\D}$ and $\aver{\O}_2$. The high compatibility between $\aver{\H^2}_1$  with $\aver{\D}$ is particularly significant: Because $\aver{\H^2}_1$  is quicker to compute than $\aver{\D}$ (cf. Fig.~\ref{fig:timeSDPs} and Table~\ref{table:dimSDP}),  its minimization can be regarded as an efficient estimator of a minimizer for $\aver{\D}$. Similarly, from the second line of Table~\ref{table:compatibility}, we see that $\aver{\O}_2$ is the most compatible measure to $\aver{\H^2}_1$, therefore analogous conclusions apply.

Finally (and somewhat sadly), we note that $\aver{\F_{\rm HS}}_1$ is the less compatible measure with any of the metrics. Once again, this was already seen in the less general context of the previous section. To mention some very rough estimates, we found that the performance drop $\Delta(\X|\F_{\rm HS})$ is usually of the order of $10\%$ for any choice of metric $\X$ if we consider transformations involving targets of pure states. This drop can get as high as $50\%$ in the more general case of transformations from mixed to mixed states, in which cases $\F_{\rm HS}$ is not a well motivated distance measure, as explained in Sec.~\ref{sec:max_close}.

\paragraph{Construction of Table~\ref{table:compatibility}.}

We only consider unbiased transformations of problems involving quantum systems of dimension ${\rm d}=2, 3, 4$ and a number of atomic transformations $I = 2, 3, 4$. For each one of the nine pairs $(I,d)$ that can be constructed, we computed the performance drops $\Delta(\X|\Y)$ for all possible combinations of $\X$ and $\Y$ in $300$ different transformations. From these, $100$ transformations were chosen to be between randomly generated sequences of pure states, $100$ from randomly generated sequences of pure states to randomly generated sequences of mixed states and $100$ between randomly generated sequences of mixed states.

We then computed the average value and standard deviation of $\Delta{\X|\Y}$ for each fixed choice of $\X$ and $\Y$, and for each type of transformation. Table~\ref{table:refmeasD} shows the results for the reference measure $\X=\D$ and the transformations where the source and target sequences were made out of mixed and pure density matrices, respectively. In Fig.~\ref{fig:barMP} the results of Table~\ref{table:refmeasD} are repeated and extended to account for other choices of reference measures (indicated in the horizontal axis), but still in the case of transformations from mixed to pure states. The vertical axis of each plot indicates the percental value of $\Delta(\X,\Y)$, and each bar represents a choice of $\Y$, as indicated in the code shown in the middle plot on the first line. From the observation of this plot (and others arising from other types of transformations), the compatibility orderings of Table~\eqref{table:compatibility} were constructed.

\begin{table}[h!]
\centering \caption[A quantitative analysis of the compatibility of $\aver{\D}$.]{A quantitative analysis of the compatibility of $\aver{\D}$ with respect to the other measures considered. The numbers in the table correspond to the mean values and standard deviation of $\Delta(\D|\Y)$ averaged over $100$ randomly generated unbiased transformations from mixed source sequences to pure target sequences. The choice of $\Y$ is indicated in the firs line of each column.}

\begin{tabular}{c|c|c|c|c|c|c|c}
\hline\hline
\multicolumn{2}{c|}{Reference: $\medmath{\aver{\D}}$} & $\medmath{\aver{\H^2}_1\;(\%)}$ & $\medmath{\pm \aver{\H^2}_1\;(\%)}$ & $\medmath{\aver{\O}_2\;(\%)}$ & $\medmath{\pm\aver{\O}_2\;(\%)}$ & $\medmath{\aver{\F_{\rm HS}}_1\;(\%)}$ & $\medmath{\pm\aver{\F_{\rm HS}}_1\;(\%)}$\\\hline
$I=2$&${\rm d}=2$&0.46&0.78&1.25&1.64&7.57&3.00\\\cline{2-8}
     &${\rm d}=3$&0.52&0.69&0.83&0.94&8.57&1.88\\\cline{2-8}
     &${\rm d}=4$&0.22&0.28&0.50&0.53&8.69&2\\\hline
$I=3$&${\rm d}=2$&1.25&1.50&4.21&3.25&7.82&2.70\\\cline{2-8}
     &${\rm d}=3$&0.81&0.58&1.80&1.67&6.77&2.40\\\cline{2-8}
     &${\rm d}=4$&0.77&0.57&1.36&1.29&6.81&2.24\\\hline
$I=4$&${\rm d}=2$&1.06&1.11&5.16&3.46&7.67&2.62\\\cline{2-8}
     &${\rm d}=3$&1.33&0.83&3.59&2.08&6.60&2.43\\\cline{2-8}
     &${\rm d}=4$&1.22&0.67&2.60&1.93&6.47&1.95\\\hline\hline
\end{tabular}\label{table:refmeasD}
\end{table}

\begin{figure}[h]
\centering
\includegraphics[width=14cm]{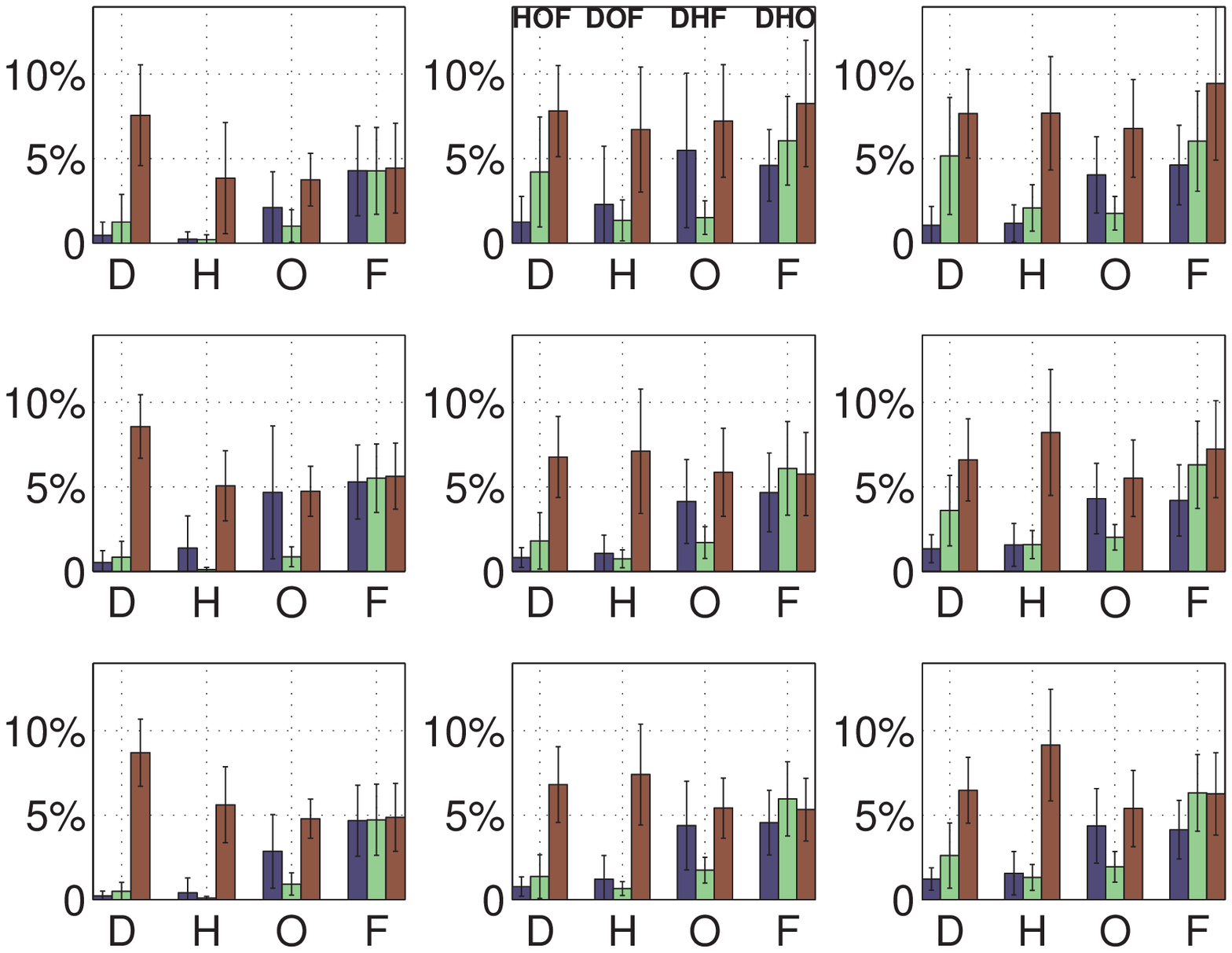}
\begin{picture}(0,0)
\put(-410,345){\vector(1,0){430}}
\put(-410,345){\vector(0,-1){350}}
\put(-350,350){$I=2$}
\put(-210,350){$I=3$}
\put(-70,350){$I=4$}
\put(-423,265){\begin{sideways}${\rm d}=2$\end{sideways}}
\put(-423,155){\begin{sideways}${\rm d}=3$\end{sideways}}
\put(-423,45){\begin{sideways}${\rm d}=4$\end{sideways}}
\end{picture}
\caption[Estimates of controllers compatibility for transformations from mixed to pure states.]{Estimates of optimal controllers compatibility for transformations between sequences of mixed and pure states. In each plot, the vertical axis gives the performance drop $\Delta(\X|\Y)$ averaged over one hundred randomly generated transformations between sequences of density matrices. $\aver{\X}$ is the reference measure specified in the horizontal axis, and $\aver{\Y}$ is one of the possible replacement measures, as specified in the middle plot on the first line. Each one of the nine plots corresponds a fixed value of $I$ and {\rm d}, indicated in the external set of axis. The error bars give the standard deviation of the averages.}\label{fig:barMP}
\end{figure} 

\chapter{Quantum control of a single qubit}\label{chap:aggie}

To a large extent, this chapter reproduces Ref.~\cite{07Branczyk012329}. Some minor notational changes were made in order to make the chapter consistent with the notation adopted in the remainder of the thesis; in addition, a few references were included and updated. More significantly, new scientific results obtained after the publication of \cite{07Branczyk012329} were included here as Sec.~\ref{sec:newsection}.
This led to a few minor additions in other parts of the text, in order to integrate the new results with the pre-existing material.\par

From a mathematical viewpoint, the main result of this chapter is an \emph{analytical solution} for the problem of optimally converting between sequences of two qubit states ($I={\rm d}=2$) in a particular setting: the target states $\brho_i$ are taken to be certain non-orthogonal pure states and the source states are taken to be $\rho_i=\mathcal{E}_p(\brho_i)$, where $\mathcal{E}_p$ is a dephasing map. The notion of optimality is captured by the maximization of $\aver{\F_{\rm HS}}_1$, which is identical to $\aver{\F}_1$ due to the purity of the targets. For ease of notation, we will denote it simply by $\aver{\F}$. The optimal solution over both feasible sets $\mathcal{Q}_2^{\rm set}$ and $\mathcal{B}_2^{\rm set}$ is obtained analytically.

Physically, the problem is phrased in terms of stabilizing an uncertain preparation of a qubit against dephasing noise. Candidate strategies are initially proposed on a physical basis, and subsequently proved to be optimal with the SDP machinery introduced in Chapter $2$. The optimal operation over $\mathcal{Q}_2^{\rm set}$ is motivated by the idea of using feedback control to optimize some quantum mechanical trade-off between information gain and disturbance. The optimal operation over $\mathcal{B}_2^{\rm set}$, in turn, arises from the classical control paradigm of maximizing the information gain.

\section[Introduction]{Introduction}\label{sec:introaggie}

Any practical quantum technology, such as quantum key distribution
or quantum computing, must function robustly in the presence of
noise.  Many modern ``classical'' technologies tolerate noise,
faulty parts, etc., by relying on \emph{feedback control} systems,
which monitor the system and use this information to control its
state.  Given the ubiquity and power of feedback control for
classical systems, it is worthwhile investigating how such control
concepts can be applied to quantum technologies as well.  However,
strategies for \emph{quantum control} must take into account
some fundamental features of quantum mechanics, namely, restrictions
on information gain, and measurement back-action.

Classically, it is possible in principle to acquire all the
information about the state of a system with certainty by using
sufficiently precise measurements.  That is, the state of a single
classical system can be precisely determined via measurement. For
quantum systems, however, this is not always possible: if the system
is prepared in one of several non-orthogonal states, no measurement
can determine which preparation occurred with certainty.

In addition, for quantum systems, monitoring comes at a price: any
measurement that acquires information about a system must
necessarily disturb it uncontrollably.  This feature is often
referred to as \emph{back-action} --- the fundamental noise induced
on a system through any measurement, which  maintains the
uncertainty relations.  This feature of quantum measurement is also
distinct from the classical situation, wherein measurements that do
not alter the state of the system can in principle be performed.

These two fundamental features of quantum systems --- that
non-orthogonal states cannot be perfectly discriminated, and that
any information gain via measurement necessarily implies disturbance
to the system --- require a reevaluation of conventional methods and
techniques from control theory when developing the theory of quantum
control.

In this chapter, we investigate the use of measurement and feedback
control of a single qubit, prepared in one of two non-orthogonal
states and subsequently subjected to noise.  Our main result is
that, in order to optimize the performance of the control scheme (as
quantified by the average fidelity of the corrected state compared
to the initial state), one must use non-projective measurements with
a strength that balances the trade-off between information gain and
disturbance.

Belavkin was the first to recognise the importance of feedback
control for quantum systems and describe a theoretical framework for
analysing both discrete and continuous time models
\cite{83Belavkin178,99Belavkin405}. Despite this early start, it is
only recently that the degree of control and isolation of quantum
systems has progressed to the point that the experimental
exploration of quantum control tasks has been possible
\cite{02Armen133602,02Smith133601,04Geremia270,04Reiner023819,04LaHaye74,06Bushev043003},
and the field is now undergoing rapid development (see for example
\cite{05job}).

The specific control problem we are interested in here is the
stabilization against noise of states of a single two level system.
Similar problems have been considered in continuous time feedback
models, e.g., the stabilization of a single state of a driven and
damped two-level atom~\cite{wang02a,02wiseman013807} and the
maintenance of the coherence of a noisy qubit using tracking
control~\cite{05lidar350}. Several recent papers have investigated
state preparation and feedback stabilization onto eigenstates of a
continuously-measured observable in higher-dimensional
systems~\cite{vanhandel05a,mirrahimi05a}.

In contrast to these prior investigations, we investigate a feedback
scheme to stabilize \emph{two} non-orthogonal states of a two-level
system. We work in a discrete-time setting, rather than
continuous-time as considered in most prior work, which considerably
simplifies the problem and most clearly illustrates the central
concepts. In significant earlier work in a discrete time setting,
Barnum and Knill proposed near-optimal strategies to correct
ensembles of orthogonal states after a general noise
process~\cite{02Barnum2097}. While Gregoratti and Werner have
investigated this kind of model of recovering the state of the
system after interaction with the
environment~\cite{03Gregoratti915,04Gregoratti2600}, their investigation
considered the case where it is possible to make measurements on the
environment. In our setting we imagine that the environment that
causes the initial decoherence is not available subsequently for the
feedback protocol.  Very recently, Ticozzi and Viola~\cite{06Ticozzi052328}
have applied both dynamical decoupling and feedback methods to
suppress unwanted dynamics of a single qubit in discrete time.

Our main interest is to investigate the effects of the kind of
trade-off between information and disturbance that is ubiquitous in
quantum information in a concrete optimal control problem. Related
information-disturbance trade-offs in quantum feedback control are
discussed in~\cite{01Doherty062306}. Finally, we note that implementing
quantum operations on a single qubit through the use of measurement
and feedback control as considered here has been investigated for
eavesdropping strategies in quantum cryptography~\cite{99Niu2764} and
for engineering general open-system dynamics~\cite{02Lloyd010101}.

Note that there is a fundamental difference between the kind of
quantum control problem we are considering here and the related task
of quantum error correction. (For an introduction to the latter,
see~\cite{00Nielsen}.)  The essence of quantum error correction is to
\emph{encode} abstract quantum information into a physical quantum
system and to choose degrees of freedom that are unaffected by the
relevant noise, or upon which errors can be deterministically
corrected. However, it can be the case that one wishes to protect
particular physical degrees of freedom of quantum systems and one is
not free to choose an arbitrary encoding.  (One such example is
\emph{reference frame distribution} via the exchange of quantum
systems~\cite{07Bartlett555}). The quantum states required for these
schemes cannot be encoded into quantum error correcting codes or
noiseless subsystems~\cite{Pre00}; protecting such systems from
noise may therefore be an application of this kind of quantum
control.

The chapter is structured as follows.  In section~\ref{sec:task}, we
define the control task in detail; in section~\ref{sec:Classical},
we present and determine the performance of control strategies based
on ``classical'' concepts. Section~\ref{sec:quantum} introduces our
quantum strategy, investigating the use of \emph{weak} quantum
measurements, and analyses its performance against the strategies of
section~\ref{sec:Classical}.  We also demonstrate that our quantum
control scheme is optimal for the task at hand.  In
section~\ref{sec:conclusions} we discuss the implications of our
result and their relevance to other problems.

\section{A Simple Control Task}
\label{sec:task}

The aim of this chapter is to explore the key issues we will confront
when applying concepts from control theory to finite-dimensional
quantum systems.  In order to facilitate the analysis and to be able
to concentrate on the key departures from classical control, we will choose a very simple quantum system and noise model.  The emphasis is
not towards a practical task, but as an illustrative example.

Consider the following operational task: a qubit prepared in one of
two non-orthogonal states $|\overline{\psi}_1\rangle$ or $|\overline{\psi}_2\rangle$ (with
overlap $\langle\overline{\psi}_1|\overline{\psi}_2\rangle = \cos\theta$ for $0\leq \theta
\leq \pi/2$) is transmitted along a noisy quantum channel. Without
knowing which state was transmitted, we will attempt to ``correct''
the system, i.e., undo the effect of the noise, through the use of a
control scheme based on measurement and feedback~\footnote{Our use
of the term ``feedback'' based on the measurement of a system refers
to a subsequent operation performed on the \emph{same} system (as
opposed to a different, identically-prepared system). This use of
the term is standard in the quantum control literature; however, the
term ``feedforward'' is occasionally given the same meaning in the
quantum computing literature (as the operation is applied forward in
the quantum circuit).  For the purpose of this thesis, we can
consider both terms as synonymous.}; see Fig.~\ref{fig:compare}.

\begin{figure}
  \begin{center}
   \includegraphics[width=0.46\textwidth]{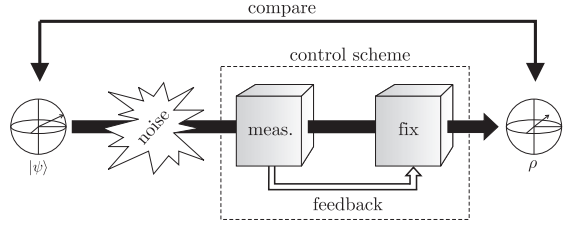}
  \end{center}
  \caption[Schematic of a quantum control procedure.]
  {Schematic of a quantum control procedure. A qubit,
  subjected to dephasing noise, is subsequently measured and corrected
  based on the results of this measurement. The output state
  $\rho^{\rm fix}$ is compared with the input state $|\overline{\psi}\rangle$ to
  characterise how well the scheme performs.}
  \label{fig:compare}
\end{figure}

The noise model that we will consider is \emph{dephasing noise}.
Let $\{|0\rangle,|1\rangle\}$ be a basis for
the qubit Hilbert space, and the Pauli operator $Z$ is the unitary
operator defined by $Z|0\rangle = |0\rangle$, $Z|1\rangle = -
|1\rangle$.  Dephasing noise is characterized as follows: with
probability $p$ a phase-flip $Z$ is applied to the system, and with
probability $1-p$ the system is unaltered.  The noise is thus
described by a quantum operation~\cite{00Nielsen}, i.e., a
completely-positive trace-preserving (CPTP) map $\mathcal{E}_p$,
that acts on a single-qubit density matrix $\rho$ as
\begin{equation}\label{eq:erho}
  \mathcal{E}_p(\rho)=p(Z\rho Z)+(1-p)\rho\,.
\end{equation}
We will consider the noisy channel to be fully characterized, meaning that $p$
is known and without loss of generality in the range $0\leq p \leq 0.5$.

We will choose the two initial states to be oriented in such a way that their
distinguishability, as measured by their trace distance, is maintained under
the action of the noise.  It is straightforward to show that this condition is
satisfied by the states
\begin{align}
  \label{eq:State1}
  |\overline{\psi}_1\rangle&=\cos\tfrac{\theta}{2}|{+}\rangle+\sin\tfrac{\theta}{2}|{-}\rangle\,,\\
  \label{eq:State2}
  |\overline{\psi}_2\rangle&=\cos\tfrac{\theta}{2}|{+}\rangle-\sin\tfrac{\theta}{2}|{-}\rangle\,,
\end{align}
where $|{\pm}\rangle=(|0\rangle\pm |1\rangle)/\sqrt{2}$.

\begin{figure}
  \begin{center}
   \includegraphics[width=0.25\textwidth]{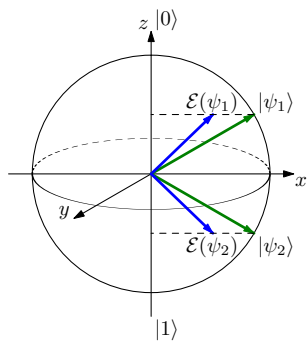}
  \end{center}
  \caption[Initial states and dephasing noise action on the Bloch sphere]
  {(Color online) Bloch sphere representation of the initial states, and the states after the noise.
  The noise shortens the Bloch vectors along the $x$-axis. We have used the notation
  $\mathcal{E}(\psi)$ as a shorthand for $\mathcal{E}(|\psi\rangle\langle\psi|)$.}
  \label{fig:bloch}
\end{figure}

Consider the Bloch sphere defined by states $|0\rangle$ and
$|1\rangle$ as the poles on the $z$-axis.   The two states
$|\overline{\psi}_1\rangle$ and $|\overline{\psi}_2\rangle$ lie in the $x{-}z$ plane and
straddle the equator of the Bloch sphere by angles $\pm\theta$; see
Fig.~\ref{fig:bloch}. On this Bloch sphere, the dephasing noise
acting on these states has the effect of decreasing the
$x$-component of their Bloch vectors. The trace distance between
these two states, given by the Euclidean distance between their
Bloch vectors, is invariant under this dephasing noise.

We now consider whether there exists a control procedure
$\mathcal{C}$ (some ``black box'') that can correct the state of
this system and counteract the noise, at least to some degree,
independent of which input state was prepared.  To quantify the
performance of any such procedure, we will use the average fidelity
to compare the noiseless input states $|\overline{\psi}_i\rangle$ with the
corrected output states $\rho^{\rm fix}_i$. Assuming an equal probability for
sending either state $|\overline{\psi}_1\rangle$ or $|\overline{\psi}_2\rangle$, the
figure of merit is
\begin{align}\label{eq:AverageFidelity}
    \aver{\F^{\mathcal{C}}} &= \tfrac{1}{2}\F(|\overline{\psi}_1\rangle,\rho^{\rm fix}_1) +
    \tfrac{1}{2}\F(|\overline{\psi}_2\rangle,\rho^{\rm fix}_2) \nonumber \\
    &= \tfrac{1}{2}\langle \overline{\psi}_1 |\rho^{\rm fix}_1 |\overline{\psi}_1\rangle +
    \tfrac{1}{2}\langle \overline{\psi}_2 |\rho^{\rm fix}_2 |\overline{\psi}_2\rangle \,,
\end{align}
where the fidelity between a pure state $|\psi\rangle$ and a mixed
state $\rho$ is defined as $\F(|\psi\rangle,\rho) \equiv
\langle\psi|\rho|\psi\rangle$.  The fidelity $\F$ ranges from $0$ to $1$
and is a measure of how much two states overlap each other (a
fidelity of 0 means the states are orthogonal, whereas a fidelity of
1 means the states are identical). It has the following simple
operational meaning when the input state is pure: the fidelity
$\F(|\psi\rangle,\rho)$ is the probability that the state
$\rho$ will yield outcome $|\psi\rangle$ from the projective
measurement
$\{|\psi\rangle\langle\psi|,|\psi^\perp\rangle\langle\psi^\perp|
\}$.

Thus, the aim is to find a control operation, described by a CPTP
map $\mathcal{C}$ independent of the choice of initial state, such
that the corrected states
\begin{equation}\label{eq:CorrectedState}
    \rho^{\rm fix}_i =
    \mathcal{C}\bigl[\mathcal{E}_p(|\overline{\psi}_i\rangle\langle\overline{\psi}_i|)\bigr] \,,
\end{equation}
for $i=1,2$ are close to the original states as quantified by the
average fidelity.  We consider control operations that consist of
two steps:  a measurement on the quantum system, followed by a
feedback operation that is conditioned on the measurement result, as
shown in Fig.~\ref{fig:compare}.

\section{Classical Control}
\label{sec:Classical}

In this section, we introduce three types of control schemes for this task, all of which are based on classical concepts, and we calculate the performance of these schemes based on the average fidelity.

\subsection{Deterministic Discriminate and Reprepare}\label{sec:DDR}

For the control of classical systems, it is always advantageous to acquire as much information about the system as possible in order to implement the best feedback scheme.  In line with this principle, a possible control strategy would be to perform a measurement on the system which attempts to discriminate between the input states, and then to reprepare the system in some state based on the measurement result. Three types of discriminate-and-reprepare schemes are investigated here and in the next section.

An important constraint imposed on the two schemes of this section is that every possible measurement outcome points to some initial preparations and is followed by the repreparation of some suitable state. Strategies of this sort are termed \emph{deterministic}. On the other hand, it is possible to design schemes where some measurement outcomes do not suggest any initial preparation; for these, only the cases where discrimination step succeeds contribute to the performance. Such schemes are termed \emph{stochastic} and will be investigated in the next section.

We first characterize all possible deterministic discriminate-and-reprepare
schemes; such schemes are associated with \emph{entanglement
breaking trace preserving} (EBTP) maps~\cite{03Horodecki629,03Ruskai643}, as
follows. Any discrimination step is described by a generalized
measurement, (or positive operator-valued measure (POVM))~\cite{00Nielsen}
yielding a classical probability distribution. The generalized
measurement is
described by the operators $\{P_a\}$ with $P_a\geq 0$ and
$\sum_a P_a = \openone_2$.
The resulting map on the quantum system is called a
\emph{quantum-classical} map $\mathcal{QC}$~\cite{98Holevo1295}, given by
\begin{equation}
  \mathcal{QC}(\rho)=\sum_a{\tr\left(\rho P_a\right)|e_a\rangle\langle e_a|}\,,
\end{equation}
where $\{ |e_a\rangle \}$ is an orthonormal basis.  The reprepare
step, in which the quantum system is re-prepared based on the
classical measurement outcome, is described by a
\emph{classical-quantum} map $\mathcal{CQ}$~\cite{98Holevo1295}, given by
\begin{equation}
  \mathcal{CQ}(\rho)=\sum_{b}{\tr\left(\rho |e_b\rangle\langle e_b|\right)Q_b}\,,
\end{equation}
where $\{Q_b\}$ are density matrices.

The concatenation $(\mathcal{CQ} \circ \mathcal{QC})(\rho)$ leads to a map of the form
\begin{equation}\label{eq:EBTP}
  \mathcal{B}(\rho) =\sum_b{\tr\left(\rho P_b\right)Q_b}\,.
\end{equation}
This map is an entanglement breaking channel. The name arises
because the output system is unentangled with any other system,
regardless of its input state.  In fact it is straightforward to see
from~\cite{03Horodecki629,03Ruskai643} that all EBTP maps can be realised by some
discriminate-and-reprepare scheme. Thus these EBTP maps formalize
our notion of deterministic discriminate-and-reprepare strategies.

The measurement for discriminating two (possibly mixed) preparations
given by Helstrom~\cite{76Helstrom} is optimal in terms of maximizing the
average probability of a success.  For our choice of states,
Helstrom's measurement is a projective measurement onto the basis
$\{ |0\rangle, |1\rangle \}$, which successfully discriminates the
states $|\overline{\psi}_1\rangle$ and $|\overline{\psi}_2\rangle$ with probability
$P_{\rm Hel} = \frac{1}{2}(1+\sin\theta)$. Note that because of
the particular choice of dephasing noise, this success probability is
independent of the noise strength $p$.

We now present and analyse two possible discriminate-and-reprepare
strategies, both of which are based on Helstrom's measurement.

\subsubsection*{Deterministic Discriminate and Reprepare Scheme 1:}\label{sec:DDR1}

With the outcome of Helstrom's measurement, one strategy is to
reprepare the qubit in either state $|\overline{\psi}_1\rangle$ or
$|\overline{\psi}_2\rangle$ based on this measurement outcome.  This scheme
yields an average fidelity
\begin{align}
  \aver{\F^{\rm DDR1}} &= P_{\rm Hel}\times 1 + (1-P_{\rm Hel})\times|\langle\overline{\psi}_1|\overline{\psi}_2\rangle|^2\nonumber\\
  &=1-\tfrac{1}{2}\left(\sin^2{\theta}-\sin^3{\theta}\right)\,.
  \label{eq:fcl}
\end{align}
Such a replacement ignores the fact
that the discrimination step can fail, with probability $1-P_{\rm  Hel}$, in which case a prepared state $|\overline{\psi}_1\rangle$ would be
reprepared as $|\overline{\psi}_2\rangle$ (or vice versa).

\subsubsection*{Deterministic Discriminate and Reprepare Scheme 2:}

We can consider other strategies that reprepare different states so
as to reduce the effect of the aforementioned error.  In particular,
we now demonstrate that the following pair of states maximizes the
average fidelity:
\begin{equation}
  |\Psi_\pm^{\rm DDR}\rangle=\sqrt{\tfrac{1}{2}\pm\tfrac{\sin^2{\theta}}{2\gamma}}
  |0\rangle+\sqrt{\tfrac{1}{2}\mp\tfrac{\sin^2{\theta}}{2\gamma}}|1\rangle\,,
\end{equation}
where $\gamma\equiv\sqrt{\sin^4{\theta}+\cos^2{\theta}}$. Note that
this replacement is also independent of $p$. Here, $|\Psi_+\rangle$
is prepared if the measurement outcome corresponds to
$|\overline{\psi}_1\rangle$, and $|\Psi_-\rangle$ is prepared otherwise. In
this strategy, the reprepared states are slightly biased towards the
alternate state to that suggested by the measurement (smaller
$\theta$) --- in a sense hedging our bet. As a proof of the
superiority of this scheme over the former, the fidelity
\begin{align}
  \aver{\F^{\rm DDR2}} & = P_{\rm Hel}\times \left(\tfrac{1}{2}|\langle\overline{\psi}_1|\Psi_+^{\rm DDR}\rangle|^2+\tfrac{1}{2}|\langle\overline{\psi}_1|\Psi_-^{\rm DDR}\rangle|^2\right) +\\
   &(1-P_{\rm Hel})\times  \left(\tfrac{1}{2}|\langle\overline{\psi}_1|\Psi_-^{\rm DDR}\rangle|^2+\tfrac{1}{2}|\langle\overline{\psi}_1|\Psi_+^{\rm DDR}\rangle|^2\right)
  \nonumber\\
  & =\frac{1}{2}+\frac{1}{2}\sqrt{\cos^2{\theta}+\sin^4{\theta}}\,,
  \label{eq:fcl2}
\end{align}
satisfies $\aver{\F^{\rm DDR2}}\geq \aver{\F^{\rm DDR1}}$ for all $\theta$. Both
$\aver{\F^{\rm DDR1}}$ and $\aver{\F^{\rm DDR2}}$ are presented in Fig.~\ref{fig:fid_all}(a).

\begin{figure}
\begin{center}
\includegraphics[width=\textwidth]{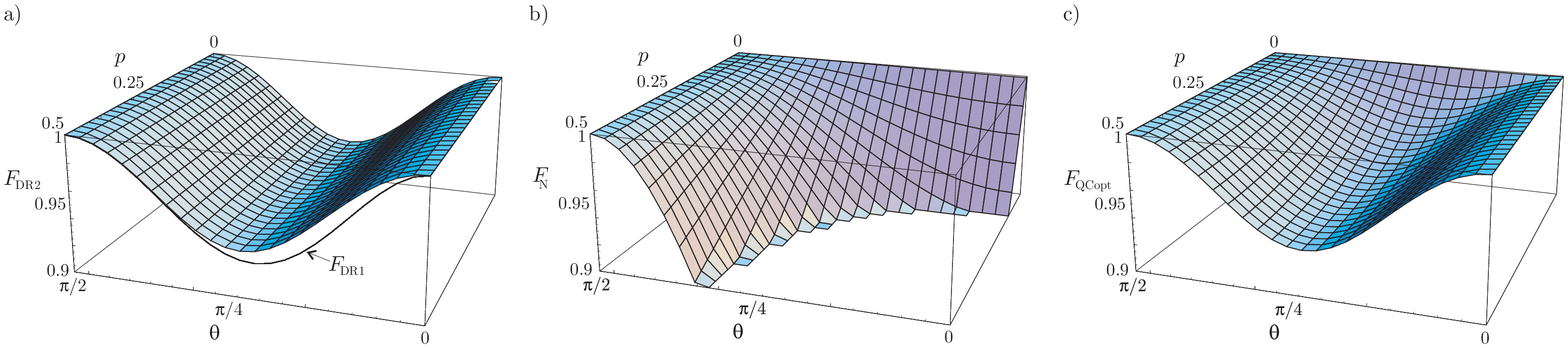}
\end{center}
\caption[Average fidelities for quantum and classical schemes.]{(Color online) The performance of the schemes, quantified by the average fidelity, as a function of the amount of noise $p$ and the angle between the input states $\theta$. a) Deterministic discriminate and reprepare scheme quantified by the average fidelity $\aver{\F^{\rm DDR2}}$ of Eq.~(\ref{eq:fcl2}). The fidelity $\aver{\F^{\rm DDR1}}$ of Eq.~(\ref{eq:fcl}) is shown as a solid line at $p=0.5$. Both average fidelities $\aver{\F^{\rm DDR2}}$ and $\aver{\F^{\rm DDR1}}$ are independent of $p$. b)``Do nothing'' scheme, quantified by the average fidelity $\aver{\F^{\rm DN}}$ of Eq.~(\ref{eq:fn}). For this scheme, the average fidelity drops to $\aver{\F^{\rm DN}} = 1/2$ for $p=0.5$ and $\theta=0$. c) Stochastic classical and deterministic quantum control schemes, quantified by the average fidelities $\aver{\F^{\rm SDR}}=\aver{\F^{\mathrm{QCopt}}}$ of Eqs.~\eqref{eq:fid_stochastic} and~\eqref{eq:fid_qc}. The range of fidelities plotted has been made identical in all the figures to aid comparison. } \label{fig:fid_all}
\end{figure}

This second discriminate-and-reprepare scheme is in fact the
\emph{optimal} deterministic discriminate and reprepare scheme, in that it
achieves the highest average fidelity
\begin{equation}
  \max_{\mathcal{B}} \aver{\F^\mathcal{B}}
  = \max_{\mathcal{B}} \tfrac{1}{2}\sum_{i=1}^{2}{\langle\overline{\psi}_i|
    \mathcal{B}\bigl[\mathcal{E}_p(|\overline{\psi}_i\rangle\langle\overline{\psi}_i|)\bigr]|\overline{\psi}_i\rangle}
    \,,
  \label{eq:EBTPoptimisation}
\end{equation}
where the maximization is over all EBTP maps $\mathcal{B}$ acting on
a single qubit.  This optimization was performed (in a different
setting) by Fuchs and Sasaki~\cite{03Fuchs377}. In the Appendix~\ref{app:optdetc}, we
provide an alternate proof of optimality using techniques from convex optimization.

\subsection{Stochastic Discriminate and Reprepare}\label{sec:newsection}

In this section we propose a particular discriminate-and-reprepare scheme in which the discrimination step can produce one of three different outcomes. In two of them, we get a suggestion of what the initial preparation was and suitably reprepare the system. In contrast, no suggestion is conveyed when the third outcome occurs, in which case we simply declare our ignorance and do not reprepare any state. The performance of the scheme is computed considering only the random occurrences of suggestive outcomes, and for this reason, the scheme is said to be \emph{stochastic} or \emph{post-selected}.

The scheme studied here can still be modeled by an entanglement breaking map, however not a trace preserving one. The trace preserving condition is relaxed to account for the fact that sometimes no output state is produced, which, on average, leads to a map that outputs ``density matrices of trace less than one''. In fact, the set of entanglement breaking \emph{trace decreasing} (EBTD) maps characterize all possible \emph{stochastic} discriminate-and-reprepare schemes. The particular scheme proposed in the following is relevant because it seems to be optimal over the set of EBTD maps. In particular, it arises from the POVM that maximizes the success rate of discrimination for a certain fixed fraction of inconclusive results~\cite{03Fiurasek012321,03Eldar042309}, followed by a replacement resulting from an optimization procedure via Langrange multipliers.

Motivated by the results of Refs.~\cite{03Fiurasek012321,03Eldar042309}, we propose the following POVM for the implementation of the discrimination step,
\begin{align}
\Pi_0&=\frac{r_x}{1+r_x}\left(\openone_2+X\right)\,,\\
\Pi_\pm&=\frac{1}{1+r_x}\left(\openone_2-r_x X \pm \sqrt{1-r_x^2} Z\right)\,,
\end{align}
where $r_x = (1-2p)\cos{\theta}$ is the $x$ component of the Bloch vector describing the system after the noise.

The outcomes `$+$' and `$-$' are interpreted to be \emph{suggestive} of the initial preparation $\ket{\overline{\psi}_1}$ and $\ket{\overline{\psi}_2}$, respectively\footnote{Recall from Ref.~\cite{04Feng012308} that in our case of two mixed states of non-orthogonal support it is impossible to design a POVM for \emph{unambiguous state discrimination}. That is why the outcomes `$\pm$' do not determine the initial preparation.}. As in the second deterministic discriminate-and-reprepare scheme of the previous section, we acknowledge the possibility of a misleading suggestion by repreparing states which are not precisely $\ket{\overline{\psi}_1}$ or $\ket{\overline{\psi}_2}$, but slightly biased to the alternate state to that suggested by the measurement. In this case, the replacements are, up to a normalization factor,
\begin{equation}
\ket{\Psi_\pm^{\rm SDR}}=\left(\frac{1}{\mp\zeta+\sqrt{1+\zeta^2}}\right)\ket{0}+\ket{1}\,,
\end{equation}
where $\zeta$ is implicitly defined as a function of $\theta$ and $p$ according to $\zeta\equiv\tan{\theta}\sin{\theta}/\sqrt{1-r_x^2}$. It follows from this that now the repreparation depends not only on the initial states, but also on the details of the dephasing map; this fact being in contrast with the replacements used in the classical deterministic schemes introduced in the previous section.

An outcome `$0$' signals an inconclusive result, in which case no repreparation step takes place. As explained before, this event is not taken into account for the characterization of the performance of the scheme.

Let us now show how the stochasticity is included in the computation of the average fidelity. We first compute the probabilities $P_{0}$, $P_{SC}$ and $P_{SI}$ of the following events: an inconclusive outcome, a suggestive outcome that correctly indicates the initial preparation and a suggestive outcome that incorrectly indicates the initial preparation,
\begin{align}
P_0&=\frac{1}{2}\tr\left[\Pi_0\mathcal{E}(\ket{\overline{\psi}_1}\!\bra{\overline{\psi}_1})\right]+ \frac{1}{2}\tr\left[\Pi_0\mathcal{E}(\ket{\overline{\psi}_2}\!\bra{\overline{\psi}_2})\right]\,,\\
P_{SC}&=\frac{1}{2}\tr\left[\Pi_+\mathcal{E}(\ket{\overline{\psi}_1}\!\bra{\overline{\psi}_1})\right]+ \frac{1}{2}\tr\left[\Pi_-\mathcal{E}(\ket{\overline{\psi}_2}\!\bra{\overline{\psi}_2})\right]\,,\\
P_{SI}&=1-P_I-P_{SR}\,.
\end{align}

Then, the average fidelity is given by
\begin{align}
\aver{\F^{\rm SDR}}&=\frac{1}{1-P_0}\left[P_{SC}\left(\frac{1}{2}|\langle\overline{\psi}_1|\Psi_+^{\rm SDR}\rangle|^2+\frac{1}{2}|\langle\overline{\psi}_2|\Psi_-^{\rm SDR}\rangle|^2\right)\right.+\nonumber\\&
P_{SI}\left.\left(\frac{1}{2}|\langle\overline{\psi}_1|\Psi_-^{\rm SDR}\rangle|^2+\frac{1}{2}|\langle\overline{\psi}_2|\Psi_+^{\rm SDR}\rangle|^2\right)\right]\,.
\end{align}
where the division by $1-P_0$ guarantees that only the suggestive outcomes are accounted. After some cumbersome manipulation, we obtain
\begin{equation}\label{eq:fid_stochastic}
\aver{\F^{\rm SDR}}=\frac{1}{2}+\frac{1}{2}\sqrt{\cos^2{\theta}+\frac{\sin^4{\theta}}{1-r_x^2}}\,.
\end{equation}

Comparing the above with the performance from the optimal deterministic discriminate and replace from Eq. \eqref{eq:fcl2}, we note that $\aver{\F^{\rm SDR}}\geq \aver{\F^{\rm DDR2}}$, since the denominator $1-r_x^2$ is obviously bounded between zero and one. Of course, it is the possibility of disregarding certain measurement outcomes that allows for this improvement.

In Sec.~\ref{sec:quantum} we will show that the fidelity~\eqref{eq:fid_stochastic} can be obtained in a deterministic framework if we switch from the classical concept of discriminate-and-reprepare to a genuinely quantum approach to control quantum systems. Quite remarkably, we will see that Eq.~\eqref{eq:fid_stochastic} gives precisely the performance of the \emph{optimal} deterministic quantum control scheme.

\subsection{Do Nothing}\label{sec:donothing}

Another control strategy would be to do nothing to correct the
states.  Although trivial, this strategy is of interest for
comparison with other schemes.  (There exist schemes that perform
worse than this strategy, because of the feature of quantum systems
that every measurement that acquires information will uncontrollably
disturb the system.)  This scheme does \emph{not} lie within the set
of discriminate-and-reprepare schemes described above (it is not
described by an entanglement breaking map) but we will nonetheless refer to it as
``classical.''

The average fidelity of this scheme is given by
\begin{equation}
   \aver{\F^{\rm DN}} = 1-p\cos^2\theta \,.
   \label{eq:fn}
\end{equation}
This performance is plotted in Fig.~\ref{fig:fid_all}(b).  Clearly,
this scheme performs best for small amounts of noise ($p\simeq 0$)
and for input states with Bloch vectors that are near the $z$-axis
(which is invariant under the dephasing noise).  In some non-trivial
regions of the $(p,\theta)$ parameter space, in particular in the
range of low noise, this ``do nothing'' scheme outperforms the
optimal deterministic discriminate-and-reprepare scheme.

\section{Deterministic Quantum Control}
\label{sec:quantum}

In the previous section, we presented control schemes based on
classical concepts.  However, using techniques that may lead to optimal control schemes for a classical system may not necessarily lead to optimal schemes for a quantum system.  As we will now demonstrate, the above deterministic classical control strategies can be outperformed by using a strategy based on quantum concepts, and the performance of the classical strategy can be obtained in a deterministic framework.

We note that the classical schemes presented in the previous section lie at the extreme ends of a spectrum:  the
``discriminate-and-reprepare'' strategy achieved maximum information
gain and induced a maximum disturbance, whereas the ``do-nothing''
strategy achieved zero disturbance but produced zero information
gain. As demonstrated by Fuchs and Peres~\cite{96Fuchs2038}, there exist
an entire range of generalized measurements that trade off
information gain and disturbance. A possible avenue for improvement
in our control schemes is to tailor the measurement in such a
way as to find a compromise, if one exists, between acquiring
information about the noise but not disturbing the system too much
as a result of the measurement.

In the following, we exploit the non-uniqueness of the Kraus decomposition of a CPTP map (cf. Sec.~\ref{sec:krausnonunique}) to re-express the noise process $\mathcal{E}_p$ in a
way that suggests a strategy for constructing such an improved
feedback protocol.

\subsection{Reexpressing the noise}

To develop an intuitive picture, we will make use of a
\emph{preferred ensemble} for the quantum operation $\mathcal{E}_p$
describing the noise.  That is, we use a decomposition of the
operation into different Kraus (error) operators than that given in
Eq.~(\ref{eq:erho}).  The resulting quantum operation
$\mathcal{E}_p$ describing the noise, however, is equivalent.

Consider the following quantum operation on a qubit, viewed on the
Bloch sphere: with probability $1/2$, the Bloch vector of the qubit
is rotated by an angle $+\alpha$ about the $z$-axis, and with
probability $1/2$ it is rotated by $-\alpha$ about the $z$-axis.
Rotations about the $z$-axis are described by the operator
\begin{equation}\label{eq:ZRotation}
    Z_{\alpha} = \mathrm{e}^{-i\alpha Z/2} = \cos(\alpha/2)\openone_2 - i
    \sin(\alpha/2)Z \,,
\end{equation}
and the quantum operation is then
\begin{align}
    \mathcal{E}_\alpha(\rho) &= \tfrac{1}{2}Z_\alpha\rho
    Z_\alpha^\dag + \tfrac{1}{2}Z_{-\alpha}\rho
    Z_{-\alpha}^\dag \nonumber \\
    &= \sin^2(\alpha/2) (Z\rho Z) + \cos^2(\alpha/2) \rho \,.
    \label{eq:PreferredEnsemble}
\end{align}
Thus, this quantum operation is equivalent to the dephasing noise
$\mathcal{E}_p$, with $p = \sin^2(\alpha/2)$.

Viewing the noise operation $\mathcal{E}_p$ with this preferred
ensemble, it is possible to describe the noise as rotating the Bloch
vector of the state by $\pm\alpha$ with equal probability.  A
possible control strategy, then, would be to attempt to acquire
information about the direction of rotation ($\pm \alpha$) via an
appropriate measurement, and then to \emph{correct} the system based
on this estimate.  Loosely, we desire a measurement that determines
whether the noise rotated the state one way ($+\alpha$) or another
($-\alpha$). Then, based on the measurement result, we apply
feedback:  a unitary operation (rotation) that takes the state of
the system back to the desired axis.

A projective measurement, wherein the state of the system collapses
to an eigenstate of the measurement, does not meet these
requirements because such a measurement destroys the
distinguishability of the two possible states.  Instead, we consider
the use of a \emph{weak} measurement, with a measurement strength
chosen to balance the competing goals of acquiring information and
leaving the system undisturbed.  We now show that such a strategy is
possible, and that there is a non-trivial optimal measurement
strength for this task.

\subsection{Weak non-destructive measurements}\label{sec:wk_meas}

For our quantum control scheme, we will make use of a
type of measurement that satisfies two
key requirements: (1)~the strength of the measurement should be
controllable, i.e., we should be able to vary the trade-off between
information gain and disturbance (back-action); and (2)~the
measurement should be non-destructive,
which leaving the measured system in an appropriate quantum state
given by the desired collapse map.
Such weak non-destructive measurements have recently been developed and
demonstrated in single-photon quantum optical
systems~\cite{pryde:190402,ralph:012113}.

Using the preferred ensemble describing the noise,
Eq.~(\ref{eq:PreferredEnsemble}), we expect intuitively that this
weak measurement should be along the $y$-axis of the Bloch sphere in
order to provide information about which direction ($\pm\alpha$) the
system was rotated, \emph{without} acquiring information about which
initial state the system was prepared in. One suitable family of
POVMs consists of two operators given by $E_m = M_m^\dag M_m$, for
$m=0,1$, where $M_m$ are the measurement operators~\cite{00Nielsen}
\begin{align}\label{eq:MeasurementOperator1}
    M_0 &= \cos(\chi/2)|{+}i\rangle\langle{+}i| +
    \sin(\chi/2)|{-}i\rangle\langle{-}i| \,, \\
    \label{eq:MeasurementOperator2}
    M_1 &= \sin(\chi/2)|{+}i\rangle\langle{+}i| +
    \cos(\chi/2)|{-}i\rangle\langle{-}i| \,.
\end{align}
The strength of the measurement depends on the choice of the parameter
$\chi$. The eigenstates of $Y$ are $|{\pm} i\rangle \equiv (|0\rangle\pm
i|1\rangle)/\sqrt{2}$.
The probabilities of obtaining the measurement results $m=0,1$ for a qubit in the state $\rho_{\rm in}$ are given by
\begin{equation}\label{eq:BornRule}
    p_m = \tr\left(E_m \rho_{\rm in}\right) \,,
\end{equation}
and the resulting state of the qubit immediately after the
measurement is
\begin{equation}\label{eq:PostMeasState}
    \rho_{\rm out}^{(m)} = \frac{M_m \rho_{\rm in}
    M_m^{\dag}}{p_m} \,.
\end{equation}

Consider the following two limits. If $\chi=\pi/2$ the two
measurement operators are the same and are proportional to the
identity. As a result the outcome probabilities are independent of
the state and the state of the signal is unaltered by the
measurement.  If $\chi=0$, a projective measurement on the signal is
induced: the signal state is projected onto the state
$|{{-}i}\rangle$ ($|{{+}i}\rangle$) when the measurement result is 0
(1). For $0 < \chi < \pi/2$, the resulting measurement on the signal
is non-projective but non-trivial.

It is illustrative to view the effect of this measurement on the
noisy input states on the Bloch sphere. In
Fig.~\ref{fig:stateprog}(a) we can see that the effect of the noise
is to shorten the length of the Bloch vector of the qubit state
(making it less pure) while \emph{increasing} the angle between the
Bloch vector and the $x$-$y$ plane from $\theta$ to $\theta'$, where
$\theta'>\theta$ . When the measurement is made, three things
happen, as can be seen in Fig.~\ref{fig:stateprog}(b): 1) the Bloch
vector is lengthened (the state becomes more pure); 2) the angle
$\theta'$ \emph{decreases} to some lesser angle $\theta''$; and 3)
the state is rotated about the $z$-axis one way or the other
depending on the result of the measurement.  The first two effects
work towards our advantage (purifying the state while decreasing
$\theta'$); the third effect we attempt to correct using
\emph{feedback}.

\begin{figure}
  \begin{center}
   \includegraphics[width=.45\textwidth]{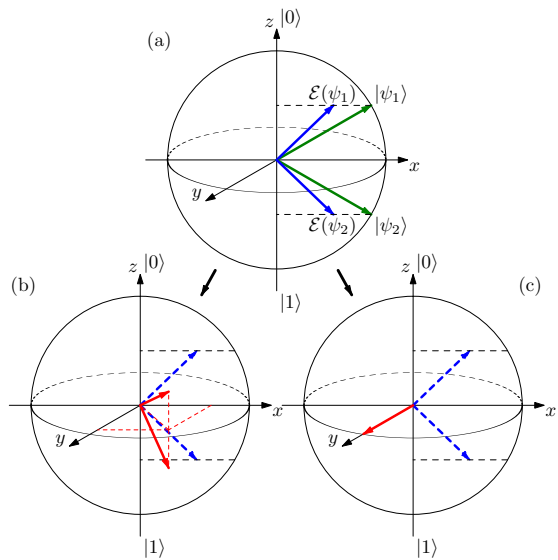}
  \end{center}
  \caption[Bloch sphere representation of the effect of a weak
  measurement.]{(Color online) Bloch sphere representation of the effect of a weak
  measurement on the system.  The transformations shown
  here correspond to having obtained the measurement result
  ``$0$''  (for the result ``$1$'',
  the behaviour would be a reflection in the $x$-$z$ plane.)
  a) The two initial states
  $|\overline{\psi}_{1,2}\rangle$ are mapped to $\rho_{1,2}$ by the noise;
  b) a weak measurement is performed
  with $0<\chi<\frac{\pi}{2}$; c) a strong projective measurement
  ($\chi=0$) is performed projecting either state into
  $|{{-}i}\rangle$. While no measurement will not yield any information
  about the system,
  a strong measurement will maximally disturb the system.
  A weak measurement will gain some information
  while also limiting the disturbance on the system.}
  \label{fig:stateprog}
\end{figure}

We will now describe how to implement this measurement using a
projective measurement on an ancillary {\it meter} qubit and an
entangling gate between the original {\it signal} qubit and the
meter. The strength of the measurement can be controlled by varying
the level of entanglement between the two qubits, which can be
implemented by initiating the meter in the state $|0\rangle$ and
subsequently applying a $Y_{\chi}$ rotation [as shown in figure
\ref{fig:schem}(a)], where
\begin{equation}
  Y_{\chi}=\mathrm{e}^{-i\chi
  Y/2}= \begin{pmatrix}\cos(\chi/2)&-\sin(\chi/2)\\
  \sin(\chi/2)&\cos(\chi/2)\end{pmatrix}\,.
\end{equation}
The parameter $\chi$ ranges from $0$ to $\pi/2$ and characterizes
the strength of the measurement, with 0 equivalent to a projective
measurement and $\pi/2$ equivalent to no measurement.

The entangling gate consists of a $X_{\frac{\pi}{2}}$ rotation on
the signal state, followed by a \textsc{cnot} gate with the signal
state as the control and the meter state as the target, followed by
a $X_{-\frac{\pi}{2}}$ on the signal state, where
\begin{equation}
  X_{\phi}=\mathrm{e}^{-i\phi
  X/2}= \begin{pmatrix} \cos(\phi/2) & -i\sin(\phi/2)\\
  -i\sin(\phi/2)&\cos(\phi/2)\end{pmatrix} \,,
\end{equation}
and where the Pauli matrix $X$ is given by $X|0\rangle = |1\rangle$
and $X|1\rangle = |0\rangle$.  The rotations $X_{\pm\frac{\pi}{2}}$
are used to ensure that the resulting weak measurement on the signal
qubit is performed in the $\{|{+}i\rangle,|{-}i\rangle\}$ basis. The
entangling gate then correlates (to a degree which depends on
$\chi$) the $\{|{+}i\rangle,|{-}i\rangle\}$ basis of the signal
qubit to the $\{|0\rangle,|1\rangle\}$ basis of the meter qubit.

\begin{figure}[h]
\centering
\subfigure[]
{
    \label{fig:sub:a}
    \includegraphics[width=0.45\textwidth]{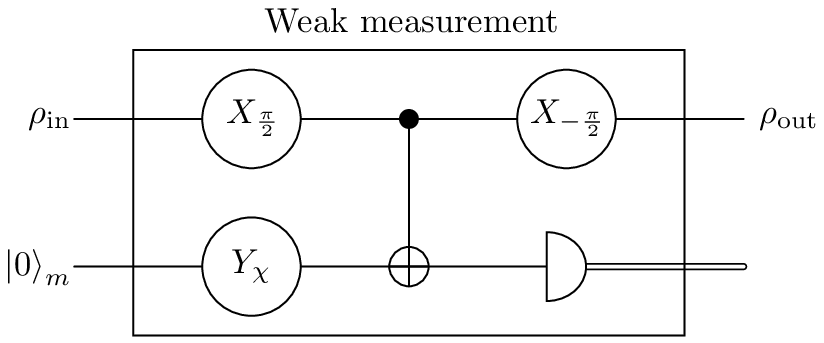}
}
\subfigure[]
{
    \label{fig:sub:b}
    \includegraphics[width=0.45\textwidth]{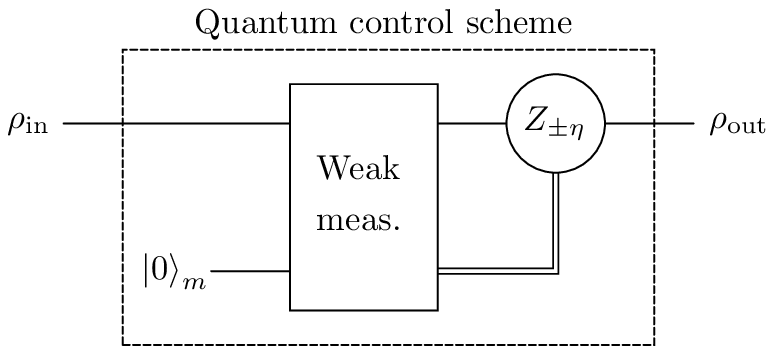}
}
\caption[Circuit diagrams of the weak measurement and quantum control scheme.]{(a) Circuit diagram of the weak measurement scheme. The input signal state $\rho_{\rm in}$ is entangled to the meter state using the \textsc{cnot} gate. The $X_{\pm\frac{\pi}{2}}$ rotations ensure that the weak measurement of the signal state is made in the desired basis $\{|{+}i\rangle,|{-}i\rangle\}$. The strength of the measurement is set using the rotation $Y_{\chi}$.  The meter state is measured in the computational basis, resulting in a classical signal (0 or 1) to be fed forward to the correction stage of the control scheme. (b) Circuit diagram of the control scheme.  A weak measurement is made on the input state and, based on the measurement results, the signal state will be rotated by $Z_{\eta}$ ($Z_{{-}\eta}$) conditional on the result of the weak measurement being 0 (1).}
\label{fig:schem}
\end{figure}

Finally the meter qubit is measured in the basis $\{|0\rangle,
|1\rangle\}$, yielding a result 0 or 1.  This measurement on the
meter induces a measurement on the signal that is precisely equal to
the generalized measurement described by the measurement operators
$M_m$ of Eq.~(\ref{eq:MeasurementOperator1}).

\subsection{Feedback control}

Once a weak measurement has been performed, a correction based on
the measurement result is performed on the quantum system: the
feedback control.  We choose the correction to be a unitary rotation
about the $z$-axis, $Z_{\pm\eta}$ where
\begin{equation}
  Z_{\eta}=\mathrm{e}^{-i\eta
  Z/2}= \begin{pmatrix} \mathrm{e}^{-i\eta/2}&0 \\
  0&\mathrm{e}^{+i\eta/2} \end{pmatrix} \,,
\end{equation}
with the aim to bring the Bloch vector of the qubit back onto the
$xz$-plane. The angle of rotation is chosen to be $\pm\eta$, depending on the
measurement result (${+}\eta$ corresponding to the measurement
result 0, and ${-}\eta$ to the measurement result 1). It is possible
to choose $\eta$ so that the system state is returned to the
$xz$-plane for all values of $p,\theta$ and $\chi$ and for both
measurement outcomes by choosing
\begin{equation}
  \label{eq:etaopt}
  \tan\eta=\frac{1}{(1-2p)\cos\theta\tan\chi}\,,
\end{equation}
with $\eta$ in the range $0 \leq \eta \leq \pi/2$.  This angle $\eta$ can be
calculated because the dephasing noise has been previously
characterised (i.e., $p$ is known).

The resulting weak measurement followed
by feedback is thus described by a quantum operation (a CPTP map)
$\mathcal{C}_{\rm QC}$ acting on a single qubit, given by
\begin{equation}\label{eq:QCasCPTP}
  \mathcal{C}_{\rm QC}(\rho) = (Z_{{+}\eta}M_0) \rho
  (Z_{{+}\eta}M_0)^\dag + (Z_{{-}\eta}M_1) \rho
  (Z_{{-}\eta}M_1)^\dag \,,
\end{equation}
where the measurement operators $M_m$ are given by
Eqs.~\eqref{eq:MeasurementOperator1} and~\eqref{eq:MeasurementOperator2}.

In summary, the quantum control scheme operates by performing a weak
measurement of the system and then correcting it based on the
results of the measurement, as in Fig.~\ref{fig:schem}b.  The weak
measurement is made by entangling an ancillary meter state with the
signal state using an entangling unitary operation, then performing
a projective measurement of the meter state.  The level of
entanglement depends on the input state of the meter, which is
controlled by a $Y_{\chi}$ rotation; this level of entanglement in
turn determines the strength of the measurement.  After measurement
of the meter, the signal state is altered due to the measurement
back-action.  To correct for this back-action, a rotation about the
$z$-axis is applied to the state, returning it back to the
$xz$-plane. To characterise how well the scheme works, we now
investigate the average fidelity.

\subsection{Performance}\label{sec:performance}

The performance of this quantum control scheme, quantified by the
average fidelity~(\ref{eq:AverageFidelity}), is
\begin{equation}\label{eq:fid}
    \aver{\F^{\rm QC}} = \tfrac{1}{2}\left[ 1 +
    \sin^2\theta \sin \chi+\cos \theta \sqrt{1-
    (1-r_x^2)\sin^2\chi}\right]\,,
\end{equation}
where $r_x = (1-2p)\cos{\theta}$.

We can see that $\aver{\F^{\rm QC}}$ is a function of the amount of noise $p$, the angle between the initial states $\theta$, and the measurement strength $\chi$.  The dependence of this fidelity on the measurement strength, for fixed $p$ and $\theta$, is illustrated in Fig.~\ref{fig:optimum_strength}.  For each value of $p$ and $\theta$, there is an \emph{optimum} measurement strength $\chi_{\mathrm{opt}}$ which maximizes the average fidelity~(\ref{eq:fid}).  This optimum measurement strength is found to be non-trivial except for the limiting cases of $p=0$ or $\theta=0,\pi/2$, and is given by
\begin{equation}\label{eq:chiopt}
   \chi_{\rm opt}(p,\theta)
   \equiv \sin^{-1}\sqrt{\frac{\sin^4\theta}{(1-r_x^2)^2\cos^2\theta +
     (1-r_x^2)\sin^4\theta}}\,,
\end{equation}
as a function of the amount of noise $p$ and the angle between the
initial states $\theta$.

\begin{figure}
  \begin{center}
   \includegraphics[width=.45\textwidth]{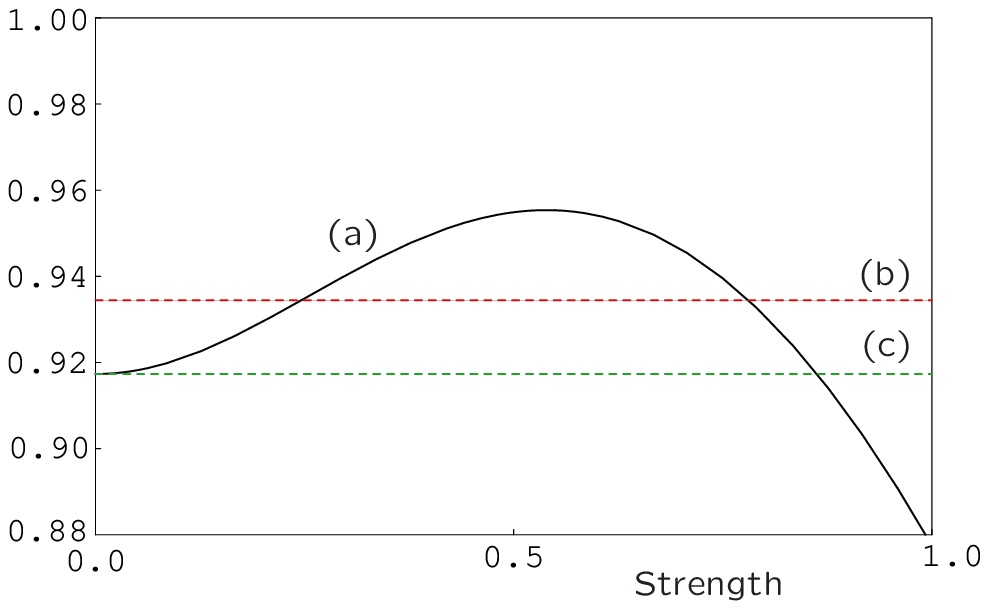}
  \end{center}
  \begin{picture}(0,0)
\put(125,90){\begin{sideways}$\aver{\F^{\rm QC}}$\end{sideways}}
\end{picture}

  \caption[Optimum measurement strength for quantum control.]{(Color online) (a) Fidelity of the quantum correction procedure with measurement
  strength ($1-2\chi/\pi$) for a representative noise value ($p=0.145$) and
  angle ($\theta=0.715$). The measurement strength ranges from a value of 0
  (corresponding to no measurement), through to a value of 1
  (corresponding to a projective measurement).
  There exists an optimum measurement strength at which we balance the amount
  of information gained with the amount of
  back-action noise introduced.  Also plotted for comparison
  are (b) the optimal ``discriminate-and-reprepare'' scheme and
  (c) the ``do nothing'' scheme for the same parameter values.}
  \label{fig:optimum_strength}
\end{figure}

Substituting $\chi_{\mathrm{opt}}$ for $\chi$ in Eq.
(\ref{eq:fid}), we get the following expression for the optimum
fidelity:
\begin{equation}\label{eq:fid_qc}
    \aver{\F^{\rm QCopt}} = \tfrac{1}{2} +\tfrac{1}{2}
    \sqrt{\cos^2\theta +
    \frac{\sin^4\theta}{1-r_x^2}}\,,
\end{equation}

Fig.~\ref{fig:fid_all}(c) plots the quantum control fidelity as a function of the input state (characterized by the angle $\theta$) and the amount of noise (characterized by $p$).

We note that $\aver{\F^{\rm QCopt}} = 1$ for three limiting cases. If $p=0$, there is no noise and so the state is not perturbed, resulting in unit fidelity for all values of $\theta$ given by simply ``doing nothing'' (zero measurement strength and no feedback). When $\theta=\pi/2$, the states are orthogonal and point along the $z$ axis.  The noise does not affect these states, again resulting in unit fidelity for all values of $p$ with a ``do nothing'' scheme. When $\theta=0$ the two states are equal and point along the $x$-axis. The control scheme reprepares this state after the noise by making a projective measurement $\chi=0$ to obtain either $|{{+}i}\rangle$ or $|{{-}i}\rangle$ and rotating back to the
$xz$-plane ($\eta=\pi/2$). This results in a fidelity of $1$ for all values of $p$.

\subsection{Comparison with Classical Schemes}

We now compare the quantum control scheme with classical schemes presented in Sec.~\ref{sec:Classical}.  Specifically, we first compare the quantum scheme with the best of the deterministic classical schemes at every point in the parameter space $(p,\theta)$, i.e., we observe the difference
in the average fidelities
\begin{equation}
  F_{\mathrm{dif}}=\aver{\F^{\rm QCopt}}
  -{\rm max}(\aver{\F^{\rm DDR2}},\aver{\F^{\rm DN}})\,,
\end{equation}
where $\aver{\F^{\rm DDR2}}$ and $\aver{\F^{\rm DN}}$ are given by Eqs.~(\ref{eq:fcl2}) and (\ref{eq:fn}), respectively.
Fig.~\ref{fig:fid_dif} reveals that $F_{\textrm{dif}}$ is always positive, and thus the quantum control scheme always outperforms the best of the classical strategies\footnote{Quite recently, a rigorous demonstration of $F_{\rm dif}\geq 0$ was given in Ref. \cite{08Xi1056}.}.

Comparing the quantum scheme with the stochastic classical scheme of Sec.~\ref{sec:newsection}, we note that they both yield precisely the same fidelity [cf. Eq.~\eqref{eq:fid_stochastic}]. In this case, the superiority of the quantum control scheme resides in producing this performance in a deterministic fashion.

\begin{figure}
\begin{center}
\includegraphics[width=0.45\textwidth]{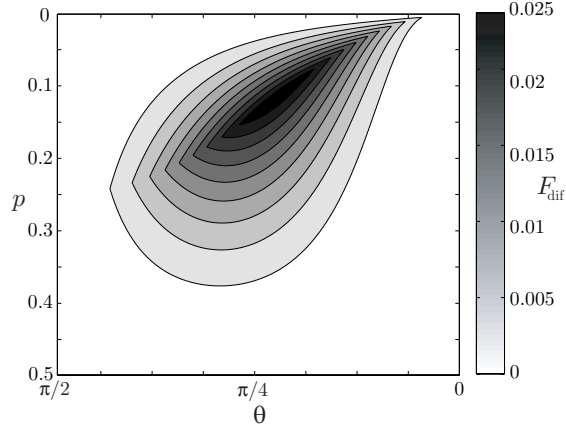}
\end{center}
\caption[Comparison between quantum and classical schemes.]{A contour plot of the difference, as a function of the
amount of noise $p$ and the angle between the initial states
$\theta$, between the average fidelities of the quantum control
scheme and the best classical scheme.  The quantum control scheme
performs significantly better for moderate values of $p$ and $\theta$
($0.05\lesssim p\lesssim 0.3$ and
$0.3\lesssim \theta \lesssim 1$).
The maximum value $ F_{\mathrm{dif}}=0.026$ occurs at $p=0.115$ and
$\theta=0.715$.}
 \label{fig:fid_dif}
\end{figure}

\subsection{Optimality}\label{sec:optimality}

We now prove that our quantum control scheme is optimal, in that it yields the maximum average fidelity of all possible quantum operations (CPTP maps).  Our proof makes use of techniques from convex optimization (specifically, those of~\cite{02Audenaert030302}) but is presented without requiring any background in this subject. In the Appendix~\ref{app:optdetq}, we provide a more detailed construction of the proof.

Consider the following optimization problem: determine the maximum average fidelity
\begin{equation}
  \aver{\F^{\rm opt}} = \max_{\mathcal{C}} \aver{\F^\mathcal{C}}
  = \max_{\mathcal{C}} \tfrac{1}{2}\sum_{i=1}^{2}{\langle\overline{\psi}_i|
  \mathcal{C}\bigl[\mathcal{E}_p(|\overline{\psi}_i\rangle\langle\overline{\psi}_i|)\bigr]
  |\overline{\psi}_i\rangle}\,,
  \label{eq:CPTPoptimisation}
\end{equation}
where the maximization is now over all CPTP maps $\mathcal{C}$
acting on a single qubit.

Recall from Sec.~\ref{sec:isomorphisms} that any CPTP map $\mathcal{C}$ acting on operators on a Hilbert space $\textsf{H}$ is in one-to-one correspondence with a (unnormalized) density operator $\mathfrak{C}$ on $\textsf{H}\otimes\textsf{H}$ via
\begin{equation}
  \mathcal{C}(\varrho)={\rm
  Tr_{1}}\left[(\varrho^{\sf T}\otimes \openone_2)\mathfrak{C}\right]\,,
\end{equation}
and is subject to the constraint $\tr_2\mathfrak{C}= \openone_2$, where the subindexes $1$ and $2$ used next to the trace operation denotes the partial trace over the first
and second subsystems, respectively~\cite{72Jamiolkowski275,01DAriano042308,00Nielsen}. With this isomorphism, the average fidelity $\aver{\F^\mathcal{C}}$ for the control scheme $\mathcal{C}$ is given by $\aver{\F^\mathcal{C}} = \tr\left(R \mathfrak{C}\right)$, where
\begin{equation}
  R\equiv\tfrac{1}{2}\sum_{i=1}^2
  {\left[\mathcal{E}_p\bigl(|\overline{\psi}_i\rangle\langle\overline{\psi}_i |\bigr)\right]^{\sf T}
  \otimes|\overline{\psi}_i\rangle\langle \overline{\psi}_i|}\,.
\end{equation}
Thus, the optimization problem (\ref{eq:CPTPoptimisation}) can be
rewritten as
\begin{equation}\label{eq:qprob_compact}
  \begin{array}{rl}
  \text{maximize}&\tr\left(R \mathfrak{C}\right)\\
  \text{subject to}&\mathfrak{C}\geq 0\\
  &\tr_2 \mathfrak{C} = \openone_2
  \,.
  \end{array}
\end{equation}
We now wish to prove that the maximum value of $\tr\left(R
\mathfrak{C}\right)$ subject to these constraints is given
by $\aver{\F^{\rm QCopt}}$ of Eq.~\eqref{eq:fid_qc}.

We note that, for any single-qubit operator $M$ satisfying $M\otimes \openone_2 - R \geq 0$, we obtain the inequality
\begin{align}
    \tr M - \tr\left(R \mathfrak{C}\right)
    &= \tr \left[(M \otimes \openone_2)\mathfrak{C}\right] - \tr\left(R \mathfrak{C}\right) \nonumber \\
    &= \tr\left[(M \otimes \openone_2 - R)\mathfrak{C}\right] \nonumber \\
    &\geq 0 \,,
\end{align}
where the first line follows from the constraint $\tr_2
\mathfrak{C} = \openone_2$, and the
inequality follows from the fact that $(M \otimes \openone_2 - R)\geq 0$ and $\mathfrak{C} \geq 0$, and thus the trace of their product is non-negative.  This inequality demonstrates that the value $\tr M$ for any matrix $M$ that satisfies the constraint $(M \otimes \openone_2 - R)\geq 0$ provides an upper bound on the solution of our optimization problem~\eqref{eq:qprob_compact}.

Consider the matrix $M = b_0 (\openone_2 + r_x X)$, where
\begin{equation}
  b_0=\frac{1}{4}
  +\frac{1}{4}\sqrt{\cos^2{\theta}+\frac{\sin^4{\theta}}{1-r_x^2}}\,,
\end{equation}
and $r_x = (1-2p)\cos\theta$ as before. It is straightforward to verify that the matrix $b_0 \openone_4 + r_x b_0 X\otimes \openone_2 - R \geq 0$, and hence the value $\tr M = 2 b_0$ provides an upper bound on the average fidelity of any control scheme. Because $2 b_0$ precisely equals the fidelity of our proposed quantum control scheme, given by Eq.~(\ref{eq:fid_qc}), this scheme necessarily gives an optimal solution to the original problem (\ref{eq:qprob_compact}). We refer the reader to the Appendix~\ref{app:optdetq} for a more constructive proof of this result.

\section{Discussion and Conclusions}
\label{sec:conclusions}

We have shown how two key characteristics of quantum physics --- that non-orthogonal states cannot be perfectly discriminated, and that any information gain via measurement necessarily implies disturbance to the system --- imply that classical strategies for control must be modified or abandoned when dealing with quantum systems.  By making use of more general measurements available in quantum mechanics, we have been able to design a deterministic quantum control strategy that outperform deterministic schemes based on classical concepts.  Quite interestingly, our quantum scheme was also shown to achieve the same performance of an (arguably) optimal stochastic classical scheme, demonstrating that for the present problem the gap between classical and quantum deterministic control is just as large as that between deterministic and stochastic control. Whether this is a general feature of more general control problems, is an interesting avenue of future research.

In constructing our quantum control scheme for the particular task
presented here, we made use of several intuitive guides.  First, we
used a preferred (and non-standard) ensemble of the dephasing noise
operator [cf. Eq.~\eqref{eq:PreferredEnsemble}], which allowed us to
view the noise as ``kicking'' the state of the qubit in one
direction or the other on the Bloch sphere.  We then made use of a
weak measurement in a basis that, loosely, attempted to acquire
information about the direction of this kick without acquiring
information about the choice of preparation of the
system.  It is remarkable (and perhaps simply lucky) that these
intuitive guides lead to a quantum control scheme that was optimal
for the task.  It is interesting to consider whether such intuition
can be applied to quantum control schemes in general, and if this
intuition can be formalized into rules for developing optimal
control schemes.

While our scheme is indeed optimal for the task presented, it is not
guaranteed to be unique; in fact, there are other decompositions of
the same CPTP map into different measurements and feedback
procedures~\cite{BlumeKohoutCombes}. In general, it is possible that
an entire class of CPTP maps may yield the optimal performance.
Also, the intuitive guides discussed above for our quantum control
scheme --- such as that the measurement essentially gains
information only about the noise and not the choice of initial state
--- may not apply to other optimal schemes.

In connection to this, we note that a similar feedback control
scheme was investigated by Niu and Griffiths~\cite{99Niu2764} for
optimal eavesdropping in a B92 quantum cryptography
protocol~\cite{92Bennett3121}, see also~\cite{96Fuchs2038}.  In their scheme, the
aim of the weak measurement was to \emph{maximize} the information
gain about which of two non-orthogonal states was transmitted for a
given amount of disturbance; in contrast, our weak measurement was
designed to acquire \emph{no} information about the choice of
non-orthogonal states. Despite these opposing aims, the obvious
similarity between these our scheme and that of Niu and Griffith
warrants further investigation, particularly since we note that
optimal feedback protocols exist based on different choices of
measurement.

It is also of interest to determine if non-trivial control schemes
exist for other types of noise processes, or if these results can be
generalized to larger numbers of initial states and to
higher-dimensional systems.

Finally, we note that the key element to our quantum control scheme
--- weak QND measurements on a qubit, and feedback onto a qubit based
on measurement results --- have both been demonstrated in recent
single-photon quantum optics experiments.  Specifically,
Pryde~\etal~\cite{pryde:190402} have demonstrated weak QND
measurements of a single photonic qubit, and have explicitly varied
the measurement strength over the full parameter range.  Also,
Pittman~\etal~\cite{pittman:052332} have demonstrated
feedback on the polarization of a single photon based on the
measurement of the polarization of another photon entangled with the
first; this feedback was used for the purposes of quantum error
correction, and is essentially identical to the feedback required
for our quantum control scheme.  Because these core essential
elements have already been demonstrated experimentally, we expect
that a demonstration of our quantum control scheme is possible in
the near future. 

\chapter{Optimal tracking for pairs of qubit states}\label{chap:tracking}

\section[Introduction]{Introduction}\label{sec:intro}

A common goal of many problems in quantum information science is the search for quantum operations that simultaneously transform a set of given input quantum states into another pre-specified set. Well known examples are tasks such as quantum cloning, state discrimination and quantum error correction.\par

In general, though, quantum mechanics forbids arbitrary quantum state dynamics. As a result, one is left with several examples of ``impossible quantum machines'' \cite{01Werner}. Not only is quantum cloning unachievable \cite{82Dieks271,82Wootters802,96Barnum2818}, but also quantum state discrimination strategies are typically subject to non-zero misidentification probabilities \cite{76Helstrom} and/or inconclusive outcomes \cite{88Dieks303,87Ivanovic257,88Peres19} and there are no quantum error correction protocols capable of fully reverting the action of an arbitrary noise model \cite{03Gregoratti915}.\par

Nevertheless, it is still possible to approximate \emph{ideal} (but unphysical) transformations with \emph{optimal} (but physical) ones. This provides quantum limits to the performance of tasks such as state discrimination, cloning and so on. In this chapter, we study the general problem of transforming the state of a single qubit into a given target state, when the system can be prepared in two different ways, and the target state depends on the choice of preparation. We call this task \emph{quantum tracking}, a term borrowed from classical control theory. Our main result is an analytical description of an optimal quantum tracking strategy.\par

More specifically, the quantum tracking problem studied here can be understood as follows. Consider that Alice prepares either a qubit state $\rho_1$ with probability $\pi_1$ or $\rho_2$ with probability $\pi_2$. Bob is allowed to interact with the system in any physically allowed way, aiming to enforce the tracking rule
\begin{equation}\label{rule}
\mbox{if Alice prepared } \rho_i\,, \mbox{ then ouput } \overline{\rho}_i
\end{equation}
for $i=1,2$ and some given qubit density matrices $\overline{\rho}_i$.\par

At his disposal, Bob has all the information about the possible preparations $\rho_i$ and their respective prior probabilities $\pi_i$, but \emph{not} the actual preparation (the value of the index $i$).

Because quantum states are generally not perfectly distinguishable, a strategy that attempts to identify Alice's preparation and then reprepare the target according to rule (\ref{rule}) is not always guaranteed to succeed. In fact,
 this limited distinguishability is an unsurpassable obstacle in the implementation of (\ref{rule}).

Throughout, an optimal solution is defined as follows.
Amongst all the physical transformations acting on the input states $\rho_i$, an optimal one is any map that outputs density matrices $\rho_i^\prime$ such that the averaged Hilbert-Schmidt inner product between $\rho_i^\prime$ and $\overline{\rho}_i$ is maximal. When $\overline{\rho}_i$ are pure states, such a figure-of-merit coincides with the averaged Uhlmann-Jozsa fidelity \cite{94Jozsa2315,76Uhlmann273}, and this notion of optimality gains an appealing operational interpretation \cite{95Schumacher2738}. Suppose that Alice [aware of her preparation and of rule (\ref{rule})], decides to check whether Bob prepared the density matrix he was supposed to, and for that purpose she performs a verification measurement on the density matrix produced by him. If Bob chooses an optimal transformation according to the above prescription, then the probability he will pass Alice's test is as large as allowed by quantum mechanics.

The tracking problem resembles the transformability problem for pairs of qubit states studied by Alberti and Uhlmann in the 80's \cite{80Alberti163} (see also Appendix \ref{app:ptc}). In \cite{80Alberti163}, a criterion based on the distinguishability between the source density matrices and the distinguishability between the target density matrices was developed in order to decide on the existence of a completely positive and trace preserving (CPTP) map simultaneously transforming each source into each target.\par

Although Alberti and Uhlmann's criterion classifies the set of states $\rho_1$, $\rho_2$, $\overline{\rho}_1$ and $\overline{\rho}_2$ for which rule (\ref{rule}) can be satisfied, it does not provide a construction of the CPTP map implementing that transformation, nor touch the problem of how to find a feasible approximation when the criterion is not satisfied. For many purposes, the requirement of perfectly converting sources into targets is unnecessarily strong, as some strictly impossible physical transformations can be very well approximated by physical ones, as illustrated in Fig \ref{fig:approx}. In fact, any experimental realization of a map is just an approximation of it.

 \begin{figure}
 \centering
\includegraphics[width=8.5cm]{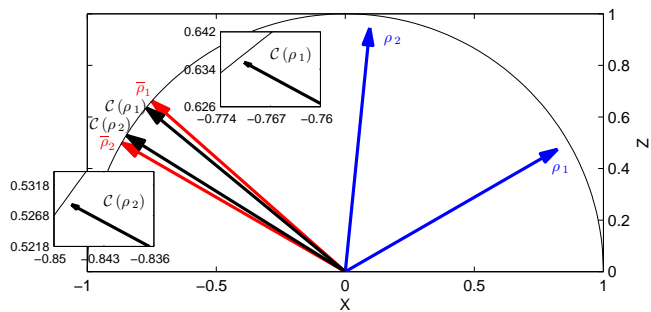}
\caption[Optimal approximation of a pair of pure qubit states]{(Color online) The transformation of the mixed states $\rho_1$ and $\rho_2$ (in blue) into the
pure states $\overline{\rho}_1$ and $\overline{\rho}_2$ (in red) is not a physical one. However, there is a physical transformation $\mathcal{C}$ capable of transforming the input states into states $\mathcal{C}(\rho_i)$ (in black) which closely approximates the targets. Note, in the detail, that $\mathcal{C}(\rho_1)$ and $\mathcal{C}(\rho_2)$ are still slightly mixed.}\label{fig:approx}
\end{figure}\par

Another problem closely related to our aims was investigated in Ref. \cite{07Branczyk012329}. Specifically, we considered the problem of determining the optimal quantum operation to stabilize the state of a single qubit, randomly prepared in one of two pure states, against the effect of dephasing noise. The results of \cite{07Branczyk012329} are here extended in several ways. The input states are allowed to be mixed and prepared with arbitrary prior probability distribution; the noise model is arbitrary and, most importantly, the stabilization task is replaced with tracking.\par

Finally, there is an intrinsic connection between the quantum tracking problem and the ``optimization approach'' \cite{05Reimpell080501,07Fletcher012338, 06Reimpell, 06Kosut,08Kosut020502,07Yamamoto012327,05Yamamoto022322} to quantum error correction \footnote{This is in contrast with the traditional approach to quantum error correction, which followed the direction of adapting classical coding techniques to the quantum domain \cite{97Knill900}.}. In these references, the encoding and recovery operations are regarded as optimization variables whose optimal values maximize a given figure-of-merit (typically a function bounded between $0$ and $1$, equal to $1$ if and only if the noise dynamics is reversible). Efficient numerical methods are then proposed to solve the optimization problem. The key differences between our work and these references is that we do not consider encoding of the initial state and focus on reverting the noise dynamics experienced by only a pair of states. By doing so, the optimization of the recovery operation can be handled analytically for a conveniently chosen figure-of-merit.\par

The chapter is structured as follows. Section \ref{sec:probandstrategy} introduces the formal statement of the problem and our working strategy, which is proved to be optimal in Section \ref{sec:optimal}. In section \ref{sec:examples} we evaluate the performance of the optimal strategy in the contexts of quantum state discrimination, quantum state stabilization in the presence noise, perfect quantum tracking, and state-dependent quantum cloning. Section \ref{sec:quantctrl} proposes a physical implementation of our strategy in terms of closed and open loop control. Section \ref{sec:discussion} discusses generalizations of the problem and concludes.

\section{Problem and Strategy}\label{sec:probandstrategy}
In this section we give a formal statement of the problem of interest and introduce our strategy.
\subsection{The Problem}\label{sec:problem}
Formally, the problem we set out to solve can be stated as follows:
\begin{problem}\label{problem}
Given qubit density matrices $\rho_1$, $\rho_2$, $\overline{\rho}_1$, $\overline{\rho}_2$ (with $\rho_1\neq\rho_2$) and probabilities $\pi_1$, $\pi_2$ with $\pi_1+\pi_2=1$, find a quantum operation $\mathcal{C}$ maximizing\footnote{As done in the previous chapter, throughout we drop the index $1$ from $\aver{\F_{\rm HS}}_1$ for ease of notation. However, we preserve the index {\rm HS} because here we do not restrict to pure target states, hence the distinction between $\F_{\rm HS}$ and $\F$ is relevant.}
\begin{equation}\label{eq:J}
\aver{\F_{\rm HS}}=\pi_1\tr\left[\mathcal{C}(\rho_1)\overline{\rho}_1\right]+\pi_2\tr\left[\mathcal{C}(\rho_2)\overline{\rho}_2\right]\,.
\end{equation}
\end{problem}
\noindent We will refer to this as ``the tracking problem''.\par

The choice of the average Hilbert-Schmidt inner product as our figure-of-merit $\aver{\F_{\rm HS}}$ is motivated by technical reasons (to be clarified later), and by the fact that for pure target states (the case of greater interest as far as applications are concerned), $\aver{\F_{\rm HS}}$ is precisely equal to the average fidelity. When the target states are mixed, $\aver{\F_{\rm HS}}$ is a lower bound to the average fidelity \cite{94Jozsa2315}. Although not as well motivated as in the case of pure target states, the determination of the quantum operation $\mathcal{C}$ maximizing $\aver{\F_{\rm HS}}$ can still be useful for mixed target states. For example, if a certain application requires tracking to be performed with average fidelity $f$ and the optimal value of our figure-of-merit is such that $\aver{\F_{\rm HS}}\geq f$, then $\mathcal{C}$ is suitable for the task.\par

As a final remark, note that we do not exclude the case $\overline{\rho}_1=\overline{\rho}_2$ from the statement of the problem. However, we will exclude the case $\overline{\rho}_1=\overline{\rho}_2 = \openone/2$ from the following analysis. Obviously, this particular transformation is always feasible and achieved with the completely depolarizing channel.\par

Next, we propose a strategy that will later be proved to be a solution of this tracking problem.

\subsection{The Strategy}\label{sec:strategy}

In this section, we provide an analytical solution of the tracking problem, i.e., we detail the structure of an optimal tracking operation $\mathcal{C}$ and derive closed forms for the associated maximal value of the figure-of-merit $\aver{\F_{\rm HS}}$. The scheme proposed here was constructed by incorporating some features observed from the numerical solution of the tracking problem in an analytical optimization procedure. In the next section, we will show that the tracking problem can be cast as a semidefinite program (SDP) \cite{96Vandenberghe49,04Boyd}, and will employ the theory for this type of optimization problem to prove that the strategy presented here actually solves the tracking problem.

A quantum operation is a description of a certain physically allowed evolution of a quantum state. For a closed quantum system (not interacting with an environment) this description is given by the familiar unitary evolution of Schr\"odinger's equation. For open quantum systems, unitary evolution alone does not account for every possible state transformation --- in this case, the set of quantum operations is identified with the more comprehensive set of completely positive and trace preserving (CPTP) maps.

Any one qubit CPTP map $\mathcal{C}$ can be decomposed as \cite{01King192,02Ruskai159}
\begin{equation}\label{eq:ruskai_decomposition} \mathcal{C}(\varsigma)=U
\mathcal{D} (V \varsigma V^\dagger) U^\dagger\,,
\end{equation} where $U$, $V$ are
unitary matrices and $\mathcal{D}$ induces an affine transformation on the
input Bloch vectors; namely, it contracts the $x$, $y$ and $z$ components via a
multiplicative factor, and subsequently adds a fixed number to them. Any
non-unitary evolution arises from a transformation of this type. In the
framework of Eq. (\ref{eq:ruskai_decomposition}), unitary dynamics is simply
obtained by making the affine map $\mathcal{D}$ redundant (e.g., multiplying by
1's and adding 0's to the $x-$, $y-$ and $z-$Bloch components).

In general, CPTP maps reduce the distinguishability of quantum states. On the
Bloch sphere this typically corresponds to a reduction of the Bloch vector
length and angles between vectors. In contrast unitary dynamics preserves
the angles between Bloch vectors and their lengths.
For the tracking problem, we can imagine that in
some cases the optimal strategy will preserve lengths and angles, i.e. it will
be some unitary correction. We will construct an ``indicator function'' which
will flag this case.

\subsubsection{Indicator function}\label{sec:indicator}
To gain some intuition, we start by constructing an indicator function for the simplest case of tracking with uniform priorities $\pi_1=\pi_2=1/2$ from pure states ($\rho_1$ and $\rho_2$) to pure states ($\overline{\rho}_1$ and $\overline{\rho}_2$). Throughout, $\Theta$ and $\overline{\Theta}$ will denote the angles between the Bloch vectors of $\rho_1$, $\rho_2$ and $\overline{\rho}_1$, $\overline{\rho}_2$, respectively. It will also be convenient to define $\theta$ and $\overline{\theta}$ to be the half-angle between the Bloch vectors, i.e., $2\theta=\Theta$ and $2\overline{\theta}=\overline{\Theta}$.\par

Given that all the states involved are pure, it is straightforward to conclude that if $\Theta=\overline{\Theta}$, then unitary dynamics is the best choice --- a suitable rotation of the Bloch vectors of the inputs can perfectly bring them to coincide with the Bloch vectors of the targets (as opposed to a non-unitary evolution that would decrease the angle, hence excluding the possibility of perfect tracking).\par

 A corollary of a theorem by Alberti and Uhlmann \cite{80Alberti163} (see appendix \ref{app:ptc}) implies that any pure state transformation such that $\Theta>\overline{\Theta}$, can be perfectly implemented. If that is the case, then this transformation must be non-unitary, since a unitary would not be able to bring the angles to perfectly match. This suggests the introduction of the function
\begin{equation}\label{eq:IF_simple}
\widetilde{\Omega}\mathrel{\mathop:}=\Theta-\overline{\Theta}\,,
\end{equation}
to indicate non-unitary dynamics whenever $\widetilde{\Omega}>0$. Next, we argue that $\widetilde{\Omega}\leq 0$ indicates unitary dynamics, thus establishing $\widetilde{\Omega}$ as an example of indicator function we were looking for.\par

 We have already seen that $\widetilde{\Omega}=0$ implies unitary dynamics. Intuitively, this conclusion can be extended to $\widetilde{\Omega}<0$ with the following reasoning. If $\Theta<\overline{\Theta}$, any further decrease of the initial angle can only further separate the resulting states from the targets. Since there is not any quantum operation capable of increasing this angle, the best policy must be to preserve it, hence a unitary.

The above discussion may suggest that the optimal indicator function is merely a comparison of the
distinguishabilities between sources and targets. If the sources are more distinguishable than the targets, then we employ a quantum measurement to decrease the distinguishability, hence approximating the targets. If the sources are no more distinguishable than the targets, then we employ a unitary operation to avoid a further decrease of the overlap between the output states and the targets.
Although this reasoning is certainly in agreement with the indicator function introduced above for the special case of
pure states, it does not extend to mixed state transformations \footnote{At least not as far as the optimization of the figure-of-merit of Eq. (\ref{eq:J}) is concerned. A possibly interesting problem would be the determination of a figure-of-merit that would preserve such behavior for general state transformations.}.

If the states are not pure and the priorities are not uniform, it is much more difficult to understand how purities, angles and priorities combine to form a meaningful decision criterion about the nature of the best dynamics. In order to introduce an indicator function for this general case (obtained from some mathematical optimization procedure, not from an heuristic argument), we first define some useful notation.\par

Let $\bm{R}_i$ be the Bloch vector of $\rho_i$ and $\overline{\bm{R}}_i$ the ``Bloch vector'' of $\pi_i\overline{\rho}_i$ [or, more precisely, the Bloch vector of the normalized density matrix  $(1 -\pi_i)\openone/2+\pi_i\brho_i$]\label{conventionRb}. Symbolically, for $\bm{\mathcal{R}}\in\{\bm{R},\overline{\bm{R}}\}$, define
\begin{subequations}
\begin{align}
\bm{\mathcal{R}}_+&\mathrel{\mathop:}=\bm{\mathcal{R}}_1+\bm{\mathcal{R}}_2\,,\\
\bm{\mathcal{R}}_-&\mathrel{\mathop:}=\bm{\mathcal{R}}_1-\bm{\mathcal{R}}_2\,,\\
\bm{\mathcal{R}}_\times&\mathrel{\mathop:}=\bm{\mathcal{R}}_1\times\bm{\mathcal{R}}_2\,,\label{eq:Rx}
\end{align}
\end{subequations}
and as usual, the corresponding unbolded type gives the Euclidean norm $\mathcal{R}_{\{+,-,\times\}}=\|\bm{\mathcal{R}}_{\{+,-,\times\}}\|$. Also, the following will be important throughout
\begin{align}
T&\mathrel{\mathop:}= \sum_{i,j=1}^{2}{(1-\bm{R}_i\cdot \bm{R}_j)(\overline{\bm{R}}_i\cdot\overline{\bm{R}}_j)}\,,\label{eq:T}\\
S&\mathrel{\mathop:}= \sqrt{T^2+4\overline{R}_\times^2(R_-^2-R_\times^2)}\,.\label{eq:S}
\end{align}

We note that $\bm{R}_1\neq\bm{R}_2$ guarantees that $R_-^2-R_\times^2 > 0$ (see Appendix \ref{app:ST}, Lemma \ref{lemma1}), hence both $S$ and $T$ are real numbers.

In terms of these quantities, we define
 \begin{equation}\label{eq:IF}
 \Omega\mathrel{\mathop:}= S+T-2\overline{R}_\times R_\times\,,
 \end{equation}
with $\Omega>0$ indicating that non-unitary dynamics are required (which will be detailed as ``procedure A'') and $\Omega\leq 0$ indicating that unitary dynamics (``procedure B'') are required. \par

 Although it would be difficult to motivate the indicator function $\Omega$ of Eq. (\ref{eq:IF}) as we did with $\widetilde{\Omega}$ in Eq. (\ref{eq:IF_simple}), it is possible to see that the former is equivalent to the latter in the case of pure qubit states. This is shown in Fig. \ref{fig:omegavsdist}, where it is also noticeable that even a simple generalization of the input states from pure to mixed states with the same level of mixedness (as measured by the norm of their Bloch vector $R$), is already sufficient to give a fairly non-trivial division line between the two types of dynamics.

\begin{figure}
\begin{center}
\includegraphics[width=8cm]{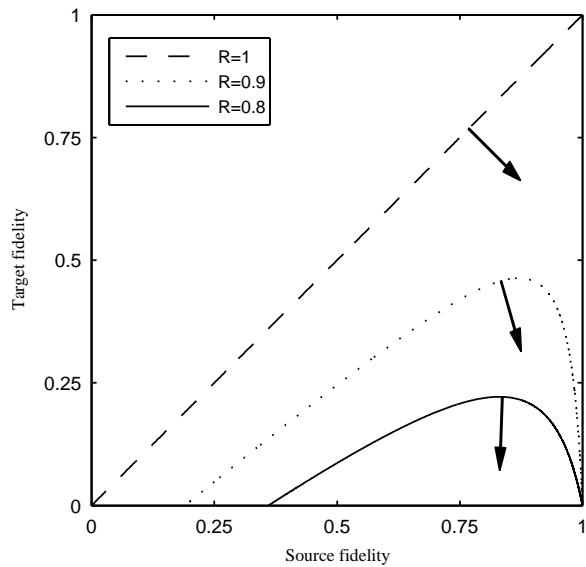}
\end{center}
\caption[Lines separating unitary and non-unitary optimal dynamics.]{Lines separating unitary and non-unitary dynamics, as prescribed by the indicator function of Eq. (\ref{eq:IF}).
For source states, we consider pairs of mixed states with Bloch vector length $R$ and separated by an angle $\Theta=2\theta$ (Bloch sphere angle).
For target states, we consider pairs of pure states separated by an angle $\overline{\Theta}=2\overline{\theta}$. The fidelity between the source states (horizontal axis) is $1-R^2\sin^2{\theta}$, and the fidelity between the target states (vertical axis) is $\cos^2{\overline{\theta}}$. The region where unitary dynamics is advisable ($\Omega\leq 0$) is indicated with an arrow. The intuitive notion that a measurement is employed when the targets are less distinguishable than the sources (and a unitary, otherwise), only holds if $R=1$ (pure sources). For $R=0.9, 0.8$ the division line moves down in such a way as to increase the portion of the parameter space where non-unitary dynamics is advisable $(\Omega > 0)$.}\label{fig:omegavsdist}
\end{figure}

\subsubsection{Procedure A}
In this section we present the details of the map $\mathcal{C}$ from Eq. (\ref{eq:ruskai_decomposition}) for $\Omega>0$ (which indicates non-unitary dynamics).

\paragraph*{Step 1.} The rotation by the unitary $V$ takes the two input Bloch vectors to vectors $\bm{R}_1^\prime$ and $\bm{R}_2^\prime$ in the $xz$-plane in such a way that they share a common positive $\bm{x}$-component and $\bm{R}_1^\prime\cdot\bm{z}>\bm{R}_2^\prime\cdot\bm{z}$, explicitly
\begin{equation}\label{eq:V}
\bm{R}_i\stackrel{V}{\longmapsto}\bm{R}_i^\prime=\frac{R_\times}{R_-}\bm{x}+\frac{\left(\bm{R}_i\cdot\bm{R}_-\right)}{R_-}\bm{z}\,.
\end{equation}

\paragraph*{Step 2.} The affine transformation $\mathcal{D}$ shortens the $\bm{x}$, $\bm{y}$ and $\bm{z}$ components of its inputs by multiplying them, respectively, by $\mu_1$, $\mu_2$, $\mu_3$ with $0\leq\mu_{\{1,2,3\}}\leq 1$ and subsequently adding $s_1$ to the $\bm{x}$ component. Applied to $\bm{R}_i^\prime$, that reads
\begin{equation}\label{eq:rill}
\bm{R}_i^\prime\stackrel{\mathcal{D}}{\longmapsto}\bm{R}_i^{\prime\prime}= \left(s_1+\mu_1\frac{R_\times}{R_-}\right)\bm{x}+\mu_3\frac{\left(\bm{R}_i\cdot\bm{R}_-\right)}{R_-}\bm{z}\,.
\end{equation}
That such a transformation can be physically implemented is not a trivial fact. Indeed, strict conditions involving the parameters $\mu_{\{1,2,3\}}$ and $s_1$ must be satisfied to guarantee the feasibility of transformation (\ref{eq:rill}) as a CPTP map \cite{02Ruskai159}. The following values can be shown to satisfy these conditions

\begin{subequations}\label{eq:nontrivial}
\begin{align}
\mu_1&=2\sqrt{\frac{2}{S(S+T)^{3}}}\overline{R}_\times^2 R_\times R_-\,,\\
\mu_2&=\left(\frac{2}{S+T}\right)\overline{R}_\times R_\times\,,\\
\mu_3&=\sqrt{\frac{2}{S(S+T)}}\overline{R}_\times R_-\,,\\
s_1&=\sqrt{\frac{1}{2S(S+T)^{3}}}\left[(S+T)^2-4\overline{R}_\times^2 R_\times^2\right]\,.
\end{align}
\end{subequations}

In Appendix \ref{app:ST} we show that the only circumstances under which the inequalities $S>0$ and $S+T>0$ are not simultaneously satisfied have $\Omega=0$. Therefore, the quantities above are real and well-defined for the present procedure ($\Omega>0$). It is not difficult to check that $\mu_1=\mu_2\mu_3$ and $s_1=\sqrt{(1-\mu_2^2)(1-\mu_3^2)}$, so the resulting map acting on density matrices is an extremal point of the convex set of CPTP maps \cite{02Ruskai159}.\par

Remarkably, if the target Bloch vectors $\overline{\bm{R}}_1$ and $\overline{\bm{R}}_2$ are parallel or anti-parallel (i.e., $\overline{R}_\times=0$), Eqs. (\ref{eq:nontrivial}) simplify to $\mu_1=\mu_2=\mu_3=0$ and $s_1=1$. This implies that $\bm{R}_i^{\prime\prime}=\bm{x}$ for $i=1,2$, or equivalently, that $\mathcal{D}$ outputs $\ket{+}=(\ket{0}+\ket{1})/\sqrt{2}$ independently of the input.

\paragraph*{Step 3.} For $\overline{R}_\times\neq 0$, the unitary $U$ rotates the input vectors $\bm{R}_i^{\prime\prime}$ to lie on the plane determined by the target vectors, and within that plane by a suitable angle. For $\overline{R}_\times = 0$, $U$ simply rotates the vector $\bm{x}$ in order to align it with $\overline{\bm{R}}_+$. In either case, $U$ can be expressed as the following map

\begin{equation}\label{eq:U}
\bm{R}_i^{\prime\prime}\stackrel{U}{\longmapsto}\bm{R}_i^{\prime\prime\prime}=k_{i1}\overline{\bm{R}}_1+k_{i2}\overline{\bm{R}}_2\,,
\end{equation}
with
\begin{align}
k_{ij}&=\tfrac{1}{\Gamma}\left[\alpha^2+\beta_i\beta_j\overline{R}_\times^2+(-1)^{i+j}\alpha\left(\beta_1-\beta_2\right)\overline{\bm{R}}_{\widetilde{i}}\cdot\overline{\bm{R}}_{\widetilde{j}}\right]\,,\label{eq:kij}\\
\Gamma&=\sqrt{\alpha^2\overline{R}_+^2+\left[\|\beta_1\overline{\bm{R}}_1+\beta_2\overline{\bm{R}}_2\|^2+2\alpha\left(\beta_1-\beta_2\right)\right]\overline{R}_\times^2}\,.\label{eq:Gamma}
\end{align}
where we have defined $\widetilde{1}=2$ and $\widetilde{2}=1$; and

\begin{align}
\alpha&=\bm{R}_i^{\prime\prime}\cdot \bm{x}=\sqrt{\frac{S+T}{2S}}\,,\label{eq:alpha_nt}\\
\beta_i&=\frac{\bm{R}_i^{\prime\prime}\cdot \bm{z}}{\overline{R}_\times}= \sqrt{\frac{2}{S(S+T)}} \bm{R}_i\cdot\bm{R}_-\,.\label{eq:beta_nt}
\end{align}

With this, the figure-of-merit of Eq. (\ref{eq:J}) can be shown to be
\begin{equation}
\aver{\F_{\rm HS}^{\rm A}}=\frac{1}{2}+\frac{\Gamma_a}{2}\,,\label{eq:fid_nt}
\end{equation}
where $\Gamma_a$ is obtained by substituting Eqs. (\ref{eq:alpha_nt}) and (\ref{eq:beta_nt}) into Eq. (\ref{eq:Gamma}) and reads
\begin{equation}\label{eq:gamma_a}
\Gamma_a=\sqrt{\overline{R}_+^2+\frac{2R_-^2\overline{R}_\times^2}{S+T}}\,.
\end{equation}

Fig. \ref{fig:solution} illustrates this sequence of transformations for the case $\overline{R}_\times\neq 0$.

\begin{figure}
\begin{center}
\includegraphics[width=\textwidth]{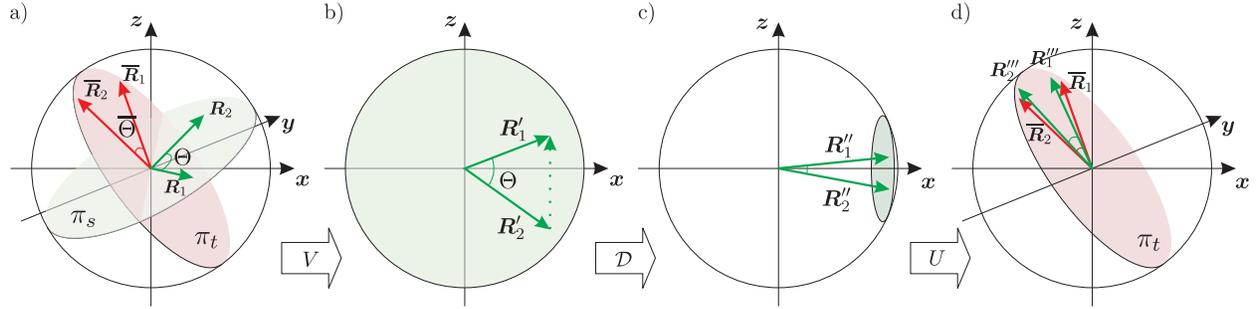}
\end{center}
\caption[Bloch sphere schematics of the non-unitary scheme.]{(Color online) Bloch sphere schematics of procedure A.  a) The source (target) Bloch vectors $\bm{R}_1$ and $\bm{R}_2$ ($\overline{\bm{R}}_1$ and $\overline{\bm{R}}_2$) determine the plane $\pi_s$ ($\pi_t$). b) $V$ implements the rotation transforming $\pi_s$ to the $xz$-plane, in such a way that the vector $\bm{R}_1^\prime-\bm{R}_2^\prime$ is \emph{parallel} to $+\bm{z}$. c) The map $\mathcal{D}$ deforms the Bloch sphere into an ellipsoid of semi-axis $\mu_1$, $\mu_2$ and $\mu_3$ and translates it by $s_1$ along the $\bm{x}$ axis. The resulting ellipsoid touches the original Bloch sphere --- a feature related to the fact that $\mathcal{D}$ is an extremal CPTP map. d) $U$ rotates the resulting states to the plane $\pi_t$, and within that plane by some angle such that the resulting states $\bm{R}_1^{\prime\prime\prime}$ and $\bm{R}_2^{\prime\prime\prime}$ approximate $\overline{\bm{R}}_1$ and $\overline{\bm{R}}_2$, respectively.}\label{fig:solution}
\end{figure}

\subsubsection{Procedure B}\label{sec:procB}

As pointed out before, if $\Omega \leq 0$ then the affine transformation $\mathcal{D}$ is not implemented and the product $VU$ gives the unitary dynamics. In this case, $V$ can be chosen precisely as in step $1$ of procedure A.\par

  For $\overline{R}_\times\neq 0$, we preserve the form of the transformation $U$ from Eq. (\ref{eq:U}),
  but the values of $\alpha$ and $\beta_i$ are given by
\begin{align}
\alpha&=\bm{R}_i^{\prime}\cdot \bm{x}=\frac{R_\times}{R_-}\,,\label{eq:alpha_t}\\
\beta_i&=\frac{\bm{R}_i^{\prime}\cdot \bm{z}}{\overline{R}_\times}=\frac{\bm{R}_i\cdot\bm{R}_-}{\overline{R}_\times R_-}\,.\label{eq:beta_t}
\end{align}

For $\overline{R}_\times=0$, assume that the Bloch vectors $\overline{\bm{R}}_1$ and $\overline{\bm{R}}_2$ are anti-parallel (this is without loss of generality, since parallel targets always exhibit $\Omega > 0$ \footnote{To see this, note that if the targets are parallel it is immediate that $\Omega=|T|+T$ and our claims holds if $T>0$. A straightforward computation shows that $T=(\overline{R}_1+\overline{R}_2)^2 - (R_1^2\overline{R}_1^2 +R_2^2\overline{R}_2^2 + 2 R_1 R_2 \overline{R}_1\overline{R}_2 \cos{\Theta})$, which clearly achieves the minimum value $0$ iff $R_1=R_2=\cos{\Theta}=1$. However, this requires the sources to be identical, which is excluded from the statement of the problem in Sec.~\ref{sec:problem}. Therefore, $T>0$.}). In particular, take $\overline{\bm{R}}_1$ parallel to $\bm{z}$ and $\overline{\bm{R}}_2$ parallel to $-\bm{z}$. The following transformation specifies $U$ in this case

\begin{equation}\label{eq:ULD}
\bm{R}_i^{\prime}\stackrel{U}{\longmapsto}\bm{R}_i^{\prime\prime}=\left[\frac{R_\times}{R_-}\cos{\vartheta}-\frac{(\bm{R}_i\cdot\bm{R}_-)}{R_-}\sin{\vartheta}\right]\bm{x}
+\left[\frac{R_\times}{R_-}\sin{\vartheta}+\frac{(\bm{R}_i\cdot\bm{R}_-)}{R_-}\cos{\vartheta}\right]\bm{z}\,,
\end{equation}
where
\begin{equation}\label{eq:sintheta}
\sin{\vartheta}=\frac{R_\times(\overline{R}_1-\overline{R}_2)}{R_-\sqrt{\overline{R}_+^2-T}}\,,
\end{equation}
and $\cos{\vartheta}=+\sqrt{1-\sin^2{\vartheta}}$. We note that $\vartheta$ is a valid angle since the rhs of Eq. (\ref{eq:sintheta}) is bounded between $-1$ and $1$ \footnote{This follows easily from the inequality $R_\times<R_-$ (Appendix \ref{app:ST}) and from $\overline{R}_+=|\overline{R}_1-\overline{R}_2|\leq \sqrt{\overline{R}_+^2-T}$ (the equality follows from the anti-parallelism of the target Bloch vectors and the inequality from the fact that $T\leq 0$ for $\overline{R}_\times = 0$ and $\Omega\leq 0$)}.

If $\overline{\bm{R}}_1$ and $\overline{\bm{R}}_2$ do not align along the $\bm{z}$ direction as specified above, we simply apply a further rotation that aligns the $\bm{z}$ axis with the direction $\overline{\bm{R}}_1/\overline{R}_1$.

 In both $\overline{R}_\times\neq 0$ and $\overline{R}_\times = 0$ cases, the average fidelity can be computed to be
\begin{equation}\label{eq:fid_t}
\aver{\F_{\rm HS}^{\rm B}}=\frac{1}{2}+\frac{\Gamma_b}{2}\,,
\end{equation}
where $\Gamma_b$ is obtained by substituting Eqs. (\ref{eq:alpha_t}) and (\ref{eq:beta_t}) into Eq. (\ref{eq:Gamma}). After some manipulation we find
\begin{equation}\label{eq:gamma_b}
\Gamma_b\mathrel{\mathop:}=\sqrt{\overline{R}_+^2-T+2R_\times\overline{R}_\times}\,.
\end{equation}

Finally, let us note that a more compact and mathematical description of procedures A and B is provided in the Appendix~\ref{app:optSSmap}.

\section{Optimality Proof}\label{sec:optimal}

In this section, we employ duality theory for SDPs to
prove the following theorem:
\begin{theorem} \label{teo:optimality}
Our tracking strategy (as described in Sec. \ref{sec:strategy}), implements optimal tracking between any pair of source and target qubit states and is, therefore, a solution of the tracking problem introduced in Sec. \ref{sec:problem}.
\end{theorem}

In the subsequent proof of this theorem, some familiarity with SDP theory is assumed. Standard reviews on the topic are \cite{96Vandenberghe49,04Boyd}. More closely related to our purposes is \cite{02Audenaert030302}, where the connection between optimization of quantum operations and SDPs was first noted. Also relevant is Ref. \cite{07Branczyk012329}, where a similar technique was used to approach a particular case of tracking.

\subsection{The tracking problem as a SDP}\label{sec:trkasSDP}

We start by showing that the tracking problem can be formulated as a SDP. Formally, it can be written as
\begin{equation}\label{eq:opt1_app}
\max_{\mathcal{C}\in\mathcal{Q}_2^{\rm set}}{\sum_{i=1}^2{\tr\left[\mathcal{C}(\rho_i)\pi_i\overline{\rho}_i\right]}}\,.
\end{equation}
It will be convenient to rewrite $\mathcal{C}(\rho_i)$ as \cite{01DAriano042308}
\begin{equation}\label{eq:ck}
\mathcal{C}(\rho_i)=\tr_1\left[(\rho_i^{\sf T}\otimes\openone)\mathfrak{C}\right]\,,
\end{equation}
where $\mathfrak{C}$ is the (unnormalized) Choi matrix \cite{75Choi285}
\begin{equation}\label{eq:kc}
\mathfrak{C}=(\mathcal{I}\otimes\mathcal{C})(\ket{\Psi^+}\!\bra{\Psi^+})\,,
\end{equation}
and $\ket{\Psi^+}=\ket{00}+\ket{11}$. Eqs. (\ref{eq:ck}) and (\ref{eq:kc}) establish a one-to-one relation between the set of CPTP maps on qubits and the set of (unnormalized) 2-qubit density matrices satisfying ${\rm Tr_2} \mathfrak{C}=\openone_2$ \cite{99Horodecki1888,99Fujiwara3290,01DAriano042308,75Choi285}. Here, $\tr_{1(2)}$ denotes the partial trace operation over the first (second) qubit.\par

Using this isomorphism, a straightforward manipulation gives for the objective function in (\ref{eq:opt1_app}) the form $-\tr\left(F_0 \mathfrak{C}\right)$, where
\begin{equation}
F_0=-\sum_{i=1}^2{\rho_i^{\sf T}\otimes\pi_i\overline{\rho}_i}\,.
\end{equation}
whereas the constraint $\mathcal{C} \in \mathcal{Q}_2^{\rm set}$ becomes $\mathfrak{C}\geq 0$ and $\tr_2\mathfrak{C}=\openone_2$. In conclusion, the tracking problem assumes the standard form of a SDP
\begin{equation}\label{eq:cvxopt}
  \begin{array}{rl}
  \text{maximize}&-\tr\left(F_0 \mathfrak{C}\right)\\
  \text{subject to}&\mathfrak{C}\geq 0\\
  &\tr\left[(\sigma_k\otimes\openone)\mathfrak{C}\right] = 2\delta_{k 0}\,,\quad k=0,1,2,3.\\
  \end{array}
\end{equation}\par
 A special feature of (\ref{eq:cvxopt}) which will be explored next is that the replacement $\rho_i\mapsto A\rho_i A^\dagger$ and $\overline{\rho}_i \mapsto B^\dagger \overline{\rho}_i B$, with $A,B \in U(2)$, yields another SDP (with $\widetilde{F}_0$ replacing $F_0$) that achieves exactly the same optimal value. This can be easily seen by noting that: (i) $\tr\left(F_0 \mathfrak{C}\right)=\tr\left(\widetilde{F}_0 \widetilde{\mathfrak{C}}\right)$, for $\widetilde{\mathfrak{C}}=(A^{\sf T}\otimes B)^\dagger \mathfrak{C} (A^{\sf T}\otimes B)$ and (ii) $\widetilde{\mathfrak{C}}$ satisfies the constraints in (\ref{eq:cvxopt}) if and only if $\mathfrak{C}$ does.\par

\subsection{The duality trick}\label{sec:dualitytrick}
As shown in Eqs.~\eqref{eq:KCKD} and~\eqref{eq:ChoiDiag}, the strategy described in Sec.~\ref{sec:strategy} can be written in terms of the Choi matrix as $\mathfrak{C}=(V^{\sf T}\otimes U) \mathfrak{D} (V^{\sf T}\otimes U)^\dagger$, with
\begin{equation}
\mathfrak{D}=\frac{1}{2}\openone\otimes\openone+\frac{s_1}{2}\openone\otimes X+\frac{\mu_1}{2}X\otimes X - \frac{\mu_2}{2}Y\otimes Y + \frac{\mu_3}{2}Z\otimes Z\,,
\end{equation}
where $X$, $Y$ and $Z$ denote the Pauli matrices. Given the reasoning of the previous section, our strategy constitutes an optimal solution to the tracking problem if and only if the following SDP is solved with $\mathfrak{K}=\mathfrak{D}$,
\begin{equation}\label{eq:cvxopt_tilde}
  \text{maximize}-\tr\left(\widetilde{F}_0 \mathfrak{K}\right)\quad
  \text{subject to}\quad\mathfrak{K}\geq 0\quad\mbox{and}\quad\tr\left[(\sigma_k\otimes\openone)\mathfrak{K}\right] = 2\delta_{k 0}\,,\quad k=0,1,2,3\,,
\end{equation}
with
\begin{equation}
\widetilde{F}_0=-\sum_{i=1}^2{(V \rho_i V^\dagger)^{\sf T}\otimes U^\dagger\pi_i\overline{\rho}_i U}\,.
\end{equation}
The above SDP has the strong duality property, i.e., its optimal value is guaranteed to be identical to the optimal value of its dual problem \cite{04Boyd}. This fact follows, for example, from the ``strict feasibility'' of the point $\widetilde{\mathfrak{K}}=\openone\otimes\openone/2$, which satisfies the constraints of (\ref{eq:cvxopt_tilde}) with the strict inequality $\widetilde{\mathfrak{K}}>0$.\par

From duality theory for SDPs (cf. Sec.~\ref{sec:LagDuality}), the problem above is solved with $\mathfrak{K}=\mathfrak{D}$ if and only if (i) $\mathfrak{D}$ satisfies the constraints of  (\ref{eq:cvxopt_tilde}) and (ii) the linear matrix inequality
\begin{equation}\label{eq:lmi}
F=\widetilde{F}_0 + \mathfrak{x}_0 \openone\otimes \openone + \mathfrak{x}_1 X\otimes\openone + \mathfrak{x}_2 Y\otimes\openone + \mathfrak{x}_3 Z\otimes\openone \geq 0
\end{equation}
is satisfied by some quadruple $(\mathfrak{x}_0,\mathfrak{x}_1,\mathfrak{x}_2,\mathfrak{x}_3)$ such that
\begin{equation}\label{eq:weak_dual}
2 \mathfrak{x}_0=-\tr\left(\widetilde{F}_0 \mathfrak{D}\right)\,.
\end{equation}
 If that is the case, then the so-called ``complementary slackness'' condition \cite{02Audenaert030302}, $\mathfrak{D}F=0$, holds for the appropriate values of coefficients $\mathfrak{x}_0$, $\mathfrak{x}_1$, $\mathfrak{x}_2$ and $\mathfrak{x}_3$.\par

To see that (i) is verified, recall that the values of $\mu_{1,2,3}$ and $s_1$ were chosen to make of $\mathcal{D}$ an (extreme) CPTP map. As mentioned before, the Choi matrix of any such map (on qubits) is characterized by the constraints of problem (\ref{eq:cvxopt_tilde}).\par

For (ii), first note that
$-\tr\left(\widetilde{F}_0 \mathfrak{D}\right)$ is merely $\aver{\F_{\rm HS}^{\rm A}}$ or $\aver{\F_{\rm HS}^{\rm B}}$ given in Eqs. (\ref{eq:fid_nt}) and (\ref{eq:fid_t}), depending on whether $\Omega>0$ or $\Omega \leq 0$. However, for later use in Ch.~\ref{chap:multistep}, we will consider the more general case where $\brho_i$ is not necessarily normalized (our intention is to show that our tracking strategy is still optimal in this case). Accounting for this generalization, we have
\begin{equation}
-{\rm Tr}\left(\widetilde{F}_0 \mathfrak{D}\right)=\left\{
\begin{array}{rcl}
\frac{1}{2}\left(c+\Gamma_a\right)&\mbox{if}& \Omega > 0\\
\frac{1}{2}\left(c+\Gamma_b\right)&\mbox{if}& \Omega \leq 0
\end{array}
\right.
\end{equation}
where $c\mathrel{\mathop:}= c_1+c_2$ and $c_i\mathrel{\mathop:}=\pi_i{\rm Tr}\overline{\rho}_i$ (for normalized $\brho_i$, we have $c_i=\pi_i$ and $c=1$).

In our particular problem, the complementary slackness condition results in sufficient independent linear equations that $\mathfrak{x}_0$, $\mathfrak{x}_1$, $\mathfrak{x}_2$ and $\mathfrak{x}_3$ are defined precisely. We find $\mathfrak{x}_2=0$ and

\begin{itemize}
\item If $\Omega>0$,
\begin{subequations}\label{eq:coefomegap}
\begin{align}
\mathfrak{x}_0&=\frac{1}{4}\left(c+\Gamma_a\right)\,,\\
\mathfrak{x}_1&=\frac{R_\times}{4R_-}\left(c+\Gamma_a\right)\,,\\
\mathfrak{x}_3&=\frac{1}{4R_-}\left[\left(c_1\bm{R}_1+c_2\bm{R}_2\right)\cdot\bm{R}_-+\frac{\Xi}{\Gamma_a}\right]\,. \label{eq:xi3p}
\end{align}
\end{subequations}
\item If $\Omega\leq 0$,
\begin{subequations}\label{eq:coefomegam}
\begin{align}
\mathfrak{x}_0&=\frac{1}{4}\left(c+\Gamma_b\right)\,,\\
\mathfrak{x}_1&=\frac{1}{4R_-}\left( cR_\times+\frac{\xi}{\Gamma_b}\right)\,,\label{eq:xi1m}\\
\mathfrak{x}_3&=\frac{1}{4R_-}\left[\left(c_1\bm{R}_1+c_2\bm{R}_2\right)\cdot\bm{R}_-+\frac{\Xi}{\Gamma_b}\right]\,.\label{eq:xi3m}
\end{align}
\end{subequations}
\end{itemize}
where, for brevity, we have defined
\begin{align}
\Xi&\mathrel{\mathop:}= \sum_{i=1}^2{\left(\bm{R}_i\cdot\bm{R}_-\right)\left(\overline{\bm{R}}_i\cdot\overline{\bm{R}}_+\right)}\,,\\
\xi&\mathrel{\mathop:}= R_\times\overline{R}_+^2+\overline{R}_\times R_-^2\label{eq:xi}\,,
\end{align}
which can be shown to satisfy the relation
\begin{equation}\label{eq:xixi}
\Xi^2+\xi^2=R_-^2\overline{R}_+^2\Gamma_b^2\,.
\end{equation}
We prove in Appendix \ref{app:welldef} that, although $\Gamma_{a}$ and $\Gamma_b$ appear in the denominator of some of the coefficients in Eqs. (\ref{eq:coefomegap}) and (\ref{eq:coefomegam}), no singularities occur if the indicated range of $\Omega$ is observed.

With the set of coefficients (\ref{eq:coefomegap}) and (\ref{eq:coefomegam}), Eq. (\ref{eq:weak_dual}) is clearly satisfied. As a result, the optimality of our tracking strategy is solely dependent on proving the linear matrix inequality $F\geq 0$ for the above set of coefficients. In Appendix \ref{app:charpoly} we study the characteristic polynomial of $F$ and conclude that all of its roots are non-negative, thus proving Theorem \ref{teo:optimality}.

\section{Examples}\label{sec:examples}
In this section we evaluate our tracking strategy at work in some physically relevant problems such as quantum state discrimination, quantum state purification, stabilization of quantum states in the presence of noise and state-dependent quantum cloning. Moreover, we also discuss the application of our strategy in circumstances where tracking is known to be perfectly achievable. The analysis presented in this section is meant to give an explicit account on the wide range of physical applications of the tracking problem and its optimal solution.

 \subsection{Quantum State Discrimination}\label{sec:ex1}
A standard result in quantum state discrimination is the Helstrom measurement \cite{76Helstrom}, which consists of a projective quantum measurement that maximizes the probability ($P_{\rm Helst}$) of correctly identifying the state of a quantum system that could have been prepared in two different states. Describing the possible preparations by $\rho_1$ with probability $p_1$ and $\rho_2$ with probability $p_2$, the Helstrom measurement gives
\begin{equation}\label{eq:helstrom_prob}
P_{\rm Helst}=\tfrac{1}{2}+\tfrac{1}{2}\|p_1\rho_1-p_2\rho_2\|_{\rm tr}\,,
\end{equation}
where $\|\cdot\|_{\rm tr}$ denotes the trace norm.\par
In this section, we propose a quantum state discrimination protocol for a pair of qubit states based on the tracking strategy introduced in Sec.~\ref{sec:strategy}. We will show that it is equivalent to Helstrom's strategy, as it will give the same correct identification probability of Eq. (\ref{eq:helstrom_prob}).

Our quantum state discrimination protocol consists of two simple steps: First we apply an optimal tracking operation $\mathcal{C}$ to approximate the states to be discriminated to some pair of orthogonal states. Without loss of generality, we take $\overline{\rho}_1=\ket{0}\!\bra{0}$ and $\overline{\rho}_2=\ket{1}\!\bra{1}$. The priority of each transformation is taken to be identical to the prior probabilities with which $\rho_1$ and $\rho_2$ are prepared, i.e., $\pi_i=p_i$. As the second and final step, we perform the quantum measurement $\{\ket{0}\!\bra{0},\ket{1}\!\bra{1}\}$, under the understanding that an outcome `0' suggests the preparation to be $\rho_1$ and an outcome `1' suggests $\rho_2$.\par
The probability of a correct identification under this tracking scheme is given by Born's rule, averaged with the prior probabilities,
\begin{align}
P_{\rm track}&=p_1\tr\left[\mathcal{C}(\rho_1)\ket{0}\!\bra{0}\right]+p_2\tr\left[\mathcal{C}(\rho_2)\ket{1}\!\bra{1}\right]\nonumber\\
&=\pi_1\tr\left[\mathcal{C}(\rho_1)\overline{\rho}_1\right]+\pi_2\tr\left[\mathcal{C}(\rho_2)\overline{\rho}_2\right]\,.\label{eq:acc_prob}
\end{align}
By comparing Eqs. (\ref{eq:acc_prob}) and (\ref{eq:J}), one promptly recognizes that $P_{\rm track}$ is precisely the performance of the operation $\mathcal{C}$ for tracking from $\rho_i$ to $\overline{\rho}_i$ with priority $\pi_i$, as measured by $\aver{\F_{\rm HS}}$. Hence, in the case of $\overline{\rho}_1=\ket{0}\!\bra{0}$, $\overline{\rho}_2=\ket{1}\!\bra{1}$ and $\pi_i=p_i$, Eqs. (\ref{eq:fid_nt}) and (\ref{eq:fid_t}) give the probability of success of our discrimination scheme for $\Omega>0$ and $\Omega \leq 0$, respectively. Next, we make these formulas more explicit.\par

Using the condition $\overline{R}_\times=0$ in Eqs. (\ref{eq:S}) and (\ref{eq:IF}), we obtain $\Omega=|T|+T$. Essentially, this means that $T$
assumes the role of the indicator function: if $T>0$, then $\Omega>0$ and we employ procedure A; if $T\leq 0$, then $\Omega=0$ and we employ procedure B.
Substituting $\overline{\bm{R}}_i=-(-1)^i p_i\bm{z}$ into Eq. (\ref{eq:T}), some simple algebra gives
\begin{equation}
T=(p_1-p_2)^2-\|p_1\bm{R}_1-p_2\bm{R}_2\|^2\,,
\end{equation}
and we can write
\begin{equation}
P_{\rm track}=\left\{\begin{array}{rcl}
\tfrac{1}{2}+\tfrac{1}{2}|p_1-p_2|&\mbox{if}&T>0\\
\tfrac{1}{2}+\tfrac{1}{2}\|p_1\bm{R}_1-p_2\bm{R}_2\|&\mbox{if}&T\leq 0
\end{array}
\right.\,,
\end{equation}
where the first line follows from Eq. (\ref{eq:fid_nt}) and the second from Eq. (\ref{eq:fid_t}).\par

 It is a tedious exercise (essentially the computation of the eigenvalues of $p_1\rho_1-p_2\rho_2$)
 to re-express Eq. (\ref{eq:helstrom_prob}) in terms of the Bloch vectors $\bm{R}_i$. The result is exactly
\begin{equation}
P_{\rm Helst}=P_{\rm track}\,,
\end{equation}
hence establishing the claimed equivalence between our strategy and Helstrom's.

Note that if $T>0$, $P_{\rm track}$ is independent of the states we are trying to distinguish, but merely dependent on the probabilities with which they occur. This can be understood by looking at the details of the affine operation taking place in procedure A. As noted before, for $\overline{R}_\times=0$ (as is the case for orthogonal targets), the affine map is such that $\mu_1=\mu_2=\mu_3=0$ and $s_1=1$; that is, the source states are completely depolarized and a new state $\ket{+}$ is prepared instead. Next, this state is rotated by the unitary $U$ and the measurement is finally performed.

It is easy to see that for $p_1=p_2=1/2$ the condition $T > 0$ (procedure A) never holds. However, as we deviate from the uniform distribution, the volume of the parameter space where procedure A is recommended grows to fully cover the space when $p_1=0$ or $p_1=1$. This is shown in Fig. \ref{fig:beat_Helst}.

\begin{figure}[h]
\centering
\subfigure[\hspace{1.5mm}$p_1<1/2$] {
    \label{fig:beat_Helst_a}
    \includegraphics[width=7.5cm]{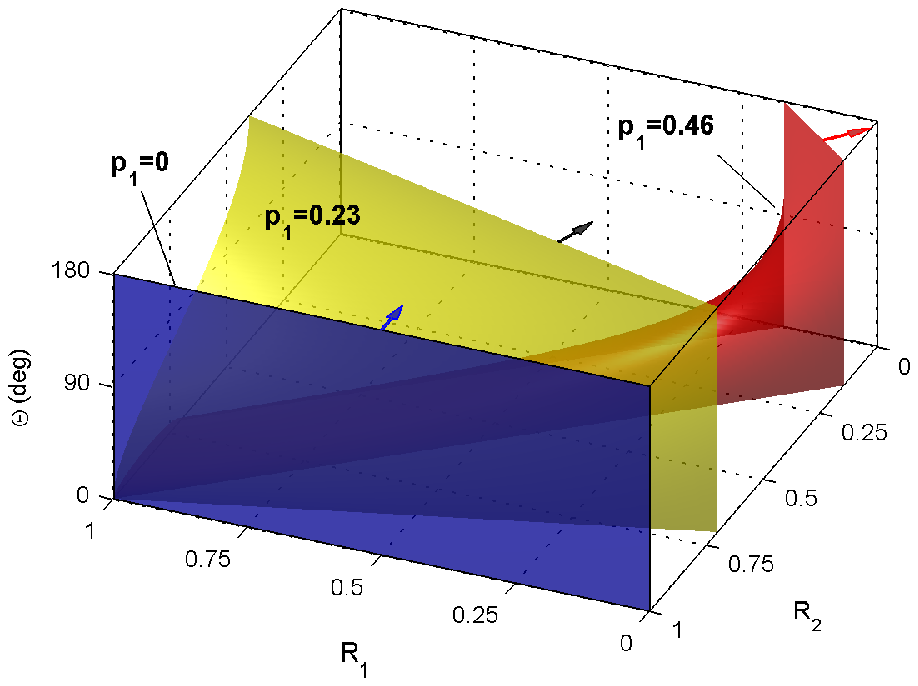}
}
\hspace{0.5cm}
\subfigure[\hspace{1.5mm} $p_1 > 1/2$] {
    \label{fig:beat_Helst_b}
    \includegraphics[width=7.5cm]{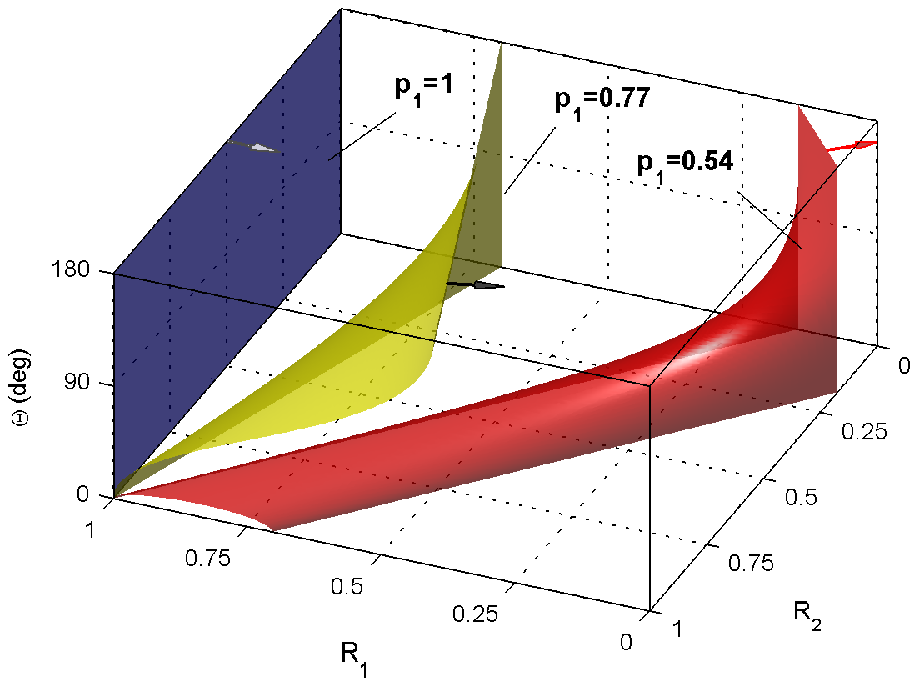}
}
\caption[Optimal regions for unitary and non-unitary schemes.]{(Color online) Each color of sheet represents a fixed deviation from the uniform probability distribution. The sheets divide the parameter space in two regions. The arrows designate the regions where procedure A (non-unitary) is recommended. On the other side, procedure B (unitary) is recommended. The larger the deviation from $p_1=1/2$, the larger the region where a non-unitary preparation for the measurement $\{\ket{0}\!\bra{0},\ket{1}\!\bra{1}\}$ is advisable.}
\label{fig:beat_Helst}
\end{figure}

\subsection{Quantum state Purification}\label{sec:ex1p5}

In this section, we consider a kind of state purification task where we aim to transform a pair of mixed source states $\rho_1$ and $\rho_2$ with the same degree of mixedness ($R_1=R_2=R<1$) that are separated in the Bloch sphere by an angle $2\theta \in (0,\pi]$ into a pair of pure target states $\overline{\rho}_1$ and $\overline{\rho}_2$ separated by the same  Bloch sphere angle $2\theta$. In other words, our purification task consists of elongating the Bloch vectors while preserving the angle between them.

For later use, it will be convenient to derive formulas for the indicator function and figure-of-merit of a slightly more general problem, where the angle between the target Bloch vectors is $2\overline{\theta} \in [0,\pi]$. The purification task can be recovered by restricting $\theta=\overline{\theta}$. In addition, we will allow the priorities of the transformations $\rho_i\to\overline{\rho}_i$ to be arbitrary positive scalars $\pi_i$ such that $\pi_1+\pi_2=1$. Later, we make $\pi_1=\pi_2$, in order to simplify the formulas.\par

In this generalized purification framework ($\theta\neq\overline{\theta}$), the indicator function is obtained from Eq. (\ref{eq:IF}), by incorporating the conditions $\overline{R}_i=\pi_i$ (purity of the targets) and  $R_i=R$ (common mixedness of the sources) in the expressions for $R_\times\overline{R}_\times$, $T$ and $S$, from Eqs. (\ref{eq:Rx}), (\ref{eq:T}) and (\ref{eq:S}), respectively. These have a particularly appealing form:
\begin{subequations}\label{eq:RSTpurif}
\begin{align}
R_\times\overline{R}_\times&=\sqrt{\mathcal{R}_s\mathcal{R}_c\left[\Pi_+\Pi_--\delta\right]}\,,\\
T&=(1-\mathcal{R}_c)\Pi_+-\mathcal{R}_s\Pi_-\,,\\
S&=\sqrt{\left[(1-\mathcal{R}_c)\Pi_++\mathcal{R}_s\Pi_-\right]^2-4\delta\mathcal{R}_s(1-\mathcal{R}_c)}\,,
\end{align}
\end{subequations}
where we have defined $\mathcal{R}_c \mathrel{\mathop:}= R^2\cos^2\theta$, $\mathcal{R}_s\mathrel{\mathop:}=R^2\sin^2\theta$, $\delta=(\pi_1-\pi_2)^2$ and
\begin{equation}
\Pi_\pm\mathrel{\mathop:}=\pi_1^2+\pi_2^2\pm2\pi_1\pi_2\cos{2\overline{\theta}}\,.
\end{equation}
From the above equations, the indicator function and the figure-of-merit can be immediately obtained. At this point, though, we specialize to the case $\delta=0$ (i.e., $\pi_1=\pi_2=1/2)$ and give explicit formulas in this particular case. From Eq. (\ref{eq:IF}),
\begin{align}
\Omega&=2\left[(1-\mathcal{R}_c)\Pi_+-\sqrt{\mathcal{R}_s\mathcal{R}_c\Pi_+\Pi_-}\right]\\
&=2\left[\cos^2\overline{\theta}-R^2\cos{\theta}\cos{\overline{\theta}\cos{\left(\theta-\overline{\theta}\right)}}\right]\,,\label{eq:Omegapurif}
\end{align}
where, in the second line, we used that $\Pi_+=\cos^2\overline{\theta}$ and $\Pi_-=\sin^2\overline{\theta}$ when $\pi_1=\pi_2=1/2$. From Eqs. (\ref{eq:fid_nt}) and (\ref{eq:fid_t}),
\begin{equation}\label{eq:fidpurif}
\aver{\F_{\rm HS}}=\left\{
\begin{array}{rcl}
\frac{1}{2}+\frac{1}{2}\sqrt{\cos^2\overline{\theta}+\frac{R^2\sin^2\theta\sin^2\overline{\theta}}{1-R^2\cos^2\theta}}&\mbox{if}&\Omega >0\\
\frac{1}{2}+\frac{1}{2}R\cos{\left(\theta-\overline{\theta}\right)}&\mbox{if}&\Omega \leq0
\end{array}\right.\,.
\end{equation}

It is now straightforward to see that, if $\theta=\overline{\theta}$, then $\Omega \geq 0$ with saturation if and only if $\theta = \pi/2$. That is, unless we are trying to purify from antipodal mixed states to orthogonal states, the best strategy is always a non-unitary transformation (procedure A). The optimal average fidelity of the purification scheme can be obtained by using $\theta=\overline{\theta}$ in Eq. (\ref{eq:fidpurif}). The resulting optimal purification performance is shown in Fig. \ref{fig:purif} and corresponds to the best achievable average fidelity allowed by quantum mechanics to the purification problem at hand.\par

\begin{figure}
\centering
\includegraphics[width=8.5cm]{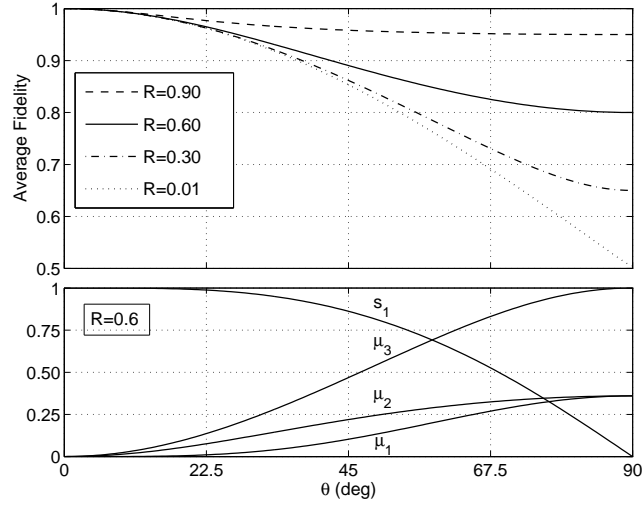}
\caption[Optimal average fidelity and controllers for purification of a pair of qubit states.]{Purifying a pair of mixed states with Bloch vectors of length $R<1$, separated by an angle $2\theta\in(0,\pi]$ with priorities $\pi_1=\pi_2=1/2$. The plot shows the optimal average fidelity for different values of $R$ and the control parameters $s_1$ and $\mu_{\{1,2,3\}}$ for $R=0.6$. The purification procedure attempts to increase, as much as possible, the length $R$ while preserving the angle $\theta$.}\label{fig:purif}
\end{figure}

From Fig. \ref{fig:purif}, we see that for small $\theta$, $\aver{\F_{\rm HS}}$ is typically high, regardless of the length $R$. This can be understood in analogy to the fact that collapsing a set of mixed states into a single pure state is always perfectly achievable. In fact, such a collapse is nearly what is needed in this domain, since a pair of pure target states separated by a small angle can be well approximated by a single pure state. Fig. \ref{fig:purif} confirms this reasoning by showing that, in the small $\theta$ domain, the source Bloch vectors are strongly compressed due to the small values of $\mu_{\{1,2,3\}}$ and then strongly elongated due to the large value of $s_1$.\par

For increasing values of $\theta$ the fidelity decreases. Such a decay is accentuated if the degree of mixedness of the source states is high (small values of $R$), reflecting the intuitive idea that it is harder to purify very mixed states. In these intermediate regions, a non-trivial combination of compressions $\mu_{\{1,2,3\}}$ and translation $s_1$ of the Bloch vectors forms the optimal purifying scheme. Noticeably, the optimal procedure has less effect on the qubit (decreasing $s_1$ and increasing $\mu_{\{1,2,3\}}$) as $\theta$ increases.\par

At $\theta=\pi/2$, we have $\Omega=0$ and an optimal unitary transformation is actually to do nothing (the unitary transformation $V$ is undone by another unitary $V^\dagger$, see Sec. \ref{sec:procB}). Note that although the plot of $\mu_1$ and $\mu_2$ in Fig. \ref{fig:purif} approaches a constant value in between $0$ and $1$, the vanishing indicator function introduces a discontinuity in the purifying operation, since now we should use procedure B, hence $\mu_1=\mu_2=1$ at $\theta=\pi/2$. Nevertheless, the values of $\mu_1$ and $\mu_2$ are utterly irrelevant in this case. At this stage both Bloch vectors are aligned with the $\bm{z}$ direction, thus any compression along $\bm{x}$ and $\bm{y}$ cannot affect the states of interest.

\subsection{Stabilizing pure states}\label{sec:ex2}

A possible use for tracking is to try to cope with the presence of noise in quantum computation and communication involving qubits. In general, noise processes (we restrict ourselves to CP processes) cannot be inverted by another CP map, not even when the noise is perfectly known \footnote{This follows from the semi-group structure of CP maps. An obvious exception arises by restricting to the group of unitary noises. In fact, a theorem by Wigner states that this is the only exception (see \cite{05Buscemi082109} for a proof, see also \cite{07Nayak103})}. However, instead of stabilizing the full Bloch sphere against noise, one may be interested at stabilizing only a limited number of states. Although not perfect, it is not uncommon that good stabilization can be achieved within this framework.\par

In this section, we consider a quantum error correction task of this type, which was studied in detail in Ref. \cite{07Branczyk012329}. We will show that the optimal correction scheme is merely a particular case of the quantum state purification procedure (with $\theta\neq\overline{\theta}$) introduced in the previous section.\par

Assume that Alice prepares (with equal probabilities) a qubit in one of the non-orthogonal pure states
\begin{subequations}\label{eq:ag_states}
\begin{align}
\ket{\psi_1}&=\cos{\frac{\overline{\theta}}{2}}\ket{+}+\sin{\frac{\overline{\theta}}{2}}\ket{-}\,,\\
\ket{\psi_2}&=\cos{\frac{\overline{\theta}}{2}}\ket{+}-\sin{\frac{\overline{\theta}}{2}}\ket{-}\,,
\end{align}
\end{subequations}
where $\ket{\pm}=(\ket{0}\pm\ket{1})/\sqrt{2}$ and $\overline{\theta}$ is the half-angle between $\ket{\psi_1}$ and $\ket{\psi_2}$ in the Bloch sphere representation, hence $\overline{\theta}\in(0,\pi/2)$. She then sends her qubit to Bob through a dephasing channel
\begin{equation}\label{eq:dephasing}
\mathcal{E}(\varsigma)=p Z \varsigma Z+(1-p)\varsigma\,,
\end{equation}
where $p$ is a constant in the range $(0,1/2]$ that has been previously determined. Bob, who does not know which of the two states was prepared, has to apply a quantum operation so as to ensure that, when Alice performs a check-measurement $\{\ket{\psi_{k}}\!\bra{\psi_{k}},\openone-\ket{\psi_{k}}\!\bra{\psi_{k}}\}$ (with $k$ labeling the identity of her actual preparation) on Bob's output, the probability of detecting her original preparation is as high as possible. This probability equals the average fidelity between the possible inputs and the outputs of Bob's operation.\par

Our tracking strategy can be of assistance to Bob if he regards the two possible noisy states as the source states $\rho_i=\mathcal{E}(\ket{\psi_i}\!\bra{\psi_i})$ and tracks (with equal priorities $\pi_i=1/2$) to the target states $\overline{\rho}_i=\ket{\psi_i}\!\bra{\psi_i}$. In this case, the target states are pure and the source states have the same degree of mixedness [this follows easily from the application of the dephasing map to the states of Eq. (\ref{eq:ag_states})], which is precisely the scenario we considered in the last section for quantum state purification.

The indicator function $\Omega$ can then be obtained from Eq. (\ref{eq:Omegapurif}) by using the following identities for the angle $\theta$ (recall that $\theta$ is the half-angle, in the Bloch sphere, between the states output by the dephasing noise),
\begin{equation}\label{eq:dephident}
\sin{\theta}=\frac{\sin{\overline{\theta}}}{R}\quad\mbox{and}\quad\cos\theta=\frac{(1-2p)\cos{\overline{\theta}}}{R}\,,
\end{equation}
where $R$ is the length of the noisy Bloch vectors. Explicitly,
\begin{equation}
\Omega=2\cos^2\overline{\theta}\left[1-R^2+2p\sin^2\overline{\theta}\right]\,.
\end{equation}
It is easy to see that, given the ranges $\overline{\theta} \in (0,\pi/2)$ and $p\in(0,1/2]$, we have $\Omega > 0$, which implies that Bob should always apply the non-unitary procedure A. The optimal performance is then obtained by substituting the identities (\ref{eq:dephident}) in the first line of Eq. (\ref{eq:fidpurif}), which gives

\begin{equation}\label{eq:fid_aggie}
\aver{\F_{\rm HS}}=\frac{1}{2}+\frac{1}{2}\sqrt{\cos^2{\theta}+\frac{\sin^4{\theta}}{1-\left(1-2p\right)^2\cos^2{\theta}}}\,.
\end{equation}
As expected, this is precisely the optimal fidelity found for this problem  in \cite{07Branczyk012329}.\par

It should be clear that our tracking strategy can be similarly applied to the stabilization of quantum states different from those of Eq. (\ref{eq:ag_states}), prepared with non-uniform prior probabilities and undergoing noise dynamics different from dephasing, in any case still providing optimal stabilization. It thus represents a significant extension of the results in \cite{07Branczyk012329}.

\subsection{Perfectly tracking quantum states}\label{sec:ex3}
In this section we evaluate the performance of our strategy in circumstances where tracking is known to be perfectly achievable. It will be convenient to split our analysis in two, namely, the case of two pure target states and the remaining cases (in which at least one of the target states is mixed).

\subsubsection{Pure target states}\label{sec:ex31}
In appendix \ref{app:ptc} we prove a corollary of Alberti and Uhlmann's theorem stating that a CPTP  map $\mathcal{A}$ perfectly transforming a pair of quantum states $\rho_i$ ($i=1,2$) into a pair of pure states $\overline{\rho}_i$ exists if and only if $\rho_i$ are also pure and $\theta\geq\overline{\theta}$. Since our tracking strategy is  optimal (cf. Theorem \ref{teo:optimality}), we can infer from Alberti and Uhlmann's theorem that it implements tracking with unit fidelity whenever $R=1$ and $\theta\geq\overline{\theta}$. This is explicitly verified in the sequence, where the indicator function and the figure-of-merit for pure state transformations are computed.

We start using Eq. (\ref{eq:RSTpurif}) with $R=1$ (pure source condition) to construct the indicator function $\Omega$ from Eq. (\ref{eq:IF}). After some straightforward manipulation, we obtain
\begin{equation}
\Omega=8\pi_1\pi_2\sin{\theta}\cos{\overline{\theta}}\sin{\left(\theta-\overline{\theta}\right)}\,.
\end{equation}
For our purposes, the only meaningful feature of $\Omega$ is whether it is strictly positive or not, in which case the above expression is equivalent to
\begin{equation}\label{eq:omegaequiv}
\widetilde{\Omega}=2\left(\theta-\overline{\theta}\right)\,,
\end{equation}
since $\theta,\overline{\theta}\in(0,\pi/2]$ and $\pi_1,\pi_2 \in (0,1)$. Recall that $\widetilde{\Omega}$ is the indicator function obtained in Sec. \ref{sec:indicator}, Eq. (\ref{eq:IF_simple}), via an heuristic argument.\par
The figure-of-merit, in turn, can be obtained from Eqs. (\ref{eq:fid_nt}) and (\ref{eq:fid_t}) to be $\aver{\F_{\rm HS}^{\rm A}} = 1$ (if $\widetilde{\Omega}>0$) and
\begin{equation}\label{eq:fidppunit}
\aver{\F_{\rm HS}^{\rm B}} =\frac{1}{2}+\frac{1}{2}\sqrt{\pi_1^2+\pi_2^2+2\pi_1\pi_2\cos{\left(2\theta-2\overline{\theta}\right)}}
\end{equation}
(if $\widetilde{\Omega} \leq 0$). Note, however, that $\aver{\F_{\rm HS}^{\rm B}}=1$ if $\theta=\overline{\theta}$ (i.e., $\widetilde{\Omega}=0$), in such a way that we can write

\begin{equation}\label{eq:fid_B}
\aver{\F_{\rm HS}}=\left\{
\begin{array}{ccl}
1&\mbox{if}&\theta\geq\overline{\theta}\\
\aver{\F_{\rm HS}^{\rm B}}&\mbox{if}&\theta<\overline{\theta}
\end{array}
\right.\,.
\end{equation}

The first line of Eq. (\ref{eq:fid_B}) is exactly the content of Alberti and Uhlmann's theorem applied to pure state transformations, whereas the second line establishes the optimal achievable average fidelity when perfect pure state transformation is impossible.

In conclusion, besides representing a construction of Alberti and Uhlmann's map $\mathcal{A}$ for perfect pure state transformations, our tracking strategy also gives the unitary map (procedure B) that optimally approximates impossible pure state transformations.

\subsubsection{Mixed target states}
The requirement of perfect tracking does not restrict the target states to be pure. In fact, the more general form of Alberti and Uhlmann's theorem states that for any given target states $\overline{\rho}_i$, there exists a CPTP map $\mathcal{A}$ that implements perfect tracking from all source states $\rho_i$ satisfying
\begin{equation}\label{eq:AU}
\|\overline{\rho}_1-t\overline{\rho}_2\|_{\rm tr}\leq\|\rho_1-t\rho_2\|_{\rm tr}\quad\forall t\in\mathbb{R}^+\,,
\end{equation}
In contrast to the previous section though, our tracking strategy is generally not a construction of the map $\mathcal{A}$ in this case. As mentioned before, this is a consequence of the fact that our figure-of-merit is not as well motivated in the case of mixed target states. For example,
in situations where perfect tracking is possible, the resulting average Hilbert-Schmidt inner product does not achieve its maximal value. This is further explored next.\par

Any CPTP map $\mathcal{C}$ implementing perfect tracking must satisfy
\begin{equation}\label{eq:perfcond}
\pi_1\tr\left[\mathcal{C}(\rho_1)\overline{\rho}_1\right]+\pi_2\tr\left[\mathcal{C}(\rho_2)\overline{\rho}_2\right]=\pi_1\tr\overline{\rho}_1^2+\pi_2\tr\overline{\rho}_2^2\,.
\end{equation}
Our strategy, though, does not arise from an attempt to enforce Eq. (\ref{eq:perfcond}), but instead to maximize its lhs (cf. Sec.~\ref{sec:problem}). Although these actions are equivalent in the case of pure target states [the rhs of Eq. (\ref{eq:perfcond}) equals $1$, which is precisely the maximum value of its lhs for states satisfying the criterion of Eq. (\ref{eq:AU})], for mixed target states this equivalence is lost. In this case, the lhs can typically be made greater than the rhs by employing an operation $\mathcal{C}$ that elongates the source Bloch vectors to nearly pure states, as illustrated in Fig. \ref{fig:mixtrack}. As a consequence, the maximization of our figure-of-merit leads to a departure from the perfect tracking operation.

\begin{figure}
\centering
\includegraphics[width=8.5cm]{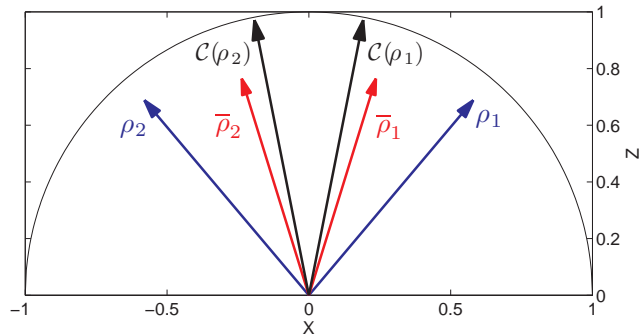}
\caption[Departure from optimality due to the lack of normalization.]{(Color online) Although perfect tracking $\rho_i\to\overline{\rho}_i$ is physically allowed, the resulting average Hilbert-Schmidt inner product [the rhs of Eq. (\ref{eq:perfcond})] is only $0.82$ for this transformation. Our tracking strategy $\mathcal{C}$ attempts to maximize this number, finding a different transformation which gives an average Hilbert-Schmidt inner product [the lhs of Eq. (\ref{eq:perfcond})] equal to approximately $0.89$. As a result, $\mathcal{C}$ does not implement perfect tracking.}\label{fig:mixtrack}
\end{figure}

Yet, recall that the average Hilbert-Schmidt inner product lower bounds the average fidelity and as such, its maximization has some beneficial impact in implementing tracking, in the sense that it ensures that the resulting average fidelity is no less than the maximal average Hilbert-Schmidt inner product.

\subsection{State-dependent Cloner}

One of the most celebrated results in quantum information science is the ``no-cloning theorem'' \cite{82Dieks271,82Wootters802}, which establishes the impossibility of copying an unknown pure quantum state. Since its inception in the literature, a lot of work has been done on the topic, both extending its range of applicability as well as attempting to weaken its impact in practical applications (see \cite{05Scarani1225} for a review). Remarkable results in each of these directions are the ``no-broadcasting theorem'' for noncommuting mixed quantum states \cite{96Barnum2818} and the Bu\v{z}ek-Hillery optimal quantum cloning machine \cite{96Buzek1844}. \par

In this section we consider a state-dependent cloning task introduced in Ref. \cite{98Bruss2368}. We will show that our tracking strategy provides a straightforward derivation of the optimal cloning fidelity obtained in that paper. Following \cite{98Bruss2368}, let
\begin{subequations}
 \begin{align}
\ket{a}&=\cos{\phi}\ket{0}+\sin{\phi}\ket{1}\,,\\
\ket{b}&=\sin{\phi}\ket{0}+\cos{\phi}\ket{1}\,,
 \end{align}
 \end{subequations}
 for $\phi\in [0,\pi/4)$, be the only two possible preparations of a single-qubit, each of which occurring with probability $1/2$. The cloning task is to output the two-qubit state $\ket{aa}\equiv\ket{a}\otimes\ket{a}$ if the initial preparation is $\ket{a}$ or $\ket{bb}\equiv\ket{b}\otimes\ket{b}$ if the initial preparation is $\ket{b}$. In \cite{98Bruss2368}, a unitary transformation $U$ was obtained such that the figure-of-merit (the so-called ``global fidelity'')
\begin{equation}\label{eq:Fg}
F_g=\frac{1}{2}\left(|\langle a a |U| a 0\rangle|^2+|\langle bb |U| b 0\rangle|^2\right)
\end{equation}
is maximal.\par

The key point that allows the application of our tracking strategy here is that, although the unitary evolution $U$ acts on the Hilbert space of a two-qubit system, it was shown in \cite[Appendix B]{98Bruss2368} that the maximizing $U$ is such that $U\ket{a0}$ and $U\ket{b0}$ lie in the two-dimensional subspace spanned by $\{\ket{aa},\ket{bb}\}$. Therefore, we can regard this cloning as a transformation from the two-dimensional subspace spanned by $\{\ket{a0},\ket{b0}\}$ to the two-dimensional subspace spanned by $\{\ket{aa},\ket{bb}\}$. By this same argument, we could have even relaxed the condition that the system to be cloned is a qubit.

Let $\ket{s_1}$ and $\ket{s_2}$ ($\ket{t_1}$ and $\ket{t_2}$) be the fictitious qubit source (target) states, and let $2\theta$ ($2\overline{\theta}$) be the Bloch sphere angle between them. Then, we must have
\begin{subequations}
\begin{align}
\langle s_1 | s_2 \rangle &= \langle a0 | b0 \rangle = \sin{(2\phi)} = \cos\theta\,,\\
\langle t_1 | t_2 \rangle &= \langle aa | bb \rangle = \sin^2{(2\phi)} = \cos\overline{\theta}\,.
\end{align}
\end{subequations}

From the above equations, the angles $\theta$ and $\overline{\theta}$ can be computed in terms of $\phi$, and the optimal value of $F_g$ is given by the optimal fidelity for tracking between pure qubit states, as described in Sec \ref{sec:ex31}. In particular, note that for the present problem, a valid indicator function is the one proposed in Eq. (\ref{eq:omegaequiv}),
\begin{equation}
\widetilde{\Omega}=2\arccos\left[\sin\left(2\phi\right)\right]-2\arccos\left[\sin^2\left(2\phi\right)\right]\leq 0\,,
\end{equation}
where the inequality holds for the specified range of $\phi$,
implying that the optimal fidelity is given by Eq. (\ref{eq:fidppunit}) with the proper values of $\theta$ and $\overline{\theta}$, explicitly
\begin{equation}
\aver{\F_{\rm HS}}=\frac{1}{2}+\frac{1}{2}\sqrt{\pi_1^2+\pi_2^2+2\pi_1\pi_2\cos{\widetilde{\Omega}}}\,.
\end{equation}

For $\pi_1=\pi_2=1/2$, the above formula can be shown to be precisely the same as Eq. (38) of \cite{98Bruss2368}, which gives the optimal global fidelity of the cloner. Thus we have not only reproduced that previous result, but also determined how it is optimally modified to incorporate an unequal probability of preparation of $\ket{a}$ and $\ket{b}$.

Finally, let us just mention that the resulting optimal tracking unitary operation (call it $W$) is not quite the optimal cloning unitary operation $U$ appearing in Eq. (\ref{eq:Fg}) and detailed in \cite{98Bruss2368} ($U$ and $W$ do not even act in Hilbert spaces of equal dimensions).
Instead, $W$ constrains how $U$ acts on the states of the form $\ket{\psi 0}$, but to fully specify $U$ we would need to choose $U\ket{11}$ and $U\ket{01}$ such that $U$ is a unitary matrix. Since this choice is not unique and does not affect the fidelity, we can say that $W$ contains all the essential information associated with the optimal cloning map.

\section{Tracking with a Control Loop}\label{sec:quantctrl}

Although the strategy introduced in Section \ref{sec:strategy} has been tailored to correspond to a CPTP map, so far no insight on how such a map can be physically implemented has been given. In this section we provide a realization in terms of a quantum control scheme. Namely, procedures A and B are shown to have the structure of closed and open loop control, respectively.\par

We start by giving a possible Kraus decompositions for the CPTP maps representing our strategy. This is relevant here because the Kraus form of a CPTP map enables us to interpret that map as some generalized quantum measurement (with no record of the outcomes) \cite{00Nielsen}. For $\Omega > 0$, the transformation $\mathcal{C}(\rho_i)=U \mathcal{D}(V \rho_i V^\dagger)U^\dagger$ from procedure A can be written as
\begin{equation}
\mathcal{C}(\rho_i)=(U M_1 V)\rho_i(U M_1 V)^\dagger+(U Y M_2 V)\rho_i(U Y M_2 V)^\dagger\,,
\end{equation}
with
\begin{subequations}\label{eq:meas_op}
\begin{align}
M_1&=\cos{\left(\frac{\chi-\eta}{2}\right)}\ket{+}\!\bra{+}+\sin{\left(\frac{\chi+\eta}{2}\right)}\ket{-}\!\bra{-}\,,\\
M_2&=\sin{\left(\frac{\chi-\eta}{2}\right)}\ket{+}\!\bra{+}-\cos{\left(\frac{\chi+\eta}{2}\right)}\ket{-}\!\bra{-}\,,
\end{align}
\end{subequations}
where  $\chi$ and $\eta$ are defined such that $\sin{\chi}=\mu_3$, $\cos{\chi}=\sqrt{1-\mu_3^2}$, $\sin{\eta}=\mu_2$ and $\cos{\eta}=\sqrt{1-\mu_2^2}$.

 For $\Omega \leq 0$, the transformation $\mathcal{C}(\rho_i)=U V \rho_i V^\dagger U^\dagger$ from procedure B is automatically in Kraus form, with a single Kraus operator $U V$.\par
We interpret these results as follows. First for $\Omega > 0$, the unitary $V$ is applied to the system and then a generalized quantum measurement with operators $M_1$ and $M_2$ is performed. Conditioned on observing the outcome `2', a Pauli $Y$ is applied to the system, followed by the unitary $U$. If the outcome is `1', the unitary $U$ is applied straight away. Due to this measurement-dependent dynamics (feedback), procedure A can be regarded as a closed loop control scheme.\par
Note that the measurement operators $M_1$ and $M_2$ are not projections, so the implementation of such a measurement requires the enlargement of the Hilbert space (by interaction with an ancilla), with subsequent (projective) measurement of the ancilla. Fig. \ref{fig:circuit} shows a possible circuit model for procedure A.

\begin{figure}[h]
\centering
\includegraphics[width=8cm]{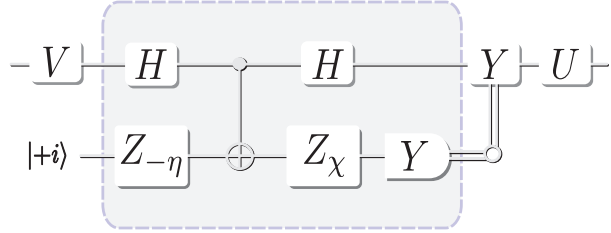}
\caption[A circuit model for the implementation of the non-unitary scheme as a feedback loop.]{A circuit model illustrating the feedback structure of procedure A. In the figure, $\ket{\pm i}=(\ket{0}\pm i\ket{1})/\sqrt{2}$ are the eigenvectors of the Pauli matrix $Y$, $H$ is the Hadamard gate and $Z_\theta=\exp{(-i\theta Z/2)}$. The highlighted circuit entangles the main system with the ancilla and projectively measures the ancilla in the basis $\{\ket{+i},\ket{-i}\}$. This induces a non-projective dynamics of the main system, and for this reason this block is referred to as a ``weak measurement''. If the measurement outcome is `$+i$', then the unitary transformations $Y$ and $U$ are applied to the main system; otherwise, only $U$ is applied.}\label{fig:circuit}
\end{figure}

For $\Omega \leq 0$, there is clearly no measurement involved, hence the control strategy is implemented independent
of acquiring extra information from the system. For this reason, procedure B can be regarded as an open loop control scheme.

\section{Discussion and Conclusions}\label{sec:discussion}

In this chapter we have introduced a simple quantum version of
a common classical control problem named tracking. Our quantum tracking problem consists of determining how to optimally enforce a certain dynamics to a qubit system, when the initial preparation of the qubit is uncertain (as modeled by a pair of states occurring with given prior probabilities) and the desired dynamics depends on the actual preparation. We presented an optimal quantum tracking strategy.\par

The tracking problem studied here is sufficiently general to provide an unifying approach to many problems in quantum information science as special cases. For example, some cases of quantum state discrimination, quantum state purification, stabilization of qubits against noise and state-dependent quantum cloning were explicitly shown to be instances of quantum tracking. As such, previously known quantum limits in the realization of these tasks were recovered via the application of our tracking strategy. Likewise, our tracking strategy can be used to obtain new and improved limits in the realization of other impossible quantum machines.

The derivation of our strategy was largely dependent on the fact that our figure-of-merit (the averaged Hilbert-Schmidt inner product) is linear in $\mathcal{C}$, which, in turn, is constrained to be an element of the \emph{convex} set of CPTP maps acting on qubits. This implies that the optimal map $\mathcal{C}$ belongs to the subset of extreme points, which has been fully characterized in \cite{02Ruskai159}. Thanks to a parametrization of these extreme points, the resulting optimization problem could be handled analytically when a few mild assumptions (supported by numerical observation) were made about the form of the optimal solution. The optimality was safeguarded \emph{a posteriori} via an argument based on the SDP structure of the tracking problem.\par

Analytical solutions for generalizations of the tracking problem studied here (e.g., other figures-of-merit and/or larger dimensional quantum systems) seem to require a modified approach from the one adopted here. For example, had we chosen to proceed with a better motivated figure-of-merit for mixed targets, such as the average fidelity, we would still have the guarantee that the optimal $\mathcal{C}$ is an extreme point, however optimality results about a possible guess would be harder to derive, since it is not known if/how the resulting optimization problem can be cast as a SDP when source and target states are mixed. Alternatively, we could have chosen, for example, to minimize the average trace distance, which can be cast as a SDP \cite{01Fazel4734,07Recht}. However, the trace distance is not concave in $\mathcal{C}$, in which case its minimum need not be an extreme point. Finally, had we kept our linear figure-of-merit but generalized from qubits to qudits for $d>2$ (or to multiple qubits), we would face the problem that the extreme points of the set of CPTP maps on higher dimensional matrix algebras are not well characterized.

A possibly simpler generalization is to preserve low dimensionality of the quantum system and linearity in the figure-of-merit, but allow for a larger number of possible sources and targets. In principle, this problem can be approached following exactly the same lines as adopted here. In fact, it is not difficult to see that a particular case of this more general problem can already be considered solved given the results of this chapter. Consider we are given two sets $\mathcal{S}_1$ and $\mathcal{S}_2$, respectively with $n_1$ and $n_2$ elements (let $N = n_1 + n_2$), of qubit density matrices $\tau_j$ ($j=1,\ldots,N$), and want to send every element of $\mathcal{S}_i$ to $\overline{\rho}_i$ for $i=1,2$. In analogy with Eq. (\ref{eq:J}), define the figure-of-merit
\begin{equation}\label{eq:figmeritger}
\aver{\F_{\rm HS}}=\sum_{j=1}^{n_1}{q_j\tr\left[\mathcal{C}(\tau_j),\overline{\rho}_1\right]}+\sum_{j=n_1+1}^N{q_j\tr\left[\mathcal{C}(\tau_j),\overline{\rho}_2\right]}\,,
\end{equation}
where the positive numbers $q_j$ set the priorities of each transformation, and $\sum_{j=1}^N q_j=1$. Due to the linearity of the trace and of quantum operations, Eq. (\ref{eq:figmeritger}) can be rewritten exactly as Eq. (\ref{eq:J}) with $\pi_1 = \sum_{j=1}^{n_1}{q_j}$, $\pi_2 = \sum_{j=n_1+1}^{N}{q_j}$,
\begin{align}
\rho_1=\frac{1}{\pi_1}\sum_{j=1}^{n_1}{q_j\tau_j} &\quad \mbox{and} \quad \rho_2=\frac{1}{\pi_2}\sum_{j=n_1+1}^{N}{q_j\tau_j}\,.
\end{align}
Note that $\pi_1,\pi_2\geq 0 $, $\pi_1+\pi_2=1$ and $\rho_1$, $\rho_2$ are valid density matrices. So, for $i=1,2$ and $j=1,\ldots,N$, the problem of optimally approximating the $N$-state transformation $\mathcal{S}_i \to \overline{\rho}_i$ with priority $q_j$ is equivalent to optimally approximating the $2$-state transformation $\rho_i\to\overline{\rho}_i$ with priority $\pi_i$.\par

\chapter{Multi-Step Tracking}\label{chap:multistep}
\section{Introduction}

So far in this thesis, we have looked at the problem of transforming quantum states with a single controlled intervention. In the last two chapters, we have seen how this can be useful to stabilize the unknown state of a quantum system undergoing some pre-characterized noisy dynamics. Our strategy consisted of waiting for the system to experience all the noise, and only after that to apply a quantum operation that optimally transformed the noisy states into the original states (or, in the case of tracking, to other desired states). In this chapter we ask whether we can do any better if instead of waiting for the full noisy evolution to take place, we actively interact with the system multiple times while the noise is still in action. This is illustrated in Fig.~\ref{fig:multistep_intro2}.
\begin{figure}[h]
\centering
\includegraphics[width=15cm]{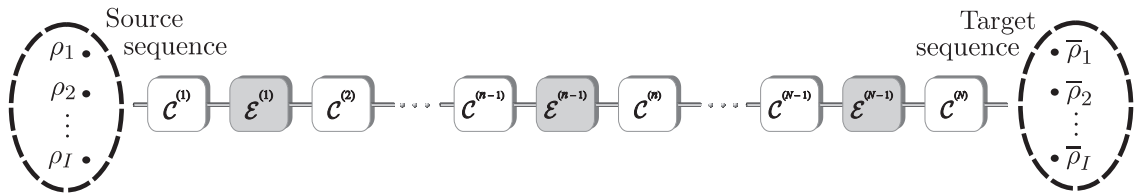}
\caption[Schematic of the multi-step tracking problem.]{Schematic of the multi-step tracking problem. Each block $\mathcal{C}^{(n)}$ represents a controlled intervention that attempts to optimally track between sequences of density matrices as the system travels through the noisy channel
$\mathcal{E}^{(N-1)}\circ\ldots\circ\mathcal{E}^{(2)}\circ\mathcal{E}^{(1)}$}\label{fig:multistep_intro2}
\end{figure}

In classical systems, a multi-step setting as the one of Fig.~\ref{fig:multistep_intro2} is the basic setup for discrete time feedback control, and is known to be a very effective scheme. Each intermediate step estimates the current state of the system, which is then suitably modified on the basis of this information. Moreover, in situations where the forthcoming noise is known, we can use the information from the state estimation step and the information about the future dynamics in order to make the system more resilient against the noise to come.

The same ideas apply for quantum systems, in this case, however, one has to consider that quantum measurements typically introduce noise on the state being measured. Because of this, it is generally not clear whether the application of multiple measurements is detrimental or beneficial for the control of quantum systems.

In this chapter, we attempt to approach this question by looking for a sequence of controllers $\mathcal{C}^{(1)}$, $\ldots$, $\mathcal{C}^{(N)}$ that provides an improved performance than that obtained in the case of an optimal single controller at the end. Clearly, the latter is recovered if the $N-1$ first controllers are equal to the identity map and the last controller is optimally chosen. Our problem is to decide whether a variation of this sequence exists (and how it can be constructed) such that the tracking task can be accomplished with higher average fidelity.

The results of this chapter are preliminary and are presented in Sec.~\ref{sec:multistep}, which is divided as follows: We start introducing the basic principle of dynamic programming and outlining how it can be used (along with an heuristic argument) to construct a multi-step tracking scheme assuming that an analytical optimal solution for the single-step tracking problem is known. In subsections~\ref{sec:backwards} and~\ref{sec:forwards} these ideas are put in practice to approach the general case of tracking between arbitrary sequences of density matrices in $N$ steps. In Sec.~\ref{sec:multi2qubits}, we restrict to multi-step tracking for \emph{pairs of qubit states} in an arbitrary number of steps and provide some numerical results for the case $N=2$.

\section{Multi-step Tracking via Optimal Single-step Tracking}\label{sec:multistep}

In this section we describe how to generate a ``good'' sequence of operations for the multi-step problem. Due to some simplifying assumptions to be made along the way, these sequences cannot be guaranteed to be optimal. Nevertheless, they are usually ``good'' in the sense that, in many cases, they produce higher fidelities than those obtained with a single optimal controller at the end.

We start with a broad description of the method, which essentially consists of recursive applications of the following fundamental idea from dynamic programming \cite{57Bellman,74jacobs}:\par
\begin{principle}
\noindent{\bf Principle of Optimality.} In an optimal sequence of controllers, whatever the initial state and the optimal first control may be, the remaining controls constitute an optimal control sequence with regard to the state resulting from the first control.
\end{principle}

We shall divide our analysis in two parts. First we look at the sequence of controllers from the end to the beginning. The application of the principle of optimality in this backwards direction --- surmounted with an heuristic argument to be explained in the next section --- will reveal each controller of the multi-step sequence as a single-step operation from certain source states $\rho_i^{(n)}$ to certain target states $\overline{\varrho}_i^{(n)}$, as illustrated in Fig.~\ref{fig:multistepdiag}.
\begin{figure}[h]
\centering
\includegraphics[width=14.5cm]{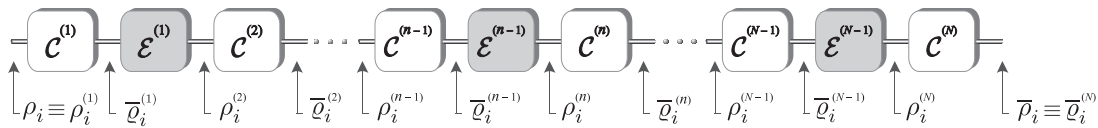}
\caption[Multi-step tracking as a sequence of single-step tracking operations.]{Multi-step tracking as a sequence of single-step tracking operations. Each controller $\mathcal{C}^{(n)}$ is seen as an optimal tracker from states $\rho_i^{(n)}$ to $\overline{\varrho}_i^{(n)}$.}\label{fig:multistepdiag}
\end{figure}
Furthermore, assuming that the analytical solution to this problem is known, we can obtain the explicit form of the function $\mathpzc{b}$, relating the target states of the $n$-th correction with the source and target states of the next correction, i.e.,
  \begin{align}\label{eq:littleb}
\overline{\varrho}_i^{(n)}=\mathpzc{b}\left(\rho_i^{(n+1)},\overline{\varrho}_i^{(n+1)}\right)\,.
  \end{align}
$N-n$ iterations of the above relation implicitly define a function $\mathpzc{B}$ such that
 \begin{align}
\overline{\varrho}_i^{(n)}&=\mathpzc{b}\left(\rho_i^{(n+1)},\mathpzc{b}\left(\rho_i^{(n+2)},\overline{\varrho}_i^{(n+2)}\right)\right)\nonumber\\
&\vdots\nonumber\\
&=\mathpzc{B}\left(\rho_i^{(n+1)},\rho_i^{(n+2)},\ldots,\rho_i^{(N)},\brho_i\right)\,,\label{eq:backward}
 \end{align}
which shows that the target states $\overline{\varrho}_i^{(n)}$ are expressed as a function of the (so far undetermined) source states at each step.

The states $\rho_i^{(n)}$ are determined in the second part of our approach, in which the control sequence is studied in the standard forward direction. Here, the analytical solution of the single-step tracking problem will yield a function $\mathpzc{f}$ such that
 \begin{equation}\label{eq:littlef}
\rho_i^{(n)}=\mathpzc{f}\left(\rho_i^{(n-1)},\overline{\varrho}_i^{(n-1)}\right)
 \end{equation}
Just as before, $n-1$ iterations of this relation induce the  function $\mathpzc{F}$, relating $\rho_i^{(n)}$ with the given source density matrices $\rho_i^{(1)}\equiv \rho_i$ and the targets $\overline{\varrho}_i^{(1,\ldots,n-1)}$
\begin{equation}
\rho_i^{(n)}=\mathpzc{F}\left(\rho_i,\overline{\varrho}_i^{(1)},\ldots,\overline{\varrho}_i^{(n-1)}\right)\label{eq:forward}
\end{equation}

Together, Eqs.~\eqref{eq:backward} and~\eqref{eq:forward} specify every $\rho_i^{(n)}$ and $\overline{\varrho}_i^{(n)}$, and hence a sequence of $N$ single step controllers. In the following, the procedure described above is explicitly applied and the heuristic argument giving rise to Eq.~\eqref{eq:littleb} is explained in detail.

\subsection{Backward direction}\label{sec:backwards}
Consider the multi-step sequence below
\[
\Qcircuit @C=0.7em @R=1.0em
{
& \gate{\scriptstyle \mathcal{C}^{(1)}}  & \gate{\scriptstyle \mathcal{E}^{(1)}} &     \qw & \push{\scriptstyle \!\!\!\cdots\rule{.2em}{0em}}&\gate{\scriptstyle \mathcal{E}^{(N-2)}}   &\gate{\scriptstyle \mathcal{C}^{(N-1)}}&\gate{\scriptstyle \mathcal{E}^{(N-1)}}&\gate{\scriptstyle \mathcal{C}^{(N)}}&\qw\gategroup{1}{2}{1}{8}{.7em}{--}
}
\]
Whatever the sequence of noises $\mathcal{E}^{(1,\ldots,N-1)}$ and optimal controllers $\mathcal{C}^{(1,\ldots,N-1)}$ are, the initial density matrices $\rho_i$ are obviously transformed into other density matrices after experiencing the action of the operations in the dashed box. Following the convention set up in Fig.~\ref{fig:multistepdiag}, we shall denote by $\rho_i^{(N)}$ the output of this sequence of operations. It then follows from the principle of optimality that $\mathcal{C}^{(N)}$ is the optimal single-step operation for the transformation $\rho_i^{(N)}\mapsto\brho_i$.\par

To see that a similar conclusion can be drawn for $\mathcal{C}^{(N-1)}$, consider the following diagram:
\[
\Qcircuit @C=0.7em @R=1.0em
{
& \gate{\scriptstyle \mathcal{C}^{(1)}}  & \gate{\scriptstyle \mathcal{E}^{(1)}} &     \qw & \push{\scriptstyle \!\!\!\cdots\rule{.2em}{0em}}&\gate{\scriptstyle \mathcal{E}^{(N-2)}} &\ustick{\scriptstyle \qquad\rho_i^{(N-1)}}\qw  &\gate{\scriptstyle \mathcal{C}^{(N-1)}}&\gate{\scriptstyle \mathcal{E}^{(N-1)}}& \ustick{\scriptstyle\quad \rho_i^{(N)}}\qw&\gate{\scriptstyle \mathcal{C}^{(N)}}&\qw\gategroup{1}{2}{1}{6}{.7em}{--}
}
\]
As before, let $\rho_i^{(N-1)}$ represent the output of the unknown sequence of operations within the dashed box. Here, the principle of optimality establishes that the combined operation $\mathcal{C}^{(N)}\circ\mathcal{E}^{(N-1)}\circ\mathcal{C}^{(N-1)}$
has to be optimal for the transformation $\rho_i^{(N-1)}\mapsto\brho_i$. Since $\mathcal{C}^{(N)}$ has already been determined [in terms of $\rho_i^{(N)}$], and $\mathcal{E}^{(N-1)}$ is not a controllable operation, we are only left with the task of determining the optimal $\mathcal{C}^{(N-1)}$. This is done by solving the optimization problem
\begin{equation}
\max_{\mathcal{C}^{(N-1)}\in\mathcal{Q}_{\rm d}^{\rm set}}\sum_{i=1}^I\pi_i{\rm Tr}\left[\left(\mathcal{C}^{(N)}\circ\mathcal{E}^{(N-1)}\circ\mathcal{C}^{(N-1)}\right)\left(\rho_i^{(N-1)}\right)\brho_i\right]\,,
\end{equation}
which, with some simple algebra can be re-expressed in terms of the Choi matrices of $\mathcal{C}^{(n)}$ and $\mathcal{E}^{(n)}$ as follows:
\begin{equation}\label{eq:probl_track_virt}
\max_{\mathfrak{C}^{(N-1)}}{\rm Tr}\left[\mathfrak{C}^{(N-1)}\sum_{i=1}^I{\rho_i^{(N-1)}}^{\sf T}\otimes \pi_i\overline{\varrho}_i^{(N-1)}\right]\quad\mbox{subject to}\quad \mathfrak{C}^{(N-1)}\geq 0 \quad\mbox{and}\quad \tr_2\mathfrak{C}^{(N-1)}=\openone_{\rm d}\,,
\end{equation}
where
\begin{equation}\label{eq:partra23bar}
\overline{\varrho}_i^{(N-1)}\defeq{\rm Tr}_{2,3}\left[\left(\openone_{\rm d}\otimes \mathfrak{C}^{(N)}\right)\left({\mathfrak{E}^{(N-1)}}^{\sf T}\otimes \brho_i\right)\right]\,,
\end{equation}
and $\tr_{2,3}$ denotes the partial trace operation over the second and third subsystems of dimension {\rm d}.\par

Up to here we have been closely following the dynamic programming recipe to optimally solve the multi-step tracking problem. To continue along these lines, though, we would now have to face the difficulty involved in solving the optimization problem~\eqref{eq:probl_track_virt}. Although this problem resembles the SDP maximizing the average Hilbert-Schmidt inner product between sequences $[\rho_i^{(N-1)}]_{i=1}^I$ and $[\overline{\varrho}_i^{(N-1)}]_{i=1}^I$ [compare with Eq.~\eqref{eq:opt_hsip}], this is just a superficial similarity. In fact, the matrices $\overline{\varrho}_i^{(N-1)}$ also depend on $\mathfrak{C}^{(N-1)}$ because they explicitly depend on $\mathfrak{C}^{(N)}$, which, in turn, depends on $\rho_i^{(N)}=(\mathcal{E}^{(N-1)}\circ\mathcal{C}^{(N-1)})(\rho_i^{(N-1)})$.\par

Our work around is to give up global optimality by relying on a simplifying assumption. Assuming that we know how to optimally solve the single-step tracking problem (which is a SDP), our goal is to exploit its solution to build a ``locally optimal'' multi-step scheme. For that purpose, we simply ignore the dependence of $\overline{\varrho}_i^{(N-1)}$ on $\mathfrak{C}^{(N-1)}$, regarding it as a fixed (but arbitrary) state. As a result, we choose $\mathfrak{C}^{(N-1)}$ to be the optimal single-step tracking operation\footnote{There is, however, a subtlety: Although the tensor product structure inside the partial trace of Eq.~\eqref{eq:partra23bar} guarantees that $\overline{\varrho}_i^{(N-1)}\geq 0$, in general $\overline{\varrho}_i^{(N-1)}$ is not normalized. As such, one should make sure that the optimal analytic solution to be used for $\mathfrak{C}^{(N-1)}$ is still optimal if the target density matrices are unnormalized. Recall that this was seen to be the case for the optimal single-step tracker constructed in Ch.~\ref{chap:tracking}} from $[\rho_i^{(N-1)}]_{i=1}^I$ to $[\overline{\varrho}_i^{(N-1)}]_{i=1}^I$.

It might be helpful to consider what has been learnt hitherto. From the first step, we have seen how $\mathcal{C}^{(N)}$ can be constructed as a function of $\rho_i^{(N)}$. From the second step, we have seen how $\mathcal{C}^{(N-1)}$ can be constructed as a function of $\rho_i^{(N-1)}$ and $\rho_i^{(N)}$ [via $\mathfrak{C}^{(N)}$]. We can now follow with this backwards approach, at each step characterizing each $\mathcal{C}^{(n)}$ as an optimal single-step tracking operation from arbitrary density matrices $\rho_i^{(n)}$ to targets $\overline{\varrho}_i^{(n)}$ satisfying
\begin{equation}\label{eq:fake_dm}
\overline{\varrho}_i^{(n)}={\rm Tr}_{2,3}\left[\left(\openone\otimes \mathfrak{C}^{(n+1)}\right)\left({\mathfrak{E}^{(n)}}^{\sf T}\otimes \overline{\varrho}_i^{(n+1)}\right)\right]\,,
\end{equation}
for $n=1,\ldots,N-1$ and $\overline{\varrho}_i^{(N)}\equiv\overline{\rho}_i$. Given that $\mathfrak{C}^{(n+1)}$ is a function of $\rho_i^{(n+1)}$ and $\overline{\varrho}_i^{(n+1)}$, the above gives a more explicit form of relation~\eqref{eq:littleb}.\par

As noted in Eq.~\eqref{eq:backward}, this recursion relation provides a way to obtain the targets of each controller as a function of the arbitrary sources introduced to the right of the controller at hand. Since each controller is fully specified by the knowledge of its sources and targets, we just need to self-consistently determine the (so far) arbitrary sources $\rho_i^{(n)}$ in order to completely characterize a sequence of controllers. This is the content of the following section.

\subsection{Forward direction}\label{sec:forwards}
It is much easier to construct relation~\eqref{eq:littlef}. In fact, this is simply the forward evolution with CPTP maps $\mathcal{C}^{(n)}$ implementing optimal tracking $\rho_i^{(n)}\mapsto\overline{\varrho}_i^{(n)}$ at each step,
\begin{equation}
\rho_i^{(n+1)}=\left(\mathcal{E}^{(n)}\circ\mathcal{C}^{(n)}\right)\left(\rho_i^{(n)}\right)\,.
\end{equation}
Re-expressed in terms of the Choi matrices of each map, the above reads
\begin{equation}\label{eq:sysforward}
\rho_i^{(n+1)}={\rm Tr}_{1,2}\left[\left({\mathfrak{C}^{(n)}}^{\sf T}\otimes\openone\right)\left(\rho_i^{(n)}\otimes \mathfrak{E}^{(n)}\right)\right]\,,
\end{equation}
which is precisely of the form of Eq.~\eqref{eq:littlef}, since $\mathfrak{C}^{(n)}$ is a function of $\rho_i^{(n)}$ and $\overline{\varrho}_i^{(n)}$.

The set of Eqs.~\eqref{eq:sysforward} can now be solved simultaneously with the set of Eqs.~\eqref{eq:fake_dm} to give each matrix $\rho_i^{(n)}$ and $\overline{\varrho}_i^{(n)}$. From the solution of this system, a sequence of controllers $\mathcal{C}^{(n)}$ can then be obtained.

\subsection{Multi-step tracking for a pair of qubit states}\label{sec:multi2qubits}

Let us now look at the more concrete example of tracking for a pair of qubits in $N$ steps, for which we shall employ the analytical solution for the single-step tracking problem obtained in  Ch.~\ref{chap:tracking}. In the following subsections we explicitly write the set of equations~\eqref{eq:fake_dm} and~\eqref{eq:sysforward} in terms of Bloch vectors and discuss some numerical results for $N=2$.

\paragraph{Equations for backwards direction.}
Eq.~\eqref{eq:fake_dm} is evaluated with the use of the following parameterizations:
\begin{align}
2\pi_i\overline{\varrho}_i^{(n+1)}&=c_i^{(n+1)}\openone_2+\overline{\bm{R}}_i^{(n+1)}\cdot\bm{\sigma}\label{eq:stokeswithpi}\,,\\
2 \mathfrak{C}^{(n+1)}&= \openone_4 + s_1^{(n+1)} \openone_2\otimes\left(\bm{u}_1^{(n+1)}\cdot\bm{\sigma}\right)+\sum_{j=1}^3 \mu_j^{(n+1)} \left(\bm{v}_j^{(n+1)}\cdot\bm{\sigma}^{\sf T}\right)\otimes\left(\bm{u}_j^{(n+1)}\cdot\bm{\sigma}\right)\label{eq:optSSS}\,,\\
2 \mathfrak{E}^{(n)}&= \openone_4 + \sum_{j=1}^{3}{t_j^{(n)} \openone_2\otimes\left(\bm{g}_j^{(n)}\cdot\bm{\sigma}\right)}+\sum_{j=1}^3 \lambda_j^{(n)} \left(\bm{h}_j^{(n)}\cdot\bm{\sigma}^{\sf T}\right)\otimes\left(\bm{g}_j^{(n)}\cdot\bm{\sigma}\right)\label{eq:choinoiseblah}\,,
\end{align}
where, in Eq.~\eqref{eq:stokeswithpi}, we kept with the convention from Ch.~\ref{chap:tracking} (cf. page~\pageref{conventionRb}) of writing $\overline{\bm{R}}_i^{(n+1)}$ to the ``Bloch vector'' of $\pi_i\overline{\varrho}_i^{(n+1)}$. Furthermore, we introduced the constant $c_i^{(n+1)}$ to account for the fact that $\overline{\varrho}_i^{(n+1)}$ is not normalized, as explained in Sec.~\ref{sec:dualitytrick} [however, we have $c_i^{(N)}=1$ to comply with the fact that $\overline{\varrho}_i^{(N)}\equiv\brho_i$]. Eq.~\eqref{eq:optSSS} gives the general form of the optimal single-step tracking solution, and was derived in Appendix~\ref{app:optSSmap}. The versors $\bm{v}_j^{(n+1)}$, $\bm{u}_j^{(n+1)}$ and the scalars $\mu_j^{(n+1)}$ and $s_1^{(n+1)}$ are functions of the Bloch vectors $\bm{R}_i^{(n+1)}$ and $\overline{\bm{R}}_i^{(n+1)}$  as described in Appendix~\ref{app:optSSmap}. Finally, Eq.~\eqref{eq:choinoiseblah} gives the general Choi matrix of a CPTP map (cf. Sec.~\ref{sec:CPTPonM2}, page~\pageref{sec:Choimatrep}). Since the noise is assumed to be known, the parameters $\bm{h}_j^{(n)}$, $\bm{g}_j^{(n)}$, $\lambda_j^{(n)}$ and $t_j^{(n)}$ are considered given.

Substituting the above formulas in Eq. (\ref{eq:fake_dm}), after some algebra we find that
\begin{equation}\label{eq:recursionRb}
\overline{\bm{R}}_i^{(n)}=\sum_{k=1}^3\left(\bm{Q}_i^{(n+1)} \cdot\bm{g}_k^{(n)}\right)\lambda_k^{(n)}\bm{h}_k^{(n)}\,,
\end{equation}
where $\bm{Q}_i^{(n+1)}$ is a real vector carrying all the undetermined parameters:
\begin{equation}\label{eq:Q}
\bm{Q}_i^{(n+1)}\defeq\sum_{j=1}^3\mu_j^{(n+1)}\left(\bm{u}_j^{(n+1)}\cdot\overline{\bm{R}}_i^{(n+1)}\right)\bm{v}_j^{(n+1)}\,.
\end{equation}
A more explicit form can be obtained by evaluating the dot product with the aid of Eq.~\eqref{eq:us}:
\begin{multline}
\bm{Q}_i=\frac{s_1}{\Gamma R_-}\left[\frac{\mu_1}{R_\times}\left(\overline{\bm{R}}_+\cdot\overline{\bm{R}}_i\right) \left(\bm{R}_\times\times\bm{R}_-\right)-(-1)^i\mu_3\overline{R}_\times\bm{R}_-\right] +\\
\frac{\mu_1}{\Gamma R_-^2}
\left[\mu_1\left(\overline{\bm{R}}_+\cdot\overline{\bm{R}}_i\right)+
      (-1)^i\mu_3\frac{\overline{R}_\times }{R_\times} \left(\bm{R}_-\cdot\bm{R}_{\widetilde{i}}\right)\right] \left(\bm{R}_\times\times\bm{R}_-\right)+\\
\frac{\mu_3}{\Gamma R_-^2}
\left[-(-1)^i\mu_1 R_\times\overline{R}_\times +\mu_3\left(\left(\bm{R}_1\cdot\bm{R}_-\right)\overline{\bm{R}}_1+ \left(\bm{R}_2\cdot\bm{R}_-\right)\overline{\bm{R}}_2\right)\cdot\overline{\bm{R}}_i\right]\bm{R}_-
\end{multline}
where, for brevity, we have omitted the index $(n+1)$ from every symbol.

\paragraph{Equations for forward direction.} Evaluation of Eq.~\eqref{eq:sysforward} is accomplished with Eqs.~\eqref{eq:choinoiseblah} and~\eqref{eq:optSSS} for $\mathfrak{E}^{(n)}$ and $\mathfrak{C}^{(n)}$, respectively, and
\begin{equation}
2\rho_i^{(n)}=\openone_{\rm d}+\bm{R}_i^{(n)}\cdot\bm{\sigma}\,,
\end{equation}
for the source density matrices, where $\bm{R}_i^{(n)}$ gives the Bloch vector of the density matrices $\rho_i^{(n)}$. After some manipulation, we find
\begin{equation}\label{eq:recursionR}
\bm{R}_i^{(n+1)}=\sum_{k=1}^3 t_k^{(n)}\bm{g}_k^{(n)}+\left(\bm{P}_i^{(n)}\cdot\bm{h}_k^{(n)}\right)\lambda_k^{(n)}\bm{g}_k^{(n)}\,,
\end{equation}
where now, $\bm{P}_i^{(n)}$ is the real vector carrying in the undetermined parameters
\begin{equation}\label{eq:P}
\bm{P}_i^{(n)}\defeq s_1^{(n)}\bm{u}_1^{(n)}+\sum_{j=1}^3 \mu_j^{(n)}\left(\bm{v}_j^{(n)}\cdot\bm{R}_i^{(n)}\right)\bm{u}_j^{(n)}\,.
\end{equation}
Computing the dot products with the aid of Eqs. (\ref{eq:vs}), gives
\begin{multline}
\bm{P}_i=\frac{s_1}{\Gamma}\left[\left(s_1+2\mu_1\frac{R_\times}{R_-}\right)\overline{\bm{R}}_+ -(-1)^i\mu_3\frac{R_-}{\overline{R}_\times}\left(\overline{\bm{R}}_{\widetilde{i}}\times\overline{\bm{R}}_\times\right)\right]- \frac{(-1)^i}{\Gamma}\mu_1\mu_3\frac{R_\times}{\overline{R}_\times}\left(\overline{\bm{R}}_{\widetilde{i}}\times\overline{\bm{R}}_\times\right)+\\
\frac{1}{\Gamma R_-^2} \left[\mu_1^2R_\times^2\overline{\bm{R}}_+ +\mu_3^2\left(\bm{R}_i\cdot\bm{R}_-\right) \left(\left(\bm{R}_1\cdot\bm{R}_-\right)\overline{\bm{R}}_1 +\left(\bm{R}_2\cdot\bm{R}_-\right)\overline{\bm{R}}_2\right)\right]
\end{multline}
where, once again, we omitted the indices $(n)$ for brevity.

\paragraph{System of Equations.} In summary, combining Eqs.~\eqref{eq:recursionRb} and~\eqref{eq:recursionR} we obtain the following non-linear system of $2I(N-1)$ vector equations with $2I(N-1)$ vector variables
\begin{equation}\label{eq:systemmultistep}
\begin{array}{ll}
\begin{array}{c}\mbox{From Eq.~\eqref{eq:recursionRb}}\\(i=1,\ldots,I)
\end{array}&\left\{\begin{array}{ll}
\overline{\bm{R}}_i^{(1)}&=\sum_{k=1}^3\left[\bm{Q}_i^{(2)}(\bm{R}_i^{(2)},\overline{\bm{R}}_i^{(2)}) \cdot\bm{g}_k^{(1)}\right]\lambda_k^{(1)}\bm{h}_k^{(1)}\\
\overline{\bm{R}}_i^{(2)}&=\sum_{k=1}^3\left[\bm{Q}_i^{(3)}(\bm{R}_i^{(3)},\overline{\bm{R}}_i^{(3)}) \cdot\bm{g}_k^{(2)}\right]\lambda_k^{(2)}\bm{h}_k^{(2)}\\
&\vdots\\
\overline{\bm{R}}_i^{(N-1)}&=\sum_{k=1}^3\left[\bm{Q}_i^{(N)}(\bm{R}_i^{(N)},\overline{\bm{R}}_i^{(N)}) \cdot\bm{g}_k^{(N-1)}\right]\lambda_k^{(N-1)}\bm{h}_k^{(N-1)}
\end{array}\right.\\
\begin{array}{c}\mbox{From Eq.~\eqref{eq:recursionR}}\\(i=1,\ldots,I) \end{array}&\left\{\begin{array}{ll}
\bm{R}_i^{(2)}&=\sum_{k=1}^3 t_k^{(1)}\bm{g}_k^{(1)}+\left[\bm{P}_i^{(1)}(\bm{R}_i^{(1)},\overline{\bm{R}}_i^{(1)})\cdot\bm{h}_k^{(1)}\right]\lambda_k^{(1)}\bm{g}_k^{(1)}\\
\bm{R}_i^{(3)}&=\sum_{k=1}^3 t_k^{(2)}\bm{g}_k^{(2)}+\left[\bm{P}_i^{(2)}(\bm{R}_i^{(2)},\overline{\bm{R}}_i^{(2)})\cdot\bm{h}_k^{(2)}\right]\lambda_k^{(2)}\bm{g}_k^{(2)}\\
&\vdots\\
\bm{R}_i^{(N)}&=\sum_{k=1}^3 t_k^{(N-1)}\bm{g}_k^{(N-1)}+\\&\left[\bm{P}_i^{(N-1)}(\bm{R}_i^{(N-1)},\overline{\bm{R}}_i^{(N-1)})\cdot\bm{h}_k^{(N-1)}\right]\lambda_k^{(N-1)}\bm{g}_k^{(N-1)}
\end{array}\right.
\end{array}
\end{equation}
where $\bm{P}_i^{(n)}$ and $\bm{Q}_i^{(n+1)}$ are respectively defined in Eqs.~\eqref{eq:P} and~\eqref{eq:Q} for $n=1,\ldots N-1$.

Due to the non-linearity posed by the complicated dependence of the vectors $\bm{P}_i$ and $\bm{Q}_i$ on the Bloch vectors, there exist many different solutions for the system~\eqref{eq:systemmultistep}. This is clearly noticed when we use Matlab (function \lstinline|fsolve|) to search for a solution in particular cases. In order to run the numerical solver, it is necessary to provide a initial guess of what the solution is, and we found that by varying the choice of this starting point the algorithm converges to different solutions (or does not converge at all). In general, different solutions for the same problem lead to different values of fidelity, and sometimes we end up converging to solutions that give fidelities smaller than the optimal single-step tracking fidelity! Nevertheless, by varying the initial conditions, we have always been able to find solutions which are at least no worse than optimally correcting only at the end.

\subsubsection{Numerical solution for some two-step cases}

In this section we look at the numerical solution of the system~\eqref{eq:systemmultistep} for $N=2$ and a control task of stabilizing (with uniform priorities $\pi_1=\pi_2=0.5$) a pair of pure qubit states lying on the XZ plane of the Bloch sphere and straddling its equator by an angle $\pm \pi/4$ [cf. Eqs.~\eqref{eq:State1} and~\eqref{eq:State2}]. The noise in between the two corrections is taken to be a diagonal non-unital extreme point of the set of CPTP maps. More specifically, it compresses the Bloch sphere by $\lambda_1$, $\lambda_2$ and $\lambda_3$ along the $x$-, $y$- and $z$- directions, respectively  ($1$ representing no compression, and $0$ representing full compression), and translates it by $t_3$ along the $z$-axis. The extremal character is imposed by choosing $\lambda_3=\lambda_1\lambda_2$ and $t_3=\sqrt{(1-\lambda_1^2)(1-\lambda_2^2)}$ [cf. Eqs.~\eqref{eq:ruskai_extreme_points}]. With this choice, the noise can be characterized with only two parameters. Each one of the three plots in Fig.~\ref{fig:systemsolution} presents the solution of the system with respect to the choices of parameters: $\lambda_1\times t_3$, $\lambda_2\times t_3$ and $\lambda_3\times t_3$. For the presented results, we provided an initial condition corresponding to a sequence of ``do-nothings operations''.

\begin{figure*}[h!]
\centering
\subfigure[\hspace{1.5mm} ] {
    \includegraphics[width=7.5cm]{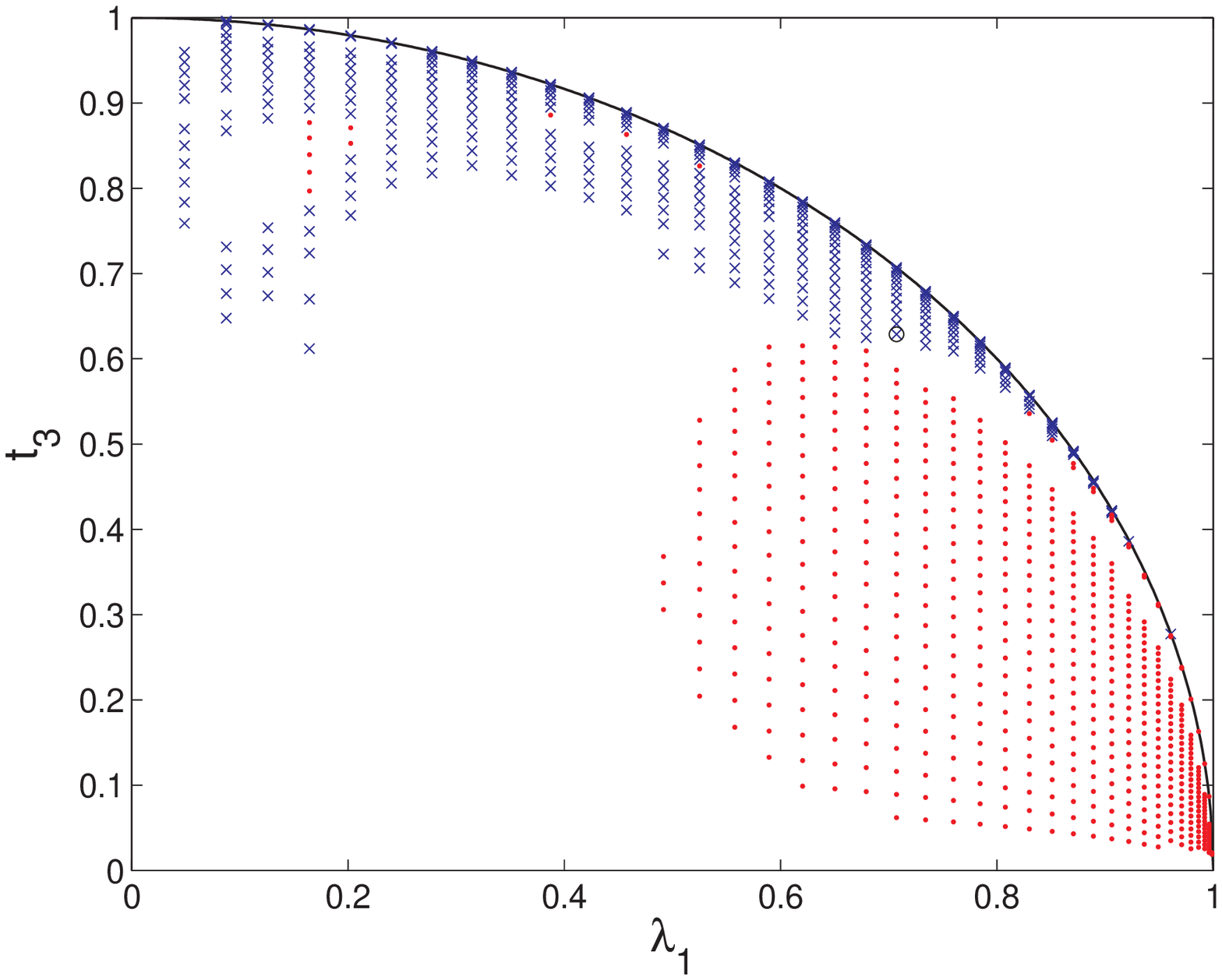}
}
\subfigure[\hspace{1.5mm}  ] {
    \includegraphics[width=7.5cm]{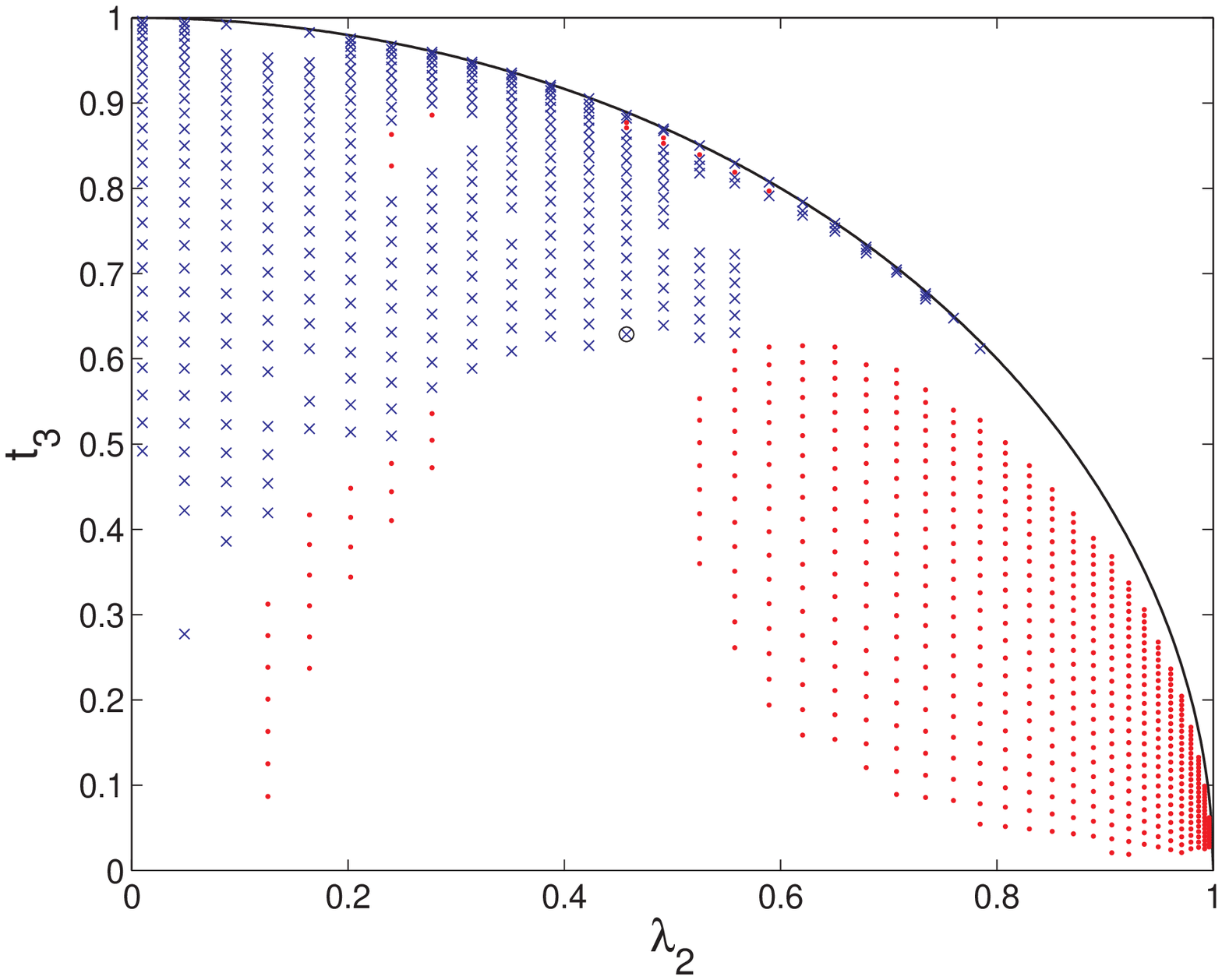}
}
\subfigure[\hspace{1.5mm}  ] {
    \includegraphics[width=7.5cm]{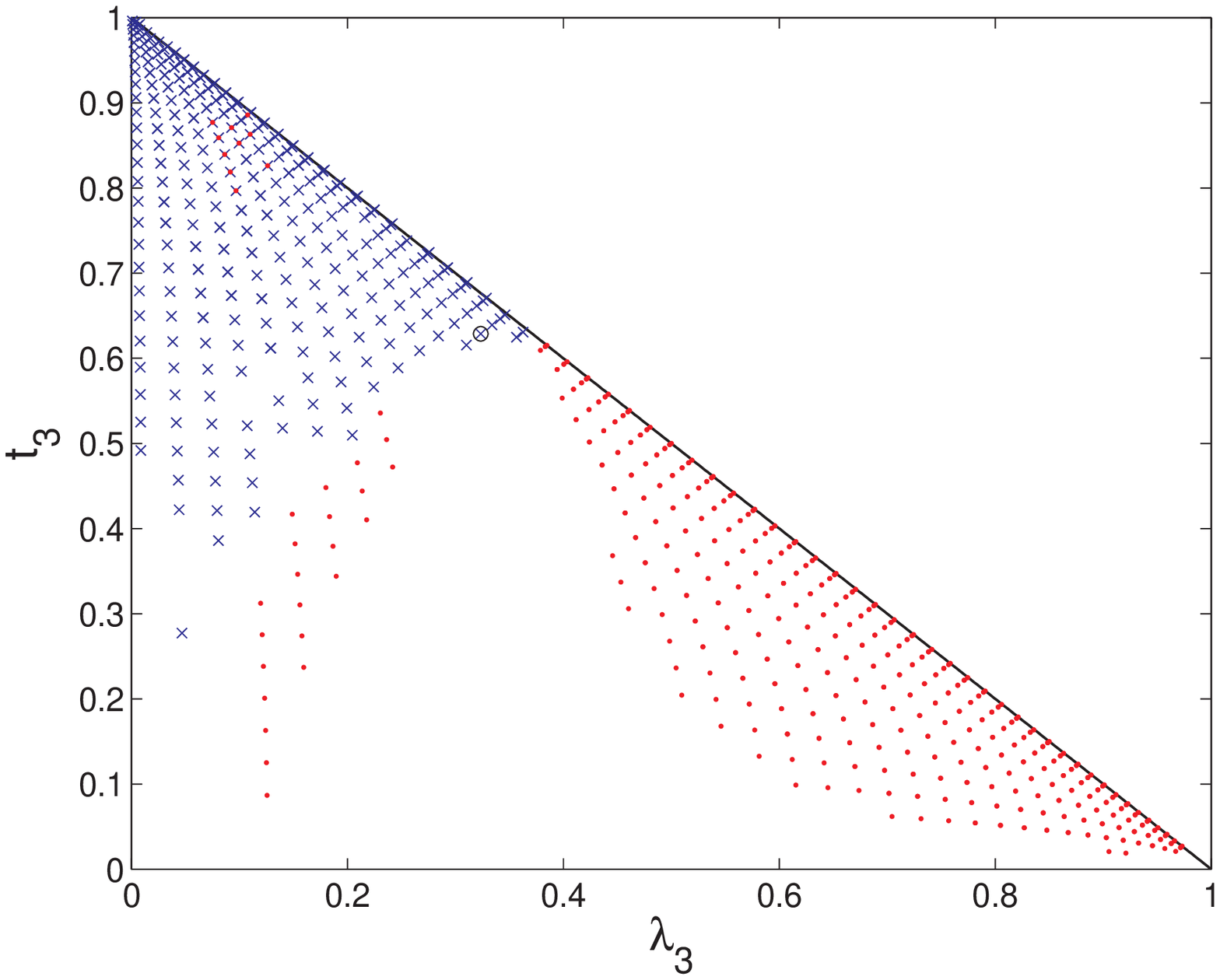}
}
\begin{picture}(0,0)
\put(-185,337){\footnotesize{No physical}}
\put(-155,327){\footnotesize{noise}}
\put(33,337){\footnotesize{No physical}}
\put(63,327){\footnotesize{noise}}
\put(-90,107){\footnotesize{No physical}}
\put(-75,87){\footnotesize{noise}}
\end{picture}

\caption[Two-steps stabilization for a pair of qubits states under extremal non-unital noise.]{(Color online) Two-steps stabilization for a pair of qubits states under extremal noise. The crosses (in blue) indicate the values of the noise parameters for which the solution of system~\eqref{eq:systemmultistep} converged to give two quantum operations producing a higher fidelity than in the case of a single optimal correction at the end. The dots (in red) illustrate a weakness of the method: for the corresponding noises, the resulting corrections turned out to give a smaller fidelity than that of a single optimal correction at the end. For the remaining empty regions (within the domain of physical noises), the solution of the system converged to give the same fidelity as that obtained by correcting only at the end. The circled crosses in each plot indicate the values of the noise parameter for which the highest advantage of approximately $10\%$ was obtained with respect to the optimal performance of a single correction at the end.}
\label{fig:systemsolution}
\end{figure*}

Each point in the plots represents a numerical solution of system~\eqref{eq:systemmultistep} for the values of noise parameters indicated in the axis. The crosses (in blue) indicate those cases for which the possibility of using a correction before the action of the noise was found to be advantageous. Although not visible from the plots, the advantage of the $2$-step schemes over the optimal $1$-step schemes was found to vary up to a maximum of approximately 10\%, which occurred at the point marked in each plot with a circle ($\lambda_{1}=0.70$, $\lambda_2=0.46$, $\lambda_3=0.32$ and $t_3=0.63$). It is a general observation (for which we do not have a satisfactory explanation) that we only produce 2-step strategies that over-perform the optimal 1-step strategy for intermediate to large values of the translation $t_3$.

The dots (in red) represent those noises for which the numerical algorithm converged to operations giving fidelities smaller than those given by the optimal single-step scheme. Of course, this simply reflects the mentioned weakness of the method of converging to suboptimal strategies. In this case, this could be remedied by repeatedly running the numerics with different initial conditions until a fidelity higher than (or equal to) that of the optimal single-step scheme was achieved. Finally, the empty regions correspond to the noises where the convergence occurred to operations that recover the optimal single-step case, i.e., to do nothing at the first step and implement the optimal single-step tracking strategy in the second step.

By looking at the details of the operations found for each noise, we found that we generally obtain correction schemes composed by two unitaries or a unitary in the first step and a non-unitary in the second step\footnote{As a matter of fact, sometimes we obtain correction schemes formed by measurements in the two steps, but by varying the initial condition we have always been able to find an alternative strategy with a unitary in the first step that over-performed the strategy with two non-unitaries. Once again, this illustrates that our method can converge to suboptimal solutions.}. Although we cannot guarantee that the \emph{optimal} multi-step scheme would not be composed of two non-unitary corrections, the fact that this is never obtained from our approach suggests that the addition of any extra-noise before the last step is generally prejudicial to stabilize the states of a qubit.

Similar conclusions were drawn by Gregoratti and Werner~\cite{04Gregoratti2600,03Gregoratti915}, who considered error correction strategies that aimed at stabilizing the \emph{entire} Bloch sphere and, different from our case, allowed the possibility of making measurements on the environment that induced the noise. In the case of qubits, they \emph{proved} that the best strategy was to correct just at the end.

In our case, this is not quite true. Since we only require  stabilization of a pair of states, we find that sometimes it is advantageous to rotate these states along the Bloch sphere and thus make them less susceptible to the noise to come. This is precisely what happens in the case of the blue crosses in Fig.~\ref{fig:systemsolution}. However, as in Ref.~\cite{04Gregoratti2600}, it seems to be better to leave any measurements to the last step. This hypothesis was also tested and verified for many other noises and pairs of qubit states, but further investigation of this matter is necessary before a general claim can be made.

\chapter{Conclusion}\label{chap:conclusion}

Motivated by a long list of successes in the classical framework, the use of feedback control in quantum systems is a promising direction for the development of new quantum technologies. However, distinguishing quantum and classical feedback is the fundamental fact that quantum measurements intrinsically disturb the system being measured. As a consequence, a naive (or classically inspired) use of feedback in quantum systems may contradict one's intentions of gaining control, and instead lead to the addition of copious amounts of noise. In this thesis, we investigated optimal ways of measuring finite dimension quantum systems, in such a way that the balance between information gain and back-action noise can be made favorable for the application of feedback control in the quantum domain. In what follows, we summarize our main results and outline some directions for future work.

The specific control problem we focused on was that of inter-converting between sequences of density matrices. This problem subsumes many situations of practical interest where one attempts to drive the dynamics of an initially unknown state: while the input (source) sequence models the initial uncertainty, the output (target) sequence models the states one would like to obtain conditioned on the identity of the initial state. Tasks such as optimal quantum state-discrimination, state-dependent quantum cloning and quantum error correction can all be formulated in terms of this problem.

In order to guide the design of optimal feedback schemes and quantify their merits, we started in Ch.~\ref{chap:dist_measure} studying distance measures for the space of density matrices. Our contribution to this topic was the proposal of an alternative definition of fidelity between mixed states. One of the most appealing properties of our ``new fidelity'' is that it is significantly easier to compute than the Uhlmann-Jozsa fidelity. In fact, it only requires the computation of traces of some products between density matrices, whereas the traditional fidelity generally involves a more expensive matrix diagonalization procedure. In addition, our new fidelity satisfies all of Jozsa's axioms, gives rise to a metric and is jointly concave. An important byproduct of the joint concavity of our fidelity was the establishment of the joint concavity property of Uhlmann-Jozsa fidelity \emph{in the case of qubit states}, settling an open problem in the field. In this chapter, we have also reviewed some known metrics on the space of density matrices, and showed how metrics for the space of \emph{sequences} of density matrices can be built from those.

Equipped with the above provisions, in Ch.~\ref{chap:assemble} we set out assembling our control problem as a particular type of convex optimization called semidefinite programs. This is a well-studied class of optimization problems for which efficient numerical methods exist to determine the optimal solutions. After showing some algebraic tricks that allowed the minimization (maximization) of several distance (closeness) measures to be written as SDPs, we exploited these methods to obtain optimal controllers. Based on these results, we investigated the sensitivity of the controllers with the choice of distance measure. On top of identifying some cases where the same controller optimizes several measures, we developed and tested a method to estimate the ``compatibility'' between optimal controllers optimizing different distance measures. In this regard, there are opportunities for refinement and future work; for example, we should enlarge the numerical samples used to compute the compatibilities and check whether the same conclusions still apply.

Our analytical investigations started in Ch.~\ref{chap:aggie}, where we considered the problem of stabilizing the state of a single qubit prepared in one of two non-orthogonal states undergoing dephasing noise. We proposed two different types of feedback strategies to approach this problem: The first was based on the classical concept of exploiting the measurement to discriminate between the two possible initial preparations, and then follow with a suitable repreparation of the system. The second was based on the idea that non-orthogonal states are fundamentally indistinguishable. As opposed to attempting to discriminate the two states, we used a quantum measurement to learn about how the noise affected the system, and fedback to counter-act the noise. We proved that our classical and quantum strategies were optimal, in the sense that no other entanglement-breaking-trace-preserving or completely-positive-trace-preserving maps, respectively, could produce a higher fidelity for the stabilization task. We have also proposed an (arguably optimal) stochastic discriminate-and-reprepare scheme, which admitted the possibility of occasional inconclusive results in the discrimination step. Quite surprisingly, this scheme produced the same optimal performance as the deterministic quantum scheme. This observation motivates some future research to find whether this is just a peculiarity of the problem considered here (two qubit states, dephasing noise), or else if in more general circumstances we can still achieve the optimal performance of quantum schemes by using stochastic discriminate-and-reprepare strategies.

In Ch.~\ref{chap:tracking} the analytical results from the previous chapter were significantly extended. We introduced an optimal strategy for optimally transforming the state of a single qubit into a given target state, when the system can be prepared in two different ways, and the target state depends on the choice of preparation. This generalizes the results of Ch.~\ref{chap:aggie} in two points: Because we now allow the source qubit states to be arbitrarily chosen, they can be regarded as the outputs of an arbitrary noise channel. Moreover, due to the arbitrariness of the target states, we were able not only to optimally \emph{stabilize} the states, but also to optimally \emph{track} them while competing against the noise dynamics. For this quantum tracking task, we found that feedback control is not always useful, and sometimes it is actually better not to measure the system at all. In these cases --- which were flagged by an analytic indicator function  --- the application of a unitary map after the noise was found to be the best strategy. Otherwise, a closed loop scheme giving optimal measurement strengths and feedback was analytically constructed. Several applications for our quantum tracking strategy were discussed in the framework of quantum information. Attesting the quality of the scheme, some optimal strategies for quantum state discrimination, purification, error correction and state-dependent cloning were recovered and extended.

Chapter~\ref{chap:multistep} concluded our scientific results with the proposal (and some preliminary numerical solutions) of a variation of the tracking problem studied in Ch.~\ref{chap:tracking}. Here, we allowed controlled interventions not only \emph{after} the system was exposed to the noise, but also \emph{before} and \emph{during} the action of the noise. In general, it is not clear whether the use of feedback at earlier stages of the dynamics is beneficial for the tracking goal, since the information gain could not compensate the back-action noise imparted on the system by multiple uses of quantum measurements. To address this problem, we developed a method (based on dynamic programming and on a heuristic argument) to derive suboptimal multi-step schemes. By applying our method to some examples involving $2$-step stabilization of a pair of qubit states, we noticed that we can usually over-perform the optimal single-step scheme. In addition, we found that our $2$-step schemes were always composed by some unitary map in the first step and (sometimes) a weak measurement in the second step, suggesting that the observation of quantum systems along the way may not be advisable. The confirmation of this hypothesis, though, is still an open problem for future research.

Throughout this thesis, we studied \emph{optimal} strategies for controlling quantum system. Optimality, however, is only possible when the details of the control task are known within a certain (high) level of accuracy, which for practical applications may not be realistic. For example, the noise model affecting the system has always been assumed to be perfectly known, giving rise to a sequence of source states onto which we relied for the construction of our control strategies. In practice, quantum process tomography is subjected to errors, and a quantum channel very accurately characterized at a certain point in time may change its properties when the system is actually running. Therefore, an important extension of this work would come from the introduction of uncertainties for the source sequence, as a result of uncertainties in the noise model that gives rise to it. The determination of control strategies in the presence of uncertainties is part of what is called \emph{robust control}.

In the context of our quantum tracking scenario, we envisage the following problem: Suppose that, as a result of a precarious characterization of the noise model, our knowledge of the source sequence is imperfect, which is modeled by a \emph{nominal} sequence and many equally possible source sequences, as illustrated in Fig.~\ref{fig:robust_intro}.
\begin{figure}[h]
\centering
\includegraphics[width=10cm]{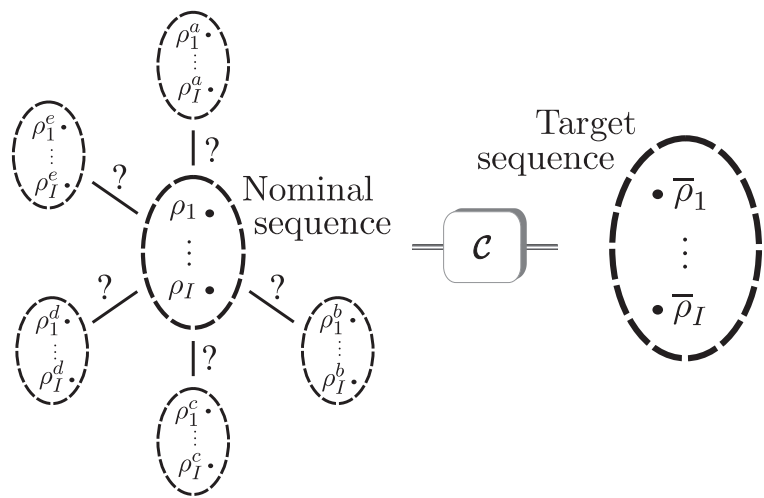}
\caption{Schematic of the robust tracking problem.}\label{fig:robust_intro}
\end{figure}
We would like to determine a quantum controller that guarantees that the minimal performance in approximating a given target sequence is above some pre-established threshold. In particular, we would like to know how high this threshold can be made while still physically achievable. The formulation of some instances of this problem as SDPs has already been accomplished, and it is our intention to further explore this topic both in the numerical and analytical frameworks.

\appendix

\chapter{Appendices to Chapter 2}\label{app:ch2}

\section{The Permutation matrix $P_{{\rm d}^4}$}\label{app:vecperm}
In this appendix we give an explicit construct of the permutation matrix $P_{{\rm d}^4}$ that establishes the equality
\begin{equation}\label{eq:def_Pd}
{\rm vec}\left( A\otimes B\right)=P_{{\rm d}^4}\left({\rm vec} A\otimes {\rm vec} B\right)\,,
\end{equation}
for $A$ and $B$ arbitrary $d\times d$ matrices. In this thesis, the need for this matrix arises in Sec.~\ref{sec:choi_composed}, where we construct the Choi matrix of a CP map formed from the composition of two CP maps.

In order to determine $P_{{\rm d}^4}$, we first numerically solved Eq.~\eqref{eq:def_Pd} with fixed matrices $A$ and $B$ of dimension ${\rm d}=2$, $3$ and $4$ (notice that $P_{{\rm d}^4}$ should not depend on the particular choice of $A$ and $B$). The resulting permutation matrices are shown in Fig.~\ref{fig:permutation}.
\begin{figure}[h]
\centering
\subfigure[${\rm d}=2$]
{
    \label{fig:sub:a}
    \includegraphics[width=4cm]{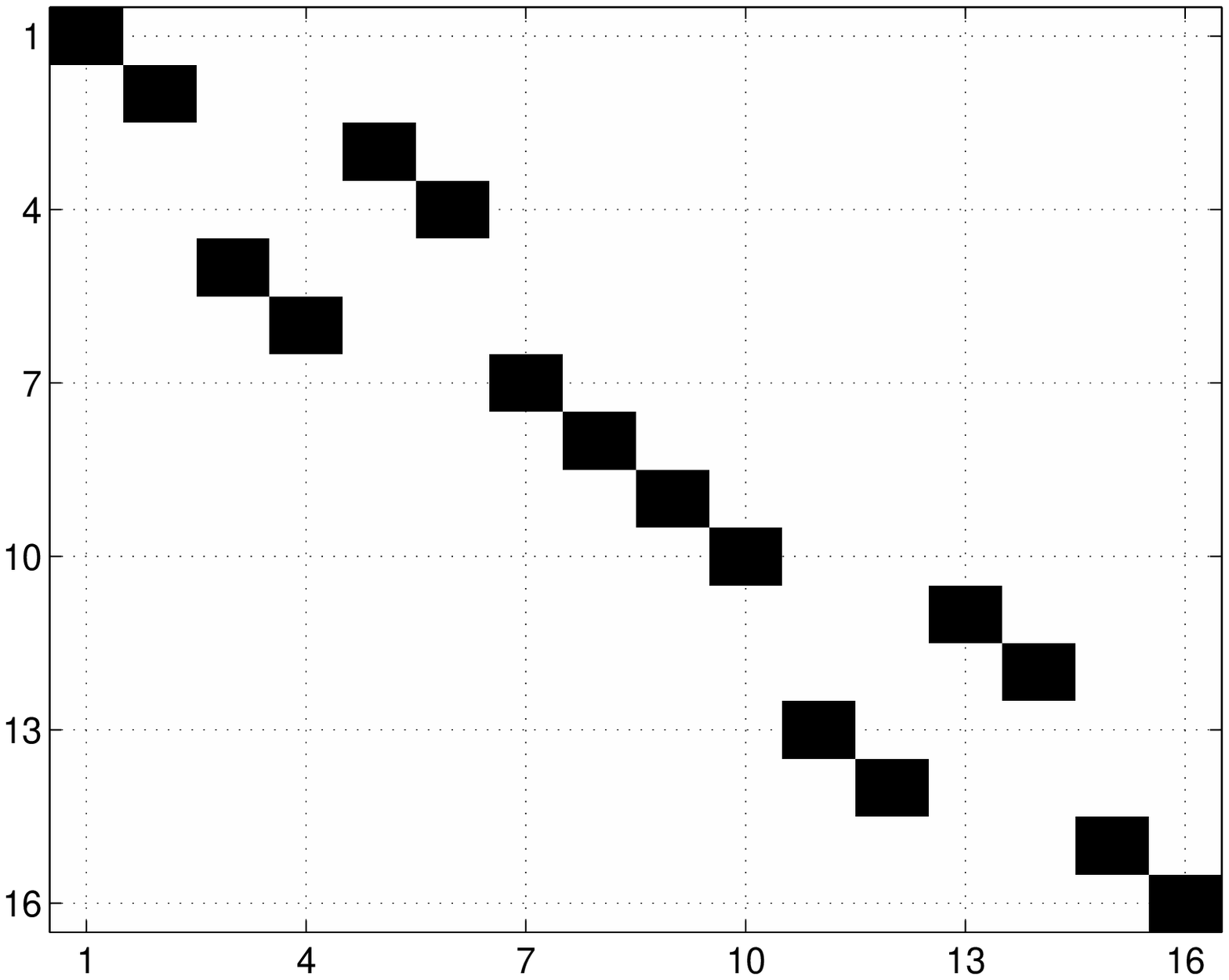}
}
\hspace{1cm}
\subfigure[${\rm d}=3$]
{
    \label{fig:sub:b}
    \includegraphics[width=4cm]{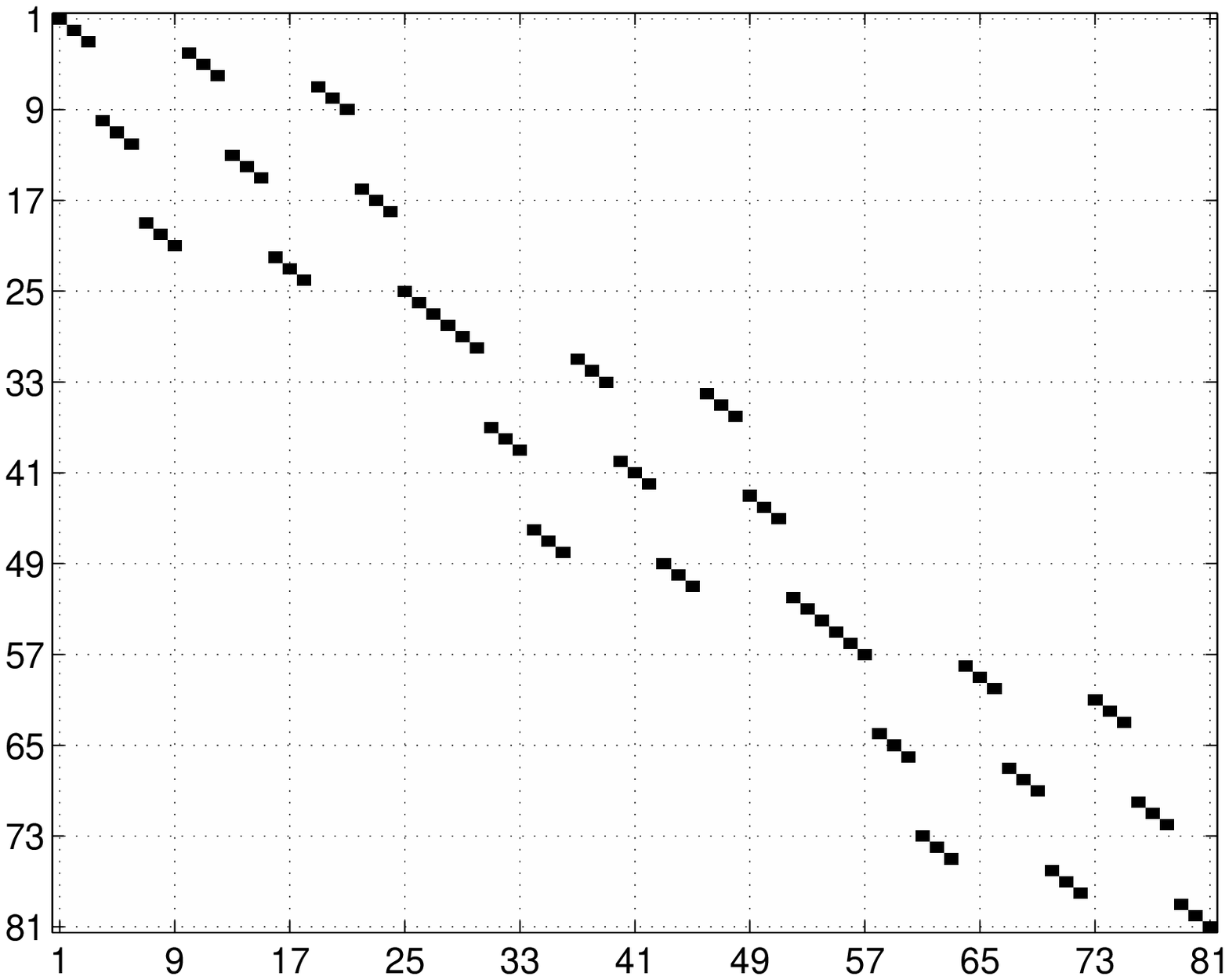}
}
\hspace{1cm}
\subfigure[${\rm d}=4$]
{
    \label{fig:sub:c}
    \includegraphics[width=4cm]{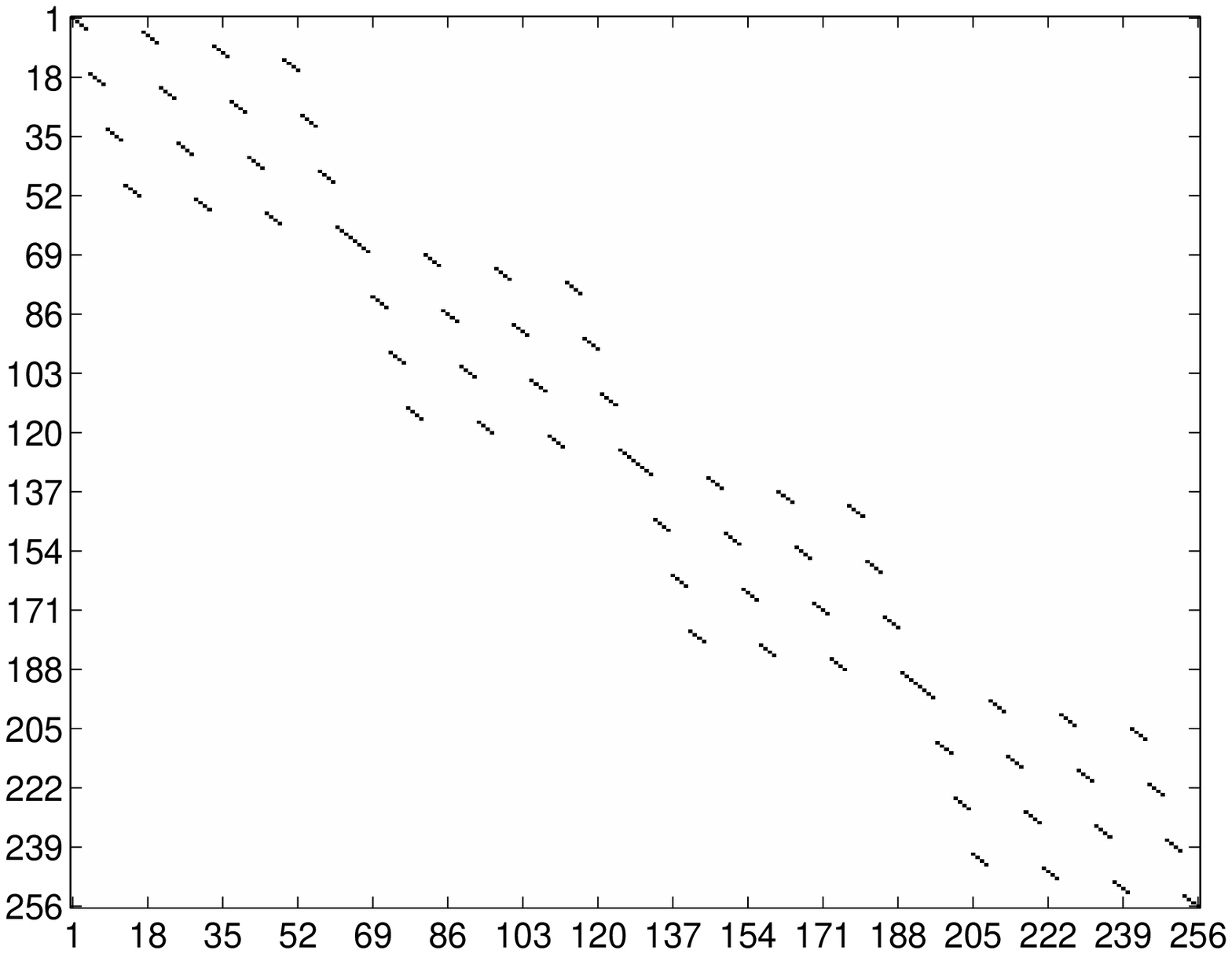}
}
\caption[The permutation matrix $P_{{\rm d}^4}$.]{The permutation matrix $P_{{\rm d}^4}$ of Eq.~\eqref{eq:def_Pd} for ${\rm d}= 2$, $3$ and $4$. The black marks correspond to $1$'s, the remaining spaces are filled with $0$'s.}
\label{fig:permutation}
\end{figure}

The form of $P_{{\rm d}^4}$ for arbitrary values of {\rm d} can be inferred from the instances shown in Fig.~\ref{fig:permutation}. For example, $P_{{\rm d}^4}$ is easily recognized as a block diagonal matrix with {\rm d} identical blocks of dimension ${\rm d}^3$. In addition, after a careful analysis of the generated pattern, Eq.~\eqref{eq:pattern1} was identified as the correct formula specifying the position $(i,j)$ of the unit elements in the first block of $P_{\rm d}$, \begin{equation}\label{eq:pattern1}
j=i+\left\lfloor\frac{i-1}{\rm d}\right\rfloor {\rm d}({\rm d}-1)-\left\lfloor\frac{i-1}{{\rm d}^2}\right\rfloor {\rm d}({\rm d}^2-1)\,,
\end{equation}
where $i=1,\ldots,{\rm d}^3$ and $\left\lfloor\cdot\right\rfloor$ denotes the floor function.

Finally, let us mention that $P_{{\rm d}^4}$ is not the same as the so-called ``vec-permutation matrix'', which is defined as the (unique) ${\rm d}^2\times {\rm d}^2$ permutation matrix $P$ such that ${\rm vec}A^{\sf T}=P{\rm vec}A$.

\chapter{Appendices to Chapter 3}\label{app:ch3}

Except for Appendix~\ref{app:metricHsq}, all the appendices presented here appear in Ref.~\cite{08Mendonca1150}. Appendix~\ref{app:codes} has been modified to include Matlab codes for the computation of the metrics $\H$ and $\O$, which were not considered in Ref.~\cite{08Mendonca1150}.

\section{Proof of Proposition \ref{prop:concFx}}\label{app:concFx}
In this appendix the joint concavity of $\F_N$ is established
via the proof of Proposition \ref{prop:concFx}.

\begin{proof}
Differentiating Eq.~\eqref{eq:Fx} twice with respect to $x$, we obtain
\begin{equation}\label{eq:Fder2}
\frac{d^{2}F\left(x\right)}{dx^{2}}=2\bm{u}\cdot\bm{v}+\frac{d^{2}f\left(x\right)}
{dx^{2}}g\left(x\right)+f\left(x\right)\frac{d^{2}g\left(x\right)}{dx^{2}} +2\frac{df\left(x\right)}{dx}\frac{dg\left(x\right)}{dx}
\end{equation}
where, for convenience, we define the functions
$f(x) \defeq \sqrt{1 - \|\bm{r} + x\bm{u}\|^2}$ and
$g(x) \defeq \sqrt{1 - \|\bm{s} + x\bm{v}\|^2}$.

After some computation we find that
\begin{equation}
    \frac{d^{2}F\left(x\right)}{dx^{2}}  = \mathfrak{F}_1(x) +\mathfrak{F}_2(x),
\end{equation}
where
\begin{align}
    \mathfrak{F}_1(x)&\defeq 2\bm{u}\cdot\bm{v}-\frac{g\left(x\right)u^{2}}{f\left(x\right)}
    -\frac{f\left(x\right)v^{2}}{g\left(x\right)}\,,\\
    \mathfrak{F}_2(x)&\defeq 2\frac{\bm{u}\cdot\left(\bm{r}+x\bm{u}\right)\bm{v}
    \cdot\left(\bm{s}+x\bm{v}\right)}{f\left(x\right)g\left(x\right)}
    -\frac{g\left(x\right)\left[\bm{u}\cdot\left(\bm{r}+x\bm{u}\right)\right]^{2}}
    {\left[f\left(x\right)\right]^{3}}
    -\frac{f\left(x\right)\left[\bm{v}\cdot\left(\bm{s}+x\bm{v}\right)\right]^{2}}
    {\left[g\left(x\right)\right]^{3}}\,.
\end{align}
The negative semidefiniteness of $d^2 F(x) /d x^2$ in the range $x \in
[0,1]$ can be observed if $\mathfrak{F}_1(x)$ and $\mathfrak{F}_2(x)
$ are written in the following alternative form:
\begin{align}
    \mathfrak{F}_1(x)&=-\left\|\sqrt{\frac{g\left(x\right)}{f\left(x\right)}}\bm{u}
    -\sqrt{\frac{f\left(x\right)}{g\left(x\right)}}\bm{v}\right\|^{2}\,,\\
    \mathfrak{F}_2(x)&=-\,\frac{1}{f\left(x\right)g\left(x\right)} \left[\frac{g\left(x\right)}{f\left(x\right)}\bm{u}\boldsymbol{\cdot}\left(\bm{r}+x\bm{u}\right)
    -\frac{f\left(x\right)}{g\left(x\right)}\bm{v}\boldsymbol{\cdot}\left(\bm{s}+x\bm{v}\right)\right]^{2}\,.
\end{align}
\end{proof}

\section{Proof of Super-multiplicativity of $\Fn$}
\label{app:super-multplicativity}

To prove that $\Fn$ is super-multiplicative, we first define $r_i\defeq
\tr\rho_i^2$ and $s_i
\defeq \tr\sigma_i^2$, such that $0 < r_i, s_i \leq 1$ (note that
here we use $r_i$ instead of $r_i^2$ as the norm square of $\bm{r}_i$,
likewise for $s_i$). Straightforward algebra gives
\begin{gather*}
\F_N(\rho_1\otimes\rho_2,\sigma_1\otimes\sigma_2)
- \F_N(\rho_1,\sigma_1)\F_N(\rho_2,\sigma_2)=\\
\sqrt{(1-r_1 r_2)(1-s_1 s_2)}
-\sqrt{(1-r_1)(1-s_1)(1-r_2)(1-s_2)}\\
-\tr\left[\rho_1\sigma_1\right]\sqrt{(1-r_2)(1-s_2)}
-\tr\left[\rho_2\sigma_2\right]\sqrt{(1-r_1)(1-s_1)}\,.
\end{gather*}
A direct application of Cauchy-Schwarz's inequality $\tr\left[\rho_i\sigma_i\right]\leq \sqrt{r_i s_i}$ gives
\begin{gather*}
    \F_N(\rho_1\otimes\rho_2,\sigma_1\otimes\sigma_2)- \F_N(\rho_1,\sigma_1)\F_N(\rho_2,\sigma_2)\geq\\
    \sqrt{(1-r_1 r_2)(1-s_1 s_2)}
    -\sqrt{(1-r_1)(1-s_1)(1-r_2)(1-s_2)}\\
    -\sqrt{r_1 s_1(1-r_2)(1-s_2)}-\sqrt{r_2 s_2(1-r_1)(1-s_1)}\,.
\end{gather*}

The super-multiplicative property is obtained by showing the positive
semi-definiteness of the rhs of the above expression. This is the content of the following proposition:
\begin{proposition} For $0\leq a,b,c,d \leq 1$, we have
\begin{equation}\label{eq:supermult}
\sqrt{(1-a b)(1-c d)}\geq\sqrt{(1-a)(1-b)(1-c)(1-d)}+\sqrt{a c
(1-b)(1-d)}+\sqrt{b d (1-a)(1-c)}\,.
\end{equation}
\end{proposition}
\begin{proof}
First note that if any of the variables equals $1$, then the validity
of the inequality is immediate. For example, let $d=1$ so that
(\ref{eq:supermult}) reduces to
\begin{equation}
\sqrt{(1-a b)(1-c)}\geq\sqrt{b (1-a)(1-c)}\,.
\end{equation}
This is trivially satisfied for all $0\leq a,b,c \leq 1$. In what
follows, we restrict to $0\leq a,b,c,d < 1$ and show that  inequality
(\ref{eq:supermult}) is equivalent to the standard inequality of
arithmetic and geometric means (hereafter referred as the AM-GM
inequality). This inequality is just an expression of the fact that the
geometric mean of a list of non-negative real numbers is never larger
than the corresponding arithmetic mean.\par Apply the substitution
$a^\prime=1-a$ (similarly for $b^\prime$, $c^\prime$ and $d^\prime$;
note that $0< a^\prime,b^\prime,c^\prime,d^\prime \leq 1$) to the
inequality (\ref{eq:supermult}) and divide the result by
$\sqrt{a^\prime b^\prime c^\prime d^\prime}$ to get the equivalent
inequality
\begin{equation}\label{eq:ineqsupermult}
\sqrt{(1+A+B)(1+C+D)}\geq 1+\sqrt{AC}+\sqrt{BD}\,,
\end{equation}
where we have defined $A=\frac{1}{a^\prime}-1$ (similarly for $B$, $C$
and $D$; note that $0\leq A,B,C,D < \infty$).  Squaring the inequality
above we find
\begin{equation}
    \frac{A\!+\!C}{2}+\frac{B\!+\!D}{2}+\frac{AD\!+\!BC}{2} \geq \sqrt{A C}\!+\! \sqrt{B D}\!+\! \sqrt{A B C D}
\end{equation}
which is clearly a sum of three AM-GM inequalities.
\end{proof}

\section{Proof of the Metric Property of $B[\F]$ and
$C[\F]$}\label{app:metric_C}

In the following, we give a new demonstration of the metric properties
of  $B[\F]$ and $C[\F]$ (see
Refs.~\cite{69Bures199,05Gilchrist062310} for the standard proofs). Our
proof consists of a simple application of Theorem~\ref{thm:schoenberg}
due to Schoenberg.

\begin{proposition}\label{prop:CFmetric}
The quantities $B[\F(\rho,\sigma)]$ and
$C[\F(\rho,\sigma)]$ given in Eqs. (\ref{eq:BF}) and
(\ref{eq:CF}),  respectively, are metrics for the space of density matrices.
\end{proposition}

\begin{proof}
For brevity, let $K[\F(\rho,\sigma)]$ represent either
$B[\F(\rho,\sigma)]$ or $C[\F(\rho,\sigma)]$.  As with $\F$, it is easy
to check that $K^2[\F(\rho,\sigma)]$ is symmetric in its two arguments,
and that $K^2[\F(\rho,\sigma)]\geq 0$ with saturation iff
$\rho=\sigma$. So, according to Theorem~\ref{thm:schoenberg},
$K[\F(\rho,\sigma)]$ is a metric if for any set of density matrices
$\{\rho_i\}_{i=1}^n$ ($n\geq 2$) and real numbers $\{c_i\}_{i=1}^n$
such that $\sum_{i=1}^n c_i=0$, it is true that
\begin{equation}\label{eq:toprove}
\sum_{i,j=1}^n K^2[\F(\rho_i,\rho_j)] c_i c_j \leq 0\,.
\end{equation}
To prove this, we derive an upper bound for $K^2[\F(\rho_i,\rho_j)]$
that can be easily seen to satisfy the condition above. First, note that
\begin{multline}\label{eq:seqofineqs}
\F(\rho_i,\rho_j)=\left(\tr|\sqrt{\rho_i}\sqrt{\rho_j}|\right)^2\geq|\tr\sqrt{\rho_i}\sqrt{\rho_j}|^2 =  \\\left(\tr\sqrt{\rho_i}\sqrt{\rho_j}\right)^2=\tr\left[\sqrt{\rho_i}\sqrt{\rho_j}\otimes\sqrt{\rho_i}\sqrt{\rho_j}\right]=\\
\tr\left[\left(\sqrt{\rho_i}\otimes\sqrt{\rho_i}\right)\left(\sqrt{\rho_j}
\otimes\sqrt{\rho_j}\right)\right]\equiv\EuScript{A}(\rho_i,\rho_j)\,,
\end{multline}
where the first equality follows from the definition
$|A|\defeq\sqrt{A^\dagger A}$ for every matrix $A$ and  the inequality
from the fact that $\tr|A|=\max_{U}|\tr UA|$ (the
maximization runs over unitary matrices $U$
\cite{94Jozsa2315,60Schatten}). Then, it follows that
\begin{align}
B^2[\F]&=2\left(1-\sqrt{\F}\right)\leq 2\left(1-\F\right)\leq 2\left(1-\EuScript{A}\right)\,,\\
C^2[\F]&=1-\F\leq 1-\EuScript{A} \leq 2\left(1-\EuScript{A}\right)\,,
\end{align}
or, in our more compact notation, $K^2[\F]\leq
2\left(1-\EuScript{A}\right)$.\par

Now, replacing $K^2[\F(\rho_i,\rho_j)]$ with the above upper bound in
the lhs of Eq. (\ref{eq:toprove}), it is easy to obtain the desired inequality:
\begin{equation}\label{eq:schoenbergA}
    \sum_{i,j=1}^n \left\{2-2\tr\left[\left(\sqrt{\rho_i}\otimes\sqrt{\rho_i}\right)\left(\sqrt{\rho_j}
    \otimes\sqrt{\rho_j}\right)\right]\right\}
    c_i c_j=-2\tr\left|\sum_{i=1}^n
    c_i\sqrt{\rho_i}\otimes\sqrt{\rho_i}\right|^2\leq 0\,,
\end{equation}
where the equality is obtained by using that $\sum_{i=1}^n c_i =0$, the
linearity of the trace operation and the hermiticity  of
$c_i\sqrt{\rho_i}\otimes\sqrt{\rho_i}$.\par
\end{proof}

Finally, let us just mention that besides establishing the metric
properties of $B[\F]$ and $C[\F]$, the present proof  also establishes
$\sqrt{2-2\left(\tr\sqrt{\rho}\sqrt{\sigma}\right)^2}$ as a metric
for the space of density matrices. In fact, by a similar application of
Schoenberg's theorem, the quantity $H(\rho,\sigma)\defeq\sqrt{2-2\tr\sqrt{\rho}\sqrt{\sigma}}$ can also be shown to be a metric.

\section{Proof of the Metric Property of $\H^2$}\label{app:metricHsq}
In Sec.~\ref{sec:HSD} we defined the Hilbert-Schmidt distance $\H$ between two density matrices $\rho$ and $\sigma$ as the Hilbert-Schmdit norm of the matrix $\rho-\sigma$. As shown on page~\pageref{proof:inducedmetric}, such a definition guarantees that $\H$ is a metric for the space of density matrices. In this appendix, we demonstrate that the function
\begin{equation}
\H^2(\rho,\sigma)=\tr\left[\left(\rho-\sigma\right)^2\right]\,,
\end{equation}
also defines a metric for the space of density matrices.

First note that properties (M1)-(M3) on page~\pageref{axioms:metric} are easily verified for $\H^2$, as a consequence of holding true for $\H$. We only have to check the triangle inequality (M4), which is necessarily satisfied if
\begin{equation}\label{eq:schoenbergHsq}
\sum_{i=1}^{n}{c_i}=0 \Rightarrow \sum_{i,j=1}^{n}{\H^4(\rho_i,\rho_j) c_i c_j}\leq 0\,,
\end{equation}
for every $n\geq 2$ and $c_i\in\mathbb{R}$. The above implication arises from Schoenberg's theorem (cf. Theorem~\ref{thm:schoenberg} on page~\pageref{thm:schoenberg}).\par

In order to prove Eq.~\eqref{eq:schoenbergHsq}, consider the inequalities
\begin{equation}\label{eq:ineqsHFA}
\left(\frac{\H}{\sqrt{2}}\right)^4\leq\left(\frac{\H}{\sqrt{2}}\right)^2\leq 2\left(1-\F\right)\leq 2\left(1-\A\right)\,,
\end{equation}
where the first inequality follows from the fact that $\H$ is bounded between $0$ and $\sqrt{2}$, the second is equivalent to Eq.~\eqref{eq:fuchs_like_bounds} (second inequality in the first line) and the third was established in Eq.~\eqref{eq:seqofineqs}, where the function $\A(\rho,\sigma)$ was defined. Inequalities~\eqref{eq:ineqsHFA} imply that Eq.~\eqref{eq:schoenbergHsq} is satisfied if the implication holds when $2(1-\A)$ replaces $\H^4$. That this is the case, has already been seen in Eq.~\eqref{eq:schoenbergA}.

\section{Matlab Codes}\label{app:codes}

In this Appendix, we present the Matlab codes that we have used to compute $\Fn$, $\F$, $\D$ and $Q$ in the numerical experiment presented in Sec.~\ref{sec:computability}. For completeness, we also present similar codes for the computation of the metrics $\H$ and $\O$.

For \lstinline|rho| and \lstinline|sigma| density matrices,
\begin{itemize}
\item $\Fn$ was computed using
\lstset{frame=none}
\begin{lstlisting}
Fn = real( rho(:)'*sigma(:) + ...
           sqrt((1 - rho(:)'*rho(:))* ...
           (1 - sigma(:)'*sigma(:))) );
\end{lstlisting}

\item $\F$ was computed using
\begin{lstlisting}
[V, D] = eig(rho);
sqrtRho = V*diag(sqrt(diag(D)))*V';
F = sum( sqrt(eig(Hermitize(sqrtRho*sigma*sqrtRho))) )^2;
\end{lstlisting}
Here \lstinline|sqrtRho*sigma*sqrtRho| is not quite Hermitian due to small numerical
errors.  We therefore employ the function \lstinline|Hermitize(M)=(M+M')/2| to turn
the almost-Hermitian matrix into a Hermitian one --- this causes Matlab to
select a more efficient algorithm for the diagonalization.

\item $\D$ was computed using\label{code:D}
\begin{lstlisting}
D=0.5*sum(abs( eig(rho-sigma) ));
\end{lstlisting}

\item $Q$ was computed using
\begin{lstlisting}
[Vr,Drho]=eig(rho); Dr=diag(Drho);
[Vs,Dsigma]=eig(sigma); Ds=diag(Dsigma);
A = abs(Vr'*Vs).^2;
[x,Q]=fminbnd(@(s) (Dr.'.^s)*A*(Ds.^(1-s)), 0, 1);
\end{lstlisting}
The algorithm used here follows from the formula for $\tr(\rho^s\sigma^{1-s})$
given in the section entitled \emph{convexity in s} of Ref.~\cite{07Audenaert160501}.

\item $\H$ can be computed using
\begin{lstlisting}
H = sqrt( (rho(:)-sigma(:))'*(rho(:)-sigma(:)) );
\end{lstlisting}

\item $\O$ can be computed using
\begin{lstlisting}
O=max(abs( eig(rho-sigma) ));
\end{lstlisting}

\end{itemize}

\chapter{Appendices to Chapter 5}\label{app:ch5}

In Sec.~\ref{sec:DDR} and \ref{sec:optimality}, the proposed
deterministic classical and quantum control schemes were shown to be optimal among the set of EBTP and CPTP maps, respectively.  Here, we provide constructive proofs of these results in further detail. All the appendices presented here appear in Ref.~\cite{07Branczyk012329}.

\section{Dual optimization for deterministic quantum control}\label{app:optdetq}

As demonstrated in Sec.~\ref{sec:optimality}, obtaining the maximum average fidelity can be expressed as the optimization problem~\eqref{eq:qprob_compact}.  For this problem (as for the classical problem which we address in the next section) the dual optimization proves to be straightforward to solve analytically and the results above can then be used to show optimality of the control scheme given by Eq.~\eqref{eq:QCasCPTP}.

We make use of some symmetry arguments to simplify the problem. This optimization problem has certain symmetry properties under the action of the group of transformations generated by the rotation $\mathfrak{C}\rightarrow (X\otimes X) \mathfrak{C} (X\otimes X)^\dagger$ and the transpose $\mathfrak{C}\rightarrow \mathfrak{C}^{\sf T}$. Specifically, the objective function is invariant under the action of this group since $\tr[R (X \otimes X)\mathfrak{C} (X \otimes X)]=\tr\left(R \mathfrak{C}\right)$ and $\tr\left(R \mathfrak{C}^{\sf T}\right)=\tr\left(R \mathfrak{C}\right)$,
because $(X \otimes X) R (X \otimes X) = R$ and $R^{\sf T}=R$,
respectively.  In addition, the constraints are covariant under the action of the group: Since conjugation with a unitary and
transposition preserve eigenvalues, $(X\otimes X) \mathfrak{C} (X\otimes X) \geq 0$ and $\mathfrak{C}^{\sf T}\geq 0$ if $\mathfrak{C} \geq 0$. To see that the equality constraints are covariant note that $\tr_2\mathfrak{C} = \openone_2$ is equivalent to the condition $\tr \left[(M\otimes \openone_2) \mathfrak{C} \right]= \tr M$ for all hermitian $M$. If $\mathfrak{C}$ obeys the partial trace constraint we have
\begin{equation}
  \tr [(M\otimes \openone_2)(X \otimes X) \mathfrak{C} (X \otimes X)] = \tr [(X M X\otimes \openone_2) \mathfrak{C} ] = \tr M \,,
\end{equation}
and
\begin{equation}
    \tr [(M\otimes \openone_2)\mathfrak{C}^{\sf T} ]
    = \tr [(M^{\sf T}\otimes \openone_2) \mathfrak{C} ]
    = \tr M \,,
\end{equation}
so both $(X \otimes X) \mathfrak{C} (X \otimes X) $ and $\mathfrak{C}^{\sf T}$ do also.
So both the objective function and the feasible set of (\ref{eq:qprob_compact}) are invariant under the action of the group. As a result there will be an invariant point $\mathfrak{C}^*_{\rm inv}=(X\otimes X)\mathfrak{C}^*_{\rm inv}(X\otimes X)=\mathfrak{C}^{*{\sf T}}_{\rm inv}$ that achieves the optimum $\mathfrak{p}^*$~\cite{04Boyd}. We do not need to optimize over the full set of $\mathfrak{C}$ but may restrict our attention to the set of invariant $\mathfrak{C}_{\rm inv}$. Gatermann and Parrilo~\cite{04Gatermann95} have investigated such invariant SDP's in detail.

The dual of our optimization problem (\ref{eq:qprob_compact}) has the form~\cite{02Audenaert030302}
\begin{equation}
  \begin{array}{rl}\label{eq:ineq_form2}
  \text{minimize}& \tr M \\
  \text{subject to}&M\otimes \openone_2 - R \geq 0
  \end{array}
\end{equation}
Notice that (as is generally the case) this semidefinite program is invariant under the same group of transformations as the original problem, under which
$M\rightarrow X M X$ and
$M\rightarrow M^{\sf T}$. For the dual problem we may likewise restrict attention to $M_{\rm inv}= b_0\openone_2+b_x X$ that are invariant under the
action of the group. This gives a simpler dual optimization
\begin{equation}\label{eq:qdual}
  \begin{array}{rl}
  \text{minimize}&2 b_0\\
  \text{subject to}&b_0 \openone_4+b_x X\otimes \openone_2 - R \geq
  0 \,,
  \end{array}
\end{equation}
where $b_0$ and $b_x$ are the new variables.  This problem is simple enough to solve analytically; the solution is
\begin{equation}
  b_0 =\frac{1}{4}+\frac{1}{4}\sqrt{\cos^2{\theta}+\frac{\sin^4{\theta}}{1-r_x^2}}\,,
\end{equation}
and $b_x = r_x b_0$ [with $r_x = (1-2p)\cos\theta$].  This may be
checked by verifying that the matrix $b_0 \openone_4+b_x
X\otimes \openone_2 - R$ is indeed positive semi-definite, hence $2 b_0$ is a valid dual feasible value. Because $2 b_0$ reproduces the fidelity of our proposed scheme, given by Eq. (\ref{eq:fid_qc}), this guess necessarily gives an optimal solution to the original problem (\ref{eq:CPTPoptimisation}).

\section{Dual optimization for deterministic classical control}\label{app:optdetc}

The same approach is used to solve the problem~(\ref{eq:EBTPoptimisation}). We start by mapping the set of trace-preserving entanglement breaking qubit
channels to bipartite states $\mathfrak{B}$. For these channels $\mathfrak{B}$ is positive, has partial trace equal to the identity, and is also \emph{separable} \cite{03Horodecki629}.  Because $\mathfrak{B}$ is an (unnormalised) state of two qubits, the separability condition is equivalent to
the positivity of the partial transpose~\cite{96Horodecki1}. We will denote the partial transpose of the operator $\mathfrak{B}$ on the second subsystem $\textsf{H}$ by $\mathfrak{B}^{T_2}$. Thus we may rephrase the optimization
problem~(\ref{eq:EBTPoptimisation}) in the form
\begin{equation}\label{eq:prob_compact}
  \begin{array}{rl}
  \text{maximize}&\tr\left(R\mathfrak{B}\right)\\
  \text{subject to}&\mathfrak{B} \geq 0\,,\quad \mathfrak{B}^{{\sf T}_2}\geq 0\\
  &\tr_2 \mathfrak{B} = \openone_2.
  \end{array}
\end{equation}
Note that the condition of positivity of the partial transpose guarantees that $\mathfrak{B}$ corresponds to an entanglement breaking map.

The new problem has the same symmetries as the full optimization (\ref{eq:qprob_compact}) with
one addition. Notice that $R^{{\sf T}_2}=R$ so the objective function of both problems is invariant under partial transpose. In our new problem the point $\mathfrak{B}^{{\sf T}_2}$ is feasible if $\mathfrak{B}$ is feasible, so the feasible set is also invariant under the partial transpose. [Note that since partial transpose does not preserve positivity this is not true of the problem~\eqref{eq:qprob_compact}]. Because of this symmetry we may restrict our attention to $\mathfrak{B}_{\rm inv}$ for which $\mathfrak{B}^{{\sf T}_2}_{\rm inv}=\mathfrak{B}_{\rm inv}$. Since the partial transpose
sends $A\otimes Y \rightarrow -A\otimes Y$ where $A$ is any Hermitian matrix, we can conclude that $\tr[(A \otimes Y)\mathfrak{B}_{\rm inv}]=0 $. It is sufficient to
check this condition for the full set of Pauli matrices
$\openone_2,X,Y,Z$ so the requirement of invariance under the partial transpose constitutes four new constraints. Notice however that the condition $\mathfrak{B}^{{\sf T}_2}_{\rm inv}\geq 0$ is now redundant since we are requiring that $\mathfrak{B}^{{\sf T}_2}_{\rm inv}=\mathfrak{B}_{\rm inv}$. So we can
replace the problem~(\ref{eq:prob_compact}) with
\begin{equation}\label{eq:prob_compact2}
  \begin{array}{rl}
  \text{maximize}&\tr\left[R\mathfrak{B}\right]\\
  \text{subject to}&\mathfrak{B} \geq 0\\
  &\tr_2 \mathfrak{B} = \openone_2 \\
& \tr[(A \otimes Y)\mathfrak{B}]=0\quad \forall A \in \{I,X,Y,Z\}.
  \end{array}
\end{equation}
Positivity of the partial transpose and hence the separability of $\mathfrak{B}$ is now guaranteed by the positivity of $\mathfrak{B}$ and the additional equality constraints.

The dual of the problem~(\ref{eq:prob_compact2}) is
\begin{equation}
  \begin{array}{rl}\label{eq:ineq_form_class}
  \text{minimize}& \tr M \\
  \text{subject to}& M\otimes \openone_2 + N\otimes Y - R \geq 0
  \end{array}
\end{equation}
This semidefinite program still has symmetries corresponding to the rotation $X\otimes X$ and the transpose (but not under the partial transpose.) These two symmetries lead to the transformations $N\rightarrow - XNX$ and $N\rightarrow -N^T$ respectively. The only invariant choices of $N$ are proportional to $Y$. As before we may restrict
attention to $M_{\rm inv}= a_0 \openone_2+a_x X$ that are invariant under the
action of the group and $N_{\rm inv}=a_y Y$. This gives a simpler dual optimization
\begin{equation}\label{eq:cdual}
  \begin{array}{rl}
  \text{minimize}&2 a_0\\
  \text{subject to}&a_0 \openone_4+a_x
  X\otimes \openone_2 +a_y Y\otimes Y - R \geq
  0 \,,
  \end{array}
\end{equation}
where $a_0,a_x$ and $a_y$ are the new variables. This problem should
 be compared to the analogous dual optimization in the quantum case~(\ref{eq:qdual}).
 Again, this
problem can be solved analytically, yielding the solution
\begin{align}
  a_0 &=\frac{1}{4}+\frac{1}{4}\sqrt{\cos^2{\theta}+\sin^{4}{\theta}}\,,\\
  a_x &=\frac{r_x}{4}+\frac{r_x}{4}
  \frac{\cos^2{\theta}}{\sqrt{\cos^2{\theta}+\sin^{4}{\theta}}}\,,\\
  a_y &=-\frac{r_x}{4}
  \frac{\cos{\theta}\sin^2{\theta}}{\sqrt{\cos^2{\theta}+\sin^{4}{\theta}}}\,.
\end{align}
Again, one can check that $a_0 \openone_4+a_x X\otimes \openone_2+a_y Y\otimes Y - R$ is positive semidefinite with these choices, which ensures that the objective function $2 a_0$ is indeed a dual feasible value. The proof
of optimality follows as before in the quantum case by: (i) observing that $2 a_0$ reproduces the fidelity $\aver{\F^{\rm DDR2}}$ of Eq.~\eqref{eq:fcl2} and (ii) applying the weak duality argument.

We note that the optimization techniques presented here may be
useful when applied to more general problems presented in Fuchs and Sasaki~\cite{03Fuchs377}.  However, when the map in question does not act on qubits, there are significant complications in characterizing the EBTP maps because the PPT condition is no longer sufficient.

\chapter{Appendices to Chapter 6}\label{app:ch6}

Except for Appendix~\ref{app:optSSmap}, all the appendices presented here appear in Ref.~\cite{08Mendonca012319}.

\section{Perfect Tracking Conditions}\label{app:ptc}

A theorem closely related to the aims of this paper has been proved by Alberti and Uhlmann \cite{80Alberti163}, consisting of a mathematical criterion for the existence of physical operations perfectly transforming between pairs of qubit states. In this appendix we briefly review this theorem and prove an important corollary that is used in a number of places in this paper (e.g., sections \ref{sec:indicator} and \ref{sec:ex3}).
\begin{theorem}[Alberti and Uhlmann]\label{teo1}
Let $\rho_1$, $\rho_2$, $\overline{\rho}_1$, $\overline{\rho}_2$ be $2 \times 2$ density matrices. Then there exists a CPTP map $\mathcal{A}$ such that
\begin{equation}
\overline{\rho}_1=\mathcal{A}(\rho_1)\qquad\mbox{and}\qquad\overline{\rho}_2=\mathcal{A}(\rho_2)\,,
\end{equation}
if and only if
\begin{equation}\label{eq:alberti}
\|\overline{\rho}_1-t\overline{\rho}_2\|_{\rm tr}\leq\|\rho_1-t\rho_2\|_{\rm tr}\quad\mbox{for all }t\in \mathbb{R}^+\,,
\end{equation}
where $\|\cdot\|_{\rm tr}$ denotes the trace norm. For higher dimensional density matrices, the above condition is necessary but not sufficient for the existence of $\mathcal{A}$.
\end{theorem}
As pointed out by Chefles, Jozsa and Winter \cite{04chefles11},
the condition (\ref{eq:alberti}) is equivalent
to the requirement that the target states are no more distinguishable than the source states by minimum error probability discrimination (Helstrom \cite{76Helstrom}), for any prior probabilities. In the particular case where $\overline{\rho}_1$ and $\overline{\rho}_2$ are pure states, this just means that the Bloch angle between $\overline{\rho}_1$ and $\overline{\rho}_2$ is smaller than the angle between $\rho_1$ and $\rho_2$. This is proved in the following.
\begin{corollary}\label{cor2}
Let $\overline{\rho}_1$ and $\overline{\rho}_2$ be any two pure distinct qubit states separated by an angle $\overline{\Theta}\in(0,\pi]$ in the Bloch representation. Let $\rho_1$ and $\rho_2$ be any (mixed or pure) qubit states separated by $\Theta\in(0,\pi]$. A CPTP map $\mathcal{A}$ such that
\begin{equation}
\overline{\rho}_1=\mathcal{A}(\rho_1)\qquad\mbox{and}\qquad\overline{\rho}_2=\mathcal{A}(\rho_2)
\end{equation}
exists if and only if $\rho_1$ and $\rho_2$ are also pure and $\overline{\Theta}\leq\Theta$.
\end{corollary}
\begin{proof}
First note that the inequality (\ref{eq:alberti}) can be equivalently written with both sides
squared. Also, since $\overline{\rho}_1-t\overline{\rho}_2$ and $\rho_1-t\rho_2$ are hermitian matrices, their trace
norm can be computed as the sum of their eigenvalues. In terms of the Bloch parameters, a straightforward computation gives
\begin{equation}
\|\overline{\rho}_1-t\overline{\rho}_2\|^2=4(1+t^2-2t\cos{\overline{\Theta}})\,, \label{eq:lhs}\\
\end{equation}
where we have made use of the fact that $t\in\mathbb{R}^+$, and
\begin{align}
\|\rho_1-t\rho_2\|^2=&\;
2\left[(1-t)^2+(R_1^2+t^2R_2^2-2 t R_1 R_2\cos{\Theta})\right]\nonumber\\
&+2\left|(1-t)^2-(R_1^2+t^2 R_2^2-2 t R_1 R_2\cos{\Theta})\right|\,, \label{eq:rhs}
\end{align}
where $R_i$ gives the magnitude of the Bloch vector for $\rho_i$, $i=1,2$.\par

Now assume that the absolute value on the right hand side of Eq. (\ref{eq:rhs}) can be removed, then the
inequality (\ref{eq:alberti}) takes the form
\begin{equation}\label{eq:ineq_impossible}
1+t^2-2t\cos{\overline{\Theta}}\leq(1-t)^2\,,
\end{equation}
which for all $t\in\mathbb{R^+}$ is satisfied if and only if $\cos{\overline{\Theta}}=1$. However, as the (pure) target states are required to be distinct, we must have $\cos{\overline{\Theta}}<1$. As a result, the inequality (\ref{eq:ineq_impossible}) is never satisfied.\par

Assume then the complementary case (when the absolute value of Eq. (\ref{eq:rhs}) is removed
at the cost of a change of sign). Then (\ref{eq:alberti}) can be written as $F(t)\leq 0$ with
\begin{equation}
F(t)=(1-R_2^2)t^2-2t(\cos{\overline{\Theta}}-R_1 R_2\cos{\Theta})+(1-R_1^2)\,.
\end{equation}
If $R_2 \neq 1$, $F(t)$ is a strictly convex function of $t$, therefore cannot be bounded from above by $0$ for all $t\in\mathbb{R^+}$, so it is necessary that $R_2=1$ ($\rho_2$ must be pure). Then, define $G(t)=\left.F(t)\right|_{R_2=1}$, explicitly
\begin{equation}
G(t)=-2t(\cos{\overline{\Theta}}-R_1\cos{\Theta})+(1-R_1^2)\,,
\end{equation}
and require $G(t)\leq 0$.\par
If $R_1\neq 1$, $G(t)$ is a linear function of $t$ with strictly positive linear coefficient. Again, such a function cannot be bounded from above by $0$ for all $t\in\mathbb{R^+}$, so it is necessary to make $R_1=1$ ($\rho_1$ must be pure). Finally, define $H(t)=\left.G(t)\right|_{R_1=1}$, i.e.,
\begin{equation}
H(t)=-2t(\cos{\overline{\Theta}}-\cos{\Theta})\,,
\end{equation}
and require $H(t)\leq 0$. Clearly, this inequality is satisfied for all $t\in\mathbb{R^+}$ if and only if $\cos{\overline{\Theta}}\geq\cos{\Theta}$, or equivalently, $\overline{\Theta}\leq\Theta$.
\end{proof}

\section{Technical details}

\subsection{Properties of $S$ and $T$}\label{app:ST}
Here, we prove that $S>0$ and $S+T>0$ if and only if one of the following holds
 \begin{enumerate}
\item[i)] $\{\overline{\bm{R}}_1,\overline{\bm{R}}_2\}$ is linearly independent; or
\item[ii)] $\{\overline{\bm{R}}_1,\overline{\bm{R}}_2\}$ is linearly dependent with $T>0$.
 \end{enumerate}
 Moreover, we show that the complementary case
\begin{enumerate}
\item[iii)] $\{\overline{\bm{R}}_1,\overline{\bm{R}}_2\}$ is linearly dependent with $T\leq 0$,
\end{enumerate}
occurs only if $\Omega=0$.\par
 This result is useful to demonstrate that the coefficients $\mu_1$, $\mu_2$, $\mu_3$ and $s_1$ defined in Eq. (\ref{eq:nontrivial}) for $\Omega>0$ (procedure A) are always (a) well-defined, (b) real (c) within the range $[0,1]$. We start with the following lemma

\begin{lemma}\label{lemma1}
Let $\bm{R}_1$, $\bm{R}_2$ be real three dimensional vectors such that $R_i\leq 1$ $(i=1,2)$. Define $\bm{R}_-\mathrel{\mathop:}=\bm{R}_1-\bm{R}_2$. If $R_-\neq0$ (i.e., $\bm{R}_1$, $\bm{R}_2$ are distinct), then $R_-^2>R_\times^2$.
\end{lemma}
\begin{proof}
Consider the triangle defined by the vectors $\bm{R}_1$, $\bm{R}_2$ and $\bm{R}_-$ as shown in Fig. \ref{fig:triangle}.
\begin{figure}
\centering
\includegraphics[width=5cm]{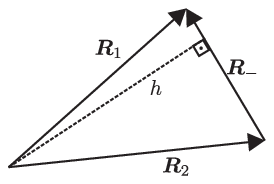}
\caption{Schematic for proof that $R_- > R_\times$}\label{fig:triangle}
\end{figure}
The magnitude of $\bm{R}_\times$ gives $2$ times the area of this triangle so that
\begin{equation}
R_\times^2=h^2R_-^2\,,
\end{equation}
where $h$ is the altitude relative to the side of length $R_-$. We write the following
\begin{equation}
R_\times^2=h^2R_-^2\leq \min{\left(R_1,R_2\right)}^2 R_-^2\leq R_-^2\,.\label{eq:ineqs}
\end{equation}

The first inequality is a direct consequence of the Pythagorean theorem, and the second follows from $R_i\leq 1$. This establishes that $R_-^2 \geq R_\times^2$. This inequality is trivially saturated if $R_-=0$. To see that this is the only case where saturation occurs, assume $R_-\neq 0$ and require saturation of both inequalities in Eq. (\ref{eq:ineqs}). The first inequality is saturated iff $R_-^2=|R_1^2-R_2^2|$ (by the Pythagorean theorem), and the second one iff $R_1=R_2=1$. Taken together, these conditions imply $R_-=0$, which contradicts the hypothesis. Therefore, if $R_-\neq 0$ (i.e., $\bm{R}_1\neq\bm{R}_2$), then  $R_-^2 > R_\times^2$.
\end{proof}
Now, recall that
\begin{align}
S = \sqrt{T^2+4\overline{R}_\times^2(R_-^2-R_\times^2)}\,.
\end{align}
Assume first linear independence of $\{\overline{\bm{R}}_1,\overline{\bm{R}}_2\}$ (i.e., $\overline{R}_\times\neq 0$). From Lemma \ref{lemma1}, it is immediate that $S>0$. Moreover,
\begin{align}
S+T = T+\sqrt{T^2+4\overline{R}_\times^2(R_-^2-R_\times^2)}>T+|T|\geq 0\,,
\end{align}
where the first inequality follows from Lemma \ref{lemma1} and the second is trivial. Therefore, $S>0$ and $S+T>0$ if condition (i) holds.\par

 For linearly dependent $\{\overline{\bm{R}}_1,\overline{\bm{R}}_2\}$, it is easy to see that $S=|T|$ and $S+T=|T|+T$, therefore $S>0$ and $S+T>0$ if condition (ii) holds.\par

  To prove the \emph{only if} part, consider the complementary case (iii). It is immediate that $S+T=0$ if $\{\overline{\bm{R}}_1,\overline{\bm{R}}_2\}$ are linearly dependent and $T \leq 0$, hence (i) and (ii) are the only situations where the premise holds.\par
  It follows trivially from the discussion above that $\Omega=0$ for condition (iii). Simply note that $S+T=0$ and the linear dependence of the targets Bloch vectors requires $2R_\times \overline{R}_\times=0$.

\subsection{Well-definedness of the dual feasible point}\label{app:welldef}

The proposed values for the coefficients $\mathfrak{x}_1$ and $\mathfrak{x}_3$ defined in Eqs. (\ref{eq:xi3p}), (\ref{eq:xi1m}) and (\ref{eq:xi3m}) have the quantities  $\Gamma_a$ and $\Gamma_b$ appearing in the denominator. In this appendix we show that this does not lead to any singularity as long as the indicated range of $\Omega$ is considered.

 To see that, note that $\Gamma_a=0$ if and only if $\overline{R}_+=\overline{R}_\times=0$. This, in turn, is equivalent to the statement that the targets have Bloch vectors of same magnitude $\overline{R}$ pointing to opposite directions, which used in Eq. (\ref{eq:T}) gives $T=-R_-^2\overline{R}^2$. In these circumstances, $\Omega$ can be easily computed to be $\Omega=T+|T|=0$. Therefore, no singularity can occur in Eq. (\ref{eq:xi3p}) in the range $\Omega>0$.\par

 Similarly, $\Gamma_b=0$ if and only if $\overline{R}_+^2=T-2R_\times\overline{R}_\times$, in which case we can write $\Omega=S+\overline{R}_+^2$. In the sequence we show that $S+\overline{R}_+^2>0$, thus no singularity can occur in Eqs. (\ref{eq:xi1m}) and (\ref{eq:xi3m}) in the range $\Omega\leq 0$.\par

  From the definition of $S$ in Eq. (\ref{eq:S}), it is immediate that the inequality $S+\overline{R}_+^2\geq 0$ holds, so we just need to show that $S+\overline{R}_+^2\neq 0$. Suppose, on the contrary, that $S=-\overline{R}_+^2$, which is possible only if $S=\overline{R}_+=0$. From Eq. (\ref{eq:S}), this can be seen to be equivalent to $T=\overline{R}_\times=\overline{R}_+=0$. To see that this leads to a contradiction, use once again the fact that $\overline{R}_\times=\overline{R}_+=0$ implies opposing target Bloch vectors of same magnitude $\overline{R}$, which gives $T=-R_-^2\overline{R}^2\neq 0$. The inequality follows from the conditions of the problem: the source states cannot be identical ($R_-\neq 0$), and the case where the two targets are identical to the maximally mixed states has been excluded from the analysis ($\overline{R}\neq 0$).

\subsection{Characteristic Polynomials for $F$}\label{app:charpoly}

In this appendix we compute the characteristic polynomials of the matrix $F$, Eq. (\ref{eq:lmi}), with the set of coefficients given in Eq. (\ref{eq:coefomegap}) (for $\Omega > 0$, procedure A) and Eq. (\ref{eq:coefomegam}) (for $\Omega \leq 0$, procedure B). By studying these polynomials, we show that $F\geq 0$, thus completing the proof of the optimality of our tracking strategy.

\subsubsection{Procedure A}
For the set of coefficients (\ref{eq:coefomegap}) (case $\Omega>0$), the characteristic equation for $F$ factorizes as
$\lambda^2 P_2(\lambda)=0$, where
\begin{equation}
P_2(\lambda)=\lambda^2-\Gamma_a\lambda+\upsilon\left[(R_-^2-R_\times^2)\Gamma_a^4-\Xi^2\right]\,,\label{eq:nont_poly}
\end{equation}
and
\begin{equation}
\upsilon=\frac{4 R_-^2 \overline{R}_\times^2+(S+T)^2}{8 R_-^2 \Gamma_a^2 S(S+T)}\,.
\end{equation}
Since both $\Gamma_a$ and $\upsilon$ are positive, the eigenvalues of $F$ are non-negative if the term in square brackets in the Eq. (\ref{eq:nont_poly}) is non-negative when $\Omega > 0$. We now show that this term is non-negative irrespective of the sign of $\Omega$.\par

First use Eq. (\ref{eq:xixi}) to substitute for $\Xi^2$, after some manipulation we find that
\begin{equation}\label{eq:computstep1}
(R_-^2-R_\times^2)\Gamma_a^4-\Xi^2=
R_-^2\left[a\left(\overline{R}_+^2+\frac{R_-^2\overline{R}_\times^2}{S+T}\right)+ R_-^2\overline{R}_\times^2+\overline{R}_+^2T\right]\,,
\end{equation}
where we have defined
\begin{equation}
a\mathrel{\mathop:}=\frac{4\overline{R}_\times^2}{S+T} \left(R_-^2-R_\times^2\right)=S-T\,,
\end{equation}
with the second equality following from Eq. (\ref{eq:S}). Note that the non-negativity of $(R_-^2-R_\times^2)\Gamma_a^4-\Xi^2$ cannot be immediately concluded from Eq. (\ref{eq:computstep1}) --- although the first and second summands are non-negative, the term $\overline{R}_+^2 T$ does admit negative values. However, using $a=S-T$ in Eq. (\ref{eq:computstep1}), after some rearrangement we get,
\begin{equation}
(R_-^2-R_\times^2)\Gamma_a^4-\Xi^2=
R_-^2\left[S\overline{R}_+^2+R_-^2\overline{R}_\times^2\left(1+\frac{S-T}{S+T}\right)\right]\geq 0\,,
\end{equation}
from which the fulfillment of the inequality is obvious. In conclusion, procedure A is optimal.\par

\subsubsection{Procedure B}
For the set of coefficients (\ref{eq:coefomegam}) (case $\Omega\leq 0$), the characteristic equation for $F$ is $\lambda P_3(\lambda)=0$, where
\begin{equation}
P_3(\lambda)=\lambda^3-\Gamma_b\lambda^2+\varpi\lambda+\omega\,,\label{eq:t_poly}
\end{equation}
and
\begin{align}
\varpi&=\tfrac{1}{4 R_-^4 \Gamma_b^2}\left[\left(1+\tfrac{R_\times \xi}{R_-^2\Gamma_b^2}\right)R_-^4\Gamma_b^4-(R_-^2+R_\times^2)(\xi^2+\Xi^2)\right],\label{eq:varpi}\\
\omega&=-\frac{(R_\times\Gamma_b^2-\xi)\left[R_-^2\Gamma_b^2\xi-R_\times(\xi^2+\Xi^2)\right]}{8 R_-^4\Gamma_b^3}\label{eq:omega}\,.
\end{align}
from which it follows that the eigenvalues of $F$ are non negative if $\varpi \geq 0$ \emph{and} $\omega \leq 0$ \emph{when} $\Omega \leq 0$. Next, we simplify Eqs. (\ref{eq:varpi}) and (\ref{eq:omega}) in order to make it clear that these conditions are satisfied.\par

It is just a matter of applying Eqs. (\ref{eq:xi}) and (\ref{eq:xixi}) to Eq. (\ref{eq:varpi}) to show that
\begin{equation}
\varpi=\tfrac{1}{4}\left(-\Omega+S+R_\times\overline{R}_\times\right)\geq 0\,,
\end{equation}
from which the inequality is clearly seen to hold if $\Omega \leq 0$.\par

To prove that $\omega\leq 0$ if $\Omega \leq 0$, consider first the term in the square brackets in Eq. (\ref{eq:omega}). Again, employing Eqs. (\ref{eq:xi}) and (\ref{eq:xixi}) this can be simplified to   $R_-^4\overline{R}_\times\Gamma_b^2$ which is obviously non-negative. Therefore, the validity of the inequality $\omega \leq 0$ if $\Omega\leq 0$ is now solely conditioned on the validity of the inequality
\begin{equation}\label{eq:ineqprocB}
R_\times\Gamma_b^2-\xi \geq 0\quad\mbox{for}\quad \Omega \leq 0\,.
\end{equation}
To see that this is so, first note that the only way to satisfy the conditions $\overline{R}_\times=0$ and $\Omega\leq 0$ is to have $S=T=\Omega=0$, which implies that (\ref{eq:ineqprocB}) is satisfied with saturation. Consider then the complementary case $\overline{R}_\times\neq 0$ and $\Omega \leq 0$. Using Eq. (\ref{eq:xi}) for $\xi$, and multiplying and dividing by $4\overline{R}_\times$, we get
\begin{equation}\label{eq:computstep2}
R_\times\Gamma_b^2-\xi=\frac{1}{4\overline{R}_\times}\left[-4 R_\times\overline{R}_\times(\Omega-S)-4\overline{R}_\times^2R_-^2\right]\,.
\end{equation}

Now, from Eq. (\ref{eq:S}), we know that $4\overline{R}_\times^2R_-^2=S^2-T^2+4R_\times^2\overline{R}_\times^2$, which used in Eq. (\ref{eq:computstep2}) gives, after some algebra,
 \begin{equation}
R_\times\Gamma_b^2-\xi=-\frac{\Omega}{4\overline{R}_\times}\left(S-T+2R_\times\overline{R}_\times\right)\geq 0\,.
 \end{equation}
Once again, the inequality is obviously true if $\Omega \leq 0$, thus establishing the optimality of procedure B.

\section{Choi matrix for optimal single-step tracking}\label{app:optSSmap}

In Sec~\ref{sec:Choimatrep}, we have seen that any quantum channel $\mathcal{C}$ on a single qubit can be represented with a Choi matrix $\mathfrak{C}$ of the form
\begin{equation}\label{eq:choigenapp}
2 \mathfrak{C} = \openone_{\rm 4} + \sum_{k=1}^{3}{s_k \openone_{\rm 2}\otimes\left(\bm{u}_k\cdot\bm{\sigma}\right)}+\sum_{k=1}^3 \mu_k \left(\bm{v}_k\cdot\bm{\sigma}^{\sf T}\right)\otimes\left(\bm{u}_k\cdot\bm{\sigma}\right)\,,
\end{equation}
where $\{\bm{v}_k\}_{k=1,2,3}$ and $\{\bm{u}_k\}_{k=1,2,3} \in \mathbb{R}^3$ are orthonormal sets of vectors and $\mu_k$ and $s_k$ are certain scalars within the range $[-1,1]$.

In this appendix, we derive an explicit formula for the vectors $\bm{v}_k$ and $\bm{u}_k$ that specify the Choi matrix of the optimal single-step tracking map introduced in Sec.~\ref{sec:probandstrategy}. Optimal values for the scalar parameters $\mu_k$ and $s_k$ have already been given in Sec.~\ref{sec:probandstrategy}, where we saw that $s_2=s_3=0$, while $s_1$ and $\mu_k$ are conditioned on the value of the indicator function $\Omega$ [cf. Eq.~\eqref{eq:IF}]: $s_1=0$ and $\mu_k=1$ if $\Omega\leq 0$, or otherwise they are given by Eqs.~\eqref{eq:nontrivial}.\par

Following Eqs.~\eqref{eq:V} and~\eqref{eq:U}, we can write
\begin{align}
V\left(\frac{\openone_{\rm 2}}{2}+\frac{\bm{R}_i\cdot\bm{\sigma}}{2}\right)V^\dagger&=\frac{\openone_{\rm 2}}{2}+\frac{R_\times}{2R_-}X+\frac{\left(\bm{R}_i\cdot\bm{R}_-\right)}{2 R_-}Z\,,\label{eq:VV}\\
U\left(\frac{\openone_{\rm 2}}{2}+\frac{\alpha}{2}X+\frac{\beta_i\overline{R}_\times}{2}Z\right)U^\dagger&=\frac{\openone_{\rm 2}}{2}+\frac{1}{2}\left(k_{i1}\overline{\bm{R}}_1+k_{i2}\overline{\bm{R}}_2\right)\cdot\bm{\sigma}\,,\label{eq:UU}
\end{align}
for $i=1,2$ and any source and target Bloch vectors $\bm{R}_i$ (with $\bm{R}_1\neq\bm{R}_2$) and $\overline{\bm{R}}_i$\footnote{Recall, however, that the optimal unitary $U$ is not of the form given in Eq.~\eqref{eq:U} when $\Omega\leq 0$ \emph{and} $\overline{R}_\times=0$. As a consequence, the vector $\bm{u}_k$ derived here is not the optimal one for this particular case.}. In the above, $k_{i1}$ and $k_{i2}$ were given in Eq.~\eqref{eq:kij}, while the symbols $\alpha$ and $\beta_i$ are shorthand notation for
\begin{equation}\label{eq:alphabeta}
\alpha=s_1+\mu_1 \frac{R_\times}{R_-}\quad\mbox{and}\quad \beta_i=\mu_3\frac{\bm{R}_i\cdot\bm{R}_-}{R_-\overline{R}_\times}\,.
\end{equation}
$\alpha$ and $\beta_i$  were more explicitly evaluated in Eqs.~\eqref{eq:alpha_nt}, \eqref{eq:beta_nt} for $\Omega>0$ and in Eqs.~\eqref{eq:alpha_t}, \eqref{eq:beta_t} for $\Omega\leq 0$.

With some vector algebra, Eqs.~\eqref{eq:VV} and~\eqref{eq:UU} can be converted into expressions for  $V^\dagger\sigma_j V$ and $U \sigma_j U^\dagger$ ($j=1,2,3$), which used in Eqs.~\eqref{eq:vk_intro} and~\eqref{eq:uk_intro} yield
 \begin{equation}\label{eq:vs}
 \bm{v}_2=\bm{R}_\times/R_\times\,,\qquad\bm{v}_3=\bm{R}_-/R_-\,,
 \end{equation}
\begin{equation}\label{eq:us}
\bm{u}_2=\overline{\bm{R}}_\times/\overline{R}_\times\quad\mbox{and}\quad
\bm{u}_3=\frac{1}{\Gamma}\left[\frac{\alpha}{\overline{R}_\times}\left(\overline{\bm{R}}_+\times \overline{\bm{R}}_\times\right)+\overline{R}_\times\left(\beta_1\overline{\bm{R}}_1+\beta_2\overline{\bm{R}}_2\right)\right]\,,
\end{equation}
with $\bm{v}_1=\bm{v}_2\times\bm{v}_3$ and $\bm{u}_1=\bm{u}_2\times\bm{u}_3$, as explained in Sec.~\ref{sec:Choimatrep}. In Eq.~\eqref{eq:us}, $\Gamma$ is a normalization factor given in Eq.~\eqref{eq:Gamma}, and repeated below
\begin{equation}\label{eq:gammaapp}
\Gamma=\sqrt{\alpha^2\overline{R}_+^2+\left[\|\beta_1\overline{\bm{R}}_1+\beta_2\overline{\bm{R}}_2\|^2+2\alpha\left(\beta_1-\beta_2\right)\right]\overline{R}_\times^2}\,.
\end{equation}
In Eq.~\eqref{eq:gamma_a} [resp. Eq.~\eqref{eq:gamma_b}] a more explicitly formula for $\Gamma$ was given in the case $\Omega>0$ [resp. $\Omega\leq 0$].

\backmatter

\chapter{List of Symbols}

The following list is neither exhaustive nor exclusive, but may be helpful.
\begin{list}{}{%
\setlength{\labelwidth}{24mm}
\setlength{\leftmargin}{35mm}}

\item[$\openone_{\rm d}$\dotfill] The identity matrix of dimension {\rm d};

\item[$X$, $Y$, $Z$\dotfill] The Pauli matrices;

\item[$\sigma_1,\sigma_2,\sigma_3$\dotfill] Alternative notation for the Pauli matrices $X$, $Y$, and $Z$, respectively.

\item[$\mathcal{M}_{\rm d}$\dotfill] The algebra of complex matrices of dimension {\rm d};

\item[$M^{\sf T}$\dotfill] The transpose of the matrix $M$;

\item[$M^{{\sf T}_k}$\dotfill] The partial transpose of the matrix $M$ with respect to the $k-$th subsystem;

\item[$\tr M$\dotfill] The trace of the matrix $M$;

\item[$\tr_k M$\dotfill] The partial trace of the matrix $M$ with respect to the $k-$th subsystem;

\item[$\mathcal{I}_{\rm d}$\dotfill] The identity map on $\mathcal{M}_{\rm d}$;

\item[$\mathcal{C}$\dotfill] An arbitrary linear map from $\mathcal{M}_{\rm d}$ to $\mathcal{M}_{\rm d}$;

\item[$\mathfrak{C}$\dotfill] The Choi matrix of the map $\mathcal{C}$;

\item[$\mathcal{K}_{\rm d}^{\rm set}$\dotfill] The set of completely positive maps from $\mathcal{M}_{\rm d}$ to $\mathcal{M}_{\rm d}$;

\item[$\mathfrak{K}$\dotfill] The Choi matrix of a map $\mathcal{K}\in\mathcal{K}_{\rm d}^{\rm set}$;

\item[$\mathcal{Q}_{\rm d}^{\rm set}$\dotfill] The subset of trace preserving maps of $\mathcal{K}_{\rm d}^{\rm set}$;

\item[$\mathfrak{Q}$\dotfill] The Choi matrix of a map $\mathcal{Q}\in\mathcal{Q}_{\rm d}^{\rm set}$;

\item[$\widetilde{\mathcal{B}}_{\rm d}^{\rm set}$\dotfill] The subset $\mathcal{Q}_{\rm d}^{\rm set}$ with elements of positive partial transposed Choi matrix;

\item[$\mathcal{B}_{\rm d}^{\rm set}$\dotfill] The set of entanglement breaking and trace preserving maps from $\mathcal{M}_{\rm d}$ to $\mathcal{M}_{\rm d}$;

\item[$\mathfrak{B}$\dotfill] The Choi matrix of a map $\mathcal{B}\in\mathcal{B}_{\rm d}^{\rm set}$;

\item[${\sf H}_{\rm d}$\dotfill] Hilbert space of dimension {\rm d};

\item[$\{\ket{h_{\rm d}^\alpha}\}_{\alpha=1}^{\rm d}$\dotfill] An orthonormal bases of $\textsf{H}_{\rm d}$;

\item[$\ket{\Psi}$\dotfill] The unnormalized maximally entangled state $\sum_{\alpha=1}^{\rm d}{\ket{h_{\rm d}^\alpha}\otimes \ket{h_{\rm d}^\alpha}}\in \textsf{H}_{\rm d}\otimes\textsf{H}_{\rm d}$;

\item[$\{H_{\rm d}^\alpha\}_{\alpha=1}^{\rm d}$\dotfill] An orthonormal bases for the Hermitian matrices of dimension {\rm d} [generators of SU({\rm d})];

\item[$\F$\dotfill] The Uhlmann-Jozsa fidelity;

\item[$A,B,C{[\F]}$\dotfill] Three metrics related to $\F$;

\item[$\Fn$\dotfill] An alternative fidelity measure between mixed states;

\item[$C{[\Fn]}$\dotfill] A metric related to $\Fn$;

\item[$Q$\dotfill] The non-logarithmic variety of the quantum Chernoff bound;

\item[$\|M\|_{\rm tr}$\dotfill] The trace norm of a matrix $M$;

\item[$\|M\|_{\rm HS}$\dotfill] The Hilbert-Schmidt (or Frobenius) norm of a matrix $M$;

\item[$\|M\|$\dotfill] The spectral (or operator) norm of a matrix $M$;

\item[$\D$\dotfill] The trace distance;

\item[$\H$\dotfill] The metric induced by the Hilbert-Schmidt norm;

\item[$\O$\dotfill] The metric induced by the Spectral norm;

\item[$\mathscr{D}$\dotfill] An arbitrary function measuring the distance between density matrices;

\item[${[\rho_i]}_{i=1}^I$\dotfill] A sequence of $I$ density matrices;

\item[$\aver{\mathscr{D}}_{1,2}$\dotfill] Two possible averaging schemes for quantifying the distance $\mathscr{D}$ between sequences of density matrices;
\end{list}


\begin{thebibliography}{100}
\expandafter\ifx\csname url\endcsname\relax
  \def\url#1{\texttt{#1}}\fi
\expandafter\ifx\csname urlprefix\endcsname\relax\def\urlprefix{URL }\fi
\providecommand{\eprint}[2][]{\url{#2}}

\bibitem{08Mendonca1150}
P.~E. M.~F. Mendonca, R.~d.~J.~Napolitano, M.~A. Marchiolli, C.~J. Foster, and
  Y.-C. Liang.
\newblock \emph{Alternative fidelity measure between quantum states}.
\newblock Physical Review A \textbf{78}(5), 052330 (2008).
\newblock \href{http://arxiv.org/abs/0806.1150}{E-print arXiv:0806.1150v2
  [quant-ph]}.

\bibitem{07Branczyk012329}
A.~M. Bra{\'n}czyk, P.~E. M.~F. Mendon\c{c}a, A.~Gilchrist, A.~C. Doherty, and
  S.~D. Bartlett.
\newblock \emph{Quantum control of a single qubit}.
\newblock Physical Review A \textbf{75}(1), 012329 (2007).
\newblock \href{http://arxiv.org/abs/quant-ph/0608037v2}{E-print
  arXiv:quant-ph/0608037v2}.

\bibitem{08Mendonca012319}
P.~E. M.~F. Mendon\c{c}a, A.~Gilchrist, and A.~C. Doherty.
\newblock \emph{Optimal tracking for pairs of qubit states}.
\newblock Physical Review A \textbf{78}, 012319 (2008).
\newblock \href{http://arxiv.org/abs/0802.3896}{E-print arXiv:0802.3896v1
  [quant-ph]}.

\bibitem{05Branczyk}
A.~M. Bra{\'n}czyk.
\newblock \emph{Quantum control of a single qubit}.
\newblock Honours Thesis, The University of Queensland (2005).

\bibitem{96Fuchs2038}
C.~A. Fuchs and A.~Peres.
\newblock \emph{Quantum-state disturbance versus information gain: Uncertainty
  relations for quantum information}.
\newblock Physical Review A \textbf{53}(4), 2038 (1996).
\newblock \href{http://arxiv.org/abs/quant-ph/9512023}{E-print
  arXiv:quant-ph/9512023v1}.

\bibitem{01Doherty062306}
A.~C. Doherty, K.~Jacobs, and G.~Jungman.
\newblock \emph{Information, disturbance, and hamiltonian quantum feedback
  control}.
\newblock Physical Review A \textbf{6306}(6), 062306 (2001).
\newblock \href{http://arxiv.org/abs/quant-ph/0006013v1}{E-print
  arXiv:quant-ph/0006013v1}.

\bibitem{01Fuchs062305}
C.~A. Fuchs and K.~Jacobs.
\newblock \emph{Information-tradeoff relations for finite-strength quantum
  measurements}.
\newblock Physical Review A \textbf{63}(6), 062305 (2001).

\bibitem{83Belavkin178}
V.~P. Belavkin.
\newblock \emph{Theory of the control of observable quantum-systems}.
\newblock Automation and Remote Control \textbf{44}(2), 178 (1983).
\newblock \href{http://arxiv.org/abs/quant-ph/0408003v2}{E-print
  arXiv:quant-ph/0408003v2}.

\bibitem{79Belavkin3}
V.~P. Belavkin.
\newblock \emph{Optimal measurement and control in quantum dynamical systems}
  (1979).
\newblock Available online at
  \href{http://www.maths.nottingham.ac.uk/conferences/qpic/talks.html}{\url{ht%
tp://www.maths.nottingham.ac.uk/conferences/qpic/talks.html}}. See
  \cite{99Belavkin405} for a later published version.

\bibitem{94Wiseman2133}
H.~M. Wiseman.
\newblock \emph{Quantum-theory of continuous feedback}.
\newblock Physical Review A \textbf{49}(3), 2133 (1994).

\bibitem{93Wiseman548}
H.~M. Wiseman and G.~J. Milburn.
\newblock \emph{Quantum-theory of optical feedback via homodyne detection}.
\newblock Physical Review Letters \textbf{70}(5), 548 (1993).

\bibitem{02Armen133602}
M.~A. Armen, J.~K. Au, J.~K. Stockton, A.~C. Doherty, and H.~Mabuchi.
\newblock \emph{Adaptive homodyne measurement of optical phase}.
\newblock Physical Review Letters \textbf{89}(13), 133602 (2002).
\newblock \href{http://arxiv.org/abs/quant-ph/0204005v1}{E-print
  arXiv:quant-ph/0204005v1}.

\bibitem{06Bushev043003}
P.~Bushev, D.~Rotter, A.~Wilson, F.~Dubin, C.~Becher, J.~Eschner, R.~Blatt,
  V.~Steixner, P.~Rabl, and P.~Zoller.
\newblock \emph{Feedback cooling of a single trapped ion}.
\newblock Physical Review Letters \textbf{96}(4), 043003 (2006).
\newblock \href{http://arxiv.org/abs/quant-ph/0509125v1}{E-print
  arXiv:quant-ph/0509125v1}.

\bibitem{04Geremia270}
J.~M. Geremia, J.~K. Stockton, and H.~Mabuchi.
\newblock \emph{Real-time quantum feedback control of atomic spin-squeezing}.
\newblock Science \textbf{304}(5668), 270 (2004).

\bibitem{04LaHaye74}
M.~D. LaHaye, O.~Buu, B.~Camarota, and K.~C. Schwab.
\newblock \emph{Approaching the quantum limit of a nanomechanical resonator}.
\newblock Science \textbf{304}(5667), 74 (2004).

\bibitem{04Reiner023819}
J.~E. Reiner, W.~P. Smith, L.~A. Orozco, H.~M. Wiseman, and J.~Gambetta.
\newblock \emph{Quantum feedback in a weakly driven cavity qed system}.
\newblock Physical Review A \textbf{70}(2), 023819 (2004).

\bibitem{02Smith133601}
W.~P. Smith, J.~E. Reiner, L.~A. Orozco, S.~Kuhr, and H.~M. Wiseman.
\newblock \emph{Capture and release of a conditional state of a cavity qed
  system by quantum feedback}.
\newblock Physical Review Letters \textbf{89}(13), 133601 (2002).
\newblock \href{http://arxiv.org/abs/quant-ph/0202063v1}{E-print
  arXiv:quant-ph/0202063v1}.

\bibitem{96Vandenberghe49}
L.~Vandenberghe and S.~Boyd.
\newblock \emph{Semidefinite programming}.
\newblock Siam Review \textbf{38}(1), 49 (1996).

\bibitem{02Audenaert030302}
K.~Audenaert and B.~De~Moor.
\newblock \emph{Optimizing completely positive maps using semidefinite
  programming}.
\newblock Physical Review A \textbf{65}(3), 030302(R) (2002).
\newblock \href{http://arxiv.org/abs/quant-ph/0109155v2}{E-print
  arXiv:quant-ph/0109155v2}.

\bibitem{04Boyd}
S.~Boyd and L.~Vandenberghe.
\newblock \emph{Convex Optimization} (Cambridge University Press, 2004).
\newblock Available online at
  \href{http://www.stanford.edu/~boyd/cvxbook/}{\url{http://www.stanford.edu/~%
boyd/cvxbook/}}.

\bibitem{00Nielsen}
M.~A. Nielsen and I.~L. Chuang.
\newblock \emph{Quantum Computation and Quantum Information} (Cambridge
  University Press, Cambridge, 2000).

\bibitem{06Bengtsson}
I.~Bengtsson and K.~\.Zyczkowski.
\newblock \emph{Geometry of quantum states: An Introduction to Quantum
  Entanglement} (Cambridge University Press, 2006).

\bibitem{06Hayashi}
M.~Hayashi.
\newblock \emph{Quantum Information: An Introduction} (Springer-Verlag, Berlin,
  2006).

\bibitem{55Stinespring211}
W.~F. Stinespring.
\newblock \emph{Positive functions on c$^\ast$-algebras}.
\newblock Proceedings of the American Mathematical Society \textbf{6}(2), 211
  (1955).

\bibitem{71Kraus311}
K.~Kraus.
\newblock \emph{General state changes in quantum theory}.
\newblock Annals of Physics \textbf{64}(2), 311 (1971).

\bibitem{75Choi285}
M.~D. Choi.
\newblock \emph{Completely positive linear maps on complex matrices}.
\newblock Linear Algebra and its Applications \textbf{10}(3), 285 (1975).

\bibitem{72Choi520}
M.~Choi.
\newblock \emph{Positive linear maps on c*-algebras}.
\newblock Canadian Journal of Mathematics \textbf{24}(3), 520 (1972).

\bibitem{02Raginsky}
M.~Raginsky.
\newblock \emph{Dynamical Aspects of Information Storage in Quantum-Mechanical
  Systems}.
\newblock Ph.D. thesis, Northwestern University (2002).
\newblock \href{http://arxiv.org/abs/quant-ph/0207162}{E-print
  arXiv:quant-ph/0207162v2}.

\bibitem{03Leung528}
D.~W. Leung.
\newblock \emph{Choi's proof as a recipe for quantum process tomography}.
\newblock Journal of Mathematical Physics \textbf{44}(2), 528 (2003).
\newblock \href{http://arxiv.org/abs/quant-ph/0201119}{E-print
  arXiv:quant-ph/0201119v1}.

\bibitem{05Salgado55}
D.~Salgado, J.~L. Sanchez-Gomez, and M.~Ferrero.
\newblock \emph{A simple proof of the jamiolkowski criterion for complete
  positivity of linear maps}.
\newblock Open Systems \& Information Dynamics \textbf{12}(1), 55 (2005).
\newblock \href{http://arxiv.org/abs/math-ph/0406010v2}{E-print
  arXiv:math-ph/0406010v2}.

\bibitem{04Salgado054102}
D.~Salgado, J.~L. Sanchez-Gomez, and M.~Ferrero.
\newblock \emph{Evolution of any finite open quantum system always admits a
  kraus-type representation, although it is not always completely positive}.
\newblock Physical Review A \textbf{70}(5), 054102 (2004).

\bibitem{07Jordan}
T.~F. Jordan.
\newblock \emph{Why quantum dynamics is linear} (2007).
\newblock \href{http://arxiv.org/abs/quant-ph/0702171}{E-print
  arXiv:quant-ph/0702171}.

\bibitem{06Jordan022101}
T.~F. Jordan.
\newblock \emph{Assumptions that imply quantum dynamics is linear}.
\newblock Physical Review A \textbf{73}(2), 022101 (2006).
\newblock \href{http://arxiv.org/abs/quant-ph/0508092}{E-print
  arXiv:quant-ph/0508092}.

\bibitem{96Peres1413}
A.~Peres.
\newblock \emph{Separability criterion for density matrices}.
\newblock Physical Review Letters \textbf{77}(8), 1413 (1996).
\newblock \href{http://arxiv.org/abs/quant-ph/9604005v2}{E-print
  arXiv:quant-ph/9604005v2}.

\bibitem{96Horodecki1}
M.~Horodecki, P.~Horodecki, and R.~Horodecki.
\newblock \emph{Separability of mixed states: Necessary and sufficient
  conditions}.
\newblock Physics Letters A \textbf{223}(1-2), 1 (1996).
\newblock \href{http://arxiv.org/abs/quant-ph/9605038v2}{E-print
  arXiv:quant-ph/9605038v2}.

\bibitem{94Pechukas1060}
P.~Pechukas.
\newblock \emph{Reduced dynamics need not be completely positive}.
\newblock Physical Review Letters \textbf{73}(8), 1060 (1994).

\bibitem{95Alicki3020}
R.~Alicki.
\newblock \emph{Reduced dynamics need not be completely positive - comments}.
\newblock Physical Review Letters \textbf{75}(16), 3020 (1995).

\bibitem{95Pechukas3021}
P.~Pechukas.
\newblock \emph{Reduced dynamics need not be completely positive - reply}.
\newblock Physical Review Letters \textbf{75}(16), 3021 (1995).

\bibitem{01Stelmachovic062106}
P.~\v{S}telmachovi\v{c} and V.~Bu\v{z}ek.
\newblock \emph{Dynamics of open quantum systems initially entangled with
  environment: Beyond the kraus representation}.
\newblock Physical Review A \textbf{64}(6), 062106 (2001).
\newblock \href{http://arxiv.org/abs/quant-ph/0108136v1}{E-print
  arXiv:quant-ph/0108136v1}.

\bibitem{03Stelmachovic029902}
P.~\v{S}telmachovi\v{c} and V.~Bu\v{z}ek.
\newblock \emph{Erratum: Dynamics of open quantum systems initially entangled
  with environment: Beyond the kraus representation (vol a 64, art no 062106,
  2001)}.
\newblock Physical Review A \textbf{67}(2), 029902 (2003).

\bibitem{04Jordan052110}
T.~F. Jordan, A.~Shaji, and E.~C.~G. Sudarshan.
\newblock \emph{Dynamics of initially entangled open quantum systems}.
\newblock Physical Review A \textbf{70}(5), 052110 (2004).
\newblock \href{http://arxiv.org/abs/quant-ph/0407083}{E-print
  arXiv:quant-ph/0407083v1}.

\bibitem{05Shaji48}
A.~Shaji and E.~Sudarshan.
\newblock \emph{Who's afraid of not completely positive maps?}
\newblock Physics Letters A \textbf{341}(1-4), 48 (2005).

\bibitem{98Holevo1295}
A.~S. Holevo.
\newblock \emph{Quantum coding theorems}.
\newblock Russian Mathematical Surveys \textbf{53}(6), 1295 (1998).

\bibitem{02Verstraete}
F.~Verstraete and H.~Verschelde.
\newblock \emph{On quantum channels} (2002).
\newblock \href{http://arxiv.org/abs/quant-ph/0202124}{E-print
  arXiv:quant-ph/0202124}.

\bibitem{03Horodecki629}
M.~Horodecki, P.~W. Shor, and M.~B. Ruskai.
\newblock \emph{Entanglement breaking channels}.
\newblock Reviews in Mathematical Physics \textbf{15}(6), 629 (2003).
\newblock \href{http://arxiv.org/abs/quant-ph/0302031}{E-print
  arXiv:quant-ph/0302031}.

\bibitem{03Ruskai643}
M.~B. Ruskai.
\newblock \emph{Qubit entanglement breaking channels}.
\newblock Reviews in Mathematical Physics \textbf{15}(6), 643 (2003).
\newblock \href{http://arxiv.org/abs/quant-ph/0302032v3}{E-print
  arXiv:quant-ph/0302032v3}.

\bibitem{03Fuchs377}
C.~A. Fuchs and M.~Sasaki.
\newblock \emph{Squeezing quantum information through a classical channel:
  Measuring the ``quantumness'' of a set of quantum states}.
\newblock Quantum Information \& Computation \textbf{3}(5), 377 (2003).
\newblock \href{http://arxiv.org/abs/quant-ph/0302092}{E-print
  arXiv:quant-ph/0302092}.

\bibitem{04Zyczkowski3}
K.~\.Zyczkowski and I.~Bengtsson.
\newblock \emph{On duality between quantum maps and quantum states}.
\newblock Open Systems \& Information Dynamics \textbf{11}(1), 3 (2004).
\newblock \href{http://arxiv.org/abs/quant-ph/0401119v1}{E-print
  arXiv:quant-ph/0401119v1}.

\bibitem{01DAriano042308}
G.~M. D'Ariano and P.~Lo~Presti.
\newblock \emph{Optimal nonuniversally covariant cloning}.
\newblock Physical Review A \textbf{64}(4), 042308 (2001).

\bibitem{99Fujiwara3290}
A.~Fujiwara and P.~Algoet.
\newblock \emph{One-to-one parametrization of quantum channels}.
\newblock Physical Review A \textbf{59}(5), 3290 (1999).

\bibitem{99Horodecki1888}
M.~Horodecki, P.~Horodecki, and R.~Horodecki.
\newblock \emph{General teleportation channel, singlet fraction, and
  quasidistillation}.
\newblock Physical Review A \textbf{60}(3), 1888 (1999).
\newblock \href{http://arxiv.org/abs/quant-ph/9807091}{E-print
  arXiv:quant-ph/9807091v2}.

\bibitem{96Bennett3824}
C.~H. Bennett, D.~P. DiVincenzo, J.~A. Smolin, and W.~K. Wootters.
\newblock \emph{Mixed-state entanglement and quantum error correction}.
\newblock Physical Review A \textbf{54}(5), 3824 (1996).
\newblock \href{http://arxiv.org/abs/quant-ph/9604024}{E-print
  arXiv:quant-ph/9604024v2}.

\bibitem{01Rains2921}
E.~M. Rains.
\newblock \emph{A semidefinite program for distillable entanglement}.
\newblock IEEE Transactions on Information Theory \textbf{47}(7), 2921 (2001).

\bibitem{01Cirac544}
J.~I. Cirac, W.~D{\"u}r, B.~Kraus, and M.~Lewenstein.
\newblock \emph{Entangling operations and their implementation using a small
  amount of entanglement}.
\newblock Physical Review Letters \textbf{86}(3), 544 (2001).
\newblock \href{http://arxiv.org/abs/quant-ph/0007057}{E-print
  arXiv:quant-ph/0007057v1}.

\bibitem{91Horn}
R.~A. Horn and C.~R. Johnson.
\newblock \emph{Topics in Matrix Analysis} (Cambridge University Press,
  Cambridge, 1991).

\bibitem{04Nielsen}
M.~A. Nielsen.
\newblock \emph{On the solution of linear matrix equations} (2004).
\newblock Available online at
  \href{http://www.qinfo.org/people/nielsen/blog/archive/notes/0401.pdf}{\url{%
http://www.qinfo.org/people/nielsen/blog/archive/notes/0401.pdf}}.

\bibitem{36Schrodinger446}
E.~Schr{\"o}dinger.
\newblock \emph{Probability relations between separated systems}.
\newblock Proceedings of the Cambridge Philosophical Society \textbf{32}, 446
  (1936).

\bibitem{57Jaynes171}
E.~T. Jaynes.
\newblock \emph{Information theory and statistical mechanics. ii}.
\newblock Physical Review \textbf{108}(2), 171 (1957).

\bibitem{93Hughston14}
L.~P. Hughston, R.~Jozsa, and W.~K. Wootters.
\newblock \emph{A complete classification of quantum ensembles having a given
  density-matrix}.
\newblock Physics Letters A \textbf{183}(1), 14 (1993).

\bibitem{06Kirkpatrick95}
K.~A. Kirkpatrick.
\newblock \emph{The schr{\"o}dinger-hjw theorem}.
\newblock Foundations of Physics Letters \textbf{19}(1), 95 (2006).
\newblock \href{http://arxiv.org/abs/quant-ph/0305068}{E-print
  arXiv:quant-ph/0305068v3}.

\bibitem{07Bhatia}
R.~Bhatia.
\newblock \emph{Positive Definite Matrices} (Princeton University Press,
  Princeton, 2007).

\bibitem{85Horn}
R.~A. Horn and C.~R. Johnson.
\newblock \emph{Matrix Analysis} (Cambridge University Press, Cambridge, 1985).

\bibitem{62Gell-Mann1067}
M.~Gell-Mann.
\newblock \emph{Symmetries of baryons and mesons}.
\newblock Physical Review \textbf{125}(3), 1067 (1962).

\bibitem{01King192}
C.~King and M.~B. Ruskai.
\newblock \emph{Minimal entropy of states emerging from noisy quantum
  channels}.
\newblock IEEE Transactions on Information Theory \textbf{47}(1), 192 (2001).
\newblock \href{http://arxiv.org/abs/quant-ph/9911079}{E-print
  arXiv:quant-ph/9911079}.

\bibitem{02Ruskai159}
M.~B. Ruskai, S.~Szarek, and E.~Werner.
\newblock \emph{An analysis of completely positive trace-preserving maps on
  $\mathcal{M}_2$}.
\newblock Linear Algebra and its Applications \textbf{347}, 159 (2002).
\newblock \href{http://arxiv.org/abs/quant-ph/0101003v2}{E-print
  arXiv:quant-ph/0101003v2}.

\bibitem{04Nemirovski}
A.~Nemirovski.
\newblock \emph{Interior point polynomial time methods in convex programming}
  (2004).
\newblock Lecture notes. Availabe on-line at
  \href{http://www2.isye.gatech.edu/~nemirovs/}{\url{http://www2.isye.gatech.e%
du/~nemirovs/}}.

\bibitem{00Horodecki032310}
P.~Horodecki, M.~Lewenstein, G.~Vidal, and I.~Cirac.
\newblock \emph{Operational criterion and constructive checks for the
  separability of low-rank density matrices - art. no. 032310}.
\newblock Physical Review A \textbf{6203}(3), 032310 (2000).

\bibitem{03Horodecki589}
P.~Horodecki, J.~A. Smolin, B.~M. Terhal, and A.~V. Thapliyal.
\newblock \emph{Rank two bipartite bound entangled states do not exist}.
\newblock Theoretical Computer Science \textbf{292}(3), 589 (2003).
\newblock \href{http://arxiv.org/abs/quant-ph/9910122v4}{E-print
  arXiv:quant-ph/9910122v4}.

\bibitem{03Gurvits10}
L.~Gurvits.
\newblock \emph{Classical deterministic complexity of edmonds' problem and
  quantum entanglement}.
\newblock In \emph{Proceedings of the Thirty-fifth ACM Symposium on Theory of
  Computing}, pp. 10--19 (ACM Press, New York, 2003).

\bibitem{04Doherty022308}
A.~C. Doherty, P.~A. Parrilo, and F.~M. Spedalieri.
\newblock \emph{Complete family of separability criteria}.
\newblock Physical Review A \textbf{69}(2), 022308 (2004).
\newblock \href{http://arxiv.org/abs/quant-ph/0308032v3}{E-print
  arXiv:quant-ph/0308032v3}.

\bibitem{04Eisert062317}
J.~Eisert, P.~Hyllus, O.~Guhne, and M.~Curty.
\newblock \emph{Complete hierarchies of efficient approximations to problems in
  entanglement theory}.
\newblock Physical Review A \textbf{70}(6), 062317 (2004).
\newblock \href{http://arxiv.org/abs/quant-ph/0407135v3}{E-print
  arXiv:quant-ph/0407135v3}.

\bibitem{07Ioannou}
L.~M. Ioannou.
\newblock \emph{Computing finite-dimensional bipartite quantum separability}.
\newblock Ph.D. thesis, University of Cambridge (2005).
\newblock \href{http://arxiv.org/abs/cs/0504110}{E-print arXiv:cs/0504110v3
  [cs.DS]}.

\bibitem{62Umegaki59}
H.~Umegaki.
\newblock \emph{Conditional expectation in an operator algebra, iv (entropy and
  information)}.
\newblock Kodai Mathematical Seminar Reports \textbf{14}(2), 59 (1962).

\bibitem{00Schumacher0004045}
B.~Schumacher and M.~D. Westmoreland.
\newblock \emph{Relative entropy in quantum information theory} (2000).
\newblock \href{http://arxiv.org/abs/quant-ph/0004045}{E-print
  arXiv:quant-ph/0004045}.

\bibitem{02Vedral197}
V.~Vedral.
\newblock \emph{The role of relative entropy in quantum information theory}.
\newblock Reviews of Modern Physics \textbf{74}(1), 197 (2002).
\newblock \href{http://arxiv.org/abs/quant-ph/0102094}{E-print
  arXiv:quant-ph/0102094}.

\bibitem{76Uhlmann273}
A.~Uhlmann.
\newblock \emph{The ``transition probability" in the state space of a
  $\ast$-algebra}.
\newblock Reports on Mathematical Physics \textbf{9}, 273 (1976).

\bibitem{94Jozsa2315}
R.~Jozsa.
\newblock \emph{Fidelity for mixed quantum states}.
\newblock Journal of Modern Optics \textbf{41}(12), 2315 (1994).

\bibitem{91Hiai99}
F.~Hiai and D.~Petz.
\newblock \emph{The proper formula for relative entropy and its asymptotics in
  quantum probability}.
\newblock Communications in Mathematical Physics \textbf{143}(1), 99 (1991).

\bibitem{00Ogawa2428}
T.~Ogawa and H.~Nagaoka.
\newblock \emph{Strong converse and stein's lemma in quantum hypothesis
  testing}.
\newblock IEEE Transactions on Information Theory \textbf{46}(7), 2428 (2000).
\newblock \href{http://arxiv.org/abs/quant-ph/9906090}{E-print
  arXiv:quant-ph/9906090}.

\bibitem{06Nussbaum0607216}
M.~Nussbaum and A.~Szkola.
\newblock \emph{A lower bound of chernoff type for symmetric quantum hypothesis
  testing} (2006).
\newblock \href{http://arxiv.org/abs/quant-ph/0607216}{E-print
  arXiv:quant-ph/0607216}.

\bibitem{07Audenaert160501}
K.~M.~R. Audenaert, J.~Calsamiglia, R.~Munoz-Tapia, E.~Bagan, L.~Masanes,
  A.~Acin, and F.~Verstraete.
\newblock \emph{Discriminating states: The quantum chernoff bound}.
\newblock Physical Review Letters \textbf{98}(16), 160501 (2007).
\newblock See also \href{http://arxiv.org/abs/quant-ph/0610027}{e-print
  arXiv:quant-ph/0610027}.

\bibitem{76Helstrom}
C.~W. Helstrom.
\newblock \emph{Quantum Detection and Estimation Theory}, vol. 123 of
  \emph{Mathematics in Science and Engineering} (Academic Press, New York,
  1976).

\bibitem{79Holevo411}
A.~S. Holevo.
\newblock \emph{On asymptotically optimal hypothesis testing in quantum
  statistics}.
\newblock Theory of Probability and its Applications \textbf{23}(2), 411
  (1979).
\newblock \urlprefix\url{http://link.aip.org/link/?TPR/23/411/1}.

\bibitem{97Vedral2275}
V.~Vedral, M.~B. Plenio, M.~A. Rippin, and P.~L. Knight.
\newblock \emph{Quantifying entanglement}.
\newblock Physical Review Letters \textbf{78}(12), 2275 (1997).
\newblock \href{http://aps.arxiv.org/abs/quant-ph/9702027v1}{E-print
  arXiv:quant-ph/9702027v1}.

\bibitem{98Vedral1619}
V.~Vedral and M.~B. Plenio.
\newblock \emph{Entanglement measures and purification procedures}.
\newblock Physical Review A \textbf{57}(3), 1619 (1998).
\newblock
  \href{http://arxiv.org/abs/quant-ph/9707035v2}{arXiv:quant-ph/9707035v2}.

\bibitem{05Reimpell080501}
M.~Reimpell and R.~F. Werner.
\newblock \emph{Iterative optimization of quantum error correcting codes}.
\newblock Physical Review Letters \textbf{94}(8), 080501 (2005).
\newblock \href{http://arxiv.org/abs/quant-ph/0307138v2}{E-print
  arXiv:quant-ph/0307138v2}.

\bibitem{07Fletcher012338}
A.~S. Fletcher, P.~W. Shor, and M.~Z. Win.
\newblock \emph{Optimum quantum error recovery using semidefinite programming}.
\newblock Physical Review A \textbf{75}(1), 012338 (2007).
\newblock \href{http://arxiv.org/abs/quant-ph/0606035v1}{E-print
  arXiv:quant-ph/0606035v1}.

\bibitem{06Reimpell}
M.~Reimpell, R.~F. Werner, and K.~Audenaert.
\newblock \emph{Comment on ``optimum quantum error recovery using semidefinite
  programming''} (2006).
\newblock \href{http://arxiv.org/abs/quant-ph/0606059v1}{E-print
  arXiv:quant-ph/0606059v1}.

\bibitem{06Kosut}
R.~L. Kosut and D.~A. Lidar.
\newblock \emph{Quantum error correction via convex optimization} (2006).
\newblock \href{http://arxiv.org/abs/quant-ph/0606078v1}{E-print
  arXiv:quant-ph/0606078v1}.

\bibitem{08Kosut020502}
R.~L. Kosut, A.~Shabani, and D.~A. Lidar.
\newblock \emph{Robust quantum error correction via convex optimization}.
\newblock Physical Review Letters \textbf{100}, 020502 (2008).
\newblock \href{http://arxiv.org/abs/quant-ph/0703274v2}{E-print
  arXiv:quant-ph/0703274v2}.

\bibitem{07Yamamoto012327}
N.~Yamamoto and M.~Fazel.
\newblock \emph{Computational approach to quantum encoder design for purity
  optimization}.
\newblock Physical Review A \textbf{76}(1), 012327 (2007).
\newblock \href{http://arxiv.org/abs/quant-ph/0606106v3}{E-print
  arXiv:quant-ph/0606106v3}.

\bibitem{05Yamamoto022322}
N.~Yamamoto, S.~Hara, and K.~Tsumura.
\newblock \emph{Suboptimal quantum-error-correcting procedure based on
  semidefinite programming}.
\newblock Physical Review A \textbf{71}(2), 022322 (2005).
\newblock \href{http://arxiv.org/abs/quant-ph/0606105v1}{E-print
  arXiv:quant-ph/0606105v1}.

\bibitem{95Fuchs}
C.~A. Fuchs.
\newblock \emph{Distinguishability and Accessible Information in Quantum
  Theory}.
\newblock Ph.D. thesis, University of New Mexico (1995).
\newblock \href{http://arxiv.org/abs/quant-ph/9601020}{E-print
  arXiv:quant-ph/9601020}.

\bibitem{83Alberti5}
P.~Alberti and A.~Uhlmann.
\newblock \emph{Transition probabilities of w$^{\ast}$- and
  c$^{\ast}$-algebras}.
\newblock In H.~Baumgartel, G.~La{\ss}ner, A.~Pietsch, and A.~Uhlmann, eds.,
  \emph{Proceedings of the Second International Conference on Operator
  Algebras, Ideals, and their Applications in Theoretical Physics} (1983).
\newblock Available online at
  \href{http://www.physik.uni-leipzig.de/~uhlmann/}{\url{http://www.physik.uni%
-leipzig.de/~uhlmann/}}.

\bibitem{83Alberti25}
P.~M. Alberti.
\newblock \emph{A note on the transition-probability over c$^\ast$-algebras}.
\newblock Letters in Mathematical Physics \textbf{7}(1), 25 (1983).

\bibitem{83Alberti107}
P.~M. Alberti and A.~Uhlmann.
\newblock \emph{Stochastic linear-maps and transition-probability}.
\newblock Letters in Mathematical Physics \textbf{7}(2), 107 (1983).

\bibitem{95Schumacher2738}
B.~Schumacher.
\newblock \emph{Quantum coding}.
\newblock Physical Review A \textbf{51}(4), 2738 (1995).

\bibitem{08Miszczak}
J.~A. Miszczak, Z.~Pucha{\l}a, P.~Horodecki, A.~Uhlmann, and K.~\.Zyczkowski.
\newblock \emph{Sub-- and super--fidelity as bounds for quantum fidelity}.
\newblock Quantum Information \& Computation \textbf{9}, 0103 (2009).
\newblock \href{http://arxiv.org/abs/0805.2037}{E-print arXiv:0805.2037v1
  [quant-ph]}.

\bibitem{00Uhlmann407}
A.~Uhlmann.
\newblock \emph{On ``partial'' fidelities}.
\newblock Reports on Mathematical Physics \textbf{45}(3), 407 (2000).
\newblock \href{http://arxiv.org/abs/quant-ph/9912114v2}{E-print
  arXiv:quant-ph/9912114v2}.

\bibitem{71Uhlmann633}
A.~Uhlmann.
\newblock \emph{S\"atze \"uber dichtematrizen}.
\newblock Math.-Naturwiss. R. \textbf{20}, 633 (1971).

\bibitem{08Carlen107}
E.~A. Carlen and E.~H. Lieb.
\newblock \emph{A minkowski type trace inequality and strong subadditivity of
  quantum entropy ii: Convexity and concavity}.
\newblock Letters in Mathematical Physics \textbf{83}, 107 (2008).

\bibitem{95Uhlmann461}
A.~Uhlmann.
\newblock \emph{Geometric phases and related structures}.
\newblock Reports on Mathematical Physics \textbf{36}, 461 (1995).

\bibitem{69Bures199}
D.~Bures.
\newblock \emph{An extension of kakutani's theorem on infinite product measures
  to the tensor product of semifinite $w^\ast$-algebras}.
\newblock Transactions of the American Mathematical Society \textbf{135}, 199
  (1969).

\bibitem{92Hubner239}
M.~H{\"u}bner.
\newblock \emph{Explicit computation of the bures distance for
  density-matrices}.
\newblock Physics Letters A \textbf{163}(4), 239 (1992).

\bibitem{05Gilchrist062310}
A.~Gilchrist, N.~K. Langford, and M.~A. Nielsen.
\newblock \emph{Distance measures to compare real and ideal quantum processes}.
\newblock Physical Review A \textbf{71}(6), 062310 (2005).
\newblock \href{http://arxiv.org/abs/quant-ph/0408063}{E-print
  arXiv:quant-ph/0408063}.

\bibitem{06Rastegin}
A.~Rastegin.
\newblock \emph{Sine distance for quantum states} (2006).
\newblock \href{http://arxiv.org/abs/quant-ph/0602112}{E-print
  arXiv:quant-ph/0602112v1}.

\bibitem{97Bhatia}
R.~Bhatia.
\newblock \emph{Matrix Analysis}, vol. 169 of \emph{Graduate Texts in
  Mathematics} (Springer-Verlag, New York, 1997).

\bibitem{94Ruskai1147}
M.~B. Ruskai.
\newblock \emph{Beyond strong subadditivity - improved bounds on the
  contraction of generalized relative entropy}.
\newblock Reviews in Mathematical Physics \textbf{6}(5A), 1147 (1994).

\bibitem{99Fuchs1216}
C.~A. Fuchs and J.~van~de Graaf.
\newblock \emph{Cryptographic distinguishability measures for
  quantum-mechanical states}.
\newblock IEEE Transactions on Information Theory \textbf{45}(4), 1216 (1999).
\newblock \href{http://arxiv.org/abs/quant-ph/9712042}{E-print
  arXiv:quant-ph/9712042}.

\bibitem{01Spekkens012310}
R.~W. Spekkens and T.~Rudolph.
\newblock \emph{Degrees of concealment and bindingness in quantum bit
  commitment protocols}.
\newblock Physical Review A \textbf{65}(1), 012310 (2001).
\newblock \href{http://arxiv.org/abs/quant-ph/0106019}{E-print
  arXiv:quant-ph/0106019v2}.

\bibitem{93Hubner226}
M.~H{\"u}bner.
\newblock \emph{Computation of uhlmann parallel transport for density-matrices
  and the bures metric on 3-dimensional hilbert-space}.
\newblock Physics Letters A \textbf{179}(4-5), 226 (1993).

\bibitem{02Chen054304}
J.~L. Chen, L.~Fu, A.~A. Ungar, and X.~G. Zhao.
\newblock \emph{Alternative fidelity measure between two states of an n-state
  quantum system}.
\newblock Physical Review A \textbf{65}(5), 054304 (2002).

\bibitem{03Byrd062322}
M.~S. Byrd and N.~Khaneja.
\newblock \emph{Characterization of the positivity of the density matrix in
  terms of the coherence vector representation}.
\newblock Physical Review A \textbf{68}(6), 062322 (2003).
\newblock \href{http://arxiv.org/abs/quant-ph/0302024v2}{E-print
  arXiv:quant-ph/0302024v2}.

\bibitem{03Kimura339}
G.~Kimura.
\newblock \emph{The bloch vector for n-level systems}.
\newblock Physics Letters A \textbf{314}(5-6), 339 (2003).
\newblock \href{http://arxiv.org/abs/quant-ph/0301152v2}{E-print
  arXiv:quant-ph/0301152v2}.

\bibitem{00Ozawa158}
M.~Ozawa.
\newblock \emph{Entanglement measures and the hilbert-schmidt distance}.
\newblock Physics Letters A \textbf{268}(3), 158 (2000).
\newblock \href{http://arxiv.org/abs/quant-ph/0002036}{E-print
  arXiv:quant-ph/0002036}.

\bibitem{99Witte14}
C.~Witte and M.~Trucks.
\newblock \emph{A new entanglement measure induced by the hilbert-schmidt
  norm}.
\newblock Physics Letters A \textbf{257}(1-2), 14 (1999).
\newblock \href{http://arxiv.org/abs/quant-ph/9811027}{E-print
  arXiv:quant-ph/9811027}.

\bibitem{38Schoenberg522}
I.~J. Schoenberg.
\newblock \emph{Metric spaces and positive definite functions}.
\newblock Transactions of the American Mathematical Society \textbf{44}(3), 522
  (1938).

\bibitem{84Berg}
C.~Berg, J.~Christensen, and P.~Ressel.
\newblock \emph{Harmonic Analysis on Semigroups} (Springer-Verlag, New York,
  1984).

\bibitem{00Topsoe1602}
F.~Tops{\o}e.
\newblock \emph{Some inequalities for information divergence and related
  measures of discrimination}.
\newblock IEEE Transactions on Information Theory \textbf{46}(4), 1602 (2000).

\bibitem{03Topsoe}
F.~Tops{\o}e.
\newblock \emph{Jensen-shannon divergence and norm-based measures of
  discrimination and variation} (2003).
\newblock Available online at
  \href{http://www.math.ku.dk/~topsoe}{\url{http://www.math.ku.dk/~topsoe}}.

\bibitem{04Fuglede}
B.~Fuglede and F.~Tops{\o}e.
\newblock \emph{Jensen-shannon divergence and hilbert space embedding} (2004).
\newblock Available online at
  \href{http://www.math.ku.dk/~topsoe}{\url{http://www.math.ku.dk/~topsoe}}.

\bibitem{09Puchala024302}
Z.~Pucha{\l}a and J.~A. Miszczak.
\newblock \emph{Bound on trace distance based on superfidelity}.
\newblock Physical Review A \textbf{79}(2), 024302 (2009).
\newblock \href{http://arxiv.org/abs/0811.2323}{E-print arXiv:0811.2323v1
  [quant-ph]}.

\bibitem{98Zyczkowski883}
K.~\.{Z}yczkowski, P.~Horodecki, A.~Sanpera, and M.~Lewenstein.
\newblock \emph{Volume of the set of separable states}.
\newblock Physical Review A \textbf{58}(2), 883 (1998).

\bibitem{C:url}
\url{http://www.physics.uq.edu.au/people/foster/}.

\bibitem{06Galassi}
M.~Galassi, J.~Davies, J.~Theiler, B.~Gough, G.~Jungman, M.~Booth, and
  F.~Rossi.
\newblock \emph{GNU Scientific Library Reference Manual} (2006).

\bibitem{00Parlett38}
B.~N. Parlett.
\newblock \emph{The qr algorithm}.
\newblock Computing in Science \& Engineering \textbf{2}, 38 (2000).

\bibitem{Mike:personal}
M.~A. Nielsen.
\newblock Private communication.

\bibitem{Uhlmann:personal}
A.~Uhlmann.
\newblock Private communication.

\bibitem{60Schatten}
R.~Schatten.
\newblock \emph{Norm ideals of completely continuous operators}.
\newblock In \emph{Ergebnisse der Mathematik und ihrer Grenzgebiete}
  (Springer-Verlag, Berlin, 1960).

\bibitem{07Rastegin9533}
A.~E. Rastegin.
\newblock \emph{Trace distance from the viewpoint of quantum operation
  techniques}.
\newblock Journal of Physics A-Mathematical and Theoretical \textbf{40}, 9533
  (2007).

\bibitem{07Recht}
B.~Recht, M.~Fazel, and P.~A. Parrilo.
\newblock \emph{Guaranteed minimum-rank solutions of linear matrix equations
  via nuclear norm minimization} (2007).
\newblock \href{http://arxiv.org/abs/0706.4138v1}{E-print arXiv:0706.4138v1
  [math.OC]}.

\bibitem{05Belavkin062106}
V.~P. Belavkin, G.~M. D'Ariano, and M.~Raginsky.
\newblock \emph{Operational distance and fidelity for quantum channels}.
\newblock Journal of Mathematical Physics \textbf{46}(6), 062106 (2005).

\bibitem{06Perez-Garcia083506}
D.~Perez-Garcia, M.~M. Wolf, D.~Petz, and M.~B. Ruskai.
\newblock \emph{Contractivity of positive and trace-preserving maps under l-p
  norms}.
\newblock Journal of Mathematical Physics \textbf{47}(8), 083506 (2006).
\newblock \href{http://arxiv.org/abs/math-ph/0601063}{E-print
  arXiv:math-ph/0601063v1}.

\bibitem{0403Nielsen}
M.~A. Nielsen.
\newblock \emph{Operator monotone and operator convex functions: a survey}
  (2004).
\newblock Available online at
  \href{http://www.qinfo.org/people/nielsen/blog/archive/000095.html}{\url{htt%
p://www.qinfo.org/people/nielsen/blog/archive/000095.html}}.

\bibitem{01Fazel4734}
M.~Fazel, H.~Hindi, and S.~Boyd.
\newblock \emph{A rank minimization heuristic with application to minimum order
  system approximation}.
\newblock In \emph{Proceedings of the American Control Conference}, vol.~6, pp.
  4734--4739 (2001).

\bibitem{05Zhang}
F.~Zhang.
\newblock \emph{The Schur Complement and Its Applications} (Springer, 2005).

\bibitem{04Gatermann95}
K.~Gatermann and P.~A. Parrilo.
\newblock \emph{Symmetry groups, semidefinite programs, and sums of squares}.
\newblock Journal of Pure and Applied Algebra \textbf{192}(1-3), 95 (2004).
\newblock \href{http://arxiv.org/abs/math/0211450v1}{E-print
  arXiv:math/0211450v1}.

\bibitem{99Sturm625}
J.~F. Sturm.
\newblock \emph{Using sedumi 1.02, a matlab toolbox for optimization over
  symmetric cones}.
\newblock Optimization Methods and Software \textbf{11--12}, 625 (1999).

\bibitem{99Belavkin405}
V.~P. Belavkin.
\newblock \emph{Measurement, filtering and control in quantum open dynamical
  systems}.
\newblock Reports on Mathematical Physics \textbf{43}(3), 405 (1999).
\newblock \href{http://arxiv.org/abs/quant-ph/0208108v1}{E-print
  arXiv:quant-ph/0208108v1}.

\bibitem{05job}
\emph{Special issue on quantum control}.
\newblock Journal of Optics B-Quantum and Semiclassical Optics {\bf 7}(10)
  (2005).

\bibitem{wang02a}
J.~Wang and H.~M. Wiseman.
\newblock \emph{Feedback-stabilization of an arbitrary pure state of a
  two-level atom.}
\newblock Physical Review A \textbf{64}, 063810 (2001).

\bibitem{02wiseman013807}
H.~M. Wiseman, S.~Mancini, and J.~Wang.
\newblock \emph{Bayesian feedback vesus markovian feedback in a two-level
  atom}.
\newblock Physical Review A \textbf{66}, 013807 (2002).

\bibitem{05lidar350}
D.~A. Lidar and S.~Schneider.
\newblock \emph{Stabilizing qubit coherence via tracking-control}.
\newblock Quantum Information \& Computation \textbf{5}, 350 (2005).

\bibitem{vanhandel05a}
R.~van Handel, J.~K. Stockton, and H.~Mabuchi.
\newblock \emph{Feedback control of quantum state reduction}.
\newblock IEEE Transactions on Automatic Control \textbf{50}(6), 768 (2005).

\bibitem{mirrahimi05a}
M.~Mirrahimi and R.~van Handel.
\newblock \emph{Stabilizing feedback controls for quantum systems}.
\newblock Math-ph/05100066.

\bibitem{02Barnum2097}
H.~Barnum and E.~Knill.
\newblock \emph{Reversing quantum dynamics with near-optimal quantum and
  classical fidelity}.
\newblock Journal of Mathematical Physics \textbf{43}(5), 2097 (2002).
\newblock \href{http://arxiv.org/abs/quant-ph/0004088v1}{E-print
  arXiv:quant-ph/0004088v1}.

\bibitem{03Gregoratti915}
M.~Gregoratti and R.~F. Werner.
\newblock \emph{Quantum lost and found}.
\newblock Journal of Modern Optics \textbf{50}(6-7), 915 (2003).
\newblock \href{http://arxiv.org/abs/quant-ph/0209025}{E-print
  arXiv:quant-ph/0209025}.

\bibitem{04Gregoratti2600}
M.~Gregoratti and R.~F. Werner.
\newblock \emph{On quantum error-correction by classical feedback in discrete
  time}.
\newblock Journal of Mathematical Physics \textbf{45}(7), 2600 (2004).
\newblock \href{http://arxiv.org/abs/quant-ph/0403092}{E-print
  arXiv:quant-ph/0403092}.

\bibitem{06Ticozzi052328}
F.~Ticozzi and L.~Viola.
\newblock \emph{Single-bit feedback and quantum-dynamical decoupling}.
\newblock Physical Review A \textbf{74}, 052328 (2006).
\newblock \href{http://arxiv.org/abs/quant-ph/0609165}{E-print
  arXiv:quant-ph/0609165v1}.

\bibitem{99Niu2764}
C.~S. Niu and R.~B. Griffiths.
\newblock \emph{Two-qubit copying machine for economical quantum
  eavesdropping}.
\newblock Physical Review A \textbf{60}(4), 2764 (1999).
\newblock \href{http://arxiv.org/abs/quant-ph/9810008}{E-print
  arXiv:quant-ph/9810008}.

\bibitem{02Lloyd010101}
S.~Lloyd and L.~Viola.
\newblock \emph{Engineering quantum dynamics}.
\newblock Physical Review A \textbf{65}(1), 010101 (2002).

\bibitem{07Bartlett555}
S.~D. Bartlett, T.~Rudolph, and R.~W. Spekkens.
\newblock \emph{Reference frames, superselection rules, and quantum
  information}.
\newblock Reviews of Modern Physics \textbf{79}(2), 555 (2007).
\newblock \href{http://arxiv.org/abs/quant-ph/0610030v3}{E-print
  arXiv:quant-ph/0610030v3}.

\bibitem{Pre00}
J.~Preskill.
\newblock \emph{Quantum clock synchronization and quantum error correction}
  (2000).
\newblock \urlprefix\url{http://www.arXiv.org/abs/quant-ph/0010098}.

\bibitem{03Fiurasek012321}
J.~Fiur\'a\v{s}ek and M.~Jezek.
\newblock \emph{Optimal discrimination of mixed quantum states involving
  inconclusive results}.
\newblock Physical Review A \textbf{67}(1), 012321 (2003).
\newblock \href{http://arxiv.org/abs/quant-ph/0208126v1}{E-print
  arXiv:quant-ph/0208126v1}.

\bibitem{03Eldar042309}
Y.~C. Eldar.
\newblock \emph{Mixed-quantum-state detection with inconclusive results}.
\newblock Physical Review A \textbf{67}(4), 042309 (2003).
\newblock \href{http://arxiv.org/abs/quant-ph/0211121}{E-print
  arXiv:quant-ph/0211121v1}.

\bibitem{04Feng012308}
Y.~A. Feng, R.~Y. Duan, and M.~S. Ying.
\newblock \emph{Unambiguous discrimination between mixed quantum states}.
\newblock Physical Review A \textbf{70}(1), 012308 (2004).
\newblock \href{http://arxiv.org/abs/quant-ph/0410073}{E-print
  arXiv:quant-ph/0410073v3}.

\bibitem{pryde:190402}
G.~J. Pryde, J.~L. O'Brien, A.~G. White, S.~D. Bartlett, and T.~C. Ralph.
\newblock \emph{Measuring a photonic qubit without destroying it}.
\newblock Physical Review Letters \textbf{92}(19), 190402 (pages~4) (2004).

\bibitem{ralph:012113}
T.~C. Ralph, S.~D. Bartlett, J.~L. O'Brien, G.~J. Pryde, and H.~M. Wiseman.
\newblock \emph{Quantum nondemolition measurements for quantum information}.
\newblock Physical Review A \textbf{73}(1), 012113 (2006).

\bibitem{08Xi1056}
Z.~R. Xi and G.~S. Jin.
\newblock \emph{Performance comparison between classical and quantum control
  for a simple quantum system}.
\newblock Physica A-Statistical Mechanics and its Applications \textbf{387},
  1056 (2008).

\bibitem{72Jamiolkowski275}
A.~Jamio{\l}kowski.
\newblock \emph{Linear transformations which preserve trace and positive
  semidefiniteness of operators}.
\newblock Reports on Mathematical Physics \textbf{3}(4), 275 (1972).

\bibitem{BlumeKohoutCombes}
R.~Blume-Kohout and J.~Combes.
\newblock \emph{private communication}.

\bibitem{92Bennett3121}
C.~H. Bennett.
\newblock \emph{Quantum cryptography using any two nonorthogonal states}.
\newblock Physical Review Letters \textbf{68}(21), 3121 (1992).

\bibitem{pittman:052332}
T.~B. Pittman, B.~C. Jacobs, and J.~D. Franson.
\newblock \emph{Demonstration of quantum error correction using linear optics}.
\newblock Physical Review A \textbf{71}(5), 052332 (2005).

\bibitem{01Werner}
R.~F. Werner.
\newblock \emph{Quantum Information --– an Introduction to Basic Theoretical
  Concepts and Experiments}, vol. 173, chap. Quantum Information Theory –-- an
  Invitation, pp. 14--57 (Springer-Verlag, Berlin, 2001).
\newblock \href{http://arxiv.org/abs/quant-ph/0101061}{E-print
  arXiv:quant-ph/0101061}.

\bibitem{82Dieks271}
D.~Dieks.
\newblock \emph{Communication by electron-paramagnetic-res devices}.
\newblock Physics Letters A \textbf{92}(6), 271 (1982).

\bibitem{82Wootters802}
W.~K. Wootters and W.~H. Zurek.
\newblock \emph{A single quantum cannot be cloned}.
\newblock Nature \textbf{299}(5886), 802 (1982).

\bibitem{96Barnum2818}
H.~Barnum, C.~M. Caves, C.~A. Fuchs, R.~Jozsa, and B.~Schumacher.
\newblock \emph{Noncommuting mixed states cannot be broadcast}.
\newblock Physical Review Letters \textbf{76}(15), 2818 (1996).
\newblock \href{http://arxiv.org/abs/quant-ph/9511010}{E-print
  arXiv:quant-ph/9511010}.

\bibitem{88Dieks303}
D.~Dieks.
\newblock \emph{Overlap and distinguishability of quantum states}.
\newblock Physics Letters A \textbf{126}(5-6), 303 (1988).

\bibitem{87Ivanovic257}
I.~D. Ivanovic.
\newblock \emph{How to differentiate between nonorthogonal states}.
\newblock Physics Letters A \textbf{123}(6), 257 (1987).

\bibitem{88Peres19}
A.~Peres.
\newblock \emph{How to differentiate between non-orthogonal states}.
\newblock Physics Letters A \textbf{128}(1-2), 19 (1988).

\bibitem{80Alberti163}
P.~M. Alberti and A.~Uhlmann.
\newblock \emph{A problem relating to positive linear maps on matrix algebras}.
\newblock Reports on Mathematical Physics \textbf{18}(2), 163 (1980).
\newblock Available online at
  \href{http://www.physik.uni-leipzig.de/~uhlmann/}{\url{http://www.physik.uni%
-leipzig.de/~uhlmann/}}.

\bibitem{97Knill900}
E.~Knill and R.~Laflamme.
\newblock \emph{Theory of quantum error-correcting codes}.
\newblock Physical Review A \textbf{55}(2), 900 (1997).
\newblock \href{http://arxiv.org/abs/quant-ph/9604034}{E-print
  arXiv:quant-ph/9604034}.

\bibitem{05Buscemi082109}
F.~Buscemi, M.~Keyl, G.~M. D'Ariano, P.~Perinotti, and R.~F. Werner.
\newblock \emph{Clean positive operator valued measures}.
\newblock Journal of Mathematical Physics \textbf{46}(8), 082109 (2005).
\newblock \href{http://arxiv.org/abs/quant-ph/0505095}{E-print
  arXiv:quant-ph/0505095}.

\bibitem{07Nayak103}
A.~Nayak and P.~Sen.
\newblock \emph{Invertible quantum operations and perfect encryption of quantum
  states}.
\newblock Quantum Information \& Computation \textbf{7}(1-2), 103 (2007).
\newblock \href{http://arxiv.org/abs/quant-ph/0605041v4}{E-print
  arXiv:quant-ph/0605041v4}.

\bibitem{05Scarani1225}
V.~Scarani, S.~Iblisdir, N.~Gisin, and A.~Acin.
\newblock \emph{Quantum cloning}.
\newblock Reviews of Modern Physics \textbf{77}, 1225 (2005).
\newblock \href{http://arxiv.org/abs/quant-ph/0511088}{E-print
  arXiv:quant-ph/0511088v1}.

\bibitem{96Buzek1844}
V.~Bu\v{z}ek and M.~Hillery.
\newblock \emph{Quantum copying: Beyond the no-cloning theorem}.
\newblock Physical Review A \textbf{54}(3), 1844 (1996).
\newblock \href{http://arxiv.org/abs/quant-ph/9607018}{E-print
  arXiv:quant-ph/9607018v1}.

\bibitem{98Bruss2368}
D.~Bru\ss{}, D.~P. DiVincenzo, A.~Ekert, C.~A. Fuchs, C.~Macchiavello, and
  J.~A. Smolin.
\newblock \emph{Optimal universal and state-dependent quantum cloning}.
\newblock Physical Review A \textbf{57}(4), 2368 (1998).
\newblock \href{http://arxiv.org/abs/quant-ph/9705038v3}{E-print
  arXiv:quant-ph/9705038v3}.

\bibitem{57Bellman}
R.~Bellman.
\newblock \emph{Dynamic Programming} (Princeton University Press, Princeton,
  1957).

\bibitem{74jacobs}
O.~L.~R. Jacobs.
\newblock \emph{Introduction to control theory} (Oxford University Press,
  1974).

\bibitem{04chefles11}
A.~Chefles, R.~Jozsa, and A.~Winter.
\newblock \emph{On the existence of physical transformations between sets of
  quantum states}.
\newblock International Journal of Quantum Information \textbf{2}(1), 11
  (2004).
\newblock \href{http://arxiv.org/abs/quant-ph/0307227}{E-print
  arXiv:quant-ph/0307227}.
\end{thebibliography}
\end{document}